%% file: main.tex
\pdfoutput=1

\newif\ifdraft \draftfalse
\newif\iffull \fullfalse
\newif\ifauthorver \authorvertrue
\newif\ifanon \anonfalse
\newif\ifappendix \appendixtrue

\newif\ifrestate \restatetrue
\newif\iffullproof \fullprooffalse

\ifauthorver
\anonfalse
\documentclass[acmsmall,review=false,screen,nonacm,authorversion]{acmart}
\else\ifanon
\documentclass[acmsmall,review,screen,anonymous]{acmart}

\else
\documentclass[acmsmall,screen]{acmart}

\setcopyright{rightsretained}
\acmPrice{}
\acmDOI{10.1145/3408999}
\acmYear{2020}
\copyrightyear{2020}
\acmSubmissionID{icfp20main-p98-p}
\acmJournal{PACMPL}
\acmVolume{4}
\acmNumber{ICFP}
\acmArticle{117}
\acmMonth{8}

\fi\fi

\usepackage{booktabs}   
\usepackage{subcaption} 

\input{algeff}

\input{macro}

\renewcommand{\ottdrule}[4][]{{\displaystyle\frac{\begin{array}{c}#2\end{array}}{#3}\ \ottdrulename{#4}}}

\begin{document}

\title{Signature Restriction for Polymorphic Algebraic Effects}


\author{Taro Sekiyama}
\orcid{0000-0001-9286-230X}             
\affiliation{
  \institution{National Institute of Informatics \& SOKENDAI}            
  \country{Japan}                    
}
\email{tsekiyama@acm.org}          

\author{Takeshi Tsukada}
\orcid{0000-0002-2824-8708}             
\affiliation{
  \institution{The University of Tokyo}           
  \country{Japan}                   
}
\email{tsukada@kb.is.s.u-tokyo.ac.jp}         

\author{Atsushi Igarashi}
\orcid{0000-0002-5143-9764}
\affiliation{
  \institution{Kyoto University}
  \country{Japan}
}
\email{igarashi@kuis.kyoto-u.ac.jp}

\input{sections/abstract}



\maketitle

\input{sections/intro}
\input{sections/overview}

\input{sections/lang}
\input{sections/polytype}

\input{sections/ext}
\input{sections/eff}
\input{sections/relwork}
\input{sections/conclusion}

\ifanon\else
\begin{acks}                            
 We would like to thank Yusuke Matsushita for a fruitful discussion at an early
 stage of the research and the anonymous reviewers of ICFP 2020 PC and AEC for
 their close reading and valuable comments.
 This work was supported in part by
 ERATO HASUO Metamathematics for Systems Design Project (No.\ JPMJER1603), JST
 and JSPS KAKENHI Grant Numbers JP19K20247 (Sekiyama), JP19K20211 (Tsukada), and
 JP15H05706 (Igarashi).
\end{acks}
\fi

\bibliography{main}

\ifappendix

\renewcommand{\refsec}[1]{Appendix~\ref{sec:#1}}

\clearpage
\appendix
\input{appendix/defn}
\clearpage
\input{appendix/proof}

\fi

\end{document}

%% file: algeff.tex
\newcommand{\ottdrule}[4][]{{\displaystyle\frac{\begin{array}{l}#2\end{array}}{#3}\quad\ottdrulename{#4}}}

\newcommand{\ottpremise}[1]{ #1 \\}
\newenvironment{ottdefnblock}[3][]{ \framebox{\mbox{#2}} \quad #3 \\[0pt]}{}

\newcommand{\ottnt}[1]{\mathit{#1}}
\newcommand{\ottmv}[1]{\mathit{#1}}

\newcommand{\ottsym}[1]{#1}

\newcommand{\ottdrulename}[1]{\textsc{#1}}

\usepackage{amssymb}
\usepackage{amsmath,bm}
\usepackage{centernot}

\DeclareFontEncoding{LS1}{}{}
\DeclareFontSubstitution{LS1}{stix}{m}{n}
\DeclareSymbolFont{symbols2}{LS1}{stixfrak}{m}{n}
\DeclareMathSymbol{\typecolon}{\mathbin}{symbols2}{"25} 

\newif\ifvector \vectorfalse

\newcommand{\algeffseqover}[1]{\ifvector\overrightarrow{#1}\else\bm{#1}\fi}

\newcommand{\algeffseqoverindex}[2]{ \algeffseqover{#1}^{#2} }

\newcommand{\ottdruleEXXEval}[1]{\ottdrule[#1]{%
\ottpremise{\ottnt{M_{{\mathrm{1}}}}  \rightsquigarrow  \ottnt{M_{{\mathrm{2}}}}}%
}{
 \ottnt{E}  [  \ottnt{M_{{\mathrm{1}}}}  ]   \longrightarrow   \ottnt{E}  [  \ottnt{M_{{\mathrm{2}}}}  ] }{%
{\ottdrulename{E\_Eval}}{}%
}}

\newcommand{\ottdruleWFXXEmpty}[1]{\ottdrule[#1]{%
}{
\vdash   \emptyset }{%
{\ottdrulename{WF\_Empty}}{}%
}}

\newcommand{\ottdruleWFXXExtVar}[1]{\ottdrule[#1]{%
\ottpremise{ \mathit{x} \,  \not\in  \,  \mathit{dom}  (  \Gamma  )   \quad  \Gamma  \vdash  \ottnt{A} }%
}{
\vdash  \Gamma  \ottsym{,}  \mathit{x} \,  \mathord{:}  \, \ottnt{A}}{%
{\ottdrulename{WF\_ExtVar}}{}%
}}

\newcommand{\ottdruleWFXXExtTyVar}[1]{\ottdrule[#1]{%
\ottpremise{ \alpha \,  \not\in  \,  \mathit{dom}  (  \Gamma  )   \quad  \vdash  \Gamma }%
}{
\vdash  \Gamma  \ottsym{,}  \alpha}{%
{\ottdrulename{WF\_ExtTyVar}}{}%
}}

\newcommand{\ottdruleTXXVar}[1]{\ottdrule[#1]{%
\ottpremise{ \vdash  \Gamma  \quad  \mathit{x} \,  \mathord{:}  \, \ottnt{A} \,  \in  \, \Gamma }%
}{
\Gamma  \vdash  \mathit{x}  \ottsym{:}  \ottnt{A}}{%
{\ottdrulename{T\_Var}}{}%
}}

\newcommand{\ottdruleTXXConst}[1]{\ottdrule[#1]{%
\ottpremise{\vdash  \Gamma}%
}{
\Gamma  \vdash  \ottnt{c}  \ottsym{:}   \mathit{ty}  (  \ottnt{c}  ) }{%
{\ottdrulename{T\_Const}}{}%
}}

\newcommand{\ottdruleTXXAbs}[1]{\ottdrule[#1]{%
\ottpremise{\Gamma  \ottsym{,}  \mathit{x} \,  \mathord{:}  \, \ottnt{A}  \vdash  \ottnt{M}  \ottsym{:}  \ottnt{B}}%
}{
\Gamma  \vdash   \lambda\!  \, \mathit{x}  \ottsym{.}  \ottnt{M}  \ottsym{:}  \ottnt{A}  \rightarrow  \ottnt{B}}{%
{\ottdrulename{T\_Abs}}{}%
}}

\newcommand{\ottdruleTXXApp}[1]{\ottdrule[#1]{%
\ottpremise{ \Gamma  \vdash  \ottnt{M_{{\mathrm{1}}}}  \ottsym{:}  \ottnt{A}  \rightarrow  \ottnt{B}  \quad  \Gamma  \vdash  \ottnt{M_{{\mathrm{2}}}}  \ottsym{:}  \ottnt{A} }%
}{
\Gamma  \vdash  \ottnt{M_{{\mathrm{1}}}} \, \ottnt{M_{{\mathrm{2}}}}  \ottsym{:}  \ottnt{B}}{%
{\ottdrulename{T\_App}}{}%
}}

\newcommand{\ottdruleTXXPair}[1]{\ottdrule[#1]{%
\ottpremise{ \Gamma  \vdash  \ottnt{M_{{\mathrm{1}}}}  \ottsym{:}  \ottnt{A}  \quad  \Gamma  \vdash  \ottnt{M_{{\mathrm{2}}}}  \ottsym{:}  \ottnt{B} }%
}{
\Gamma  \vdash  \ottsym{(}  \ottnt{M_{{\mathrm{1}}}}  \ottsym{,}  \ottnt{M_{{\mathrm{2}}}}  \ottsym{)}  \ottsym{:}   \ottnt{A}  \times  \ottnt{B} }{%
{\ottdrulename{T\_Pair}}{}%
}}

\newcommand{\ottdruleTXXProjOne}[1]{\ottdrule[#1]{%
\ottpremise{\Gamma  \vdash  \ottnt{M}  \ottsym{:}   \ottnt{A}  \times  \ottnt{B} }%
}{
\Gamma  \vdash  \pi_1  \ottnt{M}  \ottsym{:}  \ottnt{A}}{%
{\ottdrulename{T\_Proj1}}{}%
}}

\newcommand{\ottdruleTXXProjTwo}[1]{\ottdrule[#1]{%
\ottpremise{\Gamma  \vdash  \ottnt{M}  \ottsym{:}   \ottnt{A}  \times  \ottnt{B} }%
}{
\Gamma  \vdash  \pi_2  \ottnt{M}  \ottsym{:}  \ottnt{B}}{%
{\ottdrulename{T\_Proj2}}{}%
}}

\newcommand{\ottdruleTXXInL}[1]{\ottdrule[#1]{%
\ottpremise{ \Gamma  \vdash  \ottnt{M}  \ottsym{:}  \ottnt{A}  \quad  \Gamma  \vdash  \ottnt{B} }%
}{
\Gamma  \vdash  \mathsf{inl} \, \ottnt{M}  \ottsym{:}   \ottnt{A}  +  \ottnt{B} }{%
{\ottdrulename{T\_InL}}{}%
}}

\newcommand{\ottdruleTXXInR}[1]{\ottdrule[#1]{%
\ottpremise{ \Gamma  \vdash  \ottnt{M}  \ottsym{:}  \ottnt{B}  \quad  \Gamma  \vdash  \ottnt{A} }%
}{
\Gamma  \vdash  \mathsf{inr} \, \ottnt{M}  \ottsym{:}   \ottnt{A}  +  \ottnt{B} }{%
{\ottdrulename{T\_InR}}{}%
}}

\newcommand{\ottdruleTXXCase}[1]{\ottdrule[#1]{%
\ottpremise{ \Gamma  \vdash  \ottnt{M}  \ottsym{:}   \ottnt{A}  +  \ottnt{B}   \quad   \Gamma  \ottsym{,}  \mathit{x} \,  \mathord{:}  \, \ottnt{A}  \vdash  \ottnt{M_{{\mathrm{1}}}}  \ottsym{:}  \ottnt{C}  \quad  \Gamma  \ottsym{,}  \mathit{y} \,  \mathord{:}  \, \ottnt{B}  \vdash  \ottnt{M_{{\mathrm{2}}}}  \ottsym{:}  \ottnt{C}  }%
}{
\Gamma  \vdash  \mathsf{case} \, \ottnt{M} \, \mathsf{of} \, \mathsf{inl} \, \mathit{x}  \rightarrow  \ottnt{M_{{\mathrm{1}}}}  \ottsym{;} \, \mathsf{inr} \, \mathit{y}  \rightarrow  \ottnt{M_{{\mathrm{2}}}}  \ottsym{:}  \ottnt{C}}{%
{\ottdrulename{T\_Case}}{}%
}}

\newcommand{\ottdruleTXXNil}[1]{\ottdrule[#1]{%
\ottpremise{\Gamma  \vdash  \ottnt{A}}%
}{
\Gamma  \vdash  \mathsf{nil}  \ottsym{:}   \ottnt{A}  \, \mathsf{list} }{%
{\ottdrulename{T\_Nil}}{}%
}}

\newcommand{\ottdruleTXXCons}[1]{\ottdrule[#1]{%
\ottpremise{\Gamma  \vdash  \ottnt{M}  \ottsym{:}   \ottnt{A}  \times     \ottnt{A}  \, \mathsf{list}    }%
}{
\Gamma  \vdash  \mathsf{cons} \, \ottnt{M}  \ottsym{:}   \ottnt{A}  \, \mathsf{list} }{%
{\ottdrulename{T\_Cons}}{}%
}}

\newcommand{\ottdruleTXXCaseList}[1]{\ottdrule[#1]{%
\ottpremise{ \Gamma  \vdash  \ottnt{M}  \ottsym{:}   \ottnt{A}  \, \mathsf{list}   \quad   \Gamma  \vdash  \ottnt{M_{{\mathrm{1}}}}  \ottsym{:}  \ottnt{B}  \quad  \Gamma  \ottsym{,}  \mathit{x} \,  \mathord{:}  \,  \ottnt{A}  \times     \ottnt{A}  \, \mathsf{list}      \vdash  \ottnt{M_{{\mathrm{2}}}}  \ottsym{:}  \ottnt{B}  }%
}{
\Gamma  \vdash  \mathsf{case} \, \ottnt{M} \, \mathsf{of} \, \mathsf{nil} \, \rightarrow  \ottnt{M_{{\mathrm{1}}}}  \ottsym{;} \, \mathsf{cons} \, \mathit{x}  \rightarrow  \ottnt{M_{{\mathrm{2}}}}  \ottsym{:}  \ottnt{B}}{%
{\ottdrulename{T\_CaseList}}{}%
}}

\newcommand{\ottdruleTXXFix}[1]{\ottdrule[#1]{%
\ottpremise{\Gamma  \ottsym{,}  \mathit{f} \,  \mathord{:}  \, \ottnt{A}  \rightarrow  \ottnt{B}  \ottsym{,}  \mathit{x} \,  \mathord{:}  \, \ottnt{A}  \vdash  \ottnt{M}  \ottsym{:}  \ottnt{B}}%
}{
\Gamma  \vdash  \mathsf{fix} \, \mathit{f}  \ottsym{.}   \lambda\!  \, \mathit{x}  \ottsym{.}  \ottnt{M}  \ottsym{:}  \ottnt{A}  \rightarrow  \ottnt{B}}{%
{\ottdrulename{T\_Fix}}{}%
}}

\newcommand{\ottdruleTXXOp}[1]{\ottdrule[#1]{%
\ottpremise{ \mathit{ty} \, \ottsym{(}  \mathsf{op}  \ottsym{)} \,  =  \,   \text{\unboldmath$\forall$}  \,  \algeffseqover{ \alpha }   \ottsym{.} \,  \ottnt{A}  \hookrightarrow  \ottnt{B}   \quad   \Gamma  \vdash  \ottnt{M}  \ottsym{:}   \ottnt{A}    [   \algeffseqover{ \ottnt{C} }   \ottsym{/}   \algeffseqover{ \alpha }   ]    \quad  \Gamma  \vdash   \algeffseqover{ \ottnt{C} }   }%
}{
\Gamma  \vdash   \textup{\texttt{\#}\relax}  \mathsf{op}   \ottsym{(}   \ottnt{M}   \ottsym{)}   \ottsym{:}   \ottnt{B}    [   \algeffseqover{ \ottnt{C} }   \ottsym{/}   \algeffseqover{ \alpha }   ]  }{%
{\ottdrulename{T\_Op}}{}%
}}

\newcommand{\ottdruleTXXHandle}[1]{\ottdrule[#1]{%
\ottpremise{ \Gamma  \vdash  \ottnt{M}  \ottsym{:}  \ottnt{A}  \quad  \Gamma  \vdash  \ottnt{H}  \ottsym{:}  \ottnt{A}  \Rightarrow  \ottnt{B} }%
}{
\Gamma  \vdash  \mathsf{handle} \, \ottnt{M} \, \mathsf{with} \, \ottnt{H}  \ottsym{:}  \ottnt{B}}{%
{\ottdrulename{T\_Handle}}{}%
}}

\newcommand{\ottdruleTXXGen}[1]{\ottdrule[#1]{%
\ottpremise{\Gamma  \ottsym{,}  \alpha  \vdash  \ottnt{M}  \ottsym{:}  \ottnt{A}}%
}{
\Gamma  \vdash  \ottnt{M}  \ottsym{:}   \text{\unboldmath$\forall$}  \, \alpha  \ottsym{.} \, \ottnt{A}}{%
{\ottdrulename{T\_Gen}}{}%
}}

\newcommand{\ottdruleTXXInst}[1]{\ottdrule[#1]{%
\ottpremise{ \Gamma  \vdash  \ottnt{M}  \ottsym{:}  \ottnt{A}  \quad   \Gamma  \vdash  \ottnt{A}  \sqsubseteq  \ottnt{B}  \quad  \Gamma  \vdash  \ottnt{B}  }%
}{
\Gamma  \vdash  \ottnt{M}  \ottsym{:}  \ottnt{B}}{%
{\ottdrulename{T\_Inst}}{}%
}}

\newcommand{\ottdruleTHXXReturn}[1]{\ottdrule[#1]{%
\ottpremise{\Gamma  \ottsym{,}  \mathit{x} \,  \mathord{:}  \, \ottnt{A}  \vdash  \ottnt{M}  \ottsym{:}  \ottnt{B}}%
}{
\Gamma  \vdash  \mathsf{return} \, \mathit{x}  \rightarrow  \ottnt{M}  \ottsym{:}  \ottnt{A}  \Rightarrow  \ottnt{B}}{%
{\ottdrulename{TH\_Return}}{}%
}}

\newcommand{\ottdruleTHXXOp}[1]{\ottdrule[#1]{%
\ottpremise{ \Gamma  \vdash  \ottnt{H}  \ottsym{:}  \ottnt{A}  \Rightarrow  \ottnt{B}  \quad   \mathit{ty} \, \ottsym{(}  \mathsf{op}  \ottsym{)} \,  =  \,   \text{\unboldmath$\forall$}  \,  \algeffseqover{ \alpha }   \ottsym{.} \,  \ottnt{C}  \hookrightarrow  \ottnt{D}   \quad  \Gamma  \ottsym{,}   \algeffseqover{ \alpha }   \ottsym{,}  \mathit{x} \,  \mathord{:}  \, \ottnt{C}  \ottsym{,}  \mathit{k} \,  \mathord{:}  \, \ottnt{D}  \rightarrow  \ottnt{B}  \vdash  \ottnt{M}  \ottsym{:}  \ottnt{B}  }%
}{
\Gamma  \vdash  \ottnt{H}  \ottsym{;}  \mathsf{op}  \ottsym{(}  \mathit{x}  \ottsym{,}  \mathit{k}  \ottsym{)}  \rightarrow  \ottnt{M}  \ottsym{:}  \ottnt{A}  \Rightarrow  \ottnt{B}}{%
{\ottdrulename{TH\_Op}}{}%
}}

\newcommand{\ottdruleCXXRefl}[1]{\ottdrule[#1]{%
\ottpremise{\vdash  \Gamma}%
}{
\Gamma  \vdash  \ottnt{A}  \sqsubseteq  \ottnt{A}}{%
{\ottdrulename{C\_Refl}}{}%
}}

\newcommand{\ottdruleCXXTrans}[1]{\ottdrule[#1]{%
\ottpremise{ \Gamma  \vdash  \ottnt{A}  \sqsubseteq  \ottnt{C}  \quad  \Gamma  \vdash  \ottnt{C}  \sqsubseteq  \ottnt{B} }%
}{
\Gamma  \vdash  \ottnt{A}  \sqsubseteq  \ottnt{B}}{%
{\ottdrulename{C\_Trans}}{}%
}}

\newcommand{\ottdruleCXXFun}[1]{\ottdrule[#1]{%
\ottpremise{ \Gamma  \vdash  \ottnt{B_{{\mathrm{1}}}}  \sqsubseteq  \ottnt{A_{{\mathrm{1}}}}  \quad  \Gamma  \vdash  \ottnt{A_{{\mathrm{2}}}}  \sqsubseteq  \ottnt{B_{{\mathrm{2}}}} }%
}{
\Gamma  \vdash  \ottnt{A_{{\mathrm{1}}}}  \rightarrow  \ottnt{A_{{\mathrm{2}}}}  \sqsubseteq  \ottnt{B_{{\mathrm{1}}}}  \rightarrow  \ottnt{B_{{\mathrm{2}}}}}{%
{\ottdrulename{C\_Fun}}{}%
}}

\newcommand{\ottdruleCXXFunEff}[1]{\ottdrule[#1]{%
\ottpremise{ \Gamma  \vdash  \ottnt{B_{{\mathrm{1}}}}  \sqsubseteq  \ottnt{A_{{\mathrm{1}}}}  \quad  \Gamma  \vdash  \ottnt{A_{{\mathrm{2}}}}  \sqsubseteq  \ottnt{B_{{\mathrm{2}}}} }%
}{
\Gamma  \vdash   \ottnt{A_{{\mathrm{1}}}}   \rightarrow ^{ \epsilon }  \ottnt{A_{{\mathrm{2}}}}   \sqsubseteq   \ottnt{B_{{\mathrm{1}}}}   \rightarrow ^{ \epsilon }  \ottnt{B_{{\mathrm{2}}}} }{%
{\ottdrulename{C\_FunEff}}{}%
}}

\newcommand{\ottdruleCXXProd}[1]{\ottdrule[#1]{%
\ottpremise{ \Gamma  \vdash  \ottnt{A_{{\mathrm{1}}}}  \sqsubseteq  \ottnt{B_{{\mathrm{1}}}}  \quad  \Gamma  \vdash  \ottnt{A_{{\mathrm{2}}}}  \sqsubseteq  \ottnt{B_{{\mathrm{2}}}} }%
}{
\Gamma  \vdash   \ottnt{A_{{\mathrm{1}}}}  \times  \ottnt{A_{{\mathrm{2}}}}   \sqsubseteq   \ottnt{B_{{\mathrm{1}}}}  \times  \ottnt{B_{{\mathrm{2}}}} }{%
{\ottdrulename{C\_Prod}}{}%
}}

\newcommand{\ottdruleCXXSum}[1]{\ottdrule[#1]{%
\ottpremise{ \Gamma  \vdash  \ottnt{A_{{\mathrm{1}}}}  \sqsubseteq  \ottnt{B_{{\mathrm{1}}}}  \quad  \Gamma  \vdash  \ottnt{A_{{\mathrm{2}}}}  \sqsubseteq  \ottnt{B_{{\mathrm{2}}}} }%
}{
\Gamma  \vdash   \ottnt{A_{{\mathrm{1}}}}  +  \ottnt{A_{{\mathrm{2}}}}   \sqsubseteq   \ottnt{B_{{\mathrm{1}}}}  +  \ottnt{B_{{\mathrm{2}}}} }{%
{\ottdrulename{C\_Sum}}{}%
}}

\newcommand{\ottdruleCXXPoly}[1]{\ottdrule[#1]{%
\ottpremise{\Gamma  \ottsym{,}  \alpha  \vdash  \ottnt{A}  \sqsubseteq  \ottnt{B}}%
}{
\Gamma  \vdash   \text{\unboldmath$\forall$}  \, \alpha  \ottsym{.} \, \ottnt{A}  \sqsubseteq   \text{\unboldmath$\forall$}  \, \alpha  \ottsym{.} \, \ottnt{B}}{%
{\ottdrulename{C\_Poly}}{}%
}}

\newcommand{\ottdruleCXXList}[1]{\ottdrule[#1]{%
\ottpremise{\Gamma  \vdash  \ottnt{A}  \sqsubseteq  \ottnt{B}}%
}{
\Gamma  \vdash   \ottnt{A}  \, \mathsf{list}   \sqsubseteq   \ottnt{B}  \, \mathsf{list} }{%
{\ottdrulename{C\_List}}{}%
}}

\newcommand{\ottdruleCXXInst}[1]{\ottdrule[#1]{%
\ottpremise{\Gamma  \vdash  \ottnt{B}}%
}{
\Gamma  \vdash   \text{\unboldmath$\forall$}  \, \alpha  \ottsym{.} \, \ottnt{A}  \sqsubseteq   \ottnt{A}    [  \ottnt{B}  \ottsym{/}  \alpha  ]  }{%
{\ottdrulename{C\_Inst}}{}%
}}

\newcommand{\ottdruleCXXGen}[1]{\ottdrule[#1]{%
\ottpremise{ \vdash  \Gamma  \quad  \alpha \,  \not\in  \,  \mathit{ftv}  (  \ottnt{A}  )  }%
}{
\Gamma  \vdash  \ottnt{A}  \sqsubseteq   \text{\unboldmath$\forall$}  \, \alpha  \ottsym{.} \, \ottnt{A}}{%
{\ottdrulename{C\_Gen}}{}%
}}

\newcommand{\ottdruleCXXDFun}[1]{\ottdrule[#1]{%
\ottpremise{ \vdash  \Gamma  \quad  \alpha \,  \not\in  \,  \mathit{ftv}  (  \ottnt{A}  )  }%
}{
\Gamma  \vdash   \text{\unboldmath$\forall$}  \, \alpha  \ottsym{.} \, \ottnt{A}  \rightarrow  \ottnt{B}  \sqsubseteq  \ottnt{A}  \rightarrow   \text{\unboldmath$\forall$}  \, \alpha  \ottsym{.} \, \ottnt{B}}{%
{\ottdrulename{C\_DFun}}{}%
}}

\newcommand{\ottdruleCXXDFunEff}[1]{\ottdrule[#1]{%
\ottpremise{ \vdash  \Gamma  \quad   \alpha \,  \not\in  \,  \mathit{ftv}  (  \ottnt{A}  )   \quad  \mathit{SR} \, \ottsym{(}  \epsilon  \ottsym{)}  }%
}{
\Gamma  \vdash   \text{\unboldmath$\forall$}  \, \alpha  \ottsym{.} \,    \ottnt{A}   \rightarrow ^{ \epsilon }  \ottnt{B}     \sqsubseteq   \ottnt{A}   \rightarrow ^{ \epsilon }   \text{\unboldmath$\forall$}  \, \alpha  \ottsym{.} \, \ottnt{B} }{%
{\ottdrulename{C\_DFunEff}}{}%
}}

\newcommand{\ottdruleCXXDProd}[1]{\ottdrule[#1]{%
\ottpremise{\vdash  \Gamma}%
}{
\Gamma  \vdash   \text{\unboldmath$\forall$}  \, \alpha  \ottsym{.} \,    \ottnt{A}  \times  \ottnt{B}     \sqsubseteq   \ottsym{(}   \text{\unboldmath$\forall$}  \, \alpha  \ottsym{.} \, \ottnt{A}  \ottsym{)}  \times  \ottsym{(}   \text{\unboldmath$\forall$}  \, \alpha  \ottsym{.} \, \ottnt{B}  \ottsym{)} }{%
{\ottdrulename{C\_DProd}}{}%
}}

\newcommand{\ottdruleCXXDSum}[1]{\ottdrule[#1]{%
\ottpremise{\vdash  \Gamma}%
}{
\Gamma  \vdash   \text{\unboldmath$\forall$}  \, \alpha  \ottsym{.} \,    \ottnt{A}  +  \ottnt{B}     \sqsubseteq   \ottsym{(}   \text{\unboldmath$\forall$}  \, \alpha  \ottsym{.} \, \ottnt{A}  \ottsym{)}  +  \ottsym{(}   \text{\unboldmath$\forall$}  \, \alpha  \ottsym{.} \, \ottnt{B}  \ottsym{)} }{%
{\ottdrulename{C\_DSum}}{}%
}}

\newcommand{\ottdruleCXXDList}[1]{\ottdrule[#1]{%
\ottpremise{\vdash  \Gamma}%
}{
\Gamma  \vdash   \text{\unboldmath$\forall$}  \, \alpha  \ottsym{.} \,    \ottnt{A}  \, \mathsf{list}     \sqsubseteq   \ottsym{(}   \text{\unboldmath$\forall$}  \, \alpha  \ottsym{.} \, \ottnt{A}  \ottsym{)}  \, \mathsf{list} }{%
{\ottdrulename{C\_DList}}{}%
}}

\newcommand{\ottdruleTeXXVar}[1]{\ottdrule[#1]{%
\ottpremise{ \vdash  \Gamma  \quad  \mathit{x} \,  \mathord{:}  \, \ottnt{A} \,  \in  \, \Gamma }%
}{
\Gamma  \vdash  \mathit{x}  \ottsym{:}  \ottnt{A} \,  |  \, \epsilon}{%
{\ottdrulename{Te\_Var}}{}%
}}

\newcommand{\ottdruleTeXXConst}[1]{\ottdrule[#1]{%
\ottpremise{\vdash  \Gamma}%
}{
\Gamma  \vdash  \ottnt{c}  \ottsym{:}   \mathit{ty}  (  \ottnt{c}  )  \,  |  \, \epsilon}{%
{\ottdrulename{Te\_Const}}{}%
}}

\newcommand{\ottdruleTeXXAbs}[1]{\ottdrule[#1]{%
\ottpremise{\Gamma  \ottsym{,}  \mathit{x} \,  \mathord{:}  \, \ottnt{A}  \vdash  \ottnt{M}  \ottsym{:}  \ottnt{B} \,  |  \, \epsilon'}%
}{
\Gamma  \vdash   \lambda\!  \, \mathit{x}  \ottsym{.}  \ottnt{M}  \ottsym{:}   \ottnt{A}   \rightarrow ^{ \epsilon' }  \ottnt{B}  \,  |  \, \epsilon}{%
{\ottdrulename{Te\_Abs}}{}%
}}

\newcommand{\ottdruleTeXXApp}[1]{\ottdrule[#1]{%
\ottpremise{ \Gamma  \vdash  \ottnt{M_{{\mathrm{1}}}}  \ottsym{:}   \ottnt{A}   \rightarrow ^{ \epsilon' }  \ottnt{B}  \,  |  \, \epsilon  \quad   \Gamma  \vdash  \ottnt{M_{{\mathrm{2}}}}  \ottsym{:}  \ottnt{A} \,  |  \, \epsilon  \quad  \epsilon' \,  \subseteq  \, \epsilon  }%
}{
\Gamma  \vdash  \ottnt{M_{{\mathrm{1}}}} \, \ottnt{M_{{\mathrm{2}}}}  \ottsym{:}  \ottnt{B} \,  |  \, \epsilon}{%
{\ottdrulename{Te\_App}}{}%
}}

\newcommand{\ottdruleTeXXPair}[1]{\ottdrule[#1]{%
\ottpremise{ \Gamma  \vdash  \ottnt{M_{{\mathrm{1}}}}  \ottsym{:}  \ottnt{A} \,  |  \, \epsilon  \quad  \Gamma  \vdash  \ottnt{M_{{\mathrm{2}}}}  \ottsym{:}  \ottnt{B} \,  |  \, \epsilon }%
}{
\Gamma  \vdash  \ottsym{(}  \ottnt{M_{{\mathrm{1}}}}  \ottsym{,}  \ottnt{M_{{\mathrm{2}}}}  \ottsym{)}  \ottsym{:}   \ottnt{A}  \times  \ottnt{B}  \,  |  \, \epsilon}{%
{\ottdrulename{Te\_Pair}}{}%
}}

\newcommand{\ottdruleTeXXProjOne}[1]{\ottdrule[#1]{%
\ottpremise{\Gamma  \vdash  \ottnt{M}  \ottsym{:}   \ottnt{A}  \times  \ottnt{B}  \,  |  \, \epsilon}%
}{
\Gamma  \vdash  \pi_1  \ottnt{M}  \ottsym{:}  \ottnt{A} \,  |  \, \epsilon}{%
{\ottdrulename{Te\_Proj1}}{}%
}}

\newcommand{\ottdruleTeXXProjTwo}[1]{\ottdrule[#1]{%
\ottpremise{\Gamma  \vdash  \ottnt{M}  \ottsym{:}   \ottnt{A}  \times  \ottnt{B}  \,  |  \, \epsilon}%
}{
\Gamma  \vdash  \pi_2  \ottnt{M}  \ottsym{:}  \ottnt{B} \,  |  \, \epsilon}{%
{\ottdrulename{Te\_Proj2}}{}%
}}

\newcommand{\ottdruleTeXXInL}[1]{\ottdrule[#1]{%
\ottpremise{ \Gamma  \vdash  \ottnt{M}  \ottsym{:}  \ottnt{A} \,  |  \, \epsilon  \quad  \Gamma  \vdash  \ottnt{B} }%
}{
\Gamma  \vdash  \mathsf{inl} \, \ottnt{M}  \ottsym{:}   \ottnt{A}  +  \ottnt{B}  \,  |  \, \epsilon}{%
{\ottdrulename{Te\_InL}}{}%
}}

\newcommand{\ottdruleTeXXInR}[1]{\ottdrule[#1]{%
\ottpremise{ \Gamma  \vdash  \ottnt{M}  \ottsym{:}  \ottnt{B} \,  |  \, \epsilon  \quad  \Gamma  \vdash  \ottnt{A} }%
}{
\Gamma  \vdash  \mathsf{inr} \, \ottnt{M}  \ottsym{:}   \ottnt{A}  +  \ottnt{B}  \,  |  \, \epsilon}{%
{\ottdrulename{Te\_InR}}{}%
}}

\newcommand{\ottdruleTeXXCase}[1]{\ottdrule[#1]{%
\ottpremise{ \Gamma  \vdash  \ottnt{M}  \ottsym{:}   \ottnt{A}  +  \ottnt{B}  \,  |  \, \epsilon  \quad   \Gamma  \ottsym{,}  \mathit{x} \,  \mathord{:}  \, \ottnt{A}  \vdash  \ottnt{M_{{\mathrm{1}}}}  \ottsym{:}  \ottnt{C} \,  |  \, \epsilon  \quad  \Gamma  \ottsym{,}  \mathit{y} \,  \mathord{:}  \, \ottnt{B}  \vdash  \ottnt{M_{{\mathrm{2}}}}  \ottsym{:}  \ottnt{C} \,  |  \, \epsilon  }%
}{
\Gamma  \vdash  \mathsf{case} \, \ottnt{M} \, \mathsf{of} \, \mathsf{inl} \, \mathit{x}  \rightarrow  \ottnt{M_{{\mathrm{1}}}}  \ottsym{;} \, \mathsf{inr} \, \mathit{y}  \rightarrow  \ottnt{M_{{\mathrm{2}}}}  \ottsym{:}  \ottnt{C} \,  |  \, \epsilon}{%
{\ottdrulename{Te\_Case}}{}%
}}

\newcommand{\ottdruleTeXXNil}[1]{\ottdrule[#1]{%
\ottpremise{\Gamma  \vdash  \ottnt{A}}%
}{
\Gamma  \vdash  \mathsf{nil}  \ottsym{:}   \ottnt{A}  \, \mathsf{list}  \,  |  \, \epsilon}{%
{\ottdrulename{Te\_Nil}}{}%
}}

\newcommand{\ottdruleTeXXCons}[1]{\ottdrule[#1]{%
\ottpremise{\Gamma  \vdash  \ottnt{M}  \ottsym{:}   \ottnt{A}  \times     \ottnt{A}  \, \mathsf{list}     \,  |  \, \epsilon}%
}{
\Gamma  \vdash  \mathsf{cons} \, \ottnt{M}  \ottsym{:}   \ottnt{A}  \, \mathsf{list}  \,  |  \, \epsilon}{%
{\ottdrulename{Te\_Cons}}{}%
}}

\newcommand{\ottdruleTeXXCaseList}[1]{\ottdrule[#1]{%
\ottpremise{ \Gamma  \vdash  \ottnt{M}  \ottsym{:}   \ottnt{A}  \, \mathsf{list}  \,  |  \, \epsilon  \quad   \Gamma  \vdash  \ottnt{M_{{\mathrm{1}}}}  \ottsym{:}  \ottnt{B} \,  |  \, \epsilon  \quad  \Gamma  \ottsym{,}  \mathit{x} \,  \mathord{:}  \,  \ottnt{A}  \times     \ottnt{A}  \, \mathsf{list}      \vdash  \ottnt{M_{{\mathrm{2}}}}  \ottsym{:}  \ottnt{B} \,  |  \, \epsilon  }%
}{
\Gamma  \vdash  \mathsf{case} \, \ottnt{M} \, \mathsf{of} \, \mathsf{nil} \, \rightarrow  \ottnt{M_{{\mathrm{1}}}}  \ottsym{;} \, \mathsf{cons} \, \mathit{x}  \rightarrow  \ottnt{M_{{\mathrm{2}}}}  \ottsym{:}  \ottnt{B} \,  |  \, \epsilon}{%
{\ottdrulename{Te\_CaseList}}{}%
}}

\newcommand{\ottdruleTeXXFix}[1]{\ottdrule[#1]{%
\ottpremise{\Gamma  \ottsym{,}  \mathit{f} \,  \mathord{:}  \,  \ottnt{A}   \rightarrow ^{ \epsilon }  \ottnt{B}   \ottsym{,}  \mathit{x} \,  \mathord{:}  \, \ottnt{A}  \vdash  \ottnt{M}  \ottsym{:}  \ottnt{B} \,  |  \, \epsilon}%
}{
\Gamma  \vdash  \mathsf{fix} \, \mathit{f}  \ottsym{.}   \lambda\!  \, \mathit{x}  \ottsym{.}  \ottnt{M}  \ottsym{:}   \ottnt{A}   \rightarrow ^{ \epsilon }  \ottnt{B}  \,  |  \, \epsilon'}{%
{\ottdrulename{Te\_Fix}}{}%
}}

\newcommand{\ottdruleTeXXOp}[1]{\ottdrule[#1]{%
\ottpremise{ \mathit{ty} \, \ottsym{(}  \mathsf{op}  \ottsym{)} \,  =  \,   \text{\unboldmath$\forall$}  \,  \algeffseqover{ \alpha }   \ottsym{.} \,  \ottnt{A}  \hookrightarrow  \ottnt{B}   \quad   \mathsf{op} \,  \in  \, \epsilon  \quad   \Gamma  \vdash  \ottnt{M}  \ottsym{:}   \ottnt{A}    [   \algeffseqover{ \ottnt{C} }   \ottsym{/}   \algeffseqover{ \alpha }   ]   \,  |  \, \epsilon  \quad  \Gamma  \vdash   \algeffseqover{ \ottnt{C} }    }%
}{
\Gamma  \vdash   \textup{\texttt{\#}\relax}  \mathsf{op}   \ottsym{(}   \ottnt{M}   \ottsym{)}   \ottsym{:}   \ottnt{B}    [   \algeffseqover{ \ottnt{C} }   \ottsym{/}   \algeffseqover{ \alpha }   ]   \,  |  \, \epsilon}{%
{\ottdrulename{Te\_Op}}{}%
}}

\newcommand{\ottdruleTeXXHandle}[1]{\ottdrule[#1]{%
\ottpremise{ \Gamma  \vdash  \ottnt{M}  \ottsym{:}  \ottnt{A} \,  |  \, \epsilon  \quad  \Gamma  \vdash  \ottnt{H}  \ottsym{:}  \ottnt{A} \,  |  \, \epsilon  \Rightarrow  \ottnt{B} \,  |  \, \epsilon' }%
}{
\Gamma  \vdash  \mathsf{handle} \, \ottnt{M} \, \mathsf{with} \, \ottnt{H}  \ottsym{:}  \ottnt{B} \,  |  \, \epsilon'}{%
{\ottdrulename{Te\_Handle}}{}%
}}

\newcommand{\ottdruleTeXXGen}[1]{\ottdrule[#1]{%
\ottpremise{ \Gamma  \ottsym{,}  \alpha  \vdash  \ottnt{M}  \ottsym{:}  \ottnt{A} \,  |  \, \epsilon  \quad  \mathit{SR} \, \ottsym{(}  \epsilon  \ottsym{)} }%
}{
\Gamma  \vdash  \ottnt{M}  \ottsym{:}   \text{\unboldmath$\forall$}  \, \alpha  \ottsym{.} \, \ottnt{A} \,  |  \, \epsilon}{%
{\ottdrulename{Te\_Gen}}{}%
}}

\newcommand{\ottdruleTeXXInst}[1]{\ottdrule[#1]{%
\ottpremise{ \Gamma  \vdash  \ottnt{M}  \ottsym{:}  \ottnt{A} \,  |  \, \epsilon  \quad   \Gamma  \vdash  \ottnt{A}  \sqsubseteq  \ottnt{B}  \quad  \Gamma  \vdash  \ottnt{B}  }%
}{
\Gamma  \vdash  \ottnt{M}  \ottsym{:}  \ottnt{B} \,  |  \, \epsilon}{%
{\ottdrulename{Te\_Inst}}{}%
}}

\newcommand{\ottdruleTeXXWeak}[1]{\ottdrule[#1]{%
\ottpremise{ \Gamma  \vdash  \ottnt{M}  \ottsym{:}  \ottnt{A} \,  |  \, \epsilon'  \quad  \epsilon' \,  \subseteq  \, \epsilon }%
}{
\Gamma  \vdash  \ottnt{M}  \ottsym{:}  \ottnt{A} \,  |  \, \epsilon}{%
{\ottdrulename{Te\_Weak}}{}%
}}

\newcommand{\ottdruleTHeXXReturn}[1]{\ottdrule[#1]{%
\ottpremise{ \Gamma  \ottsym{,}  \mathit{x} \,  \mathord{:}  \, \ottnt{A}  \vdash  \ottnt{M}  \ottsym{:}  \ottnt{B} \,  |  \, \epsilon'  \quad  \epsilon \,  \subseteq  \, \epsilon' }%
}{
\Gamma  \vdash  \mathsf{return} \, \mathit{x}  \rightarrow  \ottnt{M}  \ottsym{:}  \ottnt{A} \,  |  \, \epsilon  \Rightarrow  \ottnt{B} \,  |  \, \epsilon'}{%
{\ottdrulename{THe\_Return}}{}%
}}

\newcommand{\ottdruleTHeXXOp}[1]{\ottdrule[#1]{%
\ottpremise{ \Gamma  \vdash  \ottnt{H}  \ottsym{:}  \ottnt{A} \,  |  \, \epsilon  \Rightarrow  \ottnt{B} \,  |  \, \epsilon'  \quad   \mathit{ty} \, \ottsym{(}  \mathsf{op}  \ottsym{)} \,  =  \,   \text{\unboldmath$\forall$}  \,  \algeffseqover{ \alpha }   \ottsym{.} \,  \ottnt{C}  \hookrightarrow  \ottnt{D}   \quad  \Gamma  \ottsym{,}   \algeffseqover{ \alpha }   \ottsym{,}  \mathit{x} \,  \mathord{:}  \, \ottnt{C}  \ottsym{,}  \mathit{k} \,  \mathord{:}  \,  \ottnt{D}   \rightarrow ^{ \epsilon' }  \ottnt{B}   \vdash  \ottnt{M}  \ottsym{:}  \ottnt{B} \,  |  \, \epsilon'  }%
}{
\Gamma  \vdash  \ottnt{H}  \ottsym{;}  \mathsf{op}  \ottsym{(}  \mathit{x}  \ottsym{,}  \mathit{k}  \ottsym{)}  \rightarrow  \ottnt{M}  \ottsym{:}  \ottnt{A} \,  |  \, \epsilon \,  \mathbin{\uplus}  \, \ottsym{\{}  \mathsf{op}  \ottsym{\}}  \Rightarrow  \ottnt{B} \,  |  \, \epsilon'}{%
{\ottdrulename{THe\_Op}}{}%
}}



%% file: macro.tex
\usepackage{amsmath}
\usepackage{multirow}
\usepackage{color}
\usepackage{thmtools}
\usepackage{thm-restate}

\usepackage{listings}
\usepackage{enumitem}
\newenvironment{caseanalysis}
{\begin{description}[leftmargin=1.2em]}
{\end{description}}
\def\case#1:{\item[\textmd{Case {#1}:}]}
\def\otherwise:{\item[\textmd{Otherwise:}]}

\newcommand{\reffig}[1]{Figure~\ref{fig:#1}}
\newcommand{\refsec}[1]{Section~\ref{sec:#1}}
\newcommand{\reflem}[1]{Lemma~\ref{lem:#1}}

\newcommand{\refasm}[1]{Assumption~\ref{assum:#1}}
\newcommand{\refdef}[1]{Definition~\ref{def:#1}}
\newcommand{\refprop}[1]{Proposition~\ref{prop:#1}}

\newcommand{\lang}{$\lambda_\text{eff}$}
\newcommand{\extlang}{$\lambda_\text{eff}^\text{ext}$}
\newcommand{\handlewith}{\texttt{\textbf{handle}}--\texttt{\textbf{with}}}
\newcommand{\handlewithsf}{\textsf{handle}--\textsf{with}}
\newcommand{\emptytype}{\ensuremath{\mathsf{zero}}}

\definecolor{gray96}{gray}{.9}
\newcommand{\graybox}[1]{\!\colorbox{gray96}{\(\!#1\!\)}\!}
\newcommand{\textgray}[1]{\textcolor{gray}{#1}}

\newtheorem{defn}{Definition}
\newtheorem{assum}{Assumption}
\newtheorem{conv}{Convention}

\newtheorem{prop}{Proposition}

\newtheorem{lemm}{Lemma}
\newenvironment{lemmap}[2]
{\begin{lemm}[#1] \label{lem:#2} \noindent}
{\end{lemm}}
\newenvironment{lemma}[1]
{\begin{lemm} \label{lem:#1} \noindent}
{\end{lemm}}

\newtheorem{thm}{Theorem}

\theoremstyle{acmdefinition}
\newtheorem{remark}{Remark}

\newcommand{\Rule}[2]{\ensuremath{\text{({\sc{{#1}\_{#2}}})}}}
\newcommand{\RulewoP}[2]{\ensuremath{\text{{\sc{{#1}\_{#2}}}}}}

\newcommand{\Srule}[1]{\Rule{C}{#1}}
\newcommand{\R}[1]{\Rule{R}{#1}}
\newcommand{\RwoP}[1]{\RulewoP{R}{#1}}
\newcommand{\E}[1]{\Rule{E}{#1}}
\newcommand{\T}[1]{\Rule{T}{#1}}
\newcommand{\THrule}[1]{\Rule{TH}{#1}}
\newcommand{\Te}[1]{\Rule{Te}{#1}}

\newcommand{\TS}[1]{\ifdraft\textcolor{red}{TS:#1}\fi}
\newcommand{\AI}[1]{\ifdraft\textcolor{blue}{AI:#1}\fi}

\newcommand{\defeq}{\stackrel{\rm \tiny def}{=}}

%% file: sections/abstract.tex
\begin{abstract}
 The naive combination of polymorphic effects and polymorphic type assignment
 has been well known to break type safety.  Existing approaches to
 this problem are classified into two groups: one for restricting how effects
 are triggered and the other for restricting how they are implemented.  This
 work explores a new approach to ensuring the safety of polymorphic effects in polymorphic
 type assignment.  A novelty of our work lies in finding a restriction on
 \emph{effect interfaces}.  To formalize our idea, we employ algebraic effects
 and handlers, where an effect interface is given by a set of operations coupled
 with type signatures.  We propose \emph{signature restriction}, a new notion to
 restrict the type signatures of operations, and show that signature restriction
 is sufficient to ensure type safety of an effectful language equipped with
 unrestricted polymorphic type assignment.  We also develop a type-and-effect
 system to enable the use of both operations that satisfy and do not satisfy
 the signature restriction in a single program.
\end{abstract}

%% file: sections/intro.tex
\TS{
\begin{itemize}
 \item Discuss undecidability of type containment (and a polymorphic type system).
       We will prepare a surface language for programmers. --> Remark 2.
 \item Mention \Srule{Sum} in \refsec{ext}.
 \item Domain type --> argument type, codomain type --> return type
 \item Reconsider the external URL where the artifact is.  What site should we use?
 \item Does the supplementary material in publishing contains the full definition and proofs?
\end{itemize}
}

\section{Introduction}

\subsection{Background: Polymorphic Type Assignment with Computational Effects}

Pervasive in programming are computational effects, such as mutable memory
cells, backtracking, exception handling, concurrency/parallelism, and I/O
processing for terminals, files, networks, etc.  These effects have a variety of roles:
I/O processing enables interaction with external environments; memory
manipulation and concurrency/parallelism make software efficient; and
backtracking and exception provide reusable, general operations that make
it unnecessary to write boilerplate code.  These effects have also been proven convenient
in functional
programming~\cite{Goron/Milner/Wadsworth_1979_book,Wadler_1992_POPL,Peython-Jones/Wadler_1993_POPL}.
%

In return for convenience, however, computational effects can introduce weird, counterintuitive behavior
into programs and complicate program reasoning and verification.  For
example, incorporating effects into dependent type theory could
easily lead to inconsistency~\cite{Pedrot/Tabareau_2020_POPL}.  This fact
encourages dependent type systems to separate term-level computation from types~\cite{Xi_2007_JFP,Casinghino/Sjoberg/Weirich_2014_POPL,Swamy-et.al_2016_POPL,Sekiyama/Igarashi_2017_POPL,Ahman_2017_PhD,Cong/Asai_2018_ICFP}.
For program reasoning, the state transition caused by
effectful computations has to be tracked~\cite{Pitts/stark_1998_openfl,Ahmed/Dreyer/Rossberg_2009_POPL,Dereyer/Neis/Birkeda_2010_ICFP}.

These kinds of gaps between pure and effectful computations are also found in
our target, i.e., polymorphic type assignment:
%
%
although any pure expressions can safely be assigned polymorphic types~\cite{Leivant_1983_POPL},
unrestricted polymorphic type assignment to effectful expressions may break type
safety.  This problem with polymorphic type assignment has been
discovered in call-by-value languages with \emph{polymorphic effects}, which are
effects caused by polymorphic operations.  For example, ML-style references
are an instance of polymorphic effects because the operations for memory cell
creation, assignment, and dereference are polymorphic~\cite{Milner/Tofte/Harper_1990_SML,OCaml}.
\citet{Goron/Milner/Wadsworth_1979_book} showed that the ML-style references
cannot cooperate safely with unrestricted polymorphic type assignment owing to
the polymorphism of the operations.  Another example is control effects, which are
triggered by control operators such as
\textsf{call/cc}~\cite{Clinger/Friedman/Wand_1985} and
\textsf{shift/reset}~\cite{Danvy/Filinski_1990_LFP}.  These operators can be
assigned polymorphic types but the polymorphic control operators may cause
unsafe behavior in unrestricted polymorphic type
assignment~\cite{Harper/Lillibridge_1993_LSC}.
This fault even occurs in let-polymorphic type
assignment~\cite{Milner_1978_JCSS} where quantifiers only appear at the
outermost positions.

Many approaches to the safe use of polymorphic
effects in polymorphic type
assignment have been proposed~\cite{Tofte_1990_IC,Leroy/Weis_1991_POPL,Appel/MacQueen_1991_PLILP,Hoang/Mitchell/Viswanathan_1993_LICS,Wright_1995_LSC,Garrigue_2004_FLOPS,Asai/Kameyama_2007_APLAS,Kammar/Pretnar_2017_JFP,Sekiyama/Igarashi_2019_ESOP}.
These approaches are classified into two groups.  The first group---to which most of the
approaches belong---aims at restricting \emph{how effects are triggered}.  For
example, the value restriction~\cite{Tofte_1990_IC} restricts
polymorphic expressions to be only values in order to prevent
polymorphic expressions from triggering effects.  The other group aims at
restricting \emph{how effects are implemented}.  For example,
\citet{Sekiyama/Igarashi_2019_ESOP} proposed a type system that accepts only effects
that are \emph{safe}, i.e., that do not cause programs to get stuck no matter how they are
used.


\subsection{Our Work}

This work explores a new approach to safe polymorphic type assignment for
effectful call-by-value languages.  A novelty of our approach lies in
restriction on \emph{effect interfaces}.  In this work, the effect interfaces are represented by
sets of \emph{operations} coupled with \emph{type signatures}.  For example, an
interface for exceptions consists of a single operation \texttt{raise} to raise
an exception and its type signature $  \text{\unboldmath$\forall$}  \, \alpha  \ottsym{.} \,   \mathsf{unit}   \hookrightarrow  \alpha $, which means that
\texttt{raise} takes the unit value as an argument and returns a value of any
type $\alpha$ if the control gets back to the caller at all.  Quantification in the
signature not only provides the clients of the operation with flexibility---they can
instantiate $\alpha$ with any desired type and put a call of \texttt{raise} in any context---but also constrains its
servers in that implementations of the operation have to abstract over types.  In
fact, the type signature of \texttt{raise} is sufficiently restrictive to guarantee
that the exception effect is safe.
Generalizing this idea, we provide a criterion to decide if an effect is safe.
Our criterion is \emph{simple} in that it only mentions the occurrences of bound
type variables $\alpha$ in a type signature, \emph{robust} in that it is
independent of how effects are implemented, and \emph{permissive} in that it is
met by many safe effects---including exception, nondeterminism, and input
streaming.
We call the restriction on type signatures to meet the criterion
\emph{signature restriction}.

We formalize our idea with algebraic effects and
handlers~\cite{Plotkin/Pretnar_2009_ESOP,Plotkin/Pretnar_2013_LMCS}, which are a
programming mechanism to accommodate user-defined control effects in a modular
way.  Algebraic effects and handlers split an effect into an interface (i.e., a
set of operations with type signatures) and an interpretation, so
we can incorporate signature restriction into them
naturally.


%

We provide two polymorphic type assignment systems for a $\lambda$-calculus
equipped with algebraic effects and handlers.  The first is a simple polymorphic type
system based on Curry-style System~F~\cite{Leivant_1983_POPL} (i.e., it supports
implicit, full polymorphism).  This type system allows arbitrary terms (rather than only
values) that invoke effects to be given polymorphic types but is sound, thanks to
signature restriction.  The minimality of this simple type system reveals
the essence of signature restriction.  The second type assignment
system is a polymorphic type-and-effect system.  Using this system, we show that
effect tracking is key to apply signature restriction for programs in which both
safe and potentially unsafe polymorphic effects may happen.\footnote{As
we will show in the paper, signature restriction is permissive and
actually we find no useful effect that invalidates it.  However, the universal
enforcement of signature restriction \emph{may} give rise to inconvenience
in some cases, and we consider the capability of avoiding such (potential)
inconvenience important in designing a general-purpose programming language.}

The contributions of our work are summarized as follows.
\begin{itemize}
 \item We define a $\lambda$-calculus {\lang} with algebraic effects and
       handlers and provide a type system that supports implicit full polymorphism
       and allows any effectful expression to be polymorphic.  We formalize
       signature restriction for {\lang} and prove soundness of the type system
       under the assumption that all operations satisfy signature restriction.

 \item As a technical development to justify signature restriction, we equip the
       type system with Mitchell's type containment~\cite{Mitchell_1988_IC},
       which is an extension of type instantiation.
       In the
       literature~\cite{Peython-Jones/Vytiniotis/Weirich/Shields_2007_JFP,Dunfield/Krishnaswami_2013_ICFP},
       the proof of type soundness of a calculus equipped with type containment
       rests on translation to another calculus, such as
       System~F~\cite{Reynolds_1974_PS,Girard_1972_PhD}.\footnote{The
       translation inserts, as a replacement for type containment, functions that
       are computationally meaningless but work as type conversion statically.}
       By contrast, we show
       soundness of our type system \emph{directly}, i.e., without translation
       to any other calculus.  As far as we know, this is the first work that
       achieves it.

 \item We extend {\lang} and its type system with standard programming features
       such as products, sums, and lists to demonstrate the generality and
       extensibility of signature restriction.

 \item We develop an effect system for {\lang}, which enables a single program
       to use both safe and potentially unsafe polymorphic effects.  In this effect system, an
       expression can be polymorphic if all the effect operations performed
       by the expression satisfy signature restriction.  It also indicates that
       signature restriction can cooperate with value restriction naturally.
\end{itemize}

We employ implicit full polymorphism and type containment to show type
soundness, but either of them makes even type checking
undecidable~\cite{Wells_1994_LICS,Tiuryn/Urzyczyn_1996_LICS}.  It is thus
desirable to identify a subset of our system where type checking---and type
inference as well hopefully---is decidable.  To prove the feasibility of this
idea, we implement an interpreter for a subset of the extended {\lang} in which polymorphism is
restricted to let-polymorphism~\cite{Milner_1978_JCSS,Damas/Milner_1982_POPL}
(the effect system is not supported either).
This restriction on polymorphism ensures that both type checking and type inference
are decidable but it is still expressive so that all of the motivating well-typed
examples in this paper (except for those in \refsec{eff}, which rest on the
effect system) are typechecked.  The implementation is provided as the
supplementary material; alternatively, it can also be found at:
\url{https://github.com/skymountain/MLSR}~.

Finally, we briefly relate our work with the \emph{relaxed} value
restriction~\cite{Garrigue_2004_FLOPS} here.  It is similar to our
signature restriction in that both utilize the occurrences of type variables to
ensure soundness of polymorphic type assignment in the permissive use of
polymorphic effects.  Indeed, a \emph{strong} version of signature restriction
can be justified similarly to the relaxed value restriction.  The strong
signature restriction is, however, too restrictive and rejects many useful, safe
effects.  We generalize it to what we call signature restriction and prove its
correctness with different techniques such as type
containment.  Readers are referred to \refsec{relwork:restrict} for further details.

\paragraph{Organization.}
The remainder of this paper is organized as follows.
We start with an overview of this work (\refsec{overview}) and then define our
base calculus {\lang} (\refsec{lang}).  \refsec{polytype} introduces a
polymorphic type system for {\lang}, formalizes signature restriction, and shows
soundness of the polymorphic type system under the assumption that all
operations satisfy signature restriction.  \refsec{ext} extends {\lang}, the
polymorphic type system, and signature restriction with products, sums, and
lists.  \refsec{eff} presents an effect system to allow programs to use both safe and unsafe effects.
We finally discuss related work in
\refsec{relwork} and conclude in \refsec{conclusion}.

In this paper, we may omit the formal definitions of some well-known notions
and the statements and proofs of auxiliary lemmas for type
soundness.  The full definitions, the full statements, and the full proofs are
provided in \ifappendix{Appendix}\else{the supplementary material}\fi.

%% file: sections/overview.tex
\section{Overview}
\label{sec:overview}

This section presents an overview of our work.  After reviewing
algebraic effects and handlers, their extension to polymorphic
effects, and why a naive extension results in unsoundness, we describe
our approach of signature restriction and informally discuss why it
resolves the unsoundness problem.  All program examples in this paper follow
ML-like syntax.

\subsection{Review: Algebraic Effects and Handlers}
\label{sec:overview:algeff}

Algebraic effects and
handlers~\cite{Plotkin/Pretnar_2009_ESOP,Plotkin/Pretnar_2013_LMCS} are a
mechanism that enables users to define their own effects.  They are successfully
able to separate the syntax and semantics of effects.  The syntax of an
effect is given by a set of \emph{operations}, which are used to trigger the
effect.  For example, exception is triggered by the operation \texttt{raise} and
store manipulation is triggered by \texttt{put} and \texttt{get}, which are used to write to
and read from a store, respectively.  The semantics is given by
\emph{handlers}, which decide how to interpret operations performed by effectful
computation.

Our running example is nondeterministic computation which enumerates
all of the possible outcomes~\cite{Plotkin/Pretnar_2009_ESOP,Plotkin/Pretnar_2013_LMCS}.  This computation utilizes two
operations: \texttt{select}, which chooses an element from a given
list, and \texttt{fail}, which signals that the current control flow
is undesired and the computation should abort.\footnote{This
  describes only \emph{intended} semantics; one can also give an
  \emph{unintended} handler, e.g., one that always returns an integer 42 for a
  call of \texttt{select}.  Certain unintended handlers can be excluded in a
  polymorphic setting, as is shown in \refsec{overview:polyeff}.}
\begin{lstlisting}[numbers=left,xleftmargin=1.3\parindent]
effect select : int list (/$\hookrightarrow$/) int
effect fail   (/\,/) : unit (/$\hookrightarrow$/) unit

let filter (l : int list) (f : int (/$\rightarrow$/) bool) =
  handle
    let x = #select(l) in
    let _ = if f x then () else #fail() in x
  with
    return z (/$\rightarrow$/) [z]
    select l (/$\rightarrow$/) concat (map l ((/$\lambda$/)y. resume y))
    fail z    (/$\rightarrow$/) [(/$\,$/)]

filter [3; 5; 10] ((/$\lambda$/)x. x mod 2 = 1)  (/\quad/)   (* will evaluate to [3; 5] *)
\end{lstlisting}

The first two lines declare the operations \texttt{select} and
\texttt{fail}, which have the type signatures \texttt{int list
  $\hookrightarrow$ int} and \texttt{unit $\hookrightarrow$ unit},
respectively.  A type signature $ \ottnt{A}  \hookrightarrow  \ottnt{B} $ of an operation
signifies that the operation is called with an argument of type
$\ottnt{A}$ and, when the control gets back to the caller, it receives a
value of $\ottnt{B}$.  We refer to $\ottnt{A}$ and $\ottnt{B}$ as the \emph{domain type} and
\emph{codomain type}, respectively.\footnote{The domain and codomain types are
also called the \emph{parameter type} and the \emph{arity},
respectively~\cite{Plotkin/Pretnar_2009_ESOP}.}

The function \texttt{filter} in Lines~4--11 operates \texttt{select}
and \texttt{fail} to filter out the elements of \texttt{l} that do not
meet a given predicate \texttt{f}.  Now, let's take a closer look at
the body of the function, which consists of a single {\handlewith}
expression of the form \texttt{\textbf{handle}} $\ottnt{M}$
\texttt{\textbf{with}} $\ottnt{H}$.  This expression installs a \emph{handler}
$\ottnt{H}$ during the evaluation of $\ottnt{M}$, which we refer to as the
\emph{handled expression}.

The handled expression (Lines~6--7) chooses an integer selected from
\texttt{l} by calling \texttt{select}, tests whether the selected
integer \texttt{x} satisfies \texttt{f}, and returns \texttt{x} if
\texttt{f x} is true; otherwise, it aborts the computation by calling
\texttt{fail}.
We write $ \textup{\texttt{\#}\relax}  \mathsf{op}   \ottsym{(}   \ottnt{M}   \ottsym{)} $ to call operation $\mathsf{op}$ with argument $\ottnt{M}$.
%
%
We now explain the handler in Lines~9--11, which collects all the
values in \texttt{l} that satisfy \texttt{f} as a list, along with
an intuitive meaning of {\handlewith} expressions.

The handler $\ottnt{H}$ in \texttt{\textbf{handle}} $\ottnt{M}$
\texttt{\textbf{with}} $\ottnt{H}$ consists of a
single \emph{return clause} and zero or more \emph{operation clauses}.
The return clause takes the form \texttt{\textbf{return} x $\rightarrow$
  $\ottnt{M}$} and computes the entire result $\ottnt{M}$ of the {\handlewith}
expression using the value of the handled expression, which $\ottnt{M}$
refers to by \texttt{x}.  For example, the return clause in this
example is \texttt{\textbf{return} z $\rightarrow$ [z]}.  Because
\texttt{z} will be bound to the result of the handled expression
\texttt{x}, the entire result is the singleton list consisting of
\texttt{x}.  An operation clause of the form \texttt{op x $\rightarrow$
  $\ottnt{M}$} for an operation \texttt{op} decides how to interpret the
operation \texttt{op} called by the handled expression.  Variable \texttt{x}
will be bound to the argument of the call of \texttt{op} and $\ottnt{M}$
is the entire result of the {\handlewith} expression.  For example,
the operation clause \texttt{fail z $\rightarrow$ [$\,$]} means that, if
\texttt{fail} is called, the computation is aborted---similarly to exception
handling---and the entire {\handlewith} expression returns the empty
list, meaning there is no result that satisfies \texttt{$\mathit{f}$}.


%

Unlike exception handling, which discards the continuation of where an
exception is raised, however, handlers can \emph{resume} computation
from the point at which the operation was called.  The ability to resume a
computation suspended by the operation call provides algebraic
effects and handlers with the expressive power to implement control effects~\cite{Bauer/Pretnar_2015_JLAMP,Leijen_2017_POPL,Forster/Kammar/Lindley/Pretnar_2019_JFP}.
In our language, we use the expression \texttt{\textbf{resume}} $\ottnt{M}$ to
resume the computation of the handled expression with the value of $\ottnt{M}$.
\TS{Can this paragraph be improved?}

The operation clause for \texttt{select} enumerates all the possible
outcomes by using \texttt{\textbf{resume}}.  The clause first returns,
for each integer \texttt{y} of a given list \texttt{l}, the integer
\texttt{y} to the caller of \texttt{select} by resuming the
computation from the point at which \texttt{select} was called.
The handled expression in the example calls \texttt{select} only once, so each
resumed computation (which is performed under the same handler) returns either a
singleton list or the empty list (by calling \texttt{fail}).
The next step after the completion of all the resumed computations is to concatenate
their results.  The two steps are expressed by \texttt{concat (map l
($\lambda$y. \textbf{resume} y))}.


More formally, the suspended computation is expressed as a
\emph{delimited continuation}~\cite{Felleisen_1988_POPL,Danvy/Filinski_1990_LFP}, and \texttt{\textbf{resume}} simply
invokes it.
For example, let us consider evaluating
\lstinline[mathescape]{filter [3; 5; 10] ($\lambda$x. x mod 2 = 1)}
in the last line.  This reduces to the following expression:
\begin{lstlisting}
  handle
    let x = #select([3; 5; 10]) in
    let _ = if ((/$\lambda$/)x. x mod 2 = 1) x then () else #fail() in x
  with (/$H$/)
\end{lstlisting}
where $H$ denotes the same handler as that in the example.
At the call of \texttt{select}, the run-time system constructs the following
delimited continuation $c$
\begin{center}
\begin{tabular}{ccc}
 $c$ & $\defeq$ &
\begin{lstlisting}
handle                    
  let x = (/$[]$/) in
  let _ = if ((/$\lambda$/)x. x mod 2 = 1) x then () else #fail() in x
with (/$H$/)
\end{lstlisting}
\end{tabular}
\end{center}
(where $[]$ is the hole to be filled with resumption arguments), and
then evaluates the operation clause for \texttt{select}.  The resumption
expression in the operation clause invokes the delimited continuation $c$
after filling the hole with an integer in list \texttt{[3; 5; 10]}.  For the
case of filling the hole with \texttt{3}, the remaining computation $c[3]$ to
resume is:
\AI{The notation $c[M]$ has been removed.  Do we need it?}
\begin{lstlisting}
  handle
    let x = 3 in
    let _ = if ((/$\lambda$/)x. x mod 2 = 1) x then () else #fail() in x
  with (/$H$/) (/./)
\end{lstlisting}
Because \texttt{3} is an odd number, it satisfies the predicate
\texttt{($\lambda$x. x \textbf{mod} 2 = 1)}, and therefore the final result of this computation
is the singleton list \texttt{[3]}. The case of \texttt{5} behaves similarly
and produces \texttt{[5]}.  In the case of \texttt{10}, because the even number
\texttt{10} does not meet the given predicate, the remaining computation
$c[\texttt{10}]$ would call \texttt{fail} and, from the operation clause for
\texttt{fail}, the final result of $c[\texttt{10}]$ would be the empty list.
The operation clause for \texttt{select} concatenates all of these resulting
lists of the resumptions and finally returns \texttt{[3; 5]}.  This is
the behavior that we expect of \texttt{filter} exactly.

\TS{Can we say an operation is able to be called twice or more times?}  The
handler in the example works even when \texttt{select} is called more than once,
e.g.:
\begin{lstlisting}
  handle
    let x = #select([2; 3]) in
    let y = #select([10; 20]) in
    let z = x * y in
    let _ = if z > 50 then #fail() else () in z
  with (/$H$/) (/./)
\end{lstlisting}
This program
returns a list of the values of the handled expression that are computed with
$(\texttt{x}, \texttt{y}) \in \{ \texttt{2}, \texttt{3} \} \times \{
\texttt{10}, \texttt{20} \} $ such that the multiplication \texttt{x * y} does not
exceed \texttt{50}.

\paragraph{Typechecking.}
We also review the procedure to typecheck an operation clause \texttt{op x $\rightarrow$
$\ottnt{M}$} for \texttt{op} of type signature $\ottnt{A} \hookrightarrow \ottnt{B}$.  Since
the operation \texttt{op} is called with an argument of $\ottnt{A}$, the
typechecking assigns argument variable \texttt{x} type $\ottnt{A}$.  As the value
of $\ottnt{M}$ is the result of the entire {\handlewith} expression, the
typechecking checks $\ottnt{M}$ to have the same type as the other clauses including
the return clause.  The typechecking of resumption expressions
\texttt{\textbf{resume}} $\ottnt{M'}$ is performed as follows.  Since the value of
$\ottnt{M'}$ will be used as a result of calling \texttt{op} in a handled
expression, $\ottnt{M'}$ has to be of the type $\ottnt{B}$, the codomain type of the type
signature of \texttt{op}.  On the other hand, since the resumption expression
returns the evaluation result of the entire {\handlewith} expression, the
typechecking assumes it to have the same type as all of the clauses in the
handler.

For example, let us consider the typechecking of the operation clause for
\texttt{select} in the function \texttt{filter}.  Since the type signature of
\texttt{select} is \texttt{int list $\hookrightarrow$ int}, the variable
\texttt{l} is assigned the type \texttt{int list}.
Here, we suppose \texttt{map} and \texttt{concat} to have the following types:
\begin{center}
 \begin{tabular}{ccl}
  \texttt{map}   & : &
   \texttt{int list $\rightarrow$ (int $\rightarrow$ int list) $\rightarrow$ int list list} \\
  \texttt{concat} & : &
   \texttt{int list list $\rightarrow$ int list}
 \end{tabular}
\end{center}
(these types can be inferred automatically).
The type of \texttt{map} requires that the arguments \texttt{l} and
\lstinline[mathescape]{$\lambda$y.resume y} have the types \texttt{int list} and
\texttt{int $\rightarrow$ int list}, respectively, and they \emph{do} indeed.
The requirement to \texttt{l} is met by the type assigned to \texttt{l}.
We can derive that \lstinline[mathescape]{$\lambda$y.resume y} has type
\texttt{int $\rightarrow$ int list} as follows: first, the typechecking assigns
the bound variable \texttt{y} type \texttt{int} and checks \lstinline{resume y}
to have \texttt{int list}.  An argument of a resumption expression has to be of
the type \texttt{int}, which is the codomain type of the type signature, and \texttt{y}
has that type indeed. Then, the typechecking assumes that \lstinline{resume y}
has the same type as the clauses of the handler, which is the type \texttt{int
list}.  Thus, \lstinline[mathescape]{$\lambda$y.resume y} has the desired type.

\subsection{Polymorphic Effects}
\label{sec:overview:polyeff}
Polymorphic effects are a particular kind of effects that incorporate
polymorphism,\footnote{Another way to incorporate polymorphism is \emph{parameterized} effects, where the
declaration of an operation is parameterized over types~\cite{Wadler_1992_POPL}.} providing a
set of operations with \emph{polymorphic} type signatures.  We also call such
operations \emph{polymorphic}.

For example, we can assign \texttt{select} and \texttt{fail} polymorphic
signatures and write the program as follows:
\begin{lstlisting}[numbers=left,xleftmargin=1.3\parindent]
effect select : (/$\forall \alpha.$/) (/$\alpha$/) list (/$\hookrightarrow$/) (/$\alpha$/)
effect fail   (/\,/) : (/$\forall \alpha.$/) unit (/$\hookrightarrow$/) (/$\alpha$/)

handle
  let b = #select([true; false])
  let x = if b then #select([2; 3]) else #select([20; 30]) in
  if x > 20 then #fail() else x
with
  return z (/$\rightarrow$/) [z]
  select l (/$\rightarrow$/) concat (map l ((/$\lambda$/)y. resume y))
  fail z    (/$\rightarrow$/) [(/$\,$/)]
\end{lstlisting}
This program evaluates to the list \texttt{[2; 3; 20]} (\texttt{30} is
filtered out by the call of \texttt{fail}).

Polymorphic type signatures enable operation calls with arguments of different
types.  For example, \texttt{\#select([2; 3])} and
\texttt{\#select([true; false])} are legal operation calls that instantiate the
bound type variable $\alpha$ of the type signature with \texttt{int} and
\texttt{bool}, respectively.  The calls of the same operation are
handled by the same operation clause, even if the calls involve different type
instantiations.  It is also interesting to see that the use of polymorphic type signatures makes
programs more natural and succinct:  Thanks to its polymorphic
codomain type, a call to \texttt{fail} can be put anywhere, making it
possible to put \texttt{x} in the \text{else}-branch, unlike the
monomorphic case.

Another benefit of polymorphic type signatures is that they contribute to the
exclusion of undesired operation implementations.  For example, the polymorphic
signature of \texttt{fail} ensures that, once we call \texttt{fail}, the control
\emph{never} gets back and that of \texttt{select} ensures that no other values
than elements in an argument list are chosen.
Parametricity~\cite{Reynolds_IFIP_1983} enables formal reasoning for this; readers are referred to
\citet{Biernacki/Pirog/Polesiuk/Sieczkowski_2020_POPL} for parametricity with
the support for polymorphic algebraic effects and handlers.

\subsection{(Naive) Polymorphic Typechecking and Its Unsoundness}
\label{sec:overview:unsafety}

(Naive) typechecking of operation clauses for polymorphic operations is obtained by
extending the monomorphic setting.  The only difference is that the operation clauses have
to abstract over types.  Namely, an operation clause \texttt{op x $\rightarrow$
$\ottnt{M}$} for \texttt{op} of polymorphic type signature $\forall \alpha.\, \ottnt{A}
\hookrightarrow \ottnt{B}$ is typechecked as follows.  The typechecking process
allocates a fresh type variable $\alpha$, which is bound in $\ottnt{M}$, and assigns
variable \texttt{x} type $\ottnt{A}$ (which may refer to the bound type variable
$\alpha$).  Resumption expressions \lstinline{resume} $\ottnt{M'}$ in $\ottnt{M}$ are
typechecked as in the monomorphic setting; that is, the typechecking checks
$\ottnt{M'}$ to be of $\ottnt{B}$ (which may refer to $\alpha$) and assumes the
resumption expressions to have the same type as the clauses in the handler.
Finally, the typechecking checks whether $\ottnt{M}$ is of the same type as the other
clauses in the handler.
It is easy to see that the polymorphic version of the \texttt{select} and
\texttt{fail} example typechecks.


However, this naive extension is unsound.  In what follows, we revisit the counterexample given
by \citet{Sekiyama/Igarashi_2019_ESOP}, which is an analogue to that found by
\citet{Harper/Lillibridge_1993_LSC,Harper/Lillibridge_1991_types} with
\textsf{call/cc}~\cite{Clinger/Friedman/Wand_1985}.
\begin{lstlisting}[numbers=left,xleftmargin=1.3\parindent]
effect get_id : (/$\forall \alpha.$/) unit (/$\hookrightarrow$/) (/$(\alpha \rightarrow \alpha)$/)

handle
  let f = #get_id() in   (* f : (/$\forall \alpha.\,\alpha \rightarrow \alpha$/) *)
  if (f true) then ((f 0) + 1) else 2
with
  return x (/$\rightarrow$/) x
  get_id x (/$\rightarrow$/) resume ((/$\lambda$/)z1. let _ = resume ((/$\lambda$/)z2. z1) in z1)
\end{lstlisting}

We first check that this program is well typed.  The handled expression first
binds the variable \texttt{f} to the result returned by \texttt{get\_id}.
In polymorphic type assignment, we can assign a polymorphic type $\forall \alpha.\,
\alpha \rightarrow \alpha$ to \texttt{f} by allocating a fresh type variable
$\alpha$, instantiating the type signature of \texttt{get\_id} with $\alpha$,
and generalizing $\alpha$ finally.  The polymorphic type of \texttt{f} allows
viewing \texttt{f} both as a function of the type \texttt{bool $\rightarrow$ bool} and
of the type \texttt{int $\rightarrow$ int}.  Thus, the handled expression is well typed.
Turning to the operation clause, since the type signature of \texttt{get\_id} is
\texttt{$\forall \alpha.$ unit $\hookrightarrow \alpha \rightarrow \alpha$},
typechecking first allocates a fresh type variable $\alpha$ and assigns the argument
variable \texttt{x} the type \texttt{unit}.  The signature also requires the
arguments of the resumption expressions to have the type $\alpha \rightarrow \alpha$, and
both arguments \texttt{$\lambda$z1. $...$ z1} and
\texttt{$\lambda$z2. z1} do indeed.  The latter function is typed at $\alpha
\rightarrow \alpha$ because the requirement to the former ensures that \texttt{z1} has
$\alpha$.  Thus, the entire program is well typed.

\TS{Remove the following?}

However, this program gets stuck.  The evaluation starts with the call of
\texttt{get\_id} in the handled expression.  It constructs the following
delimited continuation:
\begin{center}
\begin{tabular}{ccc}
 $c$ & $\defeq$ &
\begin{lstlisting}
  handle
    let f = (/$[]$/) in
    if (f true) then ((f 0) + 1) else 2
  with
    return x (/$\rightarrow$/) x
    get_id x (/$\rightarrow$/) resume ((/$\lambda$/)z1. let _ = resume ((/$\lambda$/)z2. z1) in z1) (/$.$/)
\end{lstlisting}
\end{tabular}
\end{center}
The run-time system then replaces the resumption expressions with the invocation
of the delimited continuation.  Namely, the entire program evaluates to
\begin{center}
\begin{tabular}{ccc}
 $M$ & $\defeq$ &
\begin{lstlisting}
(/$c[\lambda$/)z1. let _ = (/$c[\lambda$/)z2. z1(/$]$/) in z1(/$]$/) (/$.$/)
\end{lstlisting}
\end{tabular}
\end{center}
The evaluation of $M$ proceeds as follows.

\begin{tabular}{ccl}
 $M$ & $=$ &
\begin{lstlisting}
handle
  let f = ((/$\lambda$/)z1. let _ = (/$c[\lambda$/)z2. z1(/$]$/) in z1) in
  if (f true) then ((f 0) + 1) else 2
with (/$...$/)
\end{lstlisting} \\
     & $\longrightarrow$ &
\begin{lstlisting}
handle if ((/$\lambda$/)z1. let _ = (/$c[\lambda$/)z2. z1(/$]$/) in z1) true then (/$...$/) with (/$...$/)
\end{lstlisting} \\
     & $\longrightarrow$ &
\begin{lstlisting}
handle if (let _ = (/$c[\lambda$/)z2. true(/$]$/) in true) then (/$...$/) with (/$...$/)
\end{lstlisting}
\end{tabular} \\[.5ex]

\noindent
Subsequently, the term \lstinline[mathescape]{$c[\lambda$z2. true$]$} is to be
evaluated.  The delimited continuation $c$ expects the hole to be
filled with a polymorphic function of $\forall \alpha.\, \alpha \rightarrow
\alpha$ but the function \lstinline[mathescape]{$\lambda$z2. true} is
\emph{not} polymorphic.  As a result, the term gets stuck:
\begin{center}
\begin{tabular}{ccl}
 \texttt{$c[\lambda$z2. \textbf{true}$]$}
& $=$ &
\begin{lstlisting}
handle
  let f = (/$\lambda$/)z2. true in
  if (f true) then ((f 0) + 1) else 2
with (/$...$/)
\end{lstlisting} \\
     & $\longrightarrow^*$ &
\begin{lstlisting}
handle (((/$\lambda$/)z2. true) 0) + 1 with (/$...$/)
\end{lstlisting} \\
     & $\longrightarrow$ &
\begin{lstlisting}
handle true + 1 with (/$...$/)
\end{lstlisting} \\
\end{tabular}
\end{center}

A standard approach to this problem is to restrict operation calls in
polymorphic
expressions~\cite{Tofte_1990_IC,Leroy/Weis_1991_POPL,Appel/MacQueen_1991_PLILP,Hoang/Mitchell/Viswanathan_1993_LICS,Wright_1995_LSC,Garrigue_2004_FLOPS,Asai/Kameyama_2007_APLAS}.
While this kind of approach prevents \texttt{\#get\_id()} from having a
polymorphic type, it disallows calls of any polymorphic operation inside
polymorphic expressions even when the calls are safe; refer to
\citet{Sekiyama/Igarashi_2019_ESOP} for further discussion.
\citet{Sekiyama/Igarashi_2019_ESOP} have proposed a complementary approach to
this problem, that is, restricting, by typing, the handler of a polymorphic
operation, instead of restricting handled expressions.


\input{sections/another_exp}

%% file: sections/another_exp.tex
\subsection{Our Work: Signature Restriction}
\label{sec:overview:our-work}

\AI{We don't use $<:$.}
\TS{''Typable'' --> ``typeable''}
This work takes a new approach to ensuring the safety of any call of an operation.
Instead of restricting how it is used or implemented, we restrict its
type signature: An operation \texttt{op} : $\forall \alpha.\, A \hookrightarrow
B$ should not have a ``bad'' occurrence of $\alpha$ in $A$ and $B$.  We refer to
this restriction as \emph{signature restriction}.

To see how the signature restriction works, let us explain why type preservation
is not easy to prove with the following example, where type abstraction
$\Lambda \beta.\, M$ and type application $M\{A\}$ are explicit for the ease of
readability:
\begin{center}
\begin{tabular}{c}
\begin{lstlisting}
handle let f = (/$\Lambda \beta.\,$/)#op(/$\{\beta\}$/)((/$v$/)) in (/$M$/) with (/$H$/) (/$.$/)
\end{lstlisting}
\end{tabular}
\end{center}
Here, we suppose the type signature of \texttt{op} to be $\forall \alpha.\, A
\hookrightarrow B$.
Notice that the type variable $\alpha$ in the signature $\forall \alpha.\, A \hookrightarrow B$ is instantiated to $\beta$, which is locally bound by $\Lambda \beta$.
Handling of operation \texttt{op} constructs the following delimited continuation:
\begin{center}
\begin{tabular}{ccc}
 $c$ & $\defeq$ &
\begin{lstlisting}
handle let f = (/$\Lambda \beta.\,[]$/) in (/$M$/) with (/$H$/) (/$.$/)
\end{lstlisting}
\end{tabular}
\end{center}
The problem is that an appropriate type cannot be assigned to it under the typing
context of the handler $H$: the type of the hole should be $B[\beta/\alpha]$,
but the type variable $\beta$ is not in the scope of $H$.  This is a
kind of \emph{scope extrusion}.  We have focused on the scope extrusion
via the continuation, but the operation argument $v$ may cause a similar problem
when its type $A[\beta/\alpha]$ contains $\beta$.

This analysis suggests that instantiating polymorphic operations with
\emph{closed types}, i.e., types without free type variables (especially $\beta$
here), is safe because then the types of the hole and the operation argument
should not contain $\beta$ and, thus, the continuation and the argument could be
typed under the typing context of $H$.\footnote{More precisely, the argument may
contain free type variables even when its type does not.  However, we could address
this situation successfully by eliminating them with closing type substitution
as in \citet{Sekiyama/Igarashi_2019_ESOP}.}
However, allowing only instantiation with closed types is too restrictive.  For
example, it even disallows a function wrapping \texttt{select}, \texttt{$\lambda
x.\,$\texttt{\#select}($x$)}, to have a polymorphic type {$\forall \alpha.\,
\alpha \, \texttt{list} \rightarrow \alpha$} because, for the function to have
this type, the bound type variable of the type signature of \texttt{select} has to be
instantiated with a non-closed type $\alpha$.

As another approach to addressing the scope extrusion, we introduce \emph{strong
signature restriction}, which requires that, for each polymorphic operation
\texttt{op} : $\forall \alpha.\, A \hookrightarrow B$, the type variable
$\alpha$ occur \emph{only negatively} in $A$ and \emph{only positively} in $B$.
This is a sufficient condition to prove type preservation.
Consider, for example, the expression
\begin{flushleft}
 \qquad\qquad
\begin{tabular}{ccc}
  $M_1$ & $\defeq$ &
\begin{lstlisting}
handle let f = (/$\Lambda \beta_1 \dots \beta_n.\,$/)#op(/$\{C\}$/)((/$v$/)) in (/$M$/) with (/$H$/)
\end{lstlisting}
\end{tabular}
\end{flushleft}
where $v$ is a value and $C$ is a type with free type variables $\beta_1, \dots, \beta_n$.
The idea is to rewrite this expression, immediately before the call of \texttt{op}, to
\begin{flushleft}
 \qquad\qquad
\begin{tabular}{ccc}
  $M_1'$ & $\defeq$ &
\begin{lstlisting}
handle let f = (/$\Lambda \beta_1 \dots \beta_n.\,$/)#op(/$\{ \, \graybox{\forall \beta_1 \dots \beta_n.} \, C\}$/)((/$v$/)) in (/$M$/) with (/$H$/)
\end{lstlisting}
\end{tabular}
\end{flushleft}
(where the rewritten part is shaded).
In $M_1'$, because the type variable $\alpha$ in
\texttt{op} : $\forall \alpha.\, A \hookrightarrow B$ is instantiated with a
closed type $\forall \beta_1 \dots \beta_n.\, C$, this operation call should be
safe provided that $M_1'$ is well typed.  This expression is indeed typable
if the strong signature restriction is enforced, as seen below.
To ensure that $M_1'$ is typable, we need to have
\begin{center}
\begin{tabular}{r@{\ }c@{\ }ll}
 $v$ &:& $A[\forall \beta_1 \dots \beta_n.\,C/\alpha]$ & (for typing \texttt{\#op} $\{\forall \beta_1 \dots \beta_n.\, C\}$\texttt{(}$v$\texttt{)}) \\
 \texttt{\#op} $\{\forall \beta_1 \dots \beta_n.\, C\}$\texttt{(}$v$\texttt{)} &:& $B[C/\alpha]$ & (for type preservation) ~.
\end{tabular}
\end{center}
To this end, we employ \emph{type containment}~\cite{Mitchell_1988_IC}, which
is also known as ``subtyping for second-order
types''~\cite{Tiuryn/Urzyczyn_1996_LICS}.
Type containment $ \sqsubseteq $ accepts the following key judgments:
\begin{align*}
  A[C/\alpha] & \sqsubseteq  A[\forall \beta_1 \dots \beta_n.\,C/\alpha] \\
  B[\forall \beta_1 \dots \beta_n.\,C/\alpha] & \sqsubseteq  B[C/\alpha]\ ,
\end{align*}
which follow from the acceptance of type instantiation $(\forall \beta_1 \dots \beta_n.\,C)  \sqsubseteq  C$ and the strong signature restriction which assumes that $\alpha$ occurs only negatively in $A$ and only positively in $B$.
Since $M_1$ is typable, we have $v : A[C/\alpha]$ and, by subsumption, $v : A[\forall \beta_1 \dots \beta_n.\,C/\alpha]$.
Therefore, the operation
\texttt{\#op$\{\forall \beta_1 \dots \beta_n.\, C\}$}
is applicable to $v$ and we have
\texttt{\#op$\{\forall \beta_1 \dots \beta_n.\, C\}$($v$)${}: B[\forall \beta_1 \dots \beta_n.\,C/\alpha]$}.
Again, by subsumption,
\texttt{\#op$\{\forall \beta_1 \dots \beta_n.\, C\}$($v$)${}: B[C/\alpha]$}
as desired.
%
Therefore, $M_1'$ is also typable.
Note that the translation from $M_1$ to $M_1'$ does not change the underlying untyped term, but only the types of (sub)expressions; hence, if $M'_1$ does not get stuck, neither does $M_1$.

However, the strong signature restriction is still unsatisfactory in that the type signatures of many operations do not conform to it.
For example, the signature of \texttt{select} : $\forall \alpha.\, \alpha \, \texttt{list} \hookrightarrow \alpha$ in \refsec{overview:polyeff} does not satisfy the requirements for the strong signature restriction;
it disallows positive occurrences of a bound type variable in the left-hand side of $\hookrightarrow$.

\emph{Signature restriction} is a relaxation of the strong signature restriction, allowing the type variable $\alpha$ in the signature $\forall \alpha.\, A \hookrightarrow B$ to occur at \emph{strictly positive} positions in $A$ in addition to negative positions.
The proof of type preservation in this generalized case is essentially the same
as above, but we need an additional type containment rule, known as the distributive law:
\[
 \forall \alpha.\, A \rightarrow B   \sqsubseteq  A \rightarrow \forall \alpha.\, B \quad \text{(if $\alpha$ does not occur free in $A$)}
\]
to derive type containment judgments such as those derived above.
The type signature of \texttt{select} conforms to this relaxed condition---$\alpha$ only
occurs at a strictly positive position in the domain type $\alpha \,
\texttt{list}$; therefore, we can ensure the safety of the operation call of
\texttt{select} in polymorphic expressions.

The signature restriction is a reasonable relaxation in that it rejects unsafe
operations as expected.  For example, \texttt{get\_id} does not conform to the
signature restriction because, in its type signature \texttt{$\forall \alpha.$
unit $\hookrightarrow \alpha \rightarrow \alpha$}, the bound type variable
$\alpha$ occurs negatively in the codomain type $\alpha \rightarrow \alpha$.

\AI{Comparison with \citet{Sekiyama/Igarashi_2019_ESOP} has been dropped.  We may want to move the paragraph to somewhere else.}

%% file: sections/lang.tex
\section{A $\lambda$-Calculus with Algebraic Effects and Handlers}
\label{sec:lang}
This section defines the syntax and semantics of our base language {\lang}, a
$\lambda$-calculus extended with algebraic effects and handlers.  They are based
on those of the core calculus of the language Koka~\cite{Leijen_2017_POPL}.  The
only difference is that the Koka core calculus is equipped with let-expressions
whereas {\lang} is not because we focus on implicit full polymorphism, rather than only on
let-polymorphism.
We will
present a polymorphic type system for {\lang} that takes into account signature
restriction in \refsec{polytype}.  We also extend {\lang} and the polymorphic type
system with products, sums, and lists in \refsec{ext} and provide an effect system
for the extended language in \refsec{eff}.
Signature restriction differentiates these systems from the typing discipline of
Koka.  Besides, contrary to Koka's effect system, which is row-based, our effect system is
not; refer to \refsec{eff} for detail.

\subsection{Syntax}

\begin{figure}[t]
 \[
 \begin{array}{lll}
  \multicolumn{3}{l}{
   \textbf{Variables} \quad \mathit{x}, \mathit{y}, \mathit{z}, \mathit{f}, \mathit{k} \qquad
   \textbf{Effect operations} \quad \mathsf{op}
  } \\[.5ex]
  \textbf{Constants} \quad \ottnt{c} & ::= &
    \mathsf{true}  \mid  \mathsf{false}  \mid  0  \mid  \mathsf{+}  \mid ... \\[.5ex]
  \textbf{Terms} \quad \ottnt{M} & ::= &
   \mathit{x} \mid \ottnt{c} \mid
    \lambda\!  \, \mathit{x}  \ottsym{.}  \ottnt{M} \mid \ottnt{M_{{\mathrm{1}}}} \, \ottnt{M_{{\mathrm{2}}}} \mid
    \textup{\texttt{\#}\relax}  \mathsf{op}   \ottsym{(}   \ottnt{M}   \ottsym{)}  \mid
   \mathsf{handle} \, \ottnt{M} \, \mathsf{with} \, \ottnt{H}
   \\[.5ex]
  \textbf{Handlers} \quad \ottnt{H} & ::= &
   \mathsf{return} \, \mathit{x}  \rightarrow  \ottnt{M} \mid \ottnt{H}  \ottsym{;}  \mathsf{op}  \ottsym{(}  \mathit{x}  \ottsym{,}  \mathit{k}  \ottsym{)}  \rightarrow  \ottnt{M}
   \\[.5ex]
  \textbf{Values} \quad \ottnt{v} & ::= &
   \ottnt{c} \mid  \lambda\!  \, \mathit{x}  \ottsym{.}  \ottnt{M}
   \\[.5ex]
  \textbf{Evaluation contexts} \quad
   \ottnt{E} & ::= &  []  \mid
                 \ottnt{E} \, \ottnt{M_{{\mathrm{2}}}} \mid \ottnt{v_{{\mathrm{1}}}} \, \ottnt{E} \mid
                  \textup{\texttt{\#}\relax}  \mathsf{op}   \ottsym{(}   \ottnt{E}   \ottsym{)}  \mid \mathsf{handle} \, \ottnt{E} \, \mathsf{with} \, \ottnt{H}
   \\[.5ex]
 \end{array}
 \]
 \caption{Syntax of {\lang}.}
 \label{fig:syntax}
\end{figure}

\reffig{syntax} presents the syntax of {\lang}.  We use the metavariables $\mathit{x}$,
$\mathit{y}$, $\mathit{z}$, $\mathit{f}$, $\mathit{k}$ for variables and $\mathsf{op}$
for effect operations.  Our language {\lang} is parameterized over constants,
which are ranged over by $\ottnt{c}$ and may include basic values, such as Boolean
and integer values, and basic operations for them, such as \textsf{not}, $+$, $-$,
\textsf{mod}, etc.

Terms, ranged over by $\ottnt{M}$, are from the $\lambda$-calculus, augmented with
constructors for algebraic effects and handlers.  They are composed of:
variables; constants; lambda abstractions $ \lambda\!  \, \mathit{x}  \ottsym{.}  \ottnt{M}$, where variable $\mathit{x}$ is
bound in $\ottnt{M}$; function applications $\ottnt{M_{{\mathrm{1}}}} \, \ottnt{M_{{\mathrm{2}}}}$; operation calls $ \textup{\texttt{\#}\relax}  \mathsf{op}   \ottsym{(}   \ottnt{M}   \ottsym{)} $
with arguments $\ottnt{M}$; and {\handlewithsf} expressions $\mathsf{handle} \, \ottnt{M} \, \mathsf{with} \, \ottnt{H}$,
which install handler $\ottnt{H}$ to interpret effect operations performed by
$\ottnt{M}$.  A resumption expression \lstinline[mathescape]{resume $\ottnt{M}$} that
appears in \refsec{overview} is the syntactic sugar of function application $\mathit{k} \, \ottnt{M}$ where $\mathit{k}$ is a variable that denotes delimited continuations and is
introduced by an operation clause in a handler (we will see the definition of operation
clauses shortly).  The definition of evaluation contexts, ranged over by $\ottnt{E}$, is standard; it indicates that the semantics of {\lang} is call-by-value and terms
are evaluated from left to right.

Handlers, ranged over by $\ottnt{H}$, consist of a single return clause and
zero or more operation clauses.
A return clause takes the form $\mathsf{return} \, \mathit{x}  \rightarrow  \ottnt{M}$, where $\mathit{x}$ is bound in
$\ottnt{M}$.  The body $\ottnt{M}$ is evaluated once a handled expression produces
a value, to which $\mathit{x}$ is bound in $\ottnt{M}$.  An operation clause $\mathsf{op}  \ottsym{(}  \mathit{x}  \ottsym{,}  \mathit{k}  \ottsym{)}  \rightarrow  \ottnt{M}$, where $\mathit{x}$ and $\mathit{k}$ are bound in $\ottnt{M}$, is an implementation of
the effect operation $\mathsf{op}$.  The body $\ottnt{M}$ is evaluated once a
handled expression performs $\mathsf{op}$, referring to the argument of $\mathsf{op}$ by
$\mathit{x}$.  Variable $\mathit{k}$ denotes the delimited continuation from the point
where $\mathsf{op}$ is called up to the {\handlewithsf} expression that installs the
operation clause.  This ability to manipulate delimited continuations enables
the implementation of various control effects.
In this paper, we suppose that a handler may contain at most one operation
clause for each operation.

Here, we introduce a few notions about syntax; they are standard, and therefore we
omit their formal definitions.
Term $ \ottnt{M_{{\mathrm{1}}}}    [  \ottnt{M_{{\mathrm{2}}}}  /  \mathit{x}  ]  $ is the one obtained by substituting $\ottnt{M_{{\mathrm{2}}}}$ for
$\mathit{x}$ in $\ottnt{M_{{\mathrm{1}}}}$ in a capture-avoiding manner.
A term $\ottnt{M}$ is closed if it has no free variable.
We also write $ \ottnt{E}  [  \ottnt{M}  ] $ and $ \ottnt{E}  [  \ottnt{E'}  ] $ for the term and evaluation context
obtained by filling the hole of $\ottnt{E}$ with $\ottnt{M}$ and $\ottnt{E'}$, respectively.

\TS{Show an example in the overview section?}

\subsection{Semantics}
\label{sec:lang:semantics}

\begin{figure}[t]
 \begin{flushleft}
  \textbf{Reduction rules} \quad \framebox{$\ottnt{M_{{\mathrm{1}}}}  \rightsquigarrow  \ottnt{M_{{\mathrm{2}}}}$}
 \end{flushleft}
 \[\begin{array}{rcll}
  \ottnt{c} \, \ottnt{v}                    &  \rightsquigarrow  &  \zeta  (  \ottnt{c}  ,  \ottnt{v}  )    & \RwoP{Const} \\[1ex]
  \ottsym{(}   \lambda\!  \, \mathit{x}  \ottsym{.}  \ottnt{M}  \ottsym{)} \, \ottnt{v}               &  \rightsquigarrow  &  \ottnt{M}    [  \ottnt{v}  /  \mathit{x}  ]        & \RwoP{Beta} \\[1ex]
  \mathsf{handle} \, \ottnt{v} \, \mathsf{with} \, \ottnt{H}        &  \rightsquigarrow  &  \ottnt{M}    [  \ottnt{v}  /  \mathit{x}  ]   \quad
   \text{(where $ \ottnt{H} ^\mathsf{return}  \,  =  \, \mathsf{return} \, \mathit{x}  \rightarrow  \ottnt{M}$)}          & \RwoP{Return} \\[1ex]
  \mathsf{handle} \,  \ottnt{E}  [   \textup{\texttt{\#}\relax}  \mathsf{op}   \ottsym{(}   \ottnt{v}   \ottsym{)}   ]  \, \mathsf{with} \, \ottnt{H} &  \rightsquigarrow  &   \ottnt{M}    [  \ottnt{v}  /  \mathit{x}  ]      [   \lambda\!  \, \mathit{y}  \ottsym{.}  \mathsf{handle} \,  \ottnt{E}  [  \mathit{y}  ]  \, \mathsf{with} \, \ottnt{H}  /  \mathit{k}  ]  
                                                        & \RwoP{Handle} \\
   \multicolumn{3}{r}{
    \text{(where $\mathsf{op} \,  \not\in  \, \ottnt{E}$ and $\ottnt{H}  \ottsym{(}  \mathsf{op}  \ottsym{)} \,  =  \, \mathsf{op}  \ottsym{(}  \mathit{x}  \ottsym{,}  \mathit{k}  \ottsym{)}  \rightarrow  \ottnt{M}$)}
   } \\[1ex]
 \end{array}\]
 \begin{flushleft}
  \textbf{Evaluation rules} \quad \framebox{$\ottnt{M_{{\mathrm{1}}}}  \longrightarrow  \ottnt{M_{{\mathrm{2}}}}$}
 \end{flushleft}
 \begin{center}
  $\ottdruleEXXEval{}$
 \end{center}
 \caption{Semantics of {\lang}.}
 \label{fig:semantics}
\end{figure}

This section defines the semantics of {\lang}.  It consists of two binary
relations over closed terms: the reduction relation $ \rightsquigarrow $, which gives the notion
of basic computation such as $\beta$-reduction, and the evaluation relation
$ \longrightarrow $, which defines how to evaluate programs.  These relations are defined
by the rules shown in \reffig{semantics}.

The reduction relation is defined by four rules.  The rule \R{Const} is for
constant applications.  The denotations of functional constants are given by
$ \zeta $, which is a mapping from pairs of a constant $\ottnt{c}$ and a value
$\ottnt{v}$ to the value that is the result of applying $\ottnt{c}$ to $\ottnt{v}$.  A function
application $\ottsym{(}   \lambda\!  \, \mathit{x}  \ottsym{.}  \ottnt{M}  \ottsym{)} \, \ottnt{v}$ reduces to $ \ottnt{M}    [  \ottnt{v}  /  \mathit{x}  ]  $, as usual, by \R{Beta}.
The other two rules are for computation in terms of algebraic effects and
handlers.  The rule \R{Return} is for the case in which a handled expression
evaluates to a value.  In such a case, the return clause of the installed
handler is evaluated with the value of the handled expression.  We write
$ \ottnt{H} ^\mathsf{return} $ for the return clause of a handler $\ottnt{H}$.  The rule \R{Handle} is
the core of effectful computation in algebraic effects and handlers, and looks for
an operation clause to interpret an operation invoked by a handled expression.
The redex is a {\handlewithsf} expression that takes the form
$\mathsf{handle} \,  \ottnt{E}  [   \textup{\texttt{\#}\relax}  \mathsf{op}   \ottsym{(}   \ottnt{v}   \ottsym{)}   ]  \, \mathsf{with} \, \ottnt{H}$ where the handled expression $ \ottnt{E}  [   \textup{\texttt{\#}\relax}  \mathsf{op}   \ottsym{(}   \ottnt{v}   \ottsym{)}   ] $
performs the operation $\mathsf{op}$ and $\ottnt{E}$ does not install handlers to interpret it.  We call evaluation contexts that install no handler to
interpret $\mathsf{op}$ \emph{$\mathsf{op}$-free}, which is formally
defined as follows.
\begin{restatable}[$\mathsf{op}$-free evaluation contexts]{defn}{defnOpFreeEvalCtx}
 Evaluation context $\ottnt{E}$ is \emph{$\mathsf{op}$-free}, written $\mathsf{op} \,  \not\in  \, \ottnt{E}$,
 if and only if, there exist no $\ottnt{E_{{\mathrm{1}}}}$, $\ottnt{E_{{\mathrm{2}}}}$, and $\ottnt{H}$ such that
 $\ottnt{E} \,  =  \,  \ottnt{E_{{\mathrm{1}}}}  [  \mathsf{handle} \, \ottnt{E_{{\mathrm{2}}}} \, \mathsf{with} \, \ottnt{H}  ] $ and $\ottnt{H}$ has an operation clause for
 $\mathsf{op}$.
\end{restatable}
We also denote the operation clause for $\mathsf{op}$ in $\ottnt{H}$ by $\ottnt{H}  \ottsym{(}  \mathsf{op}  \ottsym{)}$.
Then, the conjunction of $\mathsf{op} \,  \not\in  \, \ottnt{E}$ and $\ottnt{H}  \ottsym{(}  \mathsf{op}  \ottsym{)} \,  =  \, \mathsf{op}  \ottsym{(}  \mathit{x}  \ottsym{,}  \mathit{k}  \ottsym{)}  \rightarrow  \ottnt{M}$ in
\R{Handle} means that the operation clause $\mathsf{op}  \ottsym{(}  \mathit{x}  \ottsym{,}  \mathit{k}  \ottsym{)}  \rightarrow  \ottnt{M}$ installed by the
{\handlewithsf} expression is the innermost among the operation clauses for $\mathsf{op}$ from the
point at which $\mathsf{op}$ is invoked.
The {\handlewithsf} expression with such an operation clause reduces to the body
$\ottnt{M}$ of the operation clause after substituting the argument $\ottnt{v}$ of the
operation call for $\mathit{x}$ and the functional representation of the delimited
continuation $ \lambda\!  \, \mathit{y}  \ottsym{.}  \mathsf{handle} \,  \ottnt{E}  [  \mathit{y}  ]  \, \mathsf{with} \, \ottnt{H}$ for $\mathit{k}$.

The evaluation proceeds according to the evaluation rule \E{Eval} in
\reffig{semantics}.  A program is decomposed into the evaluation context $\ottnt{E}$
and the redex $\ottnt{M_{{\mathrm{1}}}}$ and evaluates to the term $ \ottnt{E}  [  \ottnt{M_{{\mathrm{2}}}}  ] $ obtained by
filling the hole of $\ottnt{E}$ with the resulting term $\ottnt{M_{{\mathrm{2}}}}$ of
the reduction of $\ottnt{M_{{\mathrm{1}}}}$.

\TS{Show an example in the overview section?}

%% file: sections/polytype.tex
\section{A Polymorphic Type System for Signature Restriction}
\label{sec:polytype}

This section defines a polymorphic type
system for {\lang} that incorporates type containment as subsumption.  We then
formalize signature restriction and show that the type system is sound if all
operations satisfy signature restriction.  The type system in this section does
not track effect information for simplicity, so a well-typed program may
terminate at an unhandled operation call.
%

\subsection{Type Language}

\begin{figure}[t]
 \[
 \begin{array}{llllll}
  \multicolumn{2}{l}{
   \textbf{Type variables} \quad \alpha, \beta, \gamma
  } & 
  \textbf{Base types} \quad \iota ::=  \mathsf{bool}  \mid  \mathsf{int}  \mid ... \\[.5ex]
  \textbf{Types} \quad \ottnt{A}, \ottnt{B}, \ottnt{C}, \ottnt{D} & ::= &
   \alpha \mid \iota \mid \ottnt{A}  \rightarrow  \ottnt{B} \mid  \text{\unboldmath$\forall$}  \, \alpha  \ottsym{.} \, \ottnt{A}
   \\[.5ex]
  \textbf{Typing contexts} \quad \Gamma & ::= &
    \emptyset  \mid \Gamma  \ottsym{,}  \mathit{x} \,  \mathord{:}  \, \ottnt{A} \mid
   \Gamma  \ottsym{,}  \alpha
   \\
 \end{array}
 \]
 \caption{The type language.}
 \label{fig:polytype:lang}
\end{figure}

\reffig{polytype:lang} presents the type language of the polymorphic type
system.  It is standard and in fact the same as that of System~F~\cite{Reynolds_1974_PS,Girard_1972_PhD}.
We use metavariables $\alpha$, $\beta$, $\gamma$ for type
variables and $\iota$ for base types such as $ \mathsf{bool} $ and $ \mathsf{int} $.
Types, ranged over by $\ottnt{A}$, $\ottnt{B}$, $\ottnt{C}$, $\ottnt{D}$, consist of: type
variables; base types; function types $\ottnt{A}  \rightarrow  \ottnt{B}$; and polymorphic types
$ \text{\unboldmath$\forall$}  \, \alpha  \ottsym{.} \, \ottnt{A}$, where type variable $\alpha$ is bound in $\ottnt{A}$.
Typing contexts, ranged over by $\Gamma$, are sequences of bindings of variables
coupled with their types and type variables.
We suppose that each constant $\ottnt{c}$ is assigned a first-order closed type
$ \mathit{ty}  (  \ottnt{c}  ) $ of the form $\iota  \rightarrow  \ldots  \rightarrow  \iota_{\ottmv{n}}  \rightarrow   \iota _{ \ottmv{n}  \ottsym{+}  \ottsym{1} } $
which is consistent with the denotation of $\ottnt{c}$.

We use the following shorthand and notions.
We write $ \algeffseqoverindex{ \alpha }{ \text{\unboldmath$\mathit{I}$} } $ for $ \algeffseqover{ \alpha }  = \alpha_{{\mathrm{1}}}, \cdots, \alpha_{\ottmv{n}}$ with $\text{\unboldmath$\mathit{I}$} = \{ 1,
..., n \}$.  We apply this bold-font notation to other syntax categories as
well; for example, $ \algeffseqoverindex{ \ottnt{A} }{ \text{\unboldmath$\mathit{I}$} } $ denotes a sequence of types.  We often omit the index
sets ($\text{\unboldmath$\mathit{I}$}$, $\text{\unboldmath$\mathit{J}$}$, $\text{\unboldmath$\mathit{K}$}$) if they are clear from the context or irrelevant:
for example, we may abbreviate $ \algeffseqoverindex{ \alpha }{ \text{\unboldmath$\mathit{I}$} } $ to $ \algeffseqover{ \alpha } $.
We also write $ \text{\unboldmath$\forall$}  \,  \algeffseqoverindex{ \alpha }{ \text{\unboldmath$\mathit{I}$} }   \ottsym{.} \, \ottnt{A}$ for $  \text{\unboldmath$\forall$}    \alpha_{{\mathrm{1}}}  . \, ... \,   \text{\unboldmath$\forall$}    \alpha_{\ottmv{n}}  .  \, \ottnt{A}$ with $\text{\unboldmath$\mathit{I}$} = \{ 1, ..., n
\}$.  We may omit the index sets and write $ \text{\unboldmath$\forall$}  \,  \algeffseqover{ \alpha }   \ottsym{.} \, \ottnt{A}$ simply.
We write $ \algeffseqoverindex{  \text{\unboldmath$\forall$}  \,  \algeffseqoverindex{ \alpha }{ \text{\unboldmath$\mathit{I}$} }   \ottsym{.} \, \ottnt{A} }{ \text{\unboldmath$\mathit{J}$} } $ for a sequence of types $ \text{\unboldmath$\forall$}  \,  \algeffseqoverindex{ \alpha }{ \text{\unboldmath$\mathit{I}$} }   \ottsym{.} \, \ottnt{A}_1$, \ldots,
$ \text{\unboldmath$\forall$}  \,  \algeffseqoverindex{ \alpha }{ \text{\unboldmath$\mathit{I}$} }   \ottsym{.} \, \ottnt{A}_n$ with $\text{\unboldmath$\mathit{J}$} = \{1,\ldots,n\}$.
The notions of free type variables and capture-avoiding type substitution are
defined as usual.  We write $ \mathit{ftv}  (  \ottnt{A}  ) $ for the set of free type variables of
$\ottnt{A}$ and $ \ottnt{A}    [   \algeffseqover{ \ottnt{B} }   \ottsym{/}   \algeffseqover{ \alpha }   ]  $ for the type obtained by substituting each
type of $ \algeffseqover{ \ottnt{B} } $ for the corresponding type variable of $ \algeffseqover{ \alpha } $
simultaneously (here we suppose that $ \algeffseqover{ \ottnt{B} } $ and $ \algeffseqover{ \alpha } $ share the same,
omitted index set).

\subsection{Polymorphic Type System}
\label{sec:polytype:typing}

\begin{figure}[t]
 \begin{flushleft}
  \textbf{Well-formedness} \quad
  \framebox{$\vdash  \Gamma$}
 \end{flushleft}
 \begin{center}
  $\ottdruleWFXXEmpty{}$ \hfil
  $\ottdruleWFXXExtVar{}$ \hfil
  $\ottdruleWFXXExtTyVar{}$ \hfil
 \end{center}
 \mbox{} \\

 \begin{flushleft}
  \textbf{Type containment} \quad
  \framebox{$\Gamma  \vdash  \ottnt{A}  \sqsubseteq  \ottnt{B}$}
 \end{flushleft}
 \begin{center}
  $\ottdruleCXXRefl{}$ \hfil
  $\ottdruleCXXTrans{}$ \\[1.5ex]
  $\ottdruleCXXInst{}$ \hfil
  $\ottdruleCXXGen{}$ \hfil
  $\ottdruleCXXPoly{}$ \\[1.5ex]
  $\ottdruleCXXFun{}$ \hfil
  $\ottdruleCXXDFun{}$
 \end{center}
 \mbox{} \\

 \begin{flushleft}
  \textbf{Term typing} \quad
  \framebox{$\Gamma  \vdash  \ottnt{M}  \ottsym{:}  \ottnt{A}$}
 \end{flushleft}
 \begin{center}
  $\ottdruleTXXVar{}$ \hfil
  $\ottdruleTXXConst{}$ \hfil
  $\ottdruleTXXAbs{}$ \\[1.5ex]
  $\ottdruleTXXApp{}$ \hfil
  $\ottdruleTXXInst{}$ \\[1.5ex]
  $\ottdruleTXXGen{}$ \hfil
  $\ottdruleTXXOp{}$ \\[1.5ex]
  $\ottdruleTXXHandle{}$
 \end{center}
 \mbox{} \\

 \begin{flushleft}
  \textbf{Handler typing} \quad
  \framebox{$\Gamma  \vdash  \ottnt{H}  \ottsym{:}  \ottnt{A}  \Rightarrow  \ottnt{B}$}
 \end{flushleft}
 \begin{center}
  $\ottdruleTHXXReturn{}$ \\[1.5ex]
  $\ottdruleTHXXOp{}$
 \end{center}

 \caption{Polymorphic type system for {\lang}.}
 \label{fig:polytype:typing}
\end{figure}

We present a polymorphic type system for {\lang}, which consists of four
judgments:
well-formedness judgment $\vdash  \Gamma$, which states that a typing context $\Gamma$
is well formed;
type containment judgment $\Gamma  \vdash  \ottnt{A}  \sqsubseteq  \ottnt{B}$, which states that, for the types $\ottnt{A}$ and $\ottnt{B}$, which are assumed to be well formed under $\Gamma$, the inhabitants of $\ottnt{A}$ are
contained in $\ottnt{B}$;
%
%
term typing judgment $\Gamma  \vdash  \ottnt{M}  \ottsym{:}  \ottnt{A}$, which states that term $\ottnt{M}$ evaluates
to a value of $\ottnt{A}$ after applying appropriate substitution for variables and
type variables in $\Gamma$;
and handler typing judgment $\Gamma  \vdash  \ottnt{H}  \ottsym{:}  \ottnt{A}  \Rightarrow  \ottnt{B}$, which states that handler
$\ottnt{H}$ handles operations called by a handled term of $\ottnt{A}$ and produces a
value of $\ottnt{B}$ after applying appropriate substitution according to $\Gamma$
(we refer to $\ottnt{A}$ and $\ottnt{B}$ as the \emph{input} and \emph{output types} of the handler,
respectively).
These judgments are defined as the smallest relations that satisfy the rules in
\reffig{polytype:typing}.

The well-formedness rules are standard.  A typing context is well formed if (1)
variables and type variables bound by it are unique and (2) it assigns
well-formed types to the variables.  We write $ \mathit{dom}  (  \Gamma  ) $ for the set of variables and type
variables bound by $\Gamma$.  A type $\ottnt{A}$ is well formed under typing context
$\Gamma$, which is expressed by $\Gamma  \vdash  \ottnt{A}$, if and only if $ \mathit{ftv}  (  \ottnt{A}  )  \,  \subseteq  \,  \mathit{dom}  (  \Gamma  ) $
(i.e., $\Gamma$ binds all of the free type variables in $\ottnt{A}$).

The type containment rules originate from the work of
\citet{Tiuryn/Urzyczyn_1996_LICS}, which simplifies the rules of type
containment of \citet{Mitchell_1988_IC}.  The rules
\Srule{Refl} and \Srule{Trans} indicate that type containment is a preorder.  The
rule \Srule{Inst} instantiates polymorphic types with well-formed types.  The
rule \Srule{Gen} may add a quantifier $\forall$ if it does not bind free type
variables.  The rules \Srule{Poly} and \Srule{Fun} are for compatibility; note
that type containment is a kind of subtyping and hence it is contravariant on the domain types of function types.  The rule
\Srule{DFun} is a simplified version~\cite{Tiuryn/Urzyczyn_1996_LICS} of the original ``distributive'' law~\cite{Mitchell_1988_IC}, which
is the core of type containment.  This rule allows $\forall$ that quantifies a function
type to move to its codomain type if the quantified type variable does not occur
free in the domain type.  This rule is justified by the fact that we can supply a
function from $ \text{\unboldmath$\forall$}  \, \alpha  \ottsym{.} \, \ottnt{A}  \rightarrow  \ottnt{B}$ to $\ottnt{A}  \rightarrow   \text{\unboldmath$\forall$}  \, \alpha  \ottsym{.} \, \ottnt{B}$ in System~F and the
result of applying type erasure to it is equivalent to the identity
function~\cite{Mitchell_1988_IC}.  This rule is crucial for allowing the domain type of a type
signature to refer to quantified type variables in strictly positive positions,
which makes signature restriction permissive.

The typing rules for terms are almost standard, coming from \citet{Mitchell_1988_IC} for polymorphism and \citet{Plotkin/Pretnar_2013_LMCS} for effects.  The rule \T{Inst}
converts types by type containment.  The rule \T{Op} is for operation
calls.  We formalize a type signature of an operation as follows.
\begin{restatable}[Type signature]{defn}{defnTypeSignature}
 \label{def:eff}
 Each effect operation $\mathsf{op}$ is assigned a type signature $\mathit{ty} \, \ottsym{(}  \mathsf{op}  \ottsym{)}$ of the
 form $   \text{\unboldmath$\forall$}    \alpha_{{\mathrm{1}}}  . \, ... \,   \text{\unboldmath$\forall$}    \alpha_{\ottmv{n}}  .  \,  \ottnt{A}  \hookrightarrow  \ottnt{B} $ for some $\ottmv{n}$, where $\alpha_{{\mathrm{1}}}  \ottsym{,}  ...  \ottsym{,}  \alpha_{\ottmv{n}}$ are bound
 in the domain type $\ottnt{A}$ and codomain type $\ottnt{B}$.  It may be abbreviated to
 $  \text{\unboldmath$\forall$}  \,  \algeffseqoverindex{ \alpha }{ \text{\unboldmath$\mathit{I}$} }   \ottsym{.} \,  \ottnt{A}  \hookrightarrow  \ottnt{B} $ or, more simply, to $  \text{\unboldmath$\forall$}  \,  \algeffseqover{ \alpha }   \ottsym{.} \,  \ottnt{A}  \hookrightarrow  \ottnt{B} $.  We suppose that
 $   \text{\unboldmath$\forall$}    \alpha_{{\mathrm{1}}}  . \, ... \,   \text{\unboldmath$\forall$}    \alpha_{\ottmv{n}}  .  \,  \ottnt{A}  \hookrightarrow  \ottnt{B} $ is closed, i.e., $ \mathit{ftv}  (  \ottnt{A}  ) ,  \mathit{ftv}  (  \ottnt{B}  )  \subseteq \{ \alpha_{{\mathrm{1}}}, \cdots, \alpha_{\ottmv{n}} \}$.
\end{restatable}
We note that domain and codomain types may involve polymorphic types.

The rule \T{Op} instantiates the type signature of the operation with
well-formed types and checks that an argument is typed at the domain type of the
instantiated signature.
%
%
We use notation $\Gamma  \vdash   \algeffseqover{ \ottnt{C} } $ for $\Gamma  \vdash  \ottnt{C_{{\mathrm{1}}}}, \cdots, \Gamma  \vdash  \ottnt{C_{\ottmv{n}}}$ when
$ \algeffseqover{ \ottnt{C} }  = \ottnt{C_{{\mathrm{1}}}}, \cdots, \ottnt{C_{\ottmv{n}}}$.

The typing rules for handlers are also ordinary~\cite{Plotkin/Pretnar_2013_LMCS}.  A return clause $\mathsf{return} \, \mathit{x}  \rightarrow  \ottnt{M}$ is typechecked by \THrule{Return}, which allows the body $\ottnt{M}$ to refer
to the values of the handled expression via bound variable $\mathit{x}$.  An
operation clause $\mathsf{op}  \ottsym{(}  \mathit{x}  \ottsym{,}  \mathit{k}  \ottsym{)}  \rightarrow  \ottnt{M}$ is typechecked by \THrule{Op}.  Let the type signature
of $\mathsf{op}$ be $  \text{\unboldmath$\forall$}  \,  \algeffseqover{ \alpha }   \ottsym{.} \,  \ottnt{C}  \hookrightarrow  \ottnt{D} $.  In typechecking $\ottnt{M}$, variable $\mathit{x}$ is
assigned the codomain type $\ottnt{C}$ since variable $\mathit{x}$ will be bound to
the arguments to the operation $\mathsf{op}$.  Variable $\mathit{k}$ is assigned to type $\ottnt{D}  \rightarrow  \ottnt{B}$ where
$\ottnt{B}$ is the output type of the handler.  This is because $\mathit{k}$ will be
bound to the functional representations of delimited continuations such that: the
delimited continuations suppose that their holes are filled with values of the
codomain type $\ottnt{D}$ of the type signature; and they are wrapped by the
$ \mathsf{handle} $--$ \mathsf{with} $ expression installing the handler and therefore they
would produce values of $\ottnt{B}$.

\subsection{Desired Propositions for Type Soundness}
\label{sec:polytype:addprop}
As mentioned in \refsec{overview:unsafety}, the polymorphic type system is unsound if
we impose no further restriction on it.  This section details the proof sketch of type preservation provided in \refsec{overview:our-work} and formulates two propositions
such that they do not hold in the polymorphic type system but, \emph{if they
did}, the type system would be sound.  In \refsec{polytype:signature-restriction:prop}, we show that the propositions hold
if all operations satisfy signature restriction.

We start by considering an issue that arises when proving soundness of the polymorphic type
system.  This issue relates to the handling of an operation call by \R{Handle},
which enables the following reduction:
\[
 \mathsf{handle} \,  \ottnt{E}  [   \textup{\texttt{\#}\relax}  \mathsf{op}   \ottsym{(}   \ottnt{v}   \ottsym{)}   ]  \, \mathsf{with} \, \ottnt{H}  \rightsquigarrow    \ottnt{M}    [  \ottnt{v}  /  \mathit{x}  ]      [   \lambda\!  \, \mathit{y}  \ottsym{.}  \mathsf{handle} \,  \ottnt{E}  [  \mathit{y}  ]  \, \mathsf{with} \, \ottnt{H}  /  \mathit{k}  ]  
\]
where $\mathsf{op} \,  \not\in  \, \ottnt{E}$ and $\ottnt{H}  \ottsym{(}  \mathsf{op}  \ottsym{)} \,  =  \, \mathsf{op}  \ottsym{(}  \mathit{x}  \ottsym{,}  \mathit{k}  \ottsym{)}  \rightarrow  \ottnt{M}$.
The problem is that the RHS term does not preserve the type of the LHS term.  If
this type preservation were successful, we would be able to prove soundness of the polymorphic type system, but
it is contradictory to the existence of the counterexample presented in
\refsec{overview:unsafety}.

A detailed investigation of this problem reveals two propositions that are lacking
but sufficient to make the polymorphic type system sound.
\begin{prop} \label{prop:sig-forall-move:domain}
 If\/ $\mathit{ty} \, \ottsym{(}  \mathsf{op}  \ottsym{)} \,  =  \,   \text{\unboldmath$\forall$}  \,  \algeffseqoverindex{ \alpha }{ \text{\unboldmath$\mathit{I}$} }   \ottsym{.} \,  \ottnt{A}  \hookrightarrow  \ottnt{B} $ and
 $\Gamma  \vdash  \ottnt{M}  \ottsym{:}    \text{\unboldmath$\forall$}  \,  \algeffseqoverindex{ \beta }{ \text{\unboldmath$\mathit{J}$} }   \ottsym{.} \, \ottnt{A}    [   \algeffseqoverindex{ \ottnt{C} }{ \text{\unboldmath$\mathit{I}$} }   \ottsym{/}   \algeffseqoverindex{ \alpha }{ \text{\unboldmath$\mathit{I}$} }   ]  $, then $\Gamma  \vdash  \ottnt{M}  \ottsym{:}   \ottnt{A}    [   \algeffseqoverindex{  \text{\unboldmath$\forall$}  \,  \algeffseqoverindex{ \beta }{ \text{\unboldmath$\mathit{J}$} }   \ottsym{.} \, \ottnt{C} }{ \text{\unboldmath$\mathit{I}$} }   \ottsym{/}   \algeffseqoverindex{ \alpha }{ \text{\unboldmath$\mathit{I}$} }   ]  $.
\end{prop}
\begin{prop} \label{prop:sig-forall-move:codomain}
 If\/ $\mathit{ty} \, \ottsym{(}  \mathsf{op}  \ottsym{)} \,  =  \,   \text{\unboldmath$\forall$}  \,  \algeffseqoverindex{ \alpha }{ \text{\unboldmath$\mathit{I}$} }   \ottsym{.} \,  \ottnt{A}  \hookrightarrow  \ottnt{B} $ and
 $\Gamma  \vdash  \ottnt{M}  \ottsym{:}   \ottnt{B}    [   \algeffseqoverindex{  \text{\unboldmath$\forall$}  \,  \algeffseqoverindex{ \beta }{ \text{\unboldmath$\mathit{J}$} }   \ottsym{.} \, \ottnt{C} }{ \text{\unboldmath$\mathit{I}$} }   \ottsym{/}   \algeffseqoverindex{ \alpha }{ \text{\unboldmath$\mathit{I}$} }   ]  $, then $\Gamma  \vdash  \ottnt{M}  \ottsym{:}    \text{\unboldmath$\forall$}  \,  \algeffseqoverindex{ \beta }{ \text{\unboldmath$\mathit{J}$} }   \ottsym{.} \, \ottnt{B}    [   \algeffseqoverindex{ \ottnt{C} }{ \text{\unboldmath$\mathit{I}$} }   \ottsym{/}   \algeffseqoverindex{ \alpha }{ \text{\unboldmath$\mathit{I}$} }   ]  $.
\end{prop}

In what follows, we show how these propositions allow us to prove type soundness.
Before that, we first fix and examine the type information of the terms appearing in the
LHS term.  Let us suppose that $\mathit{ty} \, \ottsym{(}  \mathsf{op}  \ottsym{)} \,  =  \,   \text{\unboldmath$\forall$}  \,  \algeffseqoverindex{ \alpha }{ \text{\unboldmath$\mathit{I}$} }   \ottsym{.} \,  \ottnt{A}  \hookrightarrow  \ottnt{B} $ and that the LHS term
has a type $\ottnt{D}$ under a typing context $\Gamma$.  We can then find that
\begin{equation}
 \Gamma  \ottsym{,}   \algeffseqoverindex{ \alpha }{ \text{\unboldmath$\mathit{I}$} }   \ottsym{,}  \mathit{x} \,  \mathord{:}  \, \ottnt{A}  \ottsym{,}  \mathit{k} \,  \mathord{:}  \, \ottnt{B}  \rightarrow  \ottnt{D}  \vdash  \ottnt{M}  \ottsym{:}  \ottnt{D}
  \label{eqn:polytype:addprop:one}
\end{equation}
is derived.
Turning to the handled expression $ \ottnt{E}  [   \textup{\texttt{\#}\relax}  \mathsf{op}   \ottsym{(}   \ottnt{v}   \ottsym{)}   ] $, we can find two facts about
the typing judgment for $\ottnt{v}$.  The first fact originates from \T{Op}: since
$\ottnt{v}$ is an argument of operation $\mathsf{op}$, it should be of
$ \ottnt{A}    [   \algeffseqoverindex{ \ottnt{C} }{ \text{\unboldmath$\mathit{I}$} }   \ottsym{/}   \algeffseqoverindex{ \alpha }{ \text{\unboldmath$\mathit{I}$} }   ]  $, which is a type obtained by substituting certain types
$ \algeffseqoverindex{ \ottnt{C} }{ \text{\unboldmath$\mathit{I}$} } $ for type variables $ \algeffseqoverindex{ \alpha }{ \text{\unboldmath$\mathit{I}$} } $ in the domain type $\ottnt{A}$ of the type
signature of $\mathsf{op}$.  The second is from \T{Gen}, which allows the generalization
of types \emph{anywhere}.  Thus, $\ottnt{v}$ is well typed under a typing context
$\Gamma  \ottsym{,}   \algeffseqoverindex{ \beta }{ \text{\unboldmath$\mathit{J}$} } $, an extension of $\Gamma$ with some type variables $ \algeffseqoverindex{ \beta }{ \text{\unboldmath$\mathit{J}$} } $ (note
that $\text{\unboldmath$\mathit{I}$} \neq \text{\unboldmath$\mathit{J}$}$ in general).  In summary, the typing judgment for
$\ottnt{v}$ takes the following form:
\begin{equation}
 \Gamma  \ottsym{,}   \algeffseqoverindex{ \beta }{ \text{\unboldmath$\mathit{J}$} }   \vdash  \ottnt{v}  \ottsym{:}   \ottnt{A}    [   \algeffseqoverindex{ \ottnt{C} }{ \text{\unboldmath$\mathit{I}$} }   \ottsym{/}   \algeffseqoverindex{ \alpha }{ \text{\unboldmath$\mathit{I}$} }   ]   ~.
  \label{eqn:polytype:addprop:two}
\end{equation}

Now, we show that \refprop{sig-forall-move:domain} makes $ \ottnt{M}    [  \ottnt{v}  /  \mathit{x}  ]  $ typed at
$\ottnt{D}$.
First, we can derive
\[
 \Gamma  \vdash  \ottnt{v}  \ottsym{:}    \text{\unboldmath$\forall$}  \,  \algeffseqoverindex{ \beta }{ \text{\unboldmath$\mathit{J}$} }   \ottsym{.} \, \ottnt{A}    [   \algeffseqoverindex{ \ottnt{C} }{ \text{\unboldmath$\mathit{I}$} }   \ottsym{/}   \algeffseqoverindex{ \alpha }{ \text{\unboldmath$\mathit{I}$} }   ]  
\]
by the typing derivation of judgment (\ref{eqn:polytype:addprop:two}) and
\T{Gen}.  \refprop{sig-forall-move:domain} enables us to prove
\begin{equation}
 \Gamma  \vdash  \ottnt{v}  \ottsym{:}   \ottnt{A}    [   \algeffseqoverindex{  \text{\unboldmath$\forall$}  \,  \algeffseqoverindex{ \beta }{ \text{\unboldmath$\mathit{J}$} }   \ottsym{.} \, \ottnt{C} }{ \text{\unboldmath$\mathit{I}$} }   \ottsym{/}   \algeffseqoverindex{ \alpha }{ \text{\unboldmath$\mathit{I}$} }   ]   ~.
  \label{eqn:polytype:addprop:three}
\end{equation}
We can also derive
\[
 \Gamma  \ottsym{,}  \mathit{x} \,  \mathord{:}  \, \ottnt{A} \,  [   \algeffseqoverindex{  \text{\unboldmath$\forall$}  \,  \algeffseqoverindex{ \beta }{ \text{\unboldmath$\mathit{J}$} }   \ottsym{.} \, \ottnt{C} }{ \text{\unboldmath$\mathit{I}$} }   \ottsym{/}   \algeffseqoverindex{ \alpha }{ \text{\unboldmath$\mathit{I}$} }   ]   \ottsym{,}  \mathit{k} \,  \mathord{:}  \,  \ottnt{B}    [   \algeffseqoverindex{  \text{\unboldmath$\forall$}  \,  \algeffseqoverindex{ \beta }{ \text{\unboldmath$\mathit{J}$} }   \ottsym{.} \, \ottnt{C} }{ \text{\unboldmath$\mathit{I}$} }   \ottsym{/}   \algeffseqoverindex{ \alpha }{ \text{\unboldmath$\mathit{I}$} }   ]    \rightarrow  \ottnt{D}  \vdash  \ottnt{M}  \ottsym{:}  \ottnt{D}
\]
by substituting $ \algeffseqoverindex{  \text{\unboldmath$\forall$}  \,  \algeffseqoverindex{ \beta }{ \text{\unboldmath$\mathit{J}$} }   \ottsym{.} \, \ottnt{C} }{ \text{\unboldmath$\mathit{I}$} } $ for $ \algeffseqoverindex{ \alpha }{ \text{\unboldmath$\mathit{I}$} } $ in the typing judgment
(\ref{eqn:polytype:addprop:one}); note that the type variables in $ \algeffseqoverindex{ \alpha }{ \text{\unboldmath$\mathit{I}$} } $ do not
occur free in $\ottnt{D}$ because they are bound by the type signature.  Thus,
we can derive
\begin{equation}
 \Gamma  \ottsym{,}  \mathit{k} \,  \mathord{:}  \,  \ottnt{B}    [   \algeffseqoverindex{  \text{\unboldmath$\forall$}  \,  \algeffseqoverindex{ \beta }{ \text{\unboldmath$\mathit{J}$} }   \ottsym{.} \, \ottnt{C} }{ \text{\unboldmath$\mathit{I}$} }   \ottsym{/}   \algeffseqoverindex{ \alpha }{ \text{\unboldmath$\mathit{I}$} }   ]    \rightarrow  \ottnt{D}  \vdash   \ottnt{M}    [  \ottnt{v}  /  \mathit{x}  ]    \ottsym{:}  \ottnt{D}
  \label{eqn:polytype:addprop:four}
\end{equation}
using an ordinary substitution lemma with the derivation for judgment
(\ref{eqn:polytype:addprop:three}).

Next, we show that \refprop{sig-forall-move:codomain} makes
$  \ottnt{M}    [  \ottnt{v}  /  \mathit{x}  ]      [   \lambda\!  \, \mathit{y}  \ottsym{.}  \mathsf{handle} \,  \ottnt{E}  [  \mathit{y}  ]  \, \mathsf{with} \, \ottnt{H}  /  \mathit{k}  ]  $ typed at $\ottnt{D}$.
This is possible if
\[
 \Gamma  \vdash   \lambda\!  \, \mathit{y}  \ottsym{.}  \mathsf{handle} \,  \ottnt{E}  [  \mathit{y}  ]  \, \mathsf{with} \, \ottnt{H}  \ottsym{:}   \ottnt{B}    [   \algeffseqoverindex{  \text{\unboldmath$\forall$}  \,  \algeffseqoverindex{ \beta }{ \text{\unboldmath$\mathit{J}$} }   \ottsym{.} \, \ottnt{C} }{ \text{\unboldmath$\mathit{I}$} }   \ottsym{/}   \algeffseqoverindex{ \alpha }{ \text{\unboldmath$\mathit{I}$} }   ]    \rightarrow  \ottnt{D}
\]
is derivable, jointly with the derivation of typing judgment
(\ref{eqn:polytype:addprop:four}).  Namely, it suffices to derive
\[
 \Gamma  \ottsym{,}  \mathit{y} \,  \mathord{:}  \, \ottnt{B} \,  [   \algeffseqoverindex{  \text{\unboldmath$\forall$}  \,  \algeffseqoverindex{ \beta }{ \text{\unboldmath$\mathit{J}$} }   \ottsym{.} \, \ottnt{C} }{ \text{\unboldmath$\mathit{I}$} }   \ottsym{/}   \algeffseqoverindex{ \alpha }{ \text{\unboldmath$\mathit{I}$} }   ]   \vdash  \mathsf{handle} \,  \ottnt{E}  [  \mathit{y}  ]  \, \mathsf{with} \, \ottnt{H}  \ottsym{:}  \ottnt{D} ~.
\]
%
By an observation similar to $\ottnt{v}$, we find that $ \textup{\texttt{\#}\relax}  \mathsf{op}   \ottsym{(}   \ottnt{v}   \ottsym{)} $
is typed at $ \ottnt{B}    [   \algeffseqoverindex{ \ottnt{C} }{ \text{\unboldmath$\mathit{I}$} }   \ottsym{/}   \algeffseqoverindex{ \alpha }{ \text{\unboldmath$\mathit{I}$} }   ]  $ under $\Gamma  \ottsym{,}   \algeffseqoverindex{ \beta }{ \text{\unboldmath$\mathit{J}$} } $ (note that $\ottnt{B}$ is the
codomain type of the type signature of $\mathsf{op}$).  Thus, for the above typing judgment
to hold, it suffices for $\mathit{y}$ to have the same type as
$ \textup{\texttt{\#}\relax}  \mathsf{op}   \ottsym{(}   \ottnt{v}   \ottsym{)} $.  Hence, we will derive
\begin{equation}
 \Gamma  \ottsym{,}  \mathit{y} \,  \mathord{:}  \, \ottnt{B} \,  [   \algeffseqoverindex{  \text{\unboldmath$\forall$}  \,  \algeffseqoverindex{ \beta }{ \text{\unboldmath$\mathit{J}$} }   \ottsym{.} \, \ottnt{C} }{ \text{\unboldmath$\mathit{I}$} }   \ottsym{/}   \algeffseqoverindex{ \alpha }{ \text{\unboldmath$\mathit{I}$} }   ]   \ottsym{,}   \algeffseqoverindex{ \beta }{ \text{\unboldmath$\mathit{J}$} }   \vdash  \mathit{y}  \ottsym{:}   \ottnt{B}    [   \algeffseqoverindex{ \ottnt{C} }{ \text{\unboldmath$\mathit{I}$} }   \ottsym{/}   \algeffseqoverindex{ \alpha }{ \text{\unboldmath$\mathit{I}$} }   ]   ~.
 \label{eqn:polytype:addprop:five}
\end{equation}
Because
$\Gamma  \ottsym{,}  \mathit{y} \,  \mathord{:}  \, \ottnt{B} \,  [   \algeffseqoverindex{  \text{\unboldmath$\forall$}  \,  \algeffseqoverindex{ \beta }{ \text{\unboldmath$\mathit{J}$} }   \ottsym{.} \, \ottnt{C} }{ \text{\unboldmath$\mathit{I}$} }   \ottsym{/}   \algeffseqoverindex{ \alpha }{ \text{\unboldmath$\mathit{I}$} }   ]   \ottsym{,}   \algeffseqoverindex{ \beta }{ \text{\unboldmath$\mathit{J}$} }   \vdash  \mathit{y}  \ottsym{:}   \ottnt{B}    [   \algeffseqoverindex{  \text{\unboldmath$\forall$}  \,  \algeffseqoverindex{ \beta }{ \text{\unboldmath$\mathit{J}$} }   \ottsym{.} \, \ottnt{C} }{ \text{\unboldmath$\mathit{I}$} }   \ottsym{/}   \algeffseqoverindex{ \alpha }{ \text{\unboldmath$\mathit{I}$} }   ]  $, we can derive
$\Gamma  \ottsym{,}  \mathit{y} \,  \mathord{:}  \, \ottnt{B} \,  [   \algeffseqoverindex{  \text{\unboldmath$\forall$}  \,  \algeffseqoverindex{ \beta }{ \text{\unboldmath$\mathit{J}$} }   \ottsym{.} \, \ottnt{C} }{ \text{\unboldmath$\mathit{I}$} }   \ottsym{/}   \algeffseqoverindex{ \alpha }{ \text{\unboldmath$\mathit{I}$} }   ]   \ottsym{,}   \algeffseqoverindex{ \beta }{ \text{\unboldmath$\mathit{J}$} }   \vdash  \mathit{y}  \ottsym{:}    \text{\unboldmath$\forall$}  \,  \algeffseqoverindex{ \beta }{ \text{\unboldmath$\mathit{J}$} }   \ottsym{.} \, \ottnt{B}    [   \algeffseqoverindex{ \ottnt{C} }{ \text{\unboldmath$\mathit{I}$} }   \ottsym{/}   \algeffseqoverindex{ \alpha }{ \text{\unboldmath$\mathit{I}$} }   ]  $ by
\refprop{sig-forall-move:codomain}; and, by instantiating
$  \text{\unboldmath$\forall$}  \,  \algeffseqoverindex{ \beta }{ \text{\unboldmath$\mathit{J}$} }   \ottsym{.} \, \ottnt{B}    [   \algeffseqoverindex{ \ottnt{C} }{ \text{\unboldmath$\mathit{I}$} }   \ottsym{/}   \algeffseqoverindex{ \alpha }{ \text{\unboldmath$\mathit{I}$} }   ]  $ to $ \ottnt{B}    [   \algeffseqoverindex{ \ottnt{C} }{ \text{\unboldmath$\mathit{I}$} }   \ottsym{/}   \algeffseqoverindex{ \alpha }{ \text{\unboldmath$\mathit{I}$} }   ]  $ with $ \algeffseqoverindex{ \beta }{ \text{\unboldmath$\mathit{J}$} } $ in the typing
context, we have succeeded in deriving the typing judgment (\ref{eqn:polytype:addprop:five}).

Thus, if Propositions~\ref{prop:sig-forall-move:domain} and
\ref{prop:sig-forall-move:codomain} held, we could derive
\[
 \Gamma  \vdash    \ottnt{M}    [  \ottnt{v}  /  \mathit{x}  ]      [   \lambda\!  \, \mathit{y}  \ottsym{.}  \mathsf{handle} \,  \ottnt{E}  [  \mathit{y}  ]  \, \mathsf{with} \, \ottnt{H}  /  \mathit{k}  ]    \ottsym{:}  \ottnt{D} ~.
\]
The polymorphic type system in \refsec{polytype:typing} does not actually have
these propositions, but imposing signature restriction produces a type system
that does have them.

\subsection{Signature Restriction}
\label{sec:polytype:signature-restriction}
This section formalizes signature restriction for {\lang} and shows that it
implies Propositions~\ref{prop:sig-forall-move:domain} and
\ref{prop:sig-forall-move:codomain}.

\subsubsection{Definition}
As described in \refsec{overview:our-work}, signature restriction rests on the
polarity of the occurrences of quantified type variables in a type signature.
The polarity is defined in a standard manner, as follows.
\begin{defn}[Polarity of type variable occurrence]
 \label{def:polarity}
 The positive and negative occurrences of a type variable in a type $\ottnt{A}$ are
 defined by induction on $\ottnt{A}$, as follows.
 \begin{itemize}
  \item The occurrence of $\alpha$ in type $\alpha$ is positive.

  \item The positive (resp.\ negative) occurrences of $\alpha$ in $\ottnt{A}  \rightarrow  \ottnt{B}$
        are the negative (resp.\ positive) occurrences of $\alpha$ in
        $\ottnt{A}$ and the positive (resp.\ negative) occurrences of $\alpha$ in
        $\ottnt{B}$.

  \item The positive (resp.\ negative) occurrences of $\alpha$ in $ \text{\unboldmath$\forall$}  \, \beta  \ottsym{.} \, \ottnt{A}$,
        where $\beta$ is supposed to be distinct from $\alpha$,
        are the positive (resp.\ negative) occurrences of $\alpha$ in
        $\ottnt{A}$.
 \end{itemize}

 The strictly positive occurrences of a type variable in a type are defined
 as follows.
 \begin{itemize}
  \item The occurrence of $\alpha$ in type $\alpha$ is strictly positive.

  \item The strictly positive occurrences of $\alpha$ in $\ottnt{A}  \rightarrow  \ottnt{B}$
        are the strictly positive occurrences of $\alpha$ in
        $\ottnt{B}$.

  \item The strictly positive occurrences of $\alpha$ in $ \text{\unboldmath$\forall$}  \, \beta  \ottsym{.} \, \ottnt{A}$,
        where $\beta$ is supposed to be distinct from $\alpha$,
        are the strictly positive occurrences of $\alpha$ in
        $\ottnt{A}$.
 \end{itemize}
\end{defn}

\begin{restatable}[Operations satisfying signature restriction]{defn}{defnSignatureRestriction}
 \label{def:signature-restriction}
 An operation $\mathsf{op}$ having type signature $\mathit{ty} \, \ottsym{(}  \mathsf{op}  \ottsym{)} \,  =  \,   \text{\unboldmath$\forall$}  \,  \algeffseqover{ \alpha }   \ottsym{.} \,  \ottnt{A}  \hookrightarrow  \ottnt{B} $ satisfies
 the signature restriction if and only if:
 (1) the occurrences of each type variable of $ \algeffseqover{ \alpha } $ in $\ottnt{A}$ are only
 negative or strictly positive; and
 (2) the occurrences of each type variable of $ \algeffseqover{ \alpha } $ in $\ottnt{B}$ are only
 positive.
\end{restatable}
The signature restriction allows quantified type variables to occur at strictly
positive positions of the domain type of a type signature.  This is crucial
for many operations, such as \texttt{raise}, \texttt{fail}, and
\texttt{select}, to conform to signature restriction.  The rule \Srule{DFun} plays
an important role to permit this capability, as seen in the next section.

We can easily confirm whether an operation satisfies the signature restriction.  For
example, it is easy to determine that \texttt{get\_id} does not satisfy the signature restriction: since its
type signature is $  \text{\unboldmath$\forall$}  \, \alpha  \ottsym{.} \,   \mathsf{unit}   \hookrightarrow  \alpha  \rightarrow  \alpha $, the quantified type variable $\alpha$
occurs not only at a positive position but also at a negative position in the
codomain type $\alpha  \rightarrow  \alpha$.  By contrast, the operations \texttt{raise} and
\texttt{fail} given in \refsec{overview} satisfy the signature restriction because
their type signature $  \text{\unboldmath$\forall$}  \, \alpha  \ottsym{.} \,   \mathsf{unit}   \hookrightarrow  \alpha $ meets the conditions in
\refdef{signature-restriction}.  To determine whether \texttt{select} satisfies the signature
restriction, we need to extend {\lang} and the polymorphic type system by
introducing other programming constructs such as lists.  Particulars of this
extension are presented in \refsec{ext}.

\subsubsection{Proofs of the Desired Propositions}
\label{sec:polytype:signature-restriction:prop}
The signature restriction enables us to prove
Propositions~\ref{prop:sig-forall-move:domain} and
\ref{prop:sig-forall-move:codomain}, which are crucial to show that reduction
preserves typing.  Below is the key lemma for that.
\begin{restatable}{lemm}{lemmSubtypingForallMove}
\label{lem:subtyping-forall-move}
 Suppose that $\alpha$ does not appear free in $\ottnt{A}$.
 \begin{enumerate}
  \item \label{lem:subtyping-forall-move:neg}
        If the occurrences of $\beta$ in $\ottnt{A}$ are only negative or strictly
        positive, then $\Gamma  \vdash    \text{\unboldmath$\forall$}  \, \alpha  \ottsym{.} \, \ottnt{A}    [  \ottnt{B}  \ottsym{/}  \beta  ]    \sqsubseteq   \ottnt{A}    [   \text{\unboldmath$\forall$}  \, \alpha  \ottsym{.} \, \ottnt{B}  \ottsym{/}  \beta  ]  $.
  \item \label{lem:subtyping-forall-move:pos}
        If the occurrences of $\beta$ in $\ottnt{A}$ are only positive,
        then $\Gamma  \vdash   \ottnt{A}    [   \text{\unboldmath$\forall$}  \, \alpha  \ottsym{.} \, \ottnt{B}  \ottsym{/}  \beta  ]    \sqsubseteq    \text{\unboldmath$\forall$}  \, \alpha  \ottsym{.} \, \ottnt{A}    [  \ottnt{B}  \ottsym{/}  \beta  ]  $.
 \end{enumerate}
\end{restatable}
This lemma means that an operation $\mathsf{op}$ conforming to the signature
restriction satisfies Propositions~\ref{prop:sig-forall-move:domain} and
\ref{prop:sig-forall-move:codomain}.  For \refprop{sig-forall-move:domain}:
suppose $\mathit{ty} \, \ottsym{(}  \mathsf{op}  \ottsym{)} \,  =  \,   \text{\unboldmath$\forall$}  \,  \algeffseqoverindex{ \alpha }{ \text{\unboldmath$\mathit{I}$} }   \ottsym{.} \,  \ottnt{A}  \hookrightarrow  \ottnt{B} $ and $\Gamma  \vdash  \ottnt{M}  \ottsym{:}    \text{\unboldmath$\forall$}  \,  \algeffseqoverindex{ \beta }{ \text{\unboldmath$\mathit{J}$} }   \ottsym{.} \, \ottnt{A}    [   \algeffseqoverindex{ \ottnt{C} }{ \text{\unboldmath$\mathit{I}$} }   \ottsym{/}   \algeffseqoverindex{ \alpha }{ \text{\unboldmath$\mathit{I}$} }   ]  $;
since $\mathsf{op}$ satisfies the signature restriction, we can apply case
(\ref{lem:subtyping-forall-move:neg}) of \reflem{subtyping-forall-move}, which
implies $\Gamma  \vdash    \text{\unboldmath$\forall$}  \,  \algeffseqoverindex{ \beta }{ \text{\unboldmath$\mathit{J}$} }   \ottsym{.} \, \ottnt{A}    [   \algeffseqoverindex{ \ottnt{C} }{ \text{\unboldmath$\mathit{I}$} }   \ottsym{/}   \algeffseqoverindex{ \alpha }{ \text{\unboldmath$\mathit{I}$} }   ]    \sqsubseteq   \ottnt{A}    [   \algeffseqoverindex{  \text{\unboldmath$\forall$}  \,  \algeffseqoverindex{ \beta }{ \text{\unboldmath$\mathit{J}$} }   \ottsym{.} \, \ottnt{C} }{ \text{\unboldmath$\mathit{I}$} }   \ottsym{/}   \algeffseqoverindex{ \alpha }{ \text{\unboldmath$\mathit{I}$} }   ]  $; thus,
we can derive $\Gamma  \vdash  \ottnt{M}  \ottsym{:}   \ottnt{A}    [   \algeffseqoverindex{  \text{\unboldmath$\forall$}  \,  \algeffseqoverindex{ \beta }{ \text{\unboldmath$\mathit{J}$} }   \ottsym{.} \, \ottnt{C} }{ \text{\unboldmath$\mathit{I}$} }   \ottsym{/}   \algeffseqoverindex{ \alpha }{ \text{\unboldmath$\mathit{I}$} }   ]  $ by \T{Inst}.
\refprop{sig-forall-move:codomain} is proven similarly by using case
(\ref{lem:subtyping-forall-move:pos}) of \reflem{subtyping-forall-move} instead
of case (\ref{lem:subtyping-forall-move:neg}).

It is easy to prove \reflem{subtyping-forall-move} if the occurrences of $\beta$
in $\ottnt{A}$ are only negative in case (\ref{lem:subtyping-forall-move:neg}).
In fact, the following lemma handles such a case (the statement is generalized
slightly).
\begin{restatable}{lemm}{lemmSubtypingForallRemove}
 \label{lem:subtyping-forall-remove}
 Suppose that $\alpha$ does not appear free in $\ottnt{A}$.
 \begin{enumerate}
  \item \label{lem:subtyping-forall-remove:neg}
        If the occurrences of $\beta$ in $\ottnt{A}$ are only negative,
        then $\Gamma_{{\mathrm{1}}}  \ottsym{,}  \alpha  \ottsym{,}  \Gamma_{{\mathrm{2}}}  \vdash   \ottnt{A}    [  \ottnt{B}  \ottsym{/}  \beta  ]    \sqsubseteq   \ottnt{A}    [   \text{\unboldmath$\forall$}  \, \alpha  \ottsym{.} \, \ottnt{B}  \ottsym{/}  \beta  ]  $.
  \item \label{lem:subtyping-forall-remove:pos}
        If the occurrences of $\beta$ in $\ottnt{A}$ are only positive,
        then $\Gamma_{{\mathrm{1}}}  \ottsym{,}  \alpha  \ottsym{,}  \Gamma_{{\mathrm{2}}}  \vdash   \ottnt{A}    [   \text{\unboldmath$\forall$}  \, \alpha  \ottsym{.} \, \ottnt{B}  \ottsym{/}  \beta  ]    \sqsubseteq   \ottnt{A}    [  \ottnt{B}  \ottsym{/}  \beta  ]  $.
 \end{enumerate}
\end{restatable}
\begin{proof}
 We prove both cases simultaneously by structural induction on $\ottnt{A}$.
 The polarity condition on the occurrences of $\beta$ ensures that, if $\ottnt{A} \,  =  \, \beta$, it suffices to show $\Gamma_{{\mathrm{1}}}  \ottsym{,}  \alpha  \ottsym{,}  \Gamma_{{\mathrm{2}}}  \vdash   \text{\unboldmath$\forall$}  \, \alpha  \ottsym{.} \, \ottnt{B}  \sqsubseteq  \ottnt{B}$, which is derived by
 \Srule{Inst}.
\end{proof}

Now, we prove \reflem{subtyping-forall-move} with
\reflem{subtyping-forall-remove} and \Srule{DFun}, which is the key rule for
the signature restriction to allow strictly positive occurrences of quantified type
variables in the domain type of a type signature.
\begin{proof}[Proof of \reflem{subtyping-forall-move}.]
 We prove both cases simultaneously by structural induction on $\ottnt{A}$.  The case
 (\ref{lem:subtyping-forall-move:pos}) can be proven by
 Lemma~\ref{lem:subtyping-forall-remove}: it enables us to show $\Gamma  \ottsym{,}  \alpha  \vdash   \ottnt{A}    [   \text{\unboldmath$\forall$}  \, \alpha  \ottsym{.} \, \ottnt{B}  \ottsym{/}  \beta  ]    \sqsubseteq   \ottnt{A}    [  \ottnt{B}  \ottsym{/}  \beta  ]  $; then, by \Srule{Poly}, \Srule{Gen}, and \Srule{Trans}, we can
 derive $\Gamma  \vdash   \ottnt{A}    [   \text{\unboldmath$\forall$}  \, \alpha  \ottsym{.} \, \ottnt{B}  \ottsym{/}  \beta  ]    \sqsubseteq    \text{\unboldmath$\forall$}  \, \alpha  \ottsym{.} \, \ottnt{A}    [  \ottnt{B}  \ottsym{/}  \beta  ]  $.

 Let us consider case (\ref{lem:subtyping-forall-move:neg}) where $\ottnt{A}$ is
 a function type $\ottnt{C}  \rightarrow  \ottnt{D}$; the other cases are easy to show.
 Suppose that the occurrences of $\beta$ in $\ottnt{C}  \rightarrow  \ottnt{D}$ are only negative or
 strictly positive.
 By definition, the occurrences of $\beta$ in $\ottnt{C}$ are only positive.  Thus,
 by the IH, $\Gamma  \vdash   \ottnt{C}    [   \text{\unboldmath$\forall$}  \, \alpha  \ottsym{.} \, \ottnt{B}  \ottsym{/}  \beta  ]    \sqsubseteq    \text{\unboldmath$\forall$}  \, \alpha  \ottsym{.} \, \ottnt{C}    [  \ottnt{B}  \ottsym{/}  \beta  ]  $.
 Furthermore, by definition, the occurrences of $\beta$ in $\ottnt{D}$
 are only negative or strictly positive.  Thus, by the IH,
 $\Gamma  \vdash    \text{\unboldmath$\forall$}  \, \alpha  \ottsym{.} \, \ottnt{D}    [  \ottnt{B}  \ottsym{/}  \beta  ]    \sqsubseteq   \ottnt{D}    [   \text{\unboldmath$\forall$}  \, \alpha  \ottsym{.} \, \ottnt{B}  \ottsym{/}  \beta  ]  $.
 By \Srule{Fun},
 \begin{equation}
  \Gamma  \vdash   \ottsym{(}    \text{\unboldmath$\forall$}  \, \alpha  \ottsym{.} \, \ottnt{C}    [  \ottnt{B}  \ottsym{/}  \beta  ]    \ottsym{)}  \rightarrow   \text{\unboldmath$\forall$}  \, \alpha  \ottsym{.} \, \ottnt{D}    [  \ottnt{B}  \ottsym{/}  \beta  ]    \sqsubseteq    \ottnt{C}    [   \text{\unboldmath$\forall$}  \, \alpha  \ottsym{.} \, \ottnt{B}  \ottsym{/}  \beta  ]    \rightarrow  \ottnt{D}    [   \text{\unboldmath$\forall$}  \, \alpha  \ottsym{.} \, \ottnt{B}  \ottsym{/}  \beta  ]   ~.
   \label{eqn:subtyping-forall-move:one}
 \end{equation}
 Thus:
 \[\begin{array}{r@{\ }ll}
  \Gamma  \vdash  &   \text{\unboldmath$\forall$}  \, \alpha  \ottsym{.} \, \ottnt{C}    [  \ottnt{B}  \ottsym{/}  \beta  ]    \rightarrow   \ottnt{D}    [  \ottnt{B}  \ottsym{/}  \beta  ]   \\
   \sqsubseteq        &  \text{\unboldmath$\forall$}  \, \alpha  \ottsym{.} \, \ottsym{(}    \text{\unboldmath$\forall$}  \, \alpha  \ottsym{.} \, \ottnt{C}    [  \ottnt{B}  \ottsym{/}  \beta  ]    \ottsym{)}  \rightarrow   \ottnt{D}    [  \ottnt{B}  \ottsym{/}  \beta  ]   &
   \text{(by \Srule{Poly}, \Srule{Fun}, and $\Gamma  \ottsym{,}  \alpha  \vdash    \text{\unboldmath$\forall$}  \, \alpha  \ottsym{.} \, \ottnt{C}    [  \ottnt{B}  \ottsym{/}  \beta  ]    \sqsubseteq   \ottnt{C}    [  \ottnt{B}  \ottsym{/}  \beta  ]  $)} \\
   \sqsubseteq        & \ottsym{(}    \text{\unboldmath$\forall$}  \, \alpha  \ottsym{.} \, \ottnt{C}    [  \ottnt{B}  \ottsym{/}  \beta  ]    \ottsym{)}  \rightarrow    \text{\unboldmath$\forall$}  \, \alpha  \ottsym{.} \, \ottnt{D}    [  \ottnt{B}  \ottsym{/}  \beta  ]   &
   \text{(by \Srule{DFun})} \\
   \sqsubseteq        &  \ottnt{C}    [   \text{\unboldmath$\forall$}  \, \alpha  \ottsym{.} \, \ottnt{B}  \ottsym{/}  \beta  ]    \rightarrow   \ottnt{D}    [   \text{\unboldmath$\forall$}  \, \alpha  \ottsym{.} \, \ottnt{B}  \ottsym{/}  \beta  ]   &
   \text{(by (\ref{eqn:subtyping-forall-move:one}))} ~.
   \end{array}\]
\end{proof}

\subsection{Type Soundness}
\label{sec:polytyp:typesound}

This section shows soundness of the polymorphic type system under the assumption
that all operations satisfy the signature restriction.  As usual, our proof
rests on two properties: progress and subject reduction~\cite{Wright/Felleisen_1994_IC}.  As discussed in
Sections~\ref{sec:polytype:addprop} and
\ref{sec:polytype:signature-restriction}, the signature restriction, together with
type containment, enables us to prove subject reduction.

In this work, type containment is thus a key notion to prove type soundness,
but it complicates certain inversion properties.  In the literature~\cite{Peython-Jones/Vytiniotis/Weirich/Shields_2007_JFP,Dunfield/Krishnaswami_2013_ICFP},
type soundness of a language with subtyping such as $ \sqsubseteq $ has been
shown by translation to another language---typically, System~F---where
the use of subtyping is replaced by ``coercions'' (i.e., certain term representations
for type conversion by subtyping).  This approach is acceptable in the prior work because
the semantics of the source language is determined by the target language.
However, this approach is \emph{not} acceptable in our setting because the terms checked by our type
system should be interpreted by the semantics of {\lang} as they are.  We thus
show soundness of the polymorphic type system directly, without translation to
other languages.

The property that is the most difficult to prove in the direct approach is the inversion of type
containment judgments for function types.
\begin{restatable}{lemm}{lemmSubtypingInvFunMono}
 \label{lem:subtyping-inv-fun-mono}
 If\/ $\Gamma  \vdash  \ottnt{A_{{\mathrm{1}}}}  \rightarrow  \ottnt{A_{{\mathrm{2}}}}  \sqsubseteq  \ottnt{B_{{\mathrm{1}}}}  \rightarrow  \ottnt{B_{{\mathrm{2}}}}$,
 then $\Gamma  \vdash  \ottnt{B_{{\mathrm{1}}}}  \sqsubseteq  \ottnt{A_{{\mathrm{1}}}}$ and $\Gamma  \vdash  \ottnt{A_{{\mathrm{2}}}}  \sqsubseteq  \ottnt{B_{{\mathrm{2}}}}$.
\end{restatable}
\TS{okay?}
We cannot prove this lemma as it is by induction on the derivation of
$\Gamma  \vdash  \ottnt{A_{{\mathrm{1}}}}  \rightarrow  \ottnt{A_{{\mathrm{2}}}}  \sqsubseteq  \ottnt{B_{{\mathrm{1}}}}  \rightarrow  \ottnt{B_{{\mathrm{2}}}}$ because a premise in the derivation may relate
the (nonpolymorphic) function type on one side to a polymorphic function
type on the other side.  Thus, we need to generalize the assumption to a type containment
judgment that may relate polymorphic function types:
$\Gamma  \vdash   \text{\unboldmath$\forall$}  \,  \algeffseqoverindex{ \alpha }{ \text{\unboldmath$\mathit{I}$} }   \ottsym{.} \, \ottnt{A_{{\mathrm{1}}}}  \rightarrow  \ottnt{A_{{\mathrm{2}}}}  \sqsubseteq   \text{\unboldmath$\forall$}  \,  \algeffseqoverindex{ \beta }{ \text{\unboldmath$\mathit{J}$} }   \ottsym{.} \, \ottnt{B_{{\mathrm{1}}}}  \rightarrow  \ottnt{B_{{\mathrm{2}}}}$.
By investigating the type containment rules, we find that $ \algeffseqoverindex{ \alpha }{ \text{\unboldmath$\mathit{I}$} } $ is split
into three sequences $ \algeffseqoverindex{ \alpha_{{\mathrm{01}}} }{ \text{\unboldmath$\mathit{I_{{\mathrm{01}}}}$} } $, $ \algeffseqoverindex{ \alpha_{{\mathrm{02}}} }{ \text{\unboldmath$\mathit{I_{{\mathrm{02}}}}$} } $, and $ \algeffseqoverindex{ \alpha_{{\mathrm{03}}} }{ \text{\unboldmath$\mathit{I_{{\mathrm{03}}}}$} } $ depending on
how the rules handle the type variables in $ \algeffseqoverindex{ \alpha }{ \text{\unboldmath$\mathit{I}$} } $: those of
$ \algeffseqoverindex{ \alpha_{{\mathrm{01}}} }{ \text{\unboldmath$\mathit{I_{{\mathrm{01}}}}$} } $ stay in $ \algeffseqoverindex{ \beta }{ \text{\unboldmath$\mathit{J}$} } $; those of $ \algeffseqoverindex{ \alpha_{{\mathrm{02}}} }{ \text{\unboldmath$\mathit{I_{{\mathrm{02}}}}$} } $ are quantified in the
return type $\ottnt{B_{{\mathrm{2}}}}$; and those of $ \algeffseqoverindex{ \alpha_{{\mathrm{03}}} }{ \text{\unboldmath$\mathit{I_{{\mathrm{03}}}}$} } $ are instantiated with
some types $ \algeffseqoverindex{ \ottnt{C_{{\mathrm{0}}}} }{ \text{\unboldmath$\mathit{I_{{\mathrm{03}}}}$} } $.
%
Furthermore, we have to take into account certain, unrevealed type variables
$ \algeffseqoverindex{ \gamma }{ \text{\unboldmath$\mathit{K}$} } $ that initially emerge at the outermost position by \T{Gen} and are
subsequently distributed into subcomponent types.  For example:
\[
 \ottnt{A_{{\mathrm{1}}}}  \rightarrow  \ottnt{A_{{\mathrm{2}}}}  \sqsubseteq   \text{\unboldmath$\forall$}  \, \gamma  \ottsym{.} \, \ottnt{A_{{\mathrm{1}}}}  \rightarrow  \ottnt{A_{{\mathrm{2}}}}  \sqsubseteq  \ottsym{(}   \text{\unboldmath$\forall$}  \, \gamma  \ottsym{.} \, \ottnt{A_{{\mathrm{1}}}}  \ottsym{)}  \rightarrow  \ottsym{(}   \text{\unboldmath$\forall$}  \, \gamma  \ottsym{.} \, \ottnt{A_{{\mathrm{2}}}}  \ottsym{)}
\]
where $\gamma \,  \not\in  \,  \mathit{ftv}  (  \ottnt{A_{{\mathrm{1}}}}  )  \,  \mathbin{\cup}  \,  \mathit{ftv}  (  \ottnt{A_{{\mathrm{2}}}}  ) $.
These observations are formulated in the
following inversion lemma for type containment, which implies
\reflem{subtyping-inv-fun-mono}.
We write $\ottsym{\{}   \algeffseqoverindex{ \alpha }{ \text{\unboldmath$\mathit{I}$} }   \ottsym{\}}$ to view the sequence $ \algeffseqoverindex{ \alpha }{ \text{\unboldmath$\mathit{I}$} } $ as a set by ignoring
the order and $\ottsym{\{}   \algeffseqoverindex{ \alpha }{ \text{\unboldmath$\mathit{I}$} }   \ottsym{\}} \,  \mathbin{\uplus}  \, \ottsym{\{}   \algeffseqoverindex{ \beta }{ \text{\unboldmath$\mathit{J}$} }   \ottsym{\}}$ for the union of disjoint sets
$\ottsym{\{}   \algeffseqoverindex{ \alpha }{ \text{\unboldmath$\mathit{I}$} }   \ottsym{\}}$ and $\ottsym{\{}   \algeffseqoverindex{ \beta }{ \text{\unboldmath$\mathit{J}$} }   \ottsym{\}}$.
\begin{restatable}[Type containment inversion: function types]{lemm}{lemmSubtypingInvFun}
 \label{lem:subtyping-inv-fun}
 If\/ $\Gamma  \vdash   \text{\unboldmath$\forall$}  \,  \algeffseqoverindex{ \alpha }{ \text{\unboldmath$\mathit{I}$} }   \ottsym{.} \, \ottnt{A_{{\mathrm{1}}}}  \rightarrow  \ottnt{A_{{\mathrm{2}}}}  \sqsubseteq   \text{\unboldmath$\forall$}  \,  \algeffseqoverindex{ \beta }{ \text{\unboldmath$\mathit{J}$} }   \ottsym{.} \, \ottnt{B_{{\mathrm{1}}}}  \rightarrow  \ottnt{B_{{\mathrm{2}}}}$,
 then there exist $ \algeffseqoverindex{ \alpha_{{\mathrm{1}}} }{ \text{\unboldmath$\mathit{I_{{\mathrm{1}}}}$} } $, $ \algeffseqoverindex{ \alpha_{{\mathrm{2}}} }{ \text{\unboldmath$\mathit{I_{{\mathrm{2}}}}$} } $, $ \algeffseqoverindex{ \gamma }{ \text{\unboldmath$\mathit{K}$} } $, and $ \algeffseqoverindex{ \ottnt{C} }{ \text{\unboldmath$\mathit{I_{{\mathrm{1}}}}$} } $
 such that
 \begin{itemize}
  \item $\ottsym{\{}   \algeffseqoverindex{ \alpha }{ \text{\unboldmath$\mathit{I}$} }   \ottsym{\}} \,  =  \, \ottsym{\{}   \algeffseqoverindex{ \alpha_{{\mathrm{1}}} }{ \text{\unboldmath$\mathit{I_{{\mathrm{1}}}}$} }   \ottsym{\}} \,  \mathbin{\uplus}  \, \ottsym{\{}   \algeffseqoverindex{ \alpha_{{\mathrm{2}}} }{ \text{\unboldmath$\mathit{I_{{\mathrm{2}}}}$} }   \ottsym{\}}$,
  \item $\Gamma  \ottsym{,}   \algeffseqoverindex{ \beta }{ \text{\unboldmath$\mathit{J}$} }   \ottsym{,}   \algeffseqoverindex{ \gamma }{ \text{\unboldmath$\mathit{K}$} }   \vdash   \algeffseqoverindex{ \ottnt{C} }{ \text{\unboldmath$\mathit{I_{{\mathrm{1}}}}$} } $,
  \item $\Gamma  \ottsym{,}   \algeffseqoverindex{ \beta }{ \text{\unboldmath$\mathit{J}$} }   \vdash  \ottnt{B_{{\mathrm{1}}}}  \sqsubseteq    \text{\unboldmath$\forall$}  \,  \algeffseqoverindex{ \gamma }{ \text{\unboldmath$\mathit{K}$} }   \ottsym{.} \, \ottnt{A_{{\mathrm{1}}}}    [   \algeffseqoverindex{ \ottnt{C} }{ \text{\unboldmath$\mathit{I_{{\mathrm{1}}}}$} }   \ottsym{/}   \algeffseqoverindex{ \alpha_{{\mathrm{1}}} }{ \text{\unboldmath$\mathit{I_{{\mathrm{1}}}}$} }   ]  $, and
  \item $\Gamma  \ottsym{,}   \algeffseqoverindex{ \beta }{ \text{\unboldmath$\mathit{J}$} }   \vdash    \text{\unboldmath$\forall$}  \,  \algeffseqoverindex{ \alpha_{{\mathrm{2}}} }{ \text{\unboldmath$\mathit{I_{{\mathrm{2}}}}$} }   \ottsym{.} \,  \text{\unboldmath$\forall$}  \,  \algeffseqoverindex{ \gamma }{ \text{\unboldmath$\mathit{K}$} }   \ottsym{.} \, \ottnt{A_{{\mathrm{2}}}}    [   \algeffseqoverindex{ \ottnt{C} }{ \text{\unboldmath$\mathit{I_{{\mathrm{1}}}}$} }   \ottsym{/}   \algeffseqoverindex{ \alpha_{{\mathrm{1}}} }{ \text{\unboldmath$\mathit{I_{{\mathrm{1}}}}$} }   ]    \sqsubseteq  \ottnt{B_{{\mathrm{2}}}}$.
 \end{itemize}
\end{restatable}
\noindent
In this statement, the sequence $ \algeffseqoverindex{ \alpha_{{\mathrm{2}}} }{ \text{\unboldmath$\mathit{I_{{\mathrm{2}}}}$} } $ corresponds to $ \algeffseqoverindex{ \alpha_{{\mathrm{02}}} }{ \text{\unboldmath$\mathit{I_{{\mathrm{02}}}}$} } $ in the above informal description and $ \algeffseqoverindex{ \alpha_{{\mathrm{1}}} }{ \text{\unboldmath$\mathit{I_{{\mathrm{1}}}}$} } $ includes the type variables that remain in
$ \algeffseqoverindex{ \beta }{ \text{\unboldmath$\mathit{J}$} } $ (i.e., $ \algeffseqoverindex{ \alpha_{{\mathrm{01}}} }{ \text{\unboldmath$\mathit{I_{{\mathrm{01}}}}$} } $) and those instantiated with some types in $ \algeffseqoverindex{ \ottnt{C} }{ \text{\unboldmath$\mathit{I_{{\mathrm{1}}}}$} } $ (i.e., $ \algeffseqoverindex{ \alpha_{{\mathrm{03}}} }{ \text{\unboldmath$\mathit{I_{{\mathrm{03}}}}$} } $).  Type substitution $ [   \algeffseqoverindex{ \ottnt{C} }{ \text{\unboldmath$\mathit{I_{{\mathrm{1}}}}$} }   \ottsym{/}   \algeffseqoverindex{ \alpha_{{\mathrm{1}}} }{ \text{\unboldmath$\mathit{I_{{\mathrm{1}}}}$} }   ] $ replaces a type variable in $ \algeffseqoverindex{ \alpha_{{\mathrm{01}}} }{ \text{\unboldmath$\mathit{I_{{\mathrm{01}}}}$} } $ with itself.

We also prove other lemmas such as weakening, substitution, canonical forms, and
value inversion.  We omit their formal statements and proofs in this
paper; the details can be found in \ifappendix Appendix~\ref{sec:app:proof:sound-typing}. \else the supplementary material. \fi

Now, we show progress and subject reduction.  In what follows, the metavariable
$\Delta$ ranges over typing contexts that consist of only type variable
bindings.  Note that the polymorphic type system is not equipped with a
mechanism to track effects, so the operations that are carried out may not be handled.
\begin{restatable}[Progress]{lemm}{lemmProgress}
 \label{lem:progress}
 If\/ $\Delta  \vdash  \ottnt{M}  \ottsym{:}  \ottnt{A}$, then:
 \begin{itemize}
  \item $\ottnt{M}  \longrightarrow  \ottnt{M'}$ for some $\ottnt{M'}$;
  \item $\ottnt{M}$ is a value; or
  \item $\ottnt{M} \,  =  \,  \ottnt{E}  [   \textup{\texttt{\#}\relax}  \mathsf{op}   \ottsym{(}   \ottnt{v}   \ottsym{)}   ] $ for some $\ottnt{E}$, $\mathsf{op}$, and $\ottnt{v}$
 such that $\mathsf{op} \,  \not\in  \, \ottnt{E}$.
 \end{itemize}
\end{restatable}
\begin{restatable}[Subject reduction]{lemm}{lemmSubjectRed}
 \label{lem:subject-red}
 Suppose that all operations satisfy the signature restriction.
 \begin{enumerate}
  \item If\/ $\Delta  \vdash  \ottnt{M_{{\mathrm{1}}}}  \ottsym{:}  \ottnt{A}$ and $\ottnt{M_{{\mathrm{1}}}}  \rightsquigarrow  \ottnt{M_{{\mathrm{2}}}}$,
        then $\Delta  \vdash  \ottnt{M_{{\mathrm{2}}}}  \ottsym{:}  \ottnt{A}$.
  \item If\/ $\Delta  \vdash  \ottnt{M_{{\mathrm{1}}}}  \ottsym{:}  \ottnt{A}$ and $\ottnt{M_{{\mathrm{1}}}}  \longrightarrow  \ottnt{M_{{\mathrm{2}}}}$,
        then $\Delta  \vdash  \ottnt{M_{{\mathrm{2}}}}  \ottsym{:}  \ottnt{A}$.
 \end{enumerate}
\end{restatable}

We write $ \longrightarrow^{*} $ for the reflexive, transitive closure of $ \longrightarrow $ and
$\ottnt{M}  \centernot\longrightarrow$ to mean that there exists no term $\ottnt{M'}$ such that $\ottnt{M}  \longrightarrow  \ottnt{M'}$.
\begin{restatable}[Type soundness]{thm}{thmTypeSoundness}
 \label{thm:type-sound}
 Suppose that all operations satisfy the signature restriction.
 If\/ $\Delta  \vdash  \ottnt{M}  \ottsym{:}  \ottnt{A}$ and $\ottnt{M}  \longrightarrow^{*}  \ottnt{M'}$ and $\ottnt{M'}  \centernot\longrightarrow$, then:
 \begin{itemize}
  \item $\ottnt{M'}$ is a value; or
  \item $\ottnt{M'} \,  =  \,  \ottnt{E}  [   \textup{\texttt{\#}\relax}  \mathsf{op}   \ottsym{(}   \ottnt{v}   \ottsym{)}   ] $ for some $\ottnt{E}$, $\mathsf{op}$, and $\ottnt{v}$
        such that $\mathsf{op} \,  \not\in  \, \ottnt{E}$.
 \end{itemize}
\end{restatable}
\begin{proof}
 By progress (\reflem{progress}) and subject reduction (\reflem{subject-red}).
\end{proof}

\input{sections/remark_domain_type}

%% file: sections/remark_domain_type.tex
\begin{remark}
  It is natural to ask whether the signature restriction can be further relaxed.
  Consider a type signature $  \text{\unboldmath$\forall$}  \, \alpha  \ottsym{.} \,  \ottnt{A}  \hookrightarrow  \ottnt{B} $.  A negative occurrence of
  $\alpha$ in $\ottnt{B}$ is problematic as \texttt{get\_id}, which has type
  signature $  \text{\unboldmath$\forall$}  \, \alpha  \ottsym{.} \,   \mathsf{unit}   \hookrightarrow  \alpha  \rightarrow  \alpha $, is unsafe (see \refsec{overview:unsafety}).
  A non-strictly positive occurrence of $\alpha$ in $\ottnt{A}$ is also problematic, as the following example shows.
  Let us consider a calculus with $ \mathsf{int} $, $ \mathsf{bool} $, and sum types $ \ottnt{D_{{\mathrm{1}}}}  +  \ottnt{D_{{\mathrm{2}}}} $
  for simplicity (we write $\mathsf{inl} \, \ottnt{M}$ and $\mathsf{inr} \, \ottnt{M}$ for injection into sum types).
  Consider an operation \texttt{op} of the signature $  \text{\unboldmath$\forall$}  \, \alpha  \ottsym{.} \,  \ottsym{(}  \ottsym{(}  \alpha  \rightarrow   \mathsf{int}   \ottsym{)}  \rightarrow  \alpha  \ottsym{)}  \hookrightarrow  \alpha $
  and let
  \[\begin{array}{lll}
    \ottnt{v} &\defeq&  \lambda\!  \, \mathit{f}  \ottsym{.}   \lambda\!  \, \mathit{x}  \ottsym{.}  \mathsf{inr} \, \ottsym{(}  \mathit{f} \, \ottsym{(}   \lambda\!  \, \mathit{y}  \ottsym{.}  \mathsf{inl} \, \mathit{x}  \ottsym{)}  \ottsym{)}
    ~:~
    \ottsym{(}  \ottsym{(}  \beta  \rightarrow  \ottsym{(}   \beta  +   \mathsf{int}    \ottsym{)}  \ottsym{)}  \rightarrow   \mathsf{int}   \ottsym{)}  \rightarrow  \ottsym{(}  \beta  \rightarrow  \ottsym{(}   \beta  +   \mathsf{int}    \ottsym{)}  \ottsym{)}
    \\
    \ottnt{M} &\defeq& \mathrm{let}\,\mathit{g} =  \textup{\texttt{\#}\relax}  \mathsf{op}   \ottsym{(}   \ottnt{v}   \ottsym{)} \,\mathrm{in}\,\mathsf{case} \, \mathit{g} \,  0  \, \mathsf{of} \, \mathsf{inl} \, \mathit{z}  \rightarrow  \mathit{z}  \ottsym{;} \, \mathsf{inr} \, \mathit{z}  \rightarrow   \ottnt{E}  [  \mathit{g} \,  \mathsf{true}   ] 
    ~:~
     \mathsf{int}  ,
   \end{array}\]
  where $\ottnt{E}$ is an evaluation context such that $\mathit{x} \,  \mathord{:}  \,   \mathsf{bool}   +   \mathsf{int}    \vdash   \ottnt{E}  [  \mathit{x}  ]   \ottsym{:}   \mathsf{int} $ and $ \ottnt{E}  [  \mathsf{inr} \,  \mathsf{true}   ] $ causes a run-time error (it is easy to construct such $\ottnt{E}$).
  It is not difficult to check that $\ottnt{M}$ has type $ \mathsf{int} $.
  In $ \textup{\texttt{\#}\relax}  \mathsf{op}   \ottsym{(}   \ottnt{v}   \ottsym{)} $, the type variable $\alpha$ bound by the type signature is instantiated with $\beta  \rightarrow  \ottsym{(}   \beta  +   \mathsf{int}    \ottsym{)}$, and thus $\mathit{g}$ has type $ \text{\unboldmath$\forall$}  \, \beta  \ottsym{.} \, \beta  \rightarrow  \ottsym{(}   \beta  +   \mathsf{int}    \ottsym{)}$.
  The type variable $\beta$ is instantiated with $ \mathsf{int} $ in $\mathit{g} \,  0 $ and with $ \mathsf{bool} $ in $\mathit{g} \,  \mathsf{true} $.
  Then the counterexample is given by
  \begin{equation*}
    \mathsf{handle} \, \ottnt{M} \, \mathsf{with} \, \mathsf{return} \, \mathit{x}  \rightarrow  \mathit{x}  \ottsym{;}  \mathsf{op}  \ottsym{(}  \mathit{x}  \ottsym{,}  \mathit{k}  \ottsym{)}  \rightarrow  \mathit{k} \, \ottsym{(}  \mathit{x} \, \mathit{k}  \ottsym{)}
    ~:~  \mathsf{int}   ,
  \end{equation*}
  which is reduced to $\mathsf{handle} \,  \ottnt{E}  [  \mathsf{inr} \,  \mathsf{true}   ]  \, \mathsf{with} \, \mathsf{return} \, \mathit{x}  \rightarrow  \mathit{x}  \ottsym{;}  \mathsf{op}  \ottsym{(}  \mathit{x}  \ottsym{,}  \mathit{k}  \ottsym{)}  \rightarrow  \mathit{k} \, \ottsym{(}  \mathit{x} \, \mathit{k}  \ottsym{)}$ and causes an error.
\end{remark}

%% file: sections/ext.tex
\section{An Extension of {\lang}}
\label{sec:ext}

This section demonstrates the extensibility of the signature restriction.  To this
end, we extend {\lang}, the polymorphic type system, and the signature restriction
with products, sums, and lists and show soundness of the extended polymorphic
type system under the extended signature restriction.  We also provide a few
examples of operations that satisfy the extended signature restriction.

\subsection{Extended Language}

\begin{figure}[t]
 \[
 \begin{array}{lll}
  \textbf{Terms} \quad \ottnt{M} & ::= &
   \textgray{
    \mathit{x} \mid \ottnt{c} \mid  \lambda\!  \, \mathit{x}  \ottsym{.}  \ottnt{M} \mid \ottnt{M_{{\mathrm{1}}}} \, \ottnt{M_{{\mathrm{2}}}} \mid
     \textup{\texttt{\#}\relax}  \mathsf{op}   \ottsym{(}   \ottnt{M}   \ottsym{)}  \mid \mathsf{handle} \, \ottnt{M} \, \mathsf{with} \, \ottnt{H} \mid
   }\;
   \ottsym{(}  \ottnt{M_{{\mathrm{1}}}}  \ottsym{,}  \ottnt{M_{{\mathrm{2}}}}  \ottsym{)} \mid \\ && \pi_1  \ottnt{M} \mid \pi_2  \ottnt{M} \mid
   \mathsf{inl} \, \ottnt{M} \mid \mathsf{inr} \, \ottnt{M} \mid
   \mathsf{case} \, \ottnt{M} \, \mathsf{of} \, \mathsf{inl} \, \mathit{x}  \rightarrow  \ottnt{M_{{\mathrm{1}}}}  \ottsym{;} \, \mathsf{inr} \, \mathit{y}  \rightarrow  \ottnt{M_{{\mathrm{2}}}} \mid \\ &&
    \mathsf{nil}  \mid \mathsf{cons} \, \ottnt{M} \mid
   \mathsf{case} \, \ottnt{M} \, \mathsf{of} \, \mathsf{nil} \, \rightarrow  \ottnt{M_{{\mathrm{1}}}}  \ottsym{;} \, \mathsf{cons} \, \mathit{x}  \rightarrow  \ottnt{M_{{\mathrm{2}}}} \mid
   \mathsf{fix} \, \mathit{f}  \ottsym{.}   \lambda\!  \, \mathit{x}  \ottsym{.}  \ottnt{M}
   \\[.5ex]
  \textbf{Values} \quad \ottnt{v} & ::= &
   \textgray{\ottnt{c} \mid  \lambda\!  \, \mathit{x}  \ottsym{.}  \ottnt{M} \mid}\;
   \ottsym{(}  \ottnt{v_{{\mathrm{1}}}}  \ottsym{,}  \ottnt{v_{{\mathrm{2}}}}  \ottsym{)} \mid
   \mathsf{inl} \, \ottnt{v} \mid \mathsf{inr} \, \ottnt{v} \mid  \mathsf{nil}  \mid \mathsf{cons} \, \ottnt{v}
   \\[.5ex]
  \textbf{Types} \quad \ottnt{A}, \ottnt{B}, \ottnt{C}, \ottnt{D} & ::= &
   \textgray{\alpha \mid \iota \mid \ottnt{A}  \rightarrow  \ottnt{B} \mid  \text{\unboldmath$\forall$}  \, \alpha  \ottsym{.} \, \ottnt{A} \mid}\;
    \ottnt{A}  \times  \ottnt{B}  \mid  \ottnt{A}  +  \ottnt{B}  \mid  \ottnt{A}  \, \mathsf{list} 
   \\[.5ex]
  \textbf{Evaluation contexts} \quad
   \ottnt{E} & ::= & \textgray{ []  \mid
                 \ottnt{E} \, \ottnt{M_{{\mathrm{2}}}} \mid \ottnt{v_{{\mathrm{1}}}} \, \ottnt{E} \mid
                  \textup{\texttt{\#}\relax}  \mathsf{op}   \ottsym{(}   \ottnt{E}   \ottsym{)}  \mid \mathsf{handle} \, \ottnt{E} \, \mathsf{with} \, \ottnt{H} \mid}\;
                 \ottsym{(}  \ottnt{E}  \ottsym{,}  \ottnt{M_{{\mathrm{2}}}}  \ottsym{)} \mid \ottsym{(}  \ottnt{v_{{\mathrm{1}}}}  \ottsym{,}  \ottnt{E}  \ottsym{)} \mid \\ &&
                 \pi_1  \ottnt{E} \mid \pi_2  \ottnt{E} \mid
                 \mathsf{inl} \, \ottnt{E} \mid \mathsf{inr} \, \ottnt{E} \mid
                 \mathsf{case} \, \ottnt{E} \, \mathsf{of} \, \mathsf{inl} \, \mathit{x}  \rightarrow  \ottnt{M_{{\mathrm{1}}}}  \ottsym{;} \, \mathsf{inr} \, \mathit{y}  \rightarrow  \ottnt{M_{{\mathrm{2}}}} \mid \\ &&
                 \mathsf{cons} \, \ottnt{E} \mid
                 \mathsf{case} \, \ottnt{E} \, \mathsf{of} \, \mathsf{nil} \, \rightarrow  \ottnt{M_{{\mathrm{1}}}}  \ottsym{;} \, \mathsf{cons} \, \mathit{x}  \rightarrow  \ottnt{M_{{\mathrm{2}}}}
 \end{array}
 \]
 \mbox{} \\
 
 \begin{flushleft}
  \textbf{Reduction rules} \quad
  \framebox{$\ottnt{M_{{\mathrm{1}}}}  \rightsquigarrow  \ottnt{M_{{\mathrm{2}}}}$}
 \end{flushleft}
 \[\begin{array}{@{}l@{}}
  \pi_1  \ottsym{(}  \ottnt{v_{{\mathrm{1}}}}  \ottsym{,}  \ottnt{v_{{\mathrm{2}}}}  \ottsym{)}  \rightsquigarrow  \ottnt{v_{{\mathrm{1}}}} \hfil
  \pi_2  \ottsym{(}  \ottnt{v_{{\mathrm{1}}}}  \ottsym{,}  \ottnt{v_{{\mathrm{2}}}}  \ottsym{)}  \rightsquigarrow  \ottnt{v_{{\mathrm{2}}}} \hfil
  \mathsf{fix} \, \mathit{f}  \ottsym{.}   \lambda\!  \, \mathit{x}  \ottsym{.}  \ottnt{M}  \rightsquigarrow   \ottsym{(}   \lambda\!  \, \mathit{x}  \ottsym{.}  \ottnt{M}  \ottsym{)}    [  \mathsf{fix} \, \mathit{f}  \ottsym{.}   \lambda\!  \, \mathit{x}  \ottsym{.}  \ottnt{M}  /  \mathit{f}  ]   \\
  \mathsf{case} \, \mathsf{inl} \, \ottnt{v} \, \mathsf{of} \, \mathsf{inl} \, \mathit{x}  \rightarrow  \ottnt{M_{{\mathrm{1}}}}  \ottsym{;} \, \mathsf{inr} \, \mathit{y}  \rightarrow  \ottnt{M_{{\mathrm{2}}}}  \rightsquigarrow   \ottnt{M_{{\mathrm{1}}}}    [  \ottnt{v}  /  \mathit{x}  ]   \\
  \mathsf{case} \, \mathsf{inr} \, \ottnt{v} \, \mathsf{of} \, \mathsf{inl} \, \mathit{x}  \rightarrow  \ottnt{M_{{\mathrm{1}}}}  \ottsym{;} \, \mathsf{inr} \, \mathit{y}  \rightarrow  \ottnt{M_{{\mathrm{2}}}}  \rightsquigarrow   \ottnt{M_{{\mathrm{2}}}}    [  \ottnt{v}  /  \mathit{y}  ]   \\
  \mathsf{case} \, \mathsf{nil} \, \mathsf{of} \, \mathsf{nil} \, \rightarrow  \ottnt{M_{{\mathrm{1}}}}  \ottsym{;} \, \mathsf{cons} \, \mathit{x}  \rightarrow  \ottnt{M_{{\mathrm{2}}}}  \rightsquigarrow  \ottnt{M_{{\mathrm{1}}}} \qquad
  \mathsf{case} \, \mathsf{cons} \, \ottnt{v} \, \mathsf{of} \, \mathsf{nil} \, \rightarrow  \ottnt{M_{{\mathrm{1}}}}  \ottsym{;} \, \mathsf{cons} \, \mathit{x}  \rightarrow  \ottnt{M_{{\mathrm{2}}}}  \rightsquigarrow   \ottnt{M_{{\mathrm{2}}}}    [  \ottnt{v}  /  \mathit{x}  ]  
 \end{array}\]
 \mbox{} \\

 \begin{flushleft}
  \textbf{Type containment} \quad
  \framebox{$\Gamma  \vdash  \ottnt{A}  \sqsubseteq  \ottnt{B}$} \\[1.5ex]
 \end{flushleft}
 \begin{center}
  $\ottdruleCXXProd{}$ \hfil
  $\ottdruleCXXDProd{}$ \\[1.5ex]
  $\ottdruleCXXSum{}$ \hfil
  $\ottdruleCXXDSum{}$ \\[1.5ex]
  $\ottdruleCXXList{}$ \hfil
  $\ottdruleCXXDList{}$
 \end{center}
 \iffull
 \mbox{} \\

 \begin{flushleft}
  \textbf{Term typing} \quad
  \framebox{$\Gamma  \vdash  \ottnt{M}  \ottsym{:}  \ottnt{A}$} \\[1.5ex]
 \end{flushleft}
 \begin{center}
  $\ottdruleTXXPair{}$ \hfil
  $\ottdruleTXXProjOne{}$ \hfil
  $\ottdruleTXXProjTwo{}$ \\[1.5ex]
  $\ottdruleTXXInL{}$ \hfil
  $\ottdruleTXXInR{}$ \\[1.5ex]
  $\ottdruleTXXCase{}$ \\[1.5ex]
  $\ottdruleTXXNil{}$ \hfil
  $\ottdruleTXXCons{}$ \hfil
  $\ottdruleTXXFix{}$ \\[1.5ex]
  $\ottdruleTXXCaseList{}$ \hfil
 \end{center}
 \fi
 \caption{The extended part.}
 \label{fig:extlang}
\end{figure}

The extension of {\lang} and the polymorphic type system is shown in
\reffig{extlang}, in which the extended part of the syntax is highlighted.
Terms
support: pairs; projections; injections; case expressions for sums;
the nil constant; $ \mathsf{cons} $ expressions; case
expressions for lists; and the fixed-point operator.  A case
expression matching injections $\mathsf{case} \, \ottnt{M} \, \mathsf{of} \, \mathsf{inl} \, \mathit{x}  \rightarrow  \ottnt{M_{{\mathrm{1}}}}  \ottsym{;} \, \mathsf{inr} \, \mathit{y}  \rightarrow  \ottnt{M_{{\mathrm{2}}}}$ binds
$\mathit{x}$ in $\ottnt{M_{{\mathrm{1}}}}$ and $\mathit{y}$ in $\ottnt{M_{{\mathrm{2}}}}$, respectively; a case expression
matching lists $\mathsf{case} \, \ottnt{M} \, \mathsf{of} \, \mathsf{nil} \, \rightarrow  \ottnt{M_{{\mathrm{1}}}}  \ottsym{;} \, \mathsf{cons} \, \mathit{x}  \rightarrow  \ottnt{M_{{\mathrm{2}}}}$ binds $\mathit{x}$ in
$\ottnt{M_{{\mathrm{2}}}}$; the fixed-point operator $\mathsf{fix} \, \mathit{f}  \ottsym{.}   \lambda\!  \, \mathit{x}  \ottsym{.}  \ottnt{M}$ binds $\mathit{f}$ and $\mathit{x}$
in $\ottnt{M}$.
Pairs, injections, and $ \mathsf{cons} $ expressions are values if their immediate
subterms are also values.  Types are extended with product types, sum types, and list
types.  The extension of evaluation contexts follows that of terms.
For the semantics, the reduction rules for projections, case
expressions, and the fixed-point operator are added.
The extension of the polymorphic type system is also straightforward.  Type containment
is extended by adding six rules: the three rules on the left in \reffig{extlang} are for
compatibility and the three rules on the right are for distributing $\forall$ over
immediate subcomponent types.  All of the additional typing rules are standard\iffull\else\relax{} and are thus omitted\fi.

\input{sections/remark_on_dsum}

We also extend the polarity of the occurrences of a type variable.  The polarity
of the occurrences in type variables, function types, and polymorphic types is
given in \refdef{polarity}.  We also define the polarity in product, sum, and
list types as follows.
\begin{defn}[Polarity of type variable occurrence in product, sum, and list types]
 \label{def:ext-polarity}
 \TS{OK?}
 The positive and negative occurrences of a type variable in a product, sum, and
 list type are defined as follows.
 \begin{itemize}



  \item The positive (resp.\ negative) occurrences of $\alpha$ in $ \ottnt{A}  \times  \ottnt{B} $
        are the positive (resp.\ negative) occurrences of $\alpha$ in
        $\ottnt{A}$ and those in $\ottnt{B}$.

  \item The positive (resp.\ negative) occurrences of $\alpha$ in $ \ottnt{A}  +  \ottnt{B} $
        are the positive (resp.\ negative) occurrences of $\alpha$ in
        $\ottnt{A}$ and those in $\ottnt{B}$.

  \item The positive (resp.\ negative) occurrences of $\alpha$ in $ \ottnt{A}  \, \mathsf{list} $
        are the positive (resp.\ negative) occurrences of $\alpha$ in
        $\ottnt{A}$.
 \end{itemize}

 The strictly positive occurrences of a type variable in a product, sum, and
 list type are defined as follows.
 \begin{itemize}



  \item The strictly positive occurrences of $\alpha$ in $ \ottnt{A}  \times  \ottnt{B} $
        are the strictly positive occurrences of $\alpha$ in
        $\ottnt{A}$ and those in $\ottnt{B}$.

  \item The strictly positive occurrences of $\alpha$ in $ \ottnt{A}  +  \ottnt{B} $
        are the strictly positive occurrences of $\alpha$ in
        $\ottnt{A}$ and those in $\ottnt{B}$.

  \item The strictly positive occurrences of $\alpha$ in $ \ottnt{A}  \, \mathsf{list} $
        are the strictly positive occurrences of $\alpha$ in
        $\ottnt{A}$.
 \end{itemize}
\end{defn}

The signature restriction for the extended language is defined as in
\refdef{signature-restriction} except that the polarity of occurrences of type
variables is defined by both of Definitions~\ref{def:polarity} and
\ref{def:ext-polarity}.

We can prove that the extended polymorphic type system satisfies type soundness
under the assumption that all operations conform to the signature restriction for the
extended language in a similar way as in \refsec{polytyp:typesound}; refer to \ifappendix Appendix~\ref{sec:app:proof:sound-typing} \else the supplementary material \fi
for the proof.

\subsection{Examples}
This section presents two operations that satisfy the signature restriction in the
extended language.

The first example is \texttt{select}, which is an operation given in
\refsec{overview:algeff} for nondeterministic computation.  The operation has
the type signature $  \text{\unboldmath$\forall$}  \, \alpha  \ottsym{.} \,   \alpha  \, \mathsf{list}   \hookrightarrow  \alpha $, where the quantified type variable
$\alpha$ occurs only at a strictly positive position in the domain type $ \alpha  \, \mathsf{list} $
and only at a positive position in the codomain type $\alpha$.  Thus, \texttt{select}
satisfies the signature restriction and, therefore, it can be safely called by any
polymorphic expression.

The second example is from \citet{Leijen_2017_POPL}, who implemented parser
combinators using algebraic effects and handlers.  The effect for parsing
provides a basic operation \texttt{satisfy} which has the type signature
\[
   \text{\unboldmath$\forall$}  \, \alpha  \ottsym{.} \,  \ottsym{(}    \mathsf{str}   \rightarrow  \ottsym{(}   \alpha  \times   \mathsf{str}    \ottsym{)}  +   \mathsf{unit}    \ottsym{)}  \hookrightarrow  \alpha 
\]
where $ \mathsf{str} $ is the type of strings.  This operation takes a parsing function
of $ \mathsf{str}   \rightarrow   \ottsym{(}   \alpha  \times   \mathsf{str}    \ottsym{)}  +   \mathsf{unit}  $ such that: the parsing function returns the unit value if an input string
does not conform to the grammar; otherwise, it returns the parsing result of $\alpha$
and the unparsed, remaining string.  The operation \texttt{satisfy} would return
the result of parsing if it succeeds.
For example, we can give \texttt{satisfy} a parsing function that returns the
first character of a given input---and returns the unit value if the input is
the empty string---as follows:
\[
  \textup{\texttt{\#}\relax}   \mathsf{satisfy}    \ottsym{(}    \lambda\!  \, \mathit{x}  \ottsym{.}  \mathsf{if} \, \ottsym{(}   \mathsf{length}  \, \mathit{x}  \ottsym{)} \,  >  \,  0  \, \mathsf{then} \, \mathsf{inl} \, \ottsym{(}   \mathsf{first}  \, \mathit{x}  \ottsym{,}   \mathsf{last}  \, \mathit{x}  \ottsym{)} \, \mathsf{else} \, \mathsf{inr} \, \ottsym{()}   \ottsym{)} .
\]
Here: $ \mathsf{length} $ is a function of $ \mathsf{str}   \rightarrow   \mathsf{int} $ that returns the length of
a given string; $ \mathsf{first} $ is of $ \mathsf{str}   \rightarrow   \mathsf{char} $ ($ \mathsf{char} $ is the type of
characters) that returns the first character of a given string; and $ \mathsf{last} $
is of $ \mathsf{str}   \rightarrow   \mathsf{str} $ that returns the same string as an input except that it
does not contain the first character of the input.  In this example, the call of
\texttt{satisfy} is of the type $ \mathsf{char} $ because the argument function is of the type
$ \mathsf{str}   \rightarrow   \ottsym{(}    \mathsf{char}   \times   \mathsf{str}    \ottsym{)}  +   \mathsf{unit}  $, which requires the quantified type variable
$\alpha$ to be instantiated to $ \mathsf{char} $.
The operation \texttt{satisfy} satisfies the signature restriction clearly.  The
quantified type variable $\alpha$ occurs only at a strictly positive position in
the domain type $ \mathsf{str}   \rightarrow   \ottsym{(}   \alpha  \times   \mathsf{str}    \ottsym{)}  +   \mathsf{unit}  $ of the type signature and it also
occurs only at a positive position in the codomain type $\alpha$.

%% file: sections/remark_on_dsum.tex
\begin{remark}
  The rule \Srule{DSum} in \reffig{extlang} may look peculiar or questionable.
  Actually, there exists no term $\ottnt{M}$ in (implicitly typed) System~F such that $\mathit{x} \,  \mathord{:}  \,  \text{\unboldmath$\forall$}  \, \alpha  \ottsym{.} \, \ottsym{(}   \ottnt{A}  +  \ottnt{B}   \ottsym{)}  \vdash  \ottnt{M}  \ottsym{:}   \ottsym{(}   \text{\unboldmath$\forall$}  \, \alpha  \ottsym{.} \, \ottnt{A}  \ottsym{)}  +  \ottsym{(}   \text{\unboldmath$\forall$}  \, \alpha  \ottsym{.} \, \ottnt{B}  \ottsym{)} $,
  and thus the expected coercion function of $\ottsym{(}   \text{\unboldmath$\forall$}  \, \alpha  \ottsym{.} \, \ottsym{(}   \ottnt{A}  +  \ottnt{B}   \ottsym{)}  \ottsym{)}  \rightarrow   \ottsym{(}   \text{\unboldmath$\forall$}  \, \alpha  \ottsym{.} \, \ottnt{A}  \ottsym{)}  +  \ottsym{(}   \text{\unboldmath$\forall$}  \, \alpha  \ottsym{.} \, \ottnt{B}  \ottsym{)} $
  is not definable in System~F.
  \AI{Temporarily removed: Hence it is hard to justify \Srule{DSum} by a static translation.}
  A justification can be given by the following fact: for every \emph{closed value}  ${}\vdash \ottnt{v} :  \text{\unboldmath$\forall$}  \, \alpha  \ottsym{.} \, \ottsym{(}   \ottnt{A}  +  \ottnt{B}   \ottsym{)}$,
  one has ${}\vdash \ottnt{v} :  \ottsym{(}   \text{\unboldmath$\forall$}  \, \alpha  \ottsym{.} \, \ottnt{A}  \ottsym{)}  +  \ottsym{(}   \text{\unboldmath$\forall$}  \, \alpha  \ottsym{.} \, \ottnt{B}  \ottsym{)} $.
  In fact ${}\vdash \ottnt{v} :  \text{\unboldmath$\forall$}  \, \alpha  \ottsym{.} \, \ottsym{(}   \ottnt{A}  +  \ottnt{B}   \ottsym{)}$ implies $\ottnt{v} = \mathsf{inl} \, \ottnt{v'}$ or $\mathsf{inr} \, \ottnt{v''}$.
  Assuming the former for definiteness, $\alpha  \vdash  \ottnt{v'}  \ottsym{:}  \ottnt{A}$ and thus ${}\vdash \ottnt{v'} :  \text{\unboldmath$\forall$}  \, \alpha  \ottsym{.} \, \ottnt{A}$.
\end{remark}

%% file: sections/eff.tex
\section{Cooperation of Safe and Unsafe Effects}
\label{sec:eff}

This section describes an effect system for {\extlang}, which enables the type-safe
cooperation of safe and unsafe effects in a single program.  Our effect system
allows expressions to be
polymorphic if their evaluation performs only operations that satisfy the signature
restriction.  This capability makes it possible for the effect system to
incorporate value restriction---i.e., any value can be polymorphic.  The
definition of signature restriction changes to take into account effect
information on types.  Soundness of the effect system enables us to ensure that
programs handle all the operations performed at run time.

Our effect system is inspired by \citet{Kammar/Lindley/Oury_2013_ICFP}, where
the effect system tracks involved effect operations by their names together with
their type signatures.  There are, however,
two differences between Kammar et al.'s and our effect systems.  The first
difference comes from that of the evaluation strategies the calculi adopt: the
calculus of Kammar et al.\ is based on call-by-push-value
(CBPV)~\cite{Levy_2001_PhD} and we adopt call-by-value (CBV).  This difference
influences the design of effect systems because the two strategies have
different notions for the value representations of suspended computations and effect
systems have to manage the effects caused by their run.  CBPV views functions as
(not suspended) computations, and thus Kammar et al.\ did not equip function
types with effect information; instead, they augmented the types of thunks (which
are value representations of suspended computations in CBPV) with it.  By
contrast, because CBV views functions as values that represent suspended
computations, our effect system equips function types with effect information.
The second difference is that we include only operation names and not their type
signatures in the effect information.  This is merely for simplifying the presentation
but it makes the calculus non-terminating~\cite{Kammar/Pretnar_2017_JFP}.

\TS{Cite Pretnar's subtyping for effects.}

\subsection{Effect System}
\reffig{eff-typing} shows only the key part of the effect
system; the full definition is found in \ifappendix{Appendix~\ref{sec:app:defn:effect}}\else{the supplementary material}\fi.

\begin{figure}[t]
 \[\begin{array}{lll}
  \textbf{Effects} \quad \epsilon & ::= & \{ \mathsf{op}_{{\mathrm{1}}}, \cdots, \mathsf{op}_{\ottmv{n}} \} \\
  \textbf{Types} \quad \ottnt{A}, \ottnt{B}, \ottnt{C}, \ottnt{D} & ::= &
   \textgray{\alpha \mid \iota \mid}\;
    \ottnt{A}   \rightarrow ^{ \epsilon }  \ottnt{B} 
   \;\textgray{\mid  \text{\unboldmath$\forall$}  \, \alpha  \ottsym{.} \, \ottnt{A} \mid  \ottnt{A}  \times  \ottnt{B}  \mid  \ottnt{A}  +  \ottnt{B}  \mid  \ottnt{A}  \, \mathsf{list} }
 \end{array}\]
 \mbox{} \\

 \begin{flushleft}
  \textbf{Type containment} \quad
  \framebox{$\Gamma  \vdash  \ottnt{A}  \sqsubseteq  \ottnt{B}$}
 \end{flushleft}
 \begin{center}
  $\ottdruleCXXFunEff{}$ \hfil
  $\ottdruleCXXDFunEff{}$ \hfil
  $\cdots$
 \end{center}
 \mbox{} \\

 \begin{flushleft}
  \textbf{Term typing} \quad
  \framebox{$\Gamma  \vdash  \ottnt{M}  \ottsym{:}  \ottnt{A} \,  |  \, \epsilon$}
 \end{flushleft}
 \begin{center}
  $\ottdruleTeXXApp{}$ \hfil
  $\ottdruleTeXXGen{}$ \\[1.5ex]
  $\ottdruleTeXXHandle{}$ \\[1.5ex]
  $\ottdruleTeXXFix{}$ \hfil
  $\ottdruleTeXXWeak{}$ \hfil
  $\cdots$
 \end{center}
 \mbox{} \\

 \begin{flushleft}
  \textbf{Handler typing} \quad
  \framebox{$\Gamma  \vdash  \ottnt{H}  \ottsym{:}  \ottnt{A} \,  |  \, \epsilon  \Rightarrow  \ottnt{B} \,  |  \, \epsilon'$}
 \end{flushleft}
 \begin{center}
  $\ottdruleTHeXXReturn{}$ \\[1.5ex]
  $\ottdruleTHeXXOp{}$
 \end{center}

 \caption{The effect system (excerpt).}
 \label{fig:eff-typing}
\end{figure}

The type language includes effect information.  Effects, ranged over by $\epsilon$, are
finite sets of operations.  Function types are augmented with effects that may
be triggered in applying the functions of those types.

Typing judgments also incorporate effects.
A term typing judgment $\Gamma  \vdash  \ottnt{M}  \ottsym{:}  \ottnt{A} \,  |  \, \epsilon$ asserts that $\ottnt{M}$ is a
computation that produces a value of $\ottnt{A}$ possibly with effect $\epsilon$.
A handler typing judgment $\Gamma  \vdash  \ottnt{H}  \ottsym{:}  \ottnt{A} \,  |  \, \epsilon  \Rightarrow  \ottnt{B} \,  |  \, \epsilon'$ asserts that
$\ottnt{H}$ handles a computation that produces values of $\ottnt{A}$ possibly with
effect $\epsilon$ and the handling produces values of $\ottnt{B}$ possibly with
effect $\epsilon'$.  Type containment judgments $\Gamma  \vdash  \ottnt{A}  \sqsubseteq  \ottnt{B}$ and well-formedness
judgments $\vdash  \Gamma$ take the same forms as those of the polymorphic type system
in \refsec{polytype}.

Most of the typing rules for terms are almost the same as those of the
polymorphic type system except that they take into account effect information.
The rule \Te{App} shows how effects are incorporated into the typing rules: the
effect triggered by a term is determined by its subterms.  Besides, \Te{App}
requires effect $\epsilon'$ triggered by a function to be a subset of the effect
$\epsilon$ of the subterms.
The rule \Te{Gen} is the key of the effect system, allowing a term to have a
polymorphic type if it triggers only safe effects.  The safety of an effect
$\epsilon$ is checked by the predicate $\mathit{SR} \, \ottsym{(}  \epsilon  \ottsym{)}$, which asserts that any
operation in $\epsilon$ satisfies the signature restriction for the type language in
\reffig{eff-typing}; we will formalize $\mathit{SR} \, \ottsym{(}  \epsilon  \ottsym{)}$ after explaining the
type containment rules.  A byproduct of adopting \Te{Gen} is that the effect system
incorporates the value restriction~\cite{Tofte_1990_IC} successfully: it allows values to have
polymorphic types because the values perform no operation (thus, their effects
can be the empty set $ \emptyset $) and $\mathit{SR} \, \ottsym{(}   \emptyset   \ottsym{)}$ obviously holds.
The rule \Te{Fix} gives any effect $\epsilon'$ to the fixed-point operator.  This
means that the fixed-point operator can be viewed as a pure computation because
it only produces a lambda abstraction without triggering effects.  The rule
\Te{Weak} weakens the effect information of a term.

There are two rules for deriving a handler typing judgment $\Gamma  \vdash  \ottnt{H}  \ottsym{:}  \ottnt{A} \,  |  \, \epsilon  \Rightarrow  \ottnt{B} \,  |  \, \epsilon'$.  They state that the effect of a $ \mathsf{handle} $--$ \mathsf{with} $
expression installing $\ottnt{H}$ consists of the operations that the handled expression
may call but $\ottnt{H}$ does not handle and those that the return clause or some
operation clause of $\ottnt{H}$ may call.  The effect $\epsilon \,  \mathbin{\uplus}  \, \ottsym{\{}  \mathsf{op}  \ottsym{\}}$ is the same
as $\epsilon \,  \mathbin{\cup}  \, \ottsym{\{}  \mathsf{op}  \ottsym{\}}$ except that it requires $\mathsf{op} \,  \not\in  \, \epsilon$.

Most of the type containment rules of the effect system are the same as those of the
polymorphic type system.  The exception is the rules for function types
\Srule{Fun} and \Srule{DFun}, which are replaced by \Srule{FunEff} and
\Srule{DFunEff} to take into account effect information.  The rule
\Srule{DFunEff} for deriving $\Gamma  \vdash    \text{\unboldmath$\forall$}  \, \alpha  \ottsym{.} \, \ottnt{A}   \rightarrow ^{ \epsilon }  \ottnt{B}   \sqsubseteq   \ottnt{A}   \rightarrow ^{ \epsilon }   \text{\unboldmath$\forall$}  \, \alpha  \ottsym{.} \, \ottnt{B} $ has an
addition conditional that $\mathit{SR} \, \ottsym{(}  \epsilon  \ottsym{)}$ must hold.  This condition originates from
\Te{Gen}.  The rule \Srule{DFunEff} allows that, if a lambda abstraction
$ \lambda\!  \, \mathit{x}  \ottsym{.}  \ottnt{M}$ has a polymorphic type $  \text{\unboldmath$\forall$}  \, \alpha  \ottsym{.} \, \ottnt{A}   \rightarrow ^{ \epsilon }  \ottnt{B} $, the body $\ottnt{M}$ may also
have another polymorphic type $ \text{\unboldmath$\forall$}  \, \alpha  \ottsym{.} \, \ottnt{B}$.  In general, $\ottnt{M}$ may be a
non-value term.  In such a case, only \Te{Gen} justifies that $\ottnt{M}$ has a
polymorphic type; however, to apply \Te{Gen} the effect $\epsilon$ triggered by
$\ottnt{M}$ has to meet $\mathit{SR} \, \ottsym{(}  \epsilon  \ottsym{)}$.  This is the reason why \Srule{DFunEff}
requires that $\mathit{SR} \, \ottsym{(}  \epsilon  \ottsym{)}$ hold.

Now, we formalize the predicate $\mathit{SR} \, \ottsym{(}  \epsilon  \ottsym{)}$, which states that any operation
in $\epsilon$ satisfies signature restriction extended by effect information.  In
what follows, we suppose the notions of positive/negative/strictly positive
occurrences of a type variable for the type language in \reffig{eff-typing};
they are defined naturally as in Definitions~\ref{def:polarity} and
\ref{def:ext-polarity}.
In addition, we can decide whether a type occurs at a strictly positive position in a
type by generalizing Definitions~\ref{def:polarity} and \ref{def:ext-polarity}
from the occurrences of type variables to those of types.
%
%
\begin{restatable}[Effects satisfying signature restriction]{defn}{defnEffSignatureRestriction}
 \label{def:eff-signature-restriction}
 The predicate $\mathit{SR} \, \ottsym{(}  \epsilon  \ottsym{)}$ holds if and only if, for any $\mathsf{op} \,  \in  \, \epsilon$ such that
 $\mathit{ty} \, \ottsym{(}  \mathsf{op}  \ottsym{)} \,  =  \,   \text{\unboldmath$\forall$}  \,  \algeffseqover{ \alpha }   \ottsym{.} \,  \ottnt{A}  \hookrightarrow  \ottnt{B} $:
 \begin{itemize}
  \item the occurrences of each type variable of $ \algeffseqover{ \alpha } $ in $\ottnt{A}$ are only
        negative or strictly positive;
  \item the occurrences of each type variable of $ \algeffseqover{ \alpha } $ in $\ottnt{B}$ are only
        positive; and
  \item for any function type $ \ottnt{C}   \rightarrow ^{ \epsilon' }  \ottnt{D} $ occurring at a strictly positive
        position in $\ottnt{A}$, if $\ottsym{\{}   \algeffseqover{ \alpha }   \ottsym{\}} \,  \mathbin{\cap}  \,  \mathit{ftv}  (  \ottnt{D}  )  \,  \not=  \,  \emptyset $, then
        $\mathit{SR} \, \ottsym{(}  \epsilon'  \ottsym{)}$.
 \end{itemize}
\end{restatable}
The first and second conditions of \refdef{eff-signature-restriction} are the
same as those of \refdef{signature-restriction}, signature restriction without
effect information.  The third condition is necessary to apply \Srule{DFunEff}.
The signature restriction in the polymorphic type system allows type variables $ \algeffseqover{ \alpha } $ in type signature $  \text{\unboldmath$\forall$}  \,  \algeffseqover{ \alpha }   \ottsym{.} \,  \ottnt{A}  \hookrightarrow  \ottnt{B} $ to
occur at a strictly positive position in $\ottnt{A}$
(see \refdef{signature-restriction}).  As discussed in
\refsec{polytype:signature-restriction:prop}, this capability originates from \Srule{DFun}. In the
effect system, the counterpart \Srule{DFunEff} is applied to retain this
capability, but \Srule{DFunEff} requires the effect of a function type to
satisfy $ \mathit{SR} $.  This is the reason why the signature restriction for the effect
system imposes the third condition.  Note that, if type variables in $ \algeffseqover{ \alpha } $
do not occur free in $\ottnt{D}$ (and they do not in $\ottnt{C}$ either), then we can
derive $\Gamma  \vdash    \text{\unboldmath$\forall$}  \,  \algeffseqover{ \alpha }   \ottsym{.} \, \ottnt{C}   \rightarrow ^{ \epsilon' }  \ottnt{D}   \sqsubseteq   \ottnt{C}   \rightarrow ^{ \epsilon' }   \text{\unboldmath$\forall$}  \,  \algeffseqover{ \alpha }   \ottsym{.} \, \ottnt{D} $ without \Srule{DFunEff}.
Thus, the third condition does not require $\mathit{SR} \, \ottsym{(}  \epsilon'  \ottsym{)}$ if $\ottsym{\{}   \algeffseqover{ \alpha }   \ottsym{\}} \,  \mathbin{\cap}  \,  \mathit{ftv}  (  \ottnt{D}  )  \,  =  \,  \emptyset $.

We finally state soundness of the effect system, which ensures that a well-typed
program handles all the operations performed at run time.  We prove it by
progress and subject reduction; their formal statements and proofs are found in
\ifappendix{Appendix~\ref{sec:app:proof:eff-sound-typing}}\else{the supplementary material}\fi.
\begin{restatable}[Type soundness]{thm}{thmEffTypeSoundness}
 \label{thm:eff-type-sound}
 If\/ $\Delta  \vdash  \ottnt{M}  \ottsym{:}  \ottnt{A} \,  |  \,  \emptyset $ and $\ottnt{M}  \longrightarrow^{*}  \ottnt{M'}$ and $\ottnt{M'}  \centernot\longrightarrow$, then
 $\ottnt{M'}$ is a value.
\end{restatable}

\subsection{Example}
The effect system allows us to use both safe and unsafe effects in a single
program.  For example, let us consider the following program (which can be represented in
{\extlang}).
\begin{lstlisting}
 let f : (/$\forall \alpha .\, \alpha \rightarrow^{\{\texttt{get\_id}\}} \alpha$/) = (/$\lambda$/)x. #get_id() x in
 let g : (/$\forall \alpha .\, \alpha \rightarrow^{\{\texttt{get\_id}\}} \alpha$/) = #select([(/$\lambda$/)x. x; f]) in
 if g true then (g 2) + 1 else 0
\end{lstlisting}
This example would be rejected if we were to enforce all operations to follow the signature
restriction as in \refsec{polytype} because it uses the unsafe operation
\texttt{get\_id}.  By contrast, the effect system accepts it because: the
polymorphic expression
\lstinline[mathescape]{$\lambda$x. #get_id() x}
calls no operation and
\texttt{\#select([$\lambda$x. x; f])}
calls only \texttt{select}, which satisfies the signature restriction, during the
evaluation; therefore, they can have the polymorphic
type $\forall \alpha .\, \alpha \rightarrow^{\{\texttt{get\_id}\}} \alpha$
by \Te{Gen}.
Note that the effect system still rejects the counterexample given in
\refsec{overview:unsafety} because it disallows polymorphic expressions to call
operations that do not satisfy the signature restriction, such as \texttt{get\_id}.

%% file: sections/relwork.tex
\section{Related Work}
\label{sec:relwork}

\subsection{Restriction for the Use of Effects in Polymorphic Type Assignment}
\label{sec:relwork:restrict}
The problem that type safety is broken in naively combining polymorphic effects
and polymorphic type assignment was initially discovered in a language with polymorphic
references~\cite{Goron/Milner/Wadsworth_1979_book} and later in one with
polymorphic control operators~\cite{Harper/Lillibridge_1991_types,Harper/Lillibridge_1993_LSC}.  Researchers
have developed many approaches to reconcile these conflicting features thus
far~\cite{Tofte_1990_IC,Leroy/Weis_1991_POPL,Appel/MacQueen_1991_PLILP,Hoang/Mitchell/Viswanathan_1993_LICS,Wright_1995_LSC,Garrigue_2004_FLOPS,Asai/Kameyama_2007_APLAS,Sekiyama/Igarashi_2019_ESOP}.

A major direction shared among them is to prevent the generalization of type
variables assigned to an expression if the type variables are used to
instantiate polymorphic effects triggered by the expression.  \citet{Leroy/Weis_1991_POPL} called such type variables \emph{dangerous}.  The value
restriction~\cite{Tofte_1990_IC,Wright_1995_LSC}, which allows only syntactic
values to be polymorphic, is justified by this idea because these values trigger
no effect and therefore no dangerous type variable exists.  Similarly,
\citet{Asai/Kameyama_2007_APLAS} and \citet{Leijen_2017_POPL} allowed only
observationally pure expressions to be polymorphic.  \citet{Tofte_1990_IC}
proposed another approach that classifies type variables into applicative ones,
which cannot be used to instantiate effects, and imperative ones, which may be
used, and allows the generalization of only applicative type variables.  Weak
polymorphism~\cite{Appel/MacQueen_1991_PLILP,Hoang/Mitchell/Viswanathan_1993_LICS}
extends this idea by assigning a type variable the number of function
applications necessary to trigger effects instantiated with the type variable.
If the numbers assigned to type variables are positive, effects instantiated
with the type variables are not triggered; therefore, they are not dangerous and can be
generalized safely.  \citet{Leroy/Weis_1991_POPL} prevented the generalization of
dangerous type variables by making the type information of free variables in
closures accessible.
These approaches focused on a specific effect (especially, the ML-style
reference effect) basically, but they can be applied to other effects as well.
We prevent the generalization of dangerous type variables by \emph{closing} the type
arguments of an operation call at run time, as discussed in
\refsec{overview:our-work}.
This type transformation is not always acceptable, but we find that it is if the
operation satisfies the signature restriction.
\TS{
For example, let us consider the following
program, where we make type information explicit for clarification, i.e., we write
$\Lambda \alpha.\,\ottnt{M}$ for type abstraction, $\ottnt{M} \{\ottnt{A}\}$ for type
application, and
for a call of
\texttt{op} that instantiates the type signature of \texttt{op} with $\ottnt{A}$.
%
%
In this example, the type variable $\alpha$ is dangerous because it is used to
instantiate the type signature of \texttt{select} that is called (note that we now
suppose an evaluation strategy where the body of a type abstraction is
evaluated).  What our work does essentially is to convert the type information
in the above example at run time, as follows:
%
where the type argument $\alpha$ of the operation call is converted to a closed
type $\forall \alpha.\, \alpha$.  If this conversion preserves typing, the
program should be safe because there is no dangerous type variable left.  It
does not always hold unfortunately, but we find that the type signature of a
called operation works as a sufficient condition of when this conversion
preserves typing and formalize it as signature restriction.
}

\citet{Garrigue_2004_FLOPS} proposed the relaxed value restriction, which allows the generalization of type
variables assigned to an expression if the type variables occur only at positive
positions in the type of the expression.  The polarity condition on generalized
type variables makes it possible to use the empty type as a surrogate of the
type variables and, as a result, prevents instantiating effects with the type
variables.
The relaxed value restriction is similar to signature restriction in that both
utilize the polarity of type variables.  In fact, the \emph{strong} signature
restriction, introduced in \refsec{overview:our-work}, is explainable by using
the empty type {\emptytype} and subtyping $<:$ for it (i.e., deeming {\emptytype} a
subtype of any type) as in the relaxed value restriction.
First, let us recall the key idea of the strong signature restriction: it is to
rewrite an operation call
\texttt{$\Lambda \beta_1 \dots \beta_n.\,$\#op$\{C\}$($v$)}
for \texttt{op} : $\forall \alpha.\, A \hookrightarrow B$
to
\texttt{$\Lambda \beta_1 \dots \beta_n.\,$\#op$\{ \forall \beta_1 \dots \beta_n. \, C\}$($v$)}
to close the type argument $C$ and to use provable type containment judgments
$A[C/\alpha] \mathrel{ \sqsubseteq } A[\forall \beta_1 \dots \beta_n.\,C/\alpha]$ and
$B[\forall \beta_1 \dots \beta_n.\,C/\alpha] \mathrel{ \sqsubseteq } B[C/\alpha]$ for
typing the latter term.
We can rephrase this idea with {\emptytype}, instead of $\forall \beta_1 \dots
\beta_n. \, C$, as follows: the operation call is rewritten to \texttt{$\Lambda
\beta_1 \dots \beta_n.\,$\#op$\{ \emptytype \}$($v$)} and this term can be
typed by using the subtyping judgments
$A[C/\alpha] \mathrel{<:} A[{\emptytype}/\alpha]$ and
$B[{\emptytype}/\alpha] \mathrel{<:} B[C/\alpha]$,
which are provable owing to the polarity condition of the strong signature
restriction (i.e., $\alpha$ occurs only negatively in $A$ and only positively in $B$).
However, this argument does \emph{not} extend to the (non-strong) signature
restriction because it allows the bound type variable $\alpha$ to occur at
strictly positive positions in the domain type $\ottnt{A}$ and then $A[C/\alpha]
\mathrel{<:} A[{\emptytype}/\alpha]$ no longer holds.
Thus, our technical contributions include the findings that the type argument
$\ottnt{C}$ can be closed by \emph{quantifying} it and that $A[C/\alpha]
\mathrel{ \sqsubseteq } A[\forall \beta_1 \dots \beta_n.\,C/\alpha]$ is provable by
type containment, where the distributive law plays a key role.  This
change---which may seem minor perhaps---renders the signature restriction quite
permissive.

\citet{Sekiyama/Igarashi_2019_ESOP} followed another line of research: they
restricted the definitions of effects instead of their usage.  They also employed
algebraic effects and handlers to accommodate effect definitions in a programmable
way and provided a type system that accepts only effects such that programs do not
get stuck even if they are instantiated with dangerous type
variables.  However, a problem with their work is that all effects have to be safe for any
usage and a program cannot use both safe and unsafe effects.  Our work, by contrast,
provides an effect system that allows the use of both operations that satisfy and do
not satisfy the signature restriction---inasmuch as they are performed appropriately.
The effect system utilizes the benefit of the signature restriction that it only
depends on the type interfaces of effects.
\TS{
A problem with their work is that resumption arguments cannot refer to
values computed with the operation arguments.  For example, let us consider an
operation \texttt{choose} with type signature $ \text{\unboldmath$\forall$}  \, \alpha  \ottsym{.} \,  \alpha  \times  \alpha   \rightarrow  \alpha$.  Their type
system accepts an operation clause
%
%
(which returns the first component of a given pair to the caller of
\texttt{choose}), but it rejects
%
%
(which is obtained from the first operation clause just by naming the resumption
argument $\pi_1$ \texttt{x} variable \texttt{y}) because \texttt{y} denotes the
value of $\pi_1$ \texttt{x} computed with operation argument \texttt{x}.  This
example indicates that the type system in the prior work is too sensitive to
slight change of effect implementations and, even worse, is not very
permissive---for example, the operation clause for \texttt{select} given in
\refsec{overview:algeff} would be ill typed because it uses each element of a
given list as a resumption argument.
}

Effect systems have been used to safely introduce effects
in polymorphic type assignment thus far.  \citet{Asai/Kameyama_2007_APLAS} and
\citet{Leijen_2017_POPL} utilized effect systems for the control operators
\textsf{shift/reset}~\cite{Danvy/Filinski_1990_LFP} and algebraic effects and
handlers to ensure that polymorphic expressions are observationally pure,
respectively.
\citet{Kammar/Pretnar_2017_JFP} proposed an effect system for \emph{parameterized}
algebraic effects, which are declared with type parameters and invoked with type
arguments.  Unlike polymorphic effects, parameterized effects invoked
with different type arguments are deemed different.  Kammar and Pretnar utilized
the effect system to prevent the generalization of the type variables involved by type
arguments of parameterized effects.

\subsection{User-Defined Effects}
Our work employs algebraic effects and handlers as a technical development to
describe a variety of effects.  Algebraic effects were originally proposed
as a denotational framework to describe the meaning of an effect by separating
the interface of an effect, which is given by a set of operations, and its
interpretation, which is given by the equational theory over the
operations~\cite{Plotkin/Power_2003_ACS}.
\citet{Plotkin/Pretnar_2013_LMCS,Plotkin/Pretnar_2009_ESOP} introduced effect
handlers in order to represent the semantics of exception handling in an
equational theory.  The idea of separating an effect interface and its
interpretation makes it possible to handle user-defined effects in a modular
way and encourages the emergence of languages equipped with algebraic effect handlers,
such as Eff~\cite{Bauer/Pretnar_2015_JLAMP}, Koka~\cite{Leijen_2017_POPL},
Frank~\cite{Lindley/MacBrid/McLaughlin_2017_POPL}, Multicore
OCaml~\cite{multicoreOCaml_2017_TFP}.  We also utilize the separation and
restrict only effect interfaces in order to achieve type safety in polymorphic
type assignment.

Another approach to user-defined effects is to use control operators, which
enable programmers to make access to continuations.  Many control operators have
been developed thus far---e.g., \textsf{call/cc}~\cite{Clinger/Friedman/Wand_1985},
\textsf{control/prompt}~\cite{Felleisen_1988_POPL},
\textsf{shift/reset}~\cite{Danvy/Filinski_1990_LFP},
\textsf{fcontrol/run}~\cite{Sitaram_1993_PLDI}, and
\textsf{cupto/prompt}~\cite{Gunter/Remy/Riecke_1995_FPCA}.
These operators are powerful and generic, but, in return for that, it is unsafe
to naively combine them with polymorphic type
assignment~\cite{Harper/Lillibridge_1993_LSC}.  They do not provide a means to
assign individual effects fine-grained type interfaces.  Thus, it is not clear how to
apply the idea of signature restriction for the effects implemented by control
operators.

Monads can also express the interpretation of an effect in a
denotational manner~\cite{Moggi_1991_IC} and have been used as a
long-established, programmable means for user-defined
effects~\cite{Wadler_1992_POPL,Peython-Jones/Wadler_1993_POPL}.
\citet{Filinski_2010_POPL} extracted the essence of monadic effects and proposed a
language equipped with a type system and an operation semantics for them.  We
expect our idea of restriction on effect interfaces to be applicable to monadic
effects as well, but for that we would first need to consider how to introduce
polymorphic effects into a monadic language because Filinski's language
supports parametric effects but not polymorphic effects.

%% file: sections/conclusion.tex
\section{Conclusion}
\label{sec:conclusion}

This work addresses a classic problem with polymorphic effects in
polymorphic type assignment.  Our key idea is to restrict the type interfaces of
effects.  We formalize our idea with polymorphic algebraic effects and handlers,
propose the signature restriction, which restricts the type signatures of operations
by the polarity of occurrences of quantified type variables, and prove that a
polymorphic type system is sound if all operations satisfy the signature restriction.
We also give an effect system in which operations performed by polymorphic
expressions have to satisfy the signature restriction but those performed by monomorphic
expressions do not have.  This effect system enables us to use both
operations that satisfy and do not satisfy the signature restriction in a single
program safely.

There are several directions for future work.
First, we are interested in analyzing the
signature restriction from a more semantic perspective.  For example, the
semantics of a language with control effects is often given by transformation to
continuation-passing style (CPS).  It would be interesting to study CPS
transformation for implicit polymorphism by taking the signature
restriction into account.
Another direction would be to apply the signature restriction to evaluation strategies other than
call-by-value.
\citet{Harper/Lillibridge_1993_POPL} showed that polymorphic type assignment and
the polymorphic version of the control operator
\textsf{call/cc} can be reconciled safely in
call-by-name at the cost of expressivity
and by changing the timing of type instantiation slightly.  However, it is
unclear---and we would imagine impossible---whether similar reconcilement is
achievable in other strategies such as call-by-need and call-by-push-value.
Exporting the idea of signature restriction to other evaluation strategies would
be beneficial also for testing the robustness and developing a more in-depth understanding of signature restriction.

%% file: appendix/defn.tex
\section{Definition}

\subsection{Syntax}

\[
 \begin{array}{lll}
  \multicolumn{3}{l}{
   \textbf{Variables} \quad \mathit{x}, \mathit{y}, \mathit{z}, \mathit{f}, \mathit{k} \qquad
   \textbf{Type variables} \quad \alpha, \beta, \gamma \qquad
   \textbf{Effect operations} \quad \mathsf{op}
  } \\
  \textbf{Base types} \quad \iota & ::= &  \mathsf{bool}  \mid  \mathsf{int}  \mid ... \\
  \textbf{Types} \quad \ottnt{A}, \ottnt{B}, \ottnt{C}, \ottnt{D} & ::= &
   \alpha \mid \iota \mid \ottnt{A}  \rightarrow  \ottnt{B} \mid  \text{\unboldmath$\forall$}  \, \alpha  \ottsym{.} \, \ottnt{A} \mid
    \ottnt{A}  \times  \ottnt{B}  \mid  \ottnt{A}  +  \ottnt{B}  \mid  \ottnt{A}  \, \mathsf{list} 
   \\
   \textbf{Constants} \quad \ottnt{c} & ::= &
    \mathsf{true}  \mid  \mathsf{false}  \mid  0  \mid  \mathsf{+}  \mid ... \\
   \textbf{Terms} \quad \ottnt{M} & ::= &
   \mathit{x} \mid \ottnt{c} \mid
    \lambda\!  \, \mathit{x}  \ottsym{.}  \ottnt{M} \mid \ottnt{M_{{\mathrm{1}}}} \, \ottnt{M_{{\mathrm{2}}}} \mid
    \textup{\texttt{\#}\relax}  \mathsf{op}   \ottsym{(}   \ottnt{M}   \ottsym{)}  \mid
   \mathsf{handle} \, \ottnt{M} \, \mathsf{with} \, \ottnt{H} \mid \\ &&
   \ottsym{(}  \ottnt{M_{{\mathrm{1}}}}  \ottsym{,}  \ottnt{M_{{\mathrm{2}}}}  \ottsym{)} \mid \pi_1  \ottnt{M} \mid \pi_2  \ottnt{M} \mid \\ &&
   \mathsf{inl} \, \ottnt{M} \mid \mathsf{inr} \, \ottnt{M} \mid
   \mathsf{case} \, \ottnt{M} \, \mathsf{of} \, \mathsf{inl} \, \mathit{x}  \rightarrow  \ottnt{M_{{\mathrm{1}}}}  \ottsym{;} \, \mathsf{inr} \, \mathit{y}  \rightarrow  \ottnt{M_{{\mathrm{2}}}} \mid \\ &&
    \mathsf{nil}  \mid \mathsf{cons} \, \ottnt{M} \mid
   \mathsf{case} \, \ottnt{M} \, \mathsf{of} \, \mathsf{nil} \, \rightarrow  \ottnt{M_{{\mathrm{1}}}}  \ottsym{;} \, \mathsf{cons} \, \mathit{x}  \rightarrow  \ottnt{M_{{\mathrm{2}}}} \mid
   \mathsf{fix} \, \mathit{f}  \ottsym{.}   \lambda\!  \, \mathit{x}  \ottsym{.}  \ottnt{M}
   \\
   \textbf{Handlers} \quad \ottnt{H} & ::= &
   \mathsf{return} \, \mathit{x}  \rightarrow  \ottnt{M} \mid \ottnt{H}  \ottsym{;}  \mathsf{op}  \ottsym{(}  \mathit{x}  \ottsym{,}  \mathit{k}  \ottsym{)}  \rightarrow  \ottnt{M}
   \\
   \textbf{Values} \quad \ottnt{v} & ::= &
   \ottnt{c} \mid
    \lambda\!  \, \mathit{x}  \ottsym{.}  \ottnt{M} \mid 
   \ottsym{(}  \ottnt{v_{{\mathrm{1}}}}  \ottsym{,}  \ottnt{v_{{\mathrm{2}}}}  \ottsym{)} \mid
   \mathsf{inl} \, \ottnt{v} \mid \mathsf{inr} \, \ottnt{v} \mid  \mathsf{nil}  \mid \mathsf{cons} \, \ottnt{v}
   \\[.5ex]
   \textbf{Typing contexts} \quad \Gamma & ::= &
    \emptyset  \mid \Gamma  \ottsym{,}  \mathit{x} \,  \mathord{:}  \, \ottnt{A} \mid 
   \Gamma  \ottsym{,}  \alpha
   \\
   \textbf{Evaluation contexts} \quad
    \ottnt{E} & ::= &  []  \mid
                  \ottnt{E} \, \ottnt{M_{{\mathrm{2}}}} \mid \ottnt{v_{{\mathrm{1}}}} \, \ottnt{E} \mid
                   \textup{\texttt{\#}\relax}  \mathsf{op}   \ottsym{(}   \ottnt{E}   \ottsym{)}  \mid \mathsf{handle} \, \ottnt{E} \, \mathsf{with} \, \ottnt{H} \mid \\ &&
                  \ottsym{(}  \ottnt{E}  \ottsym{,}  \ottnt{M_{{\mathrm{2}}}}  \ottsym{)} \mid \ottsym{(}  \ottnt{v_{{\mathrm{1}}}}  \ottsym{,}  \ottnt{E}  \ottsym{)} \mid
                  \pi_1  \ottnt{E} \mid \pi_2  \ottnt{E} \mid \\ &&
                  \mathsf{inl} \, \ottnt{E} \mid \mathsf{inr} \, \ottnt{E} \mid
                  \mathsf{case} \, \ottnt{E} \, \mathsf{of} \, \mathsf{inl} \, \mathit{x}  \rightarrow  \ottnt{M_{{\mathrm{1}}}}  \ottsym{;} \, \mathsf{inr} \, \mathit{y}  \rightarrow  \ottnt{M_{{\mathrm{2}}}} \mid \\ &&
                  \mathsf{cons} \, \ottnt{E} \mid
                  \mathsf{case} \, \ottnt{E} \, \mathsf{of} \, \mathsf{nil} \, \rightarrow  \ottnt{M_{{\mathrm{1}}}}  \ottsym{;} \, \mathsf{cons} \, \mathit{x}  \rightarrow  \ottnt{M_{{\mathrm{2}}}}
 \end{array}
\]

\begin{conv} \label{conv:surface}
 This work follows the conventions as below.
 \begin{itemize}
  \item  We write $ \algeffseqoverindex{ \alpha }{ \text{\unboldmath$\mathit{I}$} } $ for $ \algeffseqover{ \alpha }  = \alpha_{{\mathrm{1}}}, \cdots, \alpha_{\ottmv{n}}$
         with $\text{\unboldmath$\mathit{I}$} = \{ 1, ..., n \}$.
         We often omit index sets ($\text{\unboldmath$\mathit{I}$}$ and $\text{\unboldmath$\mathit{J}$}$) if they are not
         important: for example, we often abbreviate $ \algeffseqoverindex{ \alpha }{ \text{\unboldmath$\mathit{I}$} } $ to $ \algeffseqover{ \alpha } $.
         We apply this bold-font notation to other syntax categories as well;
         for example, $ \algeffseqoverindex{ \ottnt{A} }{ \text{\unboldmath$\mathit{I}$} } $ denotes a sequence of types.
         %

  \item We write $\{s\}$ to view the sequence $s$ as a set by ignoring
        the order.

  \item We write $ \text{\unboldmath$\forall$}  \,  \algeffseqoverindex{ \alpha }{ \text{\unboldmath$\mathit{I}$} }   \ottsym{.} \, \ottnt{A}$ for
        $  \text{\unboldmath$\forall$}    \alpha_{{\mathrm{1}}}  . \, ... \,   \text{\unboldmath$\forall$}    \alpha_{\ottmv{n}}  .  \, \ottnt{A}$ with $\text{\unboldmath$\mathit{I}$} = \{ 1, ..., n \}$.
        We may omit index sets ($ \text{\unboldmath$\forall$}  \,  \algeffseqover{ \alpha }   \ottsym{.} \, \ottnt{A}$).
        We write $ \algeffseqoverindex{  \text{\unboldmath$\forall$}  \,  \algeffseqoverindex{ \alpha }{ \text{\unboldmath$\mathit{I}$} }   \ottsym{.} \, \ottnt{A} }{ \text{\unboldmath$\mathit{J}$} } $ for a sequence of types
        $ \text{\unboldmath$\forall$}  \,  \algeffseqoverindex{ \alpha }{ \text{\unboldmath$\mathit{I}$} }   \ottsym{.} \, \ottnt{A}_1$, \ldots, $ \text{\unboldmath$\forall$}  \,  \algeffseqoverindex{ \alpha }{ \text{\unboldmath$\mathit{I}$} }   \ottsym{.} \, \ottnt{A}_n$ with $\text{\unboldmath$\mathit{J}$} =
        \{1,\ldots,n\}$.


  \item We write $\Gamma_{{\mathrm{1}}}  \ottsym{,}  \Gamma_{{\mathrm{2}}}$ for the concatenation of $\Gamma_{{\mathrm{1}}}$ and $\Gamma_{{\mathrm{2}}}$, and
        $\mathit{x} \,  \mathord{:}  \, \ottnt{A}$ and $ \algeffseqover{ \alpha } $ for $\ottsym{(}   \emptyset   \ottsym{,}  \mathit{x} \,  \mathord{:}  \, \ottnt{A}  \ottsym{)}$,
        $\ottsym{(}   \emptyset   \ottsym{,}   \algeffseqover{ \alpha }   \ottsym{)}$, respectively.
        %

  \item We write $ \ottnt{H} ^\mathsf{return} $ for the return clause in $\ottnt{H}$ and $\ottnt{H}  \ottsym{(}  \mathsf{op}  \ottsym{)}$
         for the operation clause of $\mathsf{op}$ in $\ottnt{H}$.


 \end{itemize}
\end{conv}


\begin{defn}[Domain of typing contexts] \label{defn:tctx-domain}
 We define $ \mathit{dom}  (  \Gamma  ) $ as follows.
 \[\begin{array}{lll}
   \mathit{dom}  (   \emptyset   )        &\defeq&  \emptyset  \\
   \mathit{dom}  (  \Gamma  \ottsym{,}  \mathit{x} \,  \mathord{:}  \, \ottnt{A}  )       &\defeq&  \mathit{dom}  (  \Gamma  )  \,  \mathbin{\cup}  \, \ottsym{\{}  \mathit{x}  \ottsym{\}} \\
   \mathit{dom}  (  \Gamma  \ottsym{,}  \alpha  )         &\defeq&  \mathit{dom}  (  \Gamma  )  \,  \mathbin{\cup}  \, \ottsym{\{}  \alpha  \ottsym{\}} \\
   \end{array}\]
\end{defn}

\begin{defn}[Free type variables and type substitution in types]
 Free type variables $ \mathit{ftv}  (  \ottnt{A}  ) $ in a type $\ottnt{A}$ and type substitution
 $ \ottnt{B}    [   \algeffseqover{ \ottnt{A} }   \ottsym{/}   \algeffseqover{ \alpha }   ]  $ of types $ \algeffseqover{ \ottnt{A} } $ for type variables $ \algeffseqover{ \alpha } $ in $\ottnt{B}$ are
 defined as usual.
 Type $\ottnt{A}$ is closed if and only if $ \mathit{ftv}  (  \ottnt{A}  ) $ is empty.
\end{defn}

\begin{assum} \label{assum:const}
 We suppose that each constant $\ottnt{c}$ is assigned a first-order closed type
 $ \mathit{ty}  (  \ottnt{c}  ) $ of the form $\iota  \rightarrow  \ldots  \rightarrow  \iota_{\ottmv{n}}  \rightarrow   \iota _{ \ottmv{n}  \ottsym{+}  \ottsym{1} } $.
 %
 We also suppose that, for any $\iota$, there exists the set $ \mathbb{K}_{ \iota } $ of
 constants of $\iota$.
 For any constant $\ottnt{c}$, $ \mathit{ty}  (  \ottnt{c}  )  \,  =  \, \iota$ if and only if $\ottnt{c} \,  \in  \,  \mathbb{K}_{ \iota } $.
 The function $ \zeta $ gives a denotation to pairs of constants.
 In particular, for any constants $\ottnt{c_{{\mathrm{1}}}}$ and $\ottnt{c_{{\mathrm{2}}}}$:
 (1) $ \zeta  (  \ottnt{c_{{\mathrm{1}}}}  ,  \ottnt{c_{{\mathrm{2}}}}  ) $ is defined if and only if
 $ \mathit{ty}  (  \ottnt{c_{{\mathrm{1}}}}  )  \,  =  \, \iota_{{\mathrm{0}}}  \rightarrow  \ottnt{A}$ and $ \mathit{ty}  (  \ottnt{c_{{\mathrm{2}}}}  )  \,  =  \, \iota_{{\mathrm{0}}}$ for some
 $\iota_{{\mathrm{0}}}$ and $\ottnt{A}$; and
 (2) if $ \zeta  (  \ottnt{c_{{\mathrm{1}}}}  ,  \ottnt{c_{{\mathrm{2}}}}  ) $ is defined, $ \zeta  (  \ottnt{c_{{\mathrm{1}}}}  ,  \ottnt{c_{{\mathrm{2}}}}  ) $ is a constant and
 $ \mathit{ty}  (   \zeta  (  \ottnt{c_{{\mathrm{1}}}}  ,  \ottnt{c_{{\mathrm{2}}}}  )   )  \,  =  \, \ottnt{A}$ where $ \mathit{ty}  (  \ottnt{c_{{\mathrm{1}}}}  )  \,  =  \, \iota_{{\mathrm{0}}}  \rightarrow  \ottnt{A}$.
\end{assum}

\begin{defn}[Polarity of type variable occurrence]
 The positive and negative occurrences of a type variable in a type $\ottnt{A}$ are defined
 by induction on $\ottnt{A}$, as follows.
 \begin{itemize}
  \item The occurrence of $\alpha$ in type $\alpha$ is positive.

  \item The positive (resp.\ negative) occurrences of $\alpha$ in $\ottnt{A}  \rightarrow  \ottnt{B}$
        are the negative (resp.\ positive) occurrences of $\alpha$ in
        $\ottnt{A}$ and the positive (resp.\ negative) occurrences of $\alpha$ in
        $\ottnt{B}$.

  \item The positive (resp.\ negative) occurrences of $\alpha$ in $ \text{\unboldmath$\forall$}  \, \beta  \ottsym{.} \, \ottnt{A}$,
        where $\beta$ is supposed to be distinct from $\alpha$,
        are the positive (resp.\ negative) occurrences of $\alpha$ in
        $\ottnt{A}$.

  \item The positive (resp.\ negative) occurrences of $\alpha$ in $ \ottnt{A}  \times  \ottnt{B} $
        are the positive (resp.\ negative) occurrences of $\alpha$ in
        $\ottnt{A}$ and those in $\ottnt{B}$.

  \item The positive (resp.\ negative) occurrences of $\alpha$ in $ \ottnt{A}  +  \ottnt{B} $
        are the positive (resp.\ negative) occurrences of $\alpha$ in
        $\ottnt{A}$ and those in $\ottnt{B}$.

  \item The positive (resp.\ negative) occurrences of $\alpha$ in $ \ottnt{A}  \, \mathsf{list} $
        are the positive (resp.\ negative) occurrences of $\alpha$ in
        $\ottnt{A}$.
 \end{itemize}

 The strictly positive occurrences of a type variable in a type $\ottnt{A}$ are
 defined by induction on $\ottnt{A}$, as follows.
 \begin{itemize}
  \item The occurrence of $\alpha$ in type $\alpha$ is strictly positive.

  \item The strictly positive occurrences of $\alpha$ in $\ottnt{A}  \rightarrow  \ottnt{B}$
        are the strictly positive occurrences of $\alpha$ in
        $\ottnt{B}$.

  \item The strictly positive occurrences of $\alpha$ in $ \text{\unboldmath$\forall$}  \, \beta  \ottsym{.} \, \ottnt{A}$,
        where $\beta$ is supposed to be distinct from $\alpha$,
        are the strictly positive occurrences of $\alpha$ in
        $\ottnt{A}$.

  \item The strictly positive occurrences of $\alpha$ in $ \ottnt{A}  \times  \ottnt{B} $
        are the strictly positive occurrences of $\alpha$ in
        $\ottnt{A}$ and those in $\ottnt{B}$.

  \item The strictly positive occurrences of $\alpha$ in $ \ottnt{A}  +  \ottnt{B} $
        are the strictly positive occurrences of $\alpha$ in
        $\ottnt{A}$ and those in $\ottnt{B}$.

  \item The strictly positive occurrences of $\alpha$ in $ \ottnt{A}  \, \mathsf{list} $
        are the strictly positive occurrences of $\alpha$ in
        $\ottnt{A}$.
 \end{itemize}
\end{defn}

\ifrestate
\defnTypeSignature*
\else
\begin{defn}[Type signature]
 \label{def:eff}
 Each effect operation $\mathsf{op}$ is assigned a type signature $\mathit{ty} \, \ottsym{(}  \mathsf{op}  \ottsym{)}$ of the
 form $   \text{\unboldmath$\forall$}    \alpha_{{\mathrm{1}}}  . \, ... \,   \text{\unboldmath$\forall$}    \alpha_{\ottmv{n}}  .  \,  \ottnt{A}  \hookrightarrow  \ottnt{B} $ for some $\ottmv{n}$, where $\alpha_{{\mathrm{1}}}  \ottsym{,}  ...  \ottsym{,}  \alpha_{\ottmv{n}}$ are bound
 in the domain type $\ottnt{A}$ and codomain type $\ottnt{B}$.  It may be abbreviated to
 $  \text{\unboldmath$\forall$}  \,  \algeffseqoverindex{ \alpha }{ \text{\unboldmath$\mathit{I}$} }   \ottsym{.} \,  \ottnt{A}  \hookrightarrow  \ottnt{B} $ or, more simply, to $  \text{\unboldmath$\forall$}  \,  \algeffseqover{ \alpha }   \ottsym{.} \,  \ottnt{A}  \hookrightarrow  \ottnt{B} $.
 We suppose that $   \text{\unboldmath$\forall$}    \alpha_{{\mathrm{1}}}  . \, ... \,   \text{\unboldmath$\forall$}    \alpha_{\ottmv{n}}  .  \,  \ottnt{A}  \hookrightarrow  \ottnt{B} $ is closed, i.e., $ \mathit{ftv}  (  \ottnt{A}  ) ,
  \mathit{ftv}  (  \ottnt{B}  )  \subseteq \{ \alpha_{{\mathrm{1}}}, \cdots, \alpha_{\ottmv{n}} \}$.
\end{defn}
\fi

\ifrestate
\defnSignatureRestriction*
\else
\begin{defn}[Operations satisfying signature restriction]
 \label{def:signature-restriction}
 An operation $\mathsf{op}$ having type signature $\mathit{ty} \, \ottsym{(}  \mathsf{op}  \ottsym{)} \,  =  \,   \text{\unboldmath$\forall$}  \,  \algeffseqover{ \alpha }   \ottsym{.} \,  \ottnt{A}  \hookrightarrow  \ottnt{B} $ satisfies
 the signature restriction if and only if:
 \begin{itemize}
  \item the occurrences of each type variable of $ \algeffseqover{ \alpha } $ in $\ottnt{A}$ are only
        negative or strictly positive; and
  \item the occurrences of each type variable of $ \algeffseqover{ \alpha } $ in $\ottnt{B}$ are only
        positive.
 \end{itemize}
\end{defn}
\fi

\subsection{Semantics}

\begin{figure}[t!]
 \begin{flushleft}
 \textbf{Reduction rules} \quad \framebox{$\ottnt{M_{{\mathrm{1}}}}  \rightsquigarrow  \ottnt{M_{{\mathrm{2}}}}$}
 \end{flushleft}
 \[\begin{array}{rcll}
  \ottnt{c} \, \ottnt{v}                    &  \rightsquigarrow  &  \zeta  (  \ottnt{c}  ,  \ottnt{v}  )    & \RwoP{Const} \\
  \ottsym{(}   \lambda\!  \, \mathit{x}  \ottsym{.}  \ottnt{M}  \ottsym{)} \, \ottnt{v}               &  \rightsquigarrow  &  \ottnt{M}    [  \ottnt{v}  /  \mathit{x}  ]        & \RwoP{Beta} \\
  %
  %
  \mathsf{handle} \, \ottnt{v} \, \mathsf{with} \, \ottnt{H}        &  \rightsquigarrow  &  \ottnt{M}    [  \ottnt{v}  /  \mathit{x}  ]        & \RwoP{Return} \\
   \multicolumn{3}{r}{\text{(where $ \ottnt{H} ^\mathsf{return}  \,  =  \, \mathsf{return} \, \mathit{x}  \rightarrow  \ottnt{M}$)}} \\
  \mathsf{handle} \,  \ottnt{E}  [   \textup{\texttt{\#}\relax}  \mathsf{op}   \ottsym{(}   \ottnt{v}   \ottsym{)}   ]  \, \mathsf{with} \, \ottnt{H} &  \rightsquigarrow  &   \ottnt{M}    [  \ottnt{v}  /  \mathit{x}  ]      [   \lambda\!  \, \mathit{y}  \ottsym{.}  \mathsf{handle} \,  \ottnt{E}  [  \mathit{y}  ]  \, \mathsf{with} \, \ottnt{H}  /  \mathit{k}  ]   & \RwoP{Handle} \\
   \multicolumn{3}{r}{\text{(where $\mathsf{op} \,  \not\in  \, \ottnt{E}$ and $\ottnt{H}  \ottsym{(}  \mathsf{op}  \ottsym{)} \,  =  \, \mathsf{op}  \ottsym{(}  \mathit{x}  \ottsym{,}  \mathit{k}  \ottsym{)}  \rightarrow  \ottnt{M}$)}} \\
  \pi_1  \ottsym{(}  \ottnt{v_{{\mathrm{1}}}}  \ottsym{,}  \ottnt{v_{{\mathrm{2}}}}  \ottsym{)}       &  \rightsquigarrow  & \ottnt{v_{{\mathrm{1}}}} & \RwoP{Proj1} \\
  \pi_2  \ottsym{(}  \ottnt{v_{{\mathrm{1}}}}  \ottsym{,}  \ottnt{v_{{\mathrm{2}}}}  \ottsym{)}       &  \rightsquigarrow  & \ottnt{v_{{\mathrm{2}}}} & \RwoP{Proj2} \\
  \mathsf{case} \, \mathsf{inl} \, \ottnt{v} \, \mathsf{of} \, \mathsf{inl} \, \mathit{x}  \rightarrow  \ottnt{M_{{\mathrm{1}}}}  \ottsym{;} \, \mathsf{inr} \, \mathit{y}  \rightarrow  \ottnt{M_{{\mathrm{2}}}} &  \rightsquigarrow  &  \ottnt{M_{{\mathrm{1}}}}    [  \ottnt{v}  /  \mathit{x}  ]   & \RwoP{CaseL} \\
  \mathsf{case} \, \mathsf{inr} \, \ottnt{v} \, \mathsf{of} \, \mathsf{inl} \, \mathit{x}  \rightarrow  \ottnt{M_{{\mathrm{1}}}}  \ottsym{;} \, \mathsf{inr} \, \mathit{y}  \rightarrow  \ottnt{M_{{\mathrm{2}}}} &  \rightsquigarrow  &
     \ottnt{M_{{\mathrm{2}}}}    [  \ottnt{v}  /  \mathit{y}  ]   & \RwoP{CaseR} \\
  \mathsf{case} \, \mathsf{nil} \, \mathsf{of} \, \mathsf{nil} \, \rightarrow  \ottnt{M_{{\mathrm{1}}}}  \ottsym{;} \, \mathsf{cons} \, \mathit{x}  \rightarrow  \ottnt{M_{{\mathrm{2}}}}    &  \rightsquigarrow  & \ottnt{M_{{\mathrm{1}}}} & \RwoP{Nil} \\
  \mathsf{case} \, \mathsf{cons} \, \ottnt{v} \, \mathsf{of} \, \mathsf{nil} \, \rightarrow  \ottnt{M_{{\mathrm{1}}}}  \ottsym{;} \, \mathsf{cons} \, \mathit{x}  \rightarrow  \ottnt{M_{{\mathrm{2}}}} &  \rightsquigarrow  &  \ottnt{M_{{\mathrm{2}}}}    [  \ottnt{v}  /  \mathit{x}  ]   & \RwoP{Cons} \\
  \mathsf{fix} \, \mathit{f}  \ottsym{.}   \lambda\!  \, \mathit{x}  \ottsym{.}  \ottnt{M}                           &  \rightsquigarrow  &  \ottsym{(}   \lambda\!  \, \mathit{x}  \ottsym{.}  \ottnt{M}  \ottsym{)}    [  \mathsf{fix} \, \mathit{f}  \ottsym{.}   \lambda\!  \, \mathit{x}  \ottsym{.}  \ottnt{M}  /  \mathit{f}  ]   & \RwoP{Fix} \\
   \end{array}\]
 \begin{flushleft}
 \textbf{Evaluation rules} \quad \framebox{$\ottnt{M_{{\mathrm{1}}}}  \longrightarrow  \ottnt{M_{{\mathrm{2}}}}$}
 \end{flushleft}
 \begin{center}
  $\ottdruleEXXEval{}$
 \end{center}
 \caption{Semantics.}
 \label{fig:app-inter-semantics}
\end{figure}

\ifrestate
\defnOpFreeEvalCtx*
\else
\begin{defn}[$\mathsf{op}$-free evaluation contexts]
 Evaluation context $\ottnt{E}$ is \emph{$\mathsf{op}$-free}, written $\mathsf{op} \,  \not\in  \, \ottnt{E}$,
 if and only if, there exist no $\ottnt{E_{{\mathrm{1}}}}$, $\ottnt{E_{{\mathrm{2}}}}$, and $\ottnt{H}$ such that
 $\ottnt{E} \,  =  \,  \ottnt{E_{{\mathrm{1}}}}  [  \mathsf{handle} \, \ottnt{E_{{\mathrm{2}}}} \, \mathsf{with} \, \ottnt{H}  ] $ and $\ottnt{H}$ has an operation clause for
 $\mathsf{op}$.
\end{defn}
\fi

\begin{defn}
 %
 Relations $ \longrightarrow $ and $ \rightsquigarrow $ are the smallest relations
 satisfying the rules in \reffig{app-inter-semantics}.
\end{defn}

\begin{defn}[Multi-step evaluation]
 Binary relation $ \longrightarrow^{*} $ over terms is the reflexive and transitive closure of $ \longrightarrow $.
\end{defn}

\begin{defn}[Nonreducible terms]
 We write $\ottnt{M}  \centernot\longrightarrow$ if there exists no term $\ottnt{M'}$ such that $\ottnt{M}  \longrightarrow  \ottnt{M'}$.
\end{defn}

\subsection{Typing}

\begin{figure}[t]
 \begin{flushleft}
 \textbf{Type containment} \quad \framebox{$\Gamma  \vdash  \ottnt{A}  \sqsubseteq  \ottnt{B}$}
 \end{flushleft}
 \begin{center}
  $\ottdruleCXXRefl{}$ \hfil
  $\ottdruleCXXTrans{}$ \hfil
  $\ottdruleCXXFun{}$ \\[2ex]
  $\ottdruleCXXInst{}$ \hfil
  $\ottdruleCXXGen{}$ \hfil
  $\ottdruleCXXPoly{}$ \\[2ex]
  $\ottdruleCXXProd{}$ \hfil
  $\ottdruleCXXSum{}$ \hfil
  $\ottdruleCXXList{}$ \\[2ex]
  $\ottdruleCXXDFun{}$ \hfil
  $\ottdruleCXXDProd{}$ \\[2ex]
  $\ottdruleCXXDSum{}$ \hfil
  $\ottdruleCXXDList{}$
 \end{center}

 \caption{Type containment.}
 \label{fig:app-subtyping}
\end{figure}


\begin{figure}[t]
 \begin{flushleft}
 \textbf{Well-formedness} \quad \framebox{$\vdash  \Gamma$}
 \end{flushleft}
 \begin{center}
  $\ottdruleWFXXEmpty{}$ \hfil
  $\ottdruleWFXXExtVar{}$ \hfil
  $\ottdruleWFXXExtTyVar{}$ \hfil
 \end{center}
 %

 \caption{Well-formedness.}
 \label{fig:app-surface-well-formed}
\end{figure}

\begin{figure}[t]
 \begin{flushleft}
 \textbf{Term typing} \quad \framebox{$\Gamma  \vdash  \ottnt{M}  \ottsym{:}  \ottnt{A}$}
 \end{flushleft}
 \begin{center}
  $\ottdruleTXXVar{}$ \hfil
  $\ottdruleTXXConst{}$ \hfil
  $\ottdruleTXXAbs{}$ \\[2ex]
  $\ottdruleTXXApp{}$ \hfil
  $\ottdruleTXXGen{}$ \hfil
  $\ottdruleTXXInst{}$ \\[2ex]
  $\ottdruleTXXOp{}$ \hfil
  $\ottdruleTXXHandle{}$ \\[2ex]
  $\ottdruleTXXPair{}$ \hfil
  $\ottdruleTXXProjOne{}$ \hfil
  $\ottdruleTXXProjTwo{}$ \\[2ex]
  $\ottdruleTXXInL{}$ \hfil
  $\ottdruleTXXInR{}$ \\[2ex]
  $\ottdruleTXXCase{}$ \\[2ex]
  $\ottdruleTXXNil{}$ \hfil
  $\ottdruleTXXCons{}$ \\[2ex]
  $\ottdruleTXXCaseList{}$ \hfil
  $\ottdruleTXXFix{}$
 \end{center}
 \begin{flushleft}
 \textbf{Handler typing} \quad \framebox{$\Gamma  \vdash  \ottnt{H}  \ottsym{:}  \ottnt{A}  \Rightarrow  \ottnt{B}$}
 \end{flushleft}
 \begin{center}
  $\ottdruleTHXXReturn{}$ \\[2ex]
  $\ottdruleTHXXOp{}$
 \end{center}

 \caption{Typing.}
 \label{fig:app-surface-typing}
\end{figure}



\begin{defn}
 Well-formedness judgment $\vdash  \Gamma$ is the smallest relations satisfying the
 rules in \reffig{app-surface-well-formed}.
 We write $\Gamma  \vdash  \ottnt{A}$ if and only if $ \mathit{ftv}  (  \ottnt{A}  )  \,  \subseteq  \,  \mathit{dom}  (  \Gamma  ) $ and $\vdash  \Gamma$ is derived.
 Type containment judgment $\Gamma  \vdash  \ottnt{A}  \sqsubseteq  \ottnt{B}$ is the least relation satisfying the
 rules in \reffig{app-subtyping}.
 Typing judgments $\Gamma  \vdash  \ottnt{M}  \ottsym{:}  \ottnt{A}$ and $\Gamma  \vdash  \ottnt{H}  \ottsym{:}  \ottnt{A}  \Rightarrow  \ottnt{B}$ are the smallest
 relations satisfying the rules in \reffig{app-surface-typing}.
 %
\end{defn}

\clearpage

\subsection{Type-and-effect System}
\label{sec:app:defn:effect}

The type language for the type-and-effect system is shown
\reffig{app-eff-typelang}.  \reffig{app-eff-subtyping} describes only the change
of the type containment rules from those of the polymorphic type system.

\ifrestate
\defnEffSignatureRestriction*
\else
\begin{defn}[Effects satisfying signature restriction]
 The predicate $\mathit{SR} \, \ottsym{(}  \epsilon  \ottsym{)}$ holds if and only if, for any $\mathsf{op} \,  \in  \, \epsilon$ such
 that $\mathit{ty} \, \ottsym{(}  \mathsf{op}  \ottsym{)} \,  =  \,   \text{\unboldmath$\forall$}  \,  \algeffseqover{ \alpha }   \ottsym{.} \,  \ottnt{A}  \hookrightarrow  \ottnt{B} $:
 \begin{itemize}
  \item the occurrences of each type variable of $ \algeffseqover{ \alpha } $ in $\ottnt{A}$ are only
        negative or strictly positive;
  \item the occurrences of each type variable of $ \algeffseqover{ \alpha } $ in $\ottnt{B}$ are only
        positive; and
  \item for any function type $ \ottnt{C}   \rightarrow ^{ \epsilon' }  \ottnt{D} $ occurring at a strictly positive
        position in $\ottnt{A}$, if $\ottsym{\{}   \algeffseqover{ \alpha }   \ottsym{\}} \,  \mathbin{\cap}  \,  \mathit{ftv}  (  \ottnt{D}  )  \,  \not=  \,  \emptyset $, then
        $\mathit{SR} \, \ottsym{(}  \epsilon'  \ottsym{)}$.
 \end{itemize}
\end{defn}
\fi

\begin{defn}
 Typing judgments $\Gamma  \vdash  \ottnt{M}  \ottsym{:}  \ottnt{A} \,  |  \, \epsilon$ and $\Gamma  \vdash  \ottnt{H}  \ottsym{:}  \ottnt{A} \,  |  \, \epsilon  \Rightarrow  \ottnt{B} \,  |  \, \epsilon'$
 are the smallest relations satisfying the rules in \reffig{app-eff-typing}.
\end{defn}

\begin{figure}[t]
 \[\begin{array}{lll}
  \textbf{Effects} \quad \epsilon & ::= & \{ \mathsf{op}_{{\mathrm{1}}}, \cdots, \mathsf{op}_{\ottmv{n}} \} \\
  \textbf{Types} \quad \ottnt{A}, \ottnt{B}, \ottnt{C}, \ottnt{D} & ::= &
   \alpha \mid \iota \mid  \ottnt{A}   \rightarrow ^{ \epsilon }  \ottnt{B}  \mid  \text{\unboldmath$\forall$}  \, \alpha  \ottsym{.} \, \ottnt{A} \mid
    \ottnt{A}  \times  \ottnt{B}  \mid  \ottnt{A}  +  \ottnt{B}  \mid  \ottnt{A}  \, \mathsf{list} 
 \end{array}\]
 \caption{Type language for the effect-and-type system.}
 \label{fig:app-eff-typelang}
\end{figure}

\begin{figure}[t]
 \begin{flushleft}
 \textbf{Type containment} \quad \framebox{$\Gamma  \vdash  \ottnt{A}  \sqsubseteq  \ottnt{B}$}
 \end{flushleft}
 \begin{center}
  $\ottdruleCXXFunEff{}$ \hfil
  $\ottdruleCXXDFunEff{}$
 \end{center}
 \caption{Change from \reffig{app-subtyping} for type containment of the
 effect-and-type system. It gets rid of \Srule{Fun} and \Srule{DFun} instead of
 adding \Srule{FunEff} and \Srule{DFunEff}.}
 \label{fig:app-eff-subtyping}
\end{figure}

\begin{figure}[t]
 \begin{flushleft}
 \textbf{Term typing} \quad \framebox{$\Gamma  \vdash  \ottnt{M}  \ottsym{:}  \ottnt{A} \,  |  \, \epsilon$}
 \end{flushleft}
 \begin{center}
  $\ottdruleTeXXVar{}$ \hfil
  $\ottdruleTeXXConst{}$ \\[2ex]
  $\ottdruleTeXXAbs{}$ \hfil
  $\ottdruleTeXXApp{}$ \\[2ex]
  $\ottdruleTeXXGen{}$ \hfil
  $\ottdruleTeXXInst{}$ \\[2ex]
  $\ottdruleTeXXOp{}$ \\[2ex]
  $\ottdruleTeXXHandle{}$ \\[2ex]
  $\ottdruleTeXXPair{}$ \hfil
  $\ottdruleTeXXProjOne{}$ \hfil
  $\ottdruleTeXXProjTwo{}$ \\[2ex]
  $\ottdruleTeXXInL{}$ \hfil
  $\ottdruleTeXXInR{}$ \\[2ex]
  $\ottdruleTeXXCase{}$ \\[2ex]
  $\ottdruleTeXXNil{}$ \hfil
  $\ottdruleTeXXCons{}$ \\[2ex]
  $\ottdruleTeXXCaseList{}$ \\[2ex]
  $\ottdruleTeXXFix{}$ \hfil
  $\ottdruleTeXXWeak{}$
 \end{center}
 \begin{flushleft}
 \textbf{Handler typing} \quad
 \framebox{$\Gamma  \vdash  \ottnt{H}  \ottsym{:}  \ottnt{A} \,  |  \, \epsilon  \Rightarrow  \ottnt{B} \,  |  \, \epsilon'$}
 \end{flushleft}
 \begin{center}
  $\ottdruleTHeXXReturn{}$ \\[2ex]
  $\ottdruleTHeXXOp{}$
 \end{center}

 \caption{Typing of the effect-and-type system.}
 \label{fig:app-eff-typing}
\end{figure}

%% file: appendix/proof.tex
\section{Proofs}

\subsection{Soundness of the Type System}
\label{sec:app:proof:sound-typing}


\begin{lemmap}{Weakening}{weakening}
 Suppose that $\vdash  \Gamma_{{\mathrm{1}}}  \ottsym{,}  \Gamma_{{\mathrm{2}}}$.
 Let $\Gamma_{{\mathrm{3}}}$ be a typing context such that
 $ \mathit{dom}  (  \Gamma_{{\mathrm{2}}}  )  \,  \mathbin{\cap}  \,  \mathit{dom}  (  \Gamma_{{\mathrm{3}}}  )  \,  =  \,  \emptyset $.
 \begin{enumerate}
  \item \label{lem:weakening:typing-context}
        If $\vdash  \Gamma_{{\mathrm{1}}}  \ottsym{,}  \Gamma_{{\mathrm{3}}}$, then $\vdash  \Gamma_{{\mathrm{1}}}  \ottsym{,}  \Gamma_{{\mathrm{2}}}  \ottsym{,}  \Gamma_{{\mathrm{3}}}$.

  \item \label{lem:weakening:type}
        If $\Gamma_{{\mathrm{1}}}  \ottsym{,}  \Gamma_{{\mathrm{3}}}  \vdash  \ottnt{A}$, then $\Gamma_{{\mathrm{1}}}  \ottsym{,}  \Gamma_{{\mathrm{2}}}  \ottsym{,}  \Gamma_{{\mathrm{3}}}  \vdash  \ottnt{A}$.

  \item \label{lem:weakening:sub}
        If $\Gamma_{{\mathrm{1}}}  \ottsym{,}  \Gamma_{{\mathrm{3}}}  \vdash  \ottnt{A}  \sqsubseteq  \ottnt{B}$,
        then $\Gamma_{{\mathrm{1}}}  \ottsym{,}  \Gamma_{{\mathrm{2}}}  \ottsym{,}  \Gamma_{{\mathrm{3}}}  \vdash  \ottnt{A}  \sqsubseteq  \ottnt{B}$.

  \item \label{lem:weakening:term}
        If $\Gamma_{{\mathrm{1}}}  \ottsym{,}  \Gamma_{{\mathrm{3}}}  \vdash  \ottnt{M}  \ottsym{:}  \ottnt{A}$,
        then $\Gamma_{{\mathrm{1}}}  \ottsym{,}  \Gamma_{{\mathrm{2}}}  \ottsym{,}  \Gamma_{{\mathrm{3}}}  \vdash  \ottnt{M}  \ottsym{:}  \ottnt{A}$.

  \item \label{lem:weakening:handler}
        If $\Gamma_{{\mathrm{1}}}  \ottsym{,}  \Gamma_{{\mathrm{3}}}  \vdash  \ottnt{H}  \ottsym{:}  \ottnt{A}  \Rightarrow  \ottnt{B}$,
        then $\Gamma_{{\mathrm{1}}}  \ottsym{,}  \Gamma_{{\mathrm{2}}}  \ottsym{,}  \Gamma_{{\mathrm{3}}}  \vdash  \ottnt{H}  \ottsym{:}  \ottnt{A}  \Rightarrow  \ottnt{B}$.
 \end{enumerate}
\end{lemmap}
\begin{proof}
 By (mutual) induction on the derivations of the judgments.
\end{proof}


\begin{lemmap}{Type substitution}{ty-subst}
 Suppose that $\Gamma_{{\mathrm{1}}}  \vdash  \ottnt{A}$.
 \begin{enumerate}
  \item \label{lem:ty-subst:typing-context}
        If $\vdash  \Gamma_{{\mathrm{1}}}  \ottsym{,}  \alpha  \ottsym{,}  \Gamma_{{\mathrm{2}}}$, then $\vdash  \Gamma_{{\mathrm{1}}}  \ottsym{,}  \Gamma_{{\mathrm{2}}} \,  [  \ottnt{A}  \ottsym{/}  \alpha  ] $.

  \item \label{lem:ty-subst:type}
        If $\Gamma_{{\mathrm{1}}}  \ottsym{,}  \alpha  \ottsym{,}  \Gamma_{{\mathrm{2}}}  \vdash  \ottnt{B}$,
        then $\Gamma_{{\mathrm{1}}}  \ottsym{,}  \Gamma_{{\mathrm{2}}} \,  [  \ottnt{A}  \ottsym{/}  \alpha  ]   \vdash   \ottnt{B}    [  \ottnt{A}  \ottsym{/}  \alpha  ]  $.

  \item \label{lem:ty-subst:sub}
        If $\Gamma_{{\mathrm{1}}}  \ottsym{,}  \alpha  \ottsym{,}  \Gamma_{{\mathrm{2}}}  \vdash  \ottnt{B}  \sqsubseteq  \ottnt{C}$,
        then $\Gamma_{{\mathrm{1}}}  \ottsym{,}  \Gamma_{{\mathrm{2}}} \,  [  \ottnt{A}  \ottsym{/}  \alpha  ]   \vdash   \ottnt{B}    [  \ottnt{A}  \ottsym{/}  \alpha  ]    \sqsubseteq   \ottnt{C}    [  \ottnt{A}  \ottsym{/}  \alpha  ]  $.

  \item \label{lem:ty-subst:term}
        If $\Gamma_{{\mathrm{1}}}  \ottsym{,}  \alpha  \ottsym{,}  \Gamma_{{\mathrm{2}}}  \vdash  \ottnt{M}  \ottsym{:}  \ottnt{B}$,
        then $\Gamma_{{\mathrm{1}}}  \ottsym{,}  \Gamma_{{\mathrm{2}}} \,  [  \ottnt{A}  \ottsym{/}  \alpha  ]   \vdash  \ottnt{M}  \ottsym{:}   \ottnt{B}    [  \ottnt{A}  \ottsym{/}  \alpha  ]  $.

  \item If $\Gamma_{{\mathrm{1}}}  \ottsym{,}  \alpha  \ottsym{,}  \Gamma_{{\mathrm{2}}}  \vdash  \ottnt{H}  \ottsym{:}  \ottnt{B}  \Rightarrow  \ottnt{C}$,
        then $\Gamma_{{\mathrm{1}}}  \ottsym{,}  \Gamma_{{\mathrm{2}}} \,  [  \ottnt{A}  \ottsym{/}  \alpha  ]   \vdash  \ottnt{H}  \ottsym{:}   \ottnt{B}    [  \ottnt{A}  \ottsym{/}  \alpha  ]    \Rightarrow   \ottnt{C}    [  \ottnt{A}  \ottsym{/}  \alpha  ]  $.

 \end{enumerate}
\end{lemmap}
\begin{proof}
 Straightforward by (mutual) induction on the derivations of the judgments.
 Note that the cases for \T{Op} and \THrule{Op} depend on \refdef{eff}, which
 states that,
 for any $\mathsf{op}$, if $\mathit{ty} \, \ottsym{(}  \mathsf{op}  \ottsym{)} \,  =  \,   \text{\unboldmath$\forall$}  \,  \algeffseqover{ \beta }   \ottsym{.} \,  \ottnt{C}  \hookrightarrow  \ottnt{D} $, $ \mathit{ftv}  (  \ottnt{C}  )  \,  \mathbin{\cup}  \,  \mathit{ftv}  (  \ottnt{D}  )  \,  \subseteq  \, \ottsym{\{}   \algeffseqover{ \beta }   \ottsym{\}}$.
\end{proof}


\begin{lemma}{strengthening}
 \begin{enumerate}
  \item If $\vdash  \Gamma_{{\mathrm{1}}}  \ottsym{,}  \mathit{x} \,  \mathord{:}  \, \ottnt{A}  \ottsym{,}  \Gamma_{{\mathrm{2}}}$, then $\vdash  \Gamma_{{\mathrm{1}}}  \ottsym{,}  \Gamma_{{\mathrm{2}}}$.
  \item If $\Gamma_{{\mathrm{1}}}  \ottsym{,}  \mathit{x} \,  \mathord{:}  \, \ottnt{A}  \ottsym{,}  \Gamma_{{\mathrm{2}}}  \vdash  \ottnt{B}$, then $\Gamma_{{\mathrm{1}}}  \ottsym{,}  \Gamma_{{\mathrm{2}}}  \vdash  \ottnt{B}$.
  \item If $\Gamma_{{\mathrm{1}}}  \ottsym{,}  \mathit{x} \,  \mathord{:}  \, \ottnt{A}  \ottsym{,}  \Gamma_{{\mathrm{2}}}  \vdash  \ottnt{B}  \sqsubseteq  \ottnt{C}$, then $\Gamma_{{\mathrm{1}}}  \ottsym{,}  \Gamma_{{\mathrm{2}}}  \vdash  \ottnt{B}  \sqsubseteq  \ottnt{C}$.
 \end{enumerate}
\end{lemma}
\begin{proof}
 By induction on the derivations of the judgments.
\end{proof}

\begin{lemmap}{Term substitution}{term-subst}
 Suppose that $\Gamma_{{\mathrm{1}}}  \vdash  \ottnt{M}  \ottsym{:}  \ottnt{A}$.
 \begin{enumerate}
  \item \label{lem:term-subst:term}
        If $\Gamma_{{\mathrm{1}}}  \ottsym{,}  \mathit{x} \,  \mathord{:}  \, \ottnt{A}  \ottsym{,}  \Gamma_{{\mathrm{2}}}  \vdash  \ottnt{M'}  \ottsym{:}  \ottnt{B}$,
        then $\Gamma_{{\mathrm{1}}}  \ottsym{,}  \Gamma_{{\mathrm{2}}}  \vdash   \ottnt{M'}    [  \ottnt{M}  /  \mathit{x}  ]    \ottsym{:}  \ottnt{B}$.

  \item If $\Gamma_{{\mathrm{1}}}  \ottsym{,}  \mathit{x} \,  \mathord{:}  \, \ottnt{A}  \ottsym{,}  \Gamma_{{\mathrm{2}}}  \vdash  \ottnt{H}  \ottsym{:}  \ottnt{B}  \Rightarrow  \ottnt{C}$,
        then $\Gamma_{{\mathrm{1}}}  \ottsym{,}  \Gamma_{{\mathrm{2}}}  \vdash   \ottnt{H}    [  \ottnt{M}  /  \mathit{x}  ]    \ottsym{:}  \ottnt{B}  \Rightarrow  \ottnt{C}$.
 \end{enumerate}
\end{lemmap}
\begin{proof}
 By mutual induction on the typing derivations with \reflem{strengthening}.
 The case for \T{Var} uses \reflem{weakening} (\ref{lem:weakening:term}).
\end{proof}


\begin{defn}
 The function $ \mathit{unqualify} $ returns the type obtained by removing all the
 $\forall$s at the top-level from a given type, defined as follows.
 \[\begin{array}{lll}
   \mathit{unqualify}  (   \text{\unboldmath$\forall$}  \, \alpha  \ottsym{.} \, \ottnt{A}  )  & \defeq &  \mathit{unqualify}  (  \ottnt{A}  )  \\
   \mathit{unqualify}  (  \ottnt{A}  )     & \defeq & \ottnt{A}
   \quad \text{(if $\ottnt{A} \,  \not=  \,  \text{\unboldmath$\forall$}  \, \alpha  \ottsym{.} \, \ottnt{B}$ for any $\alpha$ and $\ottnt{B}$)}
   \end{array}\]
\end{defn}

\begin{lemma}{subtyping-unqualify-no-tyvar}
 Suppose $\Gamma  \vdash  \ottnt{A}  \sqsubseteq  \ottnt{B}$.
 If $ \mathit{unqualify}  (  \ottnt{A}  ) $ is not a type variable,
 then $ \mathit{unqualify}  (  \ottnt{B}  ) $ is not either.
\end{lemma}
\begin{proof}
 By induction on the type containment derivation.
 Only the interesting case is for \Srule{Inst}.  In that case, we are given
 $\Gamma  \vdash   \text{\unboldmath$\forall$}  \, \alpha  \ottsym{.} \, \ottnt{C}  \sqsubseteq   \ottnt{C}    [  \ottnt{D}  \ottsym{/}  \alpha  ]  $ ($\ottnt{A} \,  =  \,  \text{\unboldmath$\forall$}  \, \alpha  \ottsym{.} \, \ottnt{C}$ and $\ottnt{B} \,  =  \,  \ottnt{C}    [  \ottnt{D}  \ottsym{/}  \alpha  ]  $) for some
 $\alpha$, $\ottnt{C}$, and $\ottnt{D}$, and, by inversion,
 $\Gamma  \vdash  \ottnt{D}$.  It is easy to see, if
 $ \mathit{unqualify}  (   \text{\unboldmath$\forall$}  \, \beta  \ottsym{.} \, \ottnt{C}  )  =  \mathit{unqualify}  (  \ottnt{C}  ) $ is not a type variable,
 then $ \mathit{unqualify}  (   \ottnt{C}    [  \ottnt{D}  \ottsym{/}  \beta  ]    ) $ is not either.
\end{proof}

\begin{lemma}{subtyping-unqualify}
 Suppose that $\Gamma  \vdash  \ottnt{A}  \sqsubseteq  \ottnt{B}$ and $ \mathit{unqualify}  (  \ottnt{A}  ) $ is not a type variable.
 \begin{enumerate}
  \item If $ \mathit{unqualify}  (  \ottnt{B}  )  \,  =  \, \iota$,
        then $ \mathit{unqualify}  (  \ottnt{A}  )  \,  =  \, \iota$.

  \item \label{lem:subtyping-unqualify:fun}
        If $ \mathit{unqualify}  (  \ottnt{B}  )  \,  =  \, \ottnt{B_{{\mathrm{1}}}}  \rightarrow  \ottnt{B_{{\mathrm{2}}}}$,
        then $ \mathit{unqualify}  (  \ottnt{A}  )  \,  =  \, \ottnt{A_{{\mathrm{1}}}}  \rightarrow  \ottnt{A_{{\mathrm{2}}}}$ for some $\ottnt{A_{{\mathrm{1}}}}$ and $\ottnt{A_{{\mathrm{2}}}}$.

  \item If $ \mathit{unqualify}  (  \ottnt{B}  )  \,  =  \,  \ottnt{B_{{\mathrm{1}}}}  \times  \ottnt{B_{{\mathrm{2}}}} $,
        then $ \mathit{unqualify}  (  \ottnt{A}  )  \,  =  \,  \ottnt{A_{{\mathrm{1}}}}  \times  \ottnt{A_{{\mathrm{2}}}} $ for some $\ottnt{A_{{\mathrm{1}}}}$ and $\ottnt{A_{{\mathrm{2}}}}$.

  \item If $ \mathit{unqualify}  (  \ottnt{B}  )  \,  =  \,  \ottnt{B_{{\mathrm{1}}}}  +  \ottnt{B_{{\mathrm{2}}}} $,
        then $ \mathit{unqualify}  (  \ottnt{A}  )  \,  =  \,  \ottnt{A_{{\mathrm{1}}}}  +  \ottnt{A_{{\mathrm{2}}}} $ for some $\ottnt{A_{{\mathrm{1}}}}$ and $\ottnt{A_{{\mathrm{2}}}}$.

  \item If $ \mathit{unqualify}  (  \ottnt{B}  )  \,  =  \,  \ottnt{B'}  \, \mathsf{list} $,
        then $ \mathit{unqualify}  (  \ottnt{A}  )  \,  =  \,  \ottnt{A'}  \, \mathsf{list} $ for some $\ottnt{A'}$.
 \end{enumerate}
\end{lemma}
\begin{proof}
 By induction on the type containment derivation.
 The case for \Srule{Trans} is shown by the IHs and
 \reflem{subtyping-unqualify-no-tyvar}.
 In the case for \Srule{Inst}, we are given
 $\Gamma  \vdash   \text{\unboldmath$\forall$}  \, \alpha  \ottsym{.} \, \ottnt{C}  \sqsubseteq   \ottnt{C}    [  \ottnt{D}  \ottsym{/}  \alpha  ]  $ for some $\alpha$, $\ottnt{C}$, and $\ottnt{D}$
 ($\ottnt{A} \,  =  \,  \text{\unboldmath$\forall$}  \, \alpha  \ottsym{.} \, \ottnt{C}$ and $\ottnt{B} \,  =  \,  \ottnt{C}    [  \ottnt{D}  \ottsym{/}  \alpha  ]  $).
 Since $ \mathit{unqualify}  (   \text{\unboldmath$\forall$}  \, \alpha  \ottsym{.} \, \ottnt{C}  )  \,  =  \,  \mathit{unqualify}  (  \ottnt{C}  ) $ is not a type variable,
 it is easy to see that the top type constructor of $ \mathit{unqualify}  (  \ottnt{C}  ) $
 is the same as that of $ \mathit{unqualify}  (   \ottnt{C}    [  \ottnt{D}  \ottsym{/}  \alpha  ]    ) $.
 Proving the other cases is straightforward.
\end{proof}

\begin{lemma}{canonical-forms-unqualify-no-tyvar}
 If $\Gamma  \vdash  \ottnt{v}  \ottsym{:}  \ottnt{A}$, then $ \mathit{unqualify}  (  \ottnt{A}  ) $ is not a type variable.
\end{lemma}
\begin{proof}
 By induction on the typing derivation for $\ottnt{v}$.
 We can show the case for \T{Inst} by the IH and \reflem{subtyping-unqualify-no-tyvar}.
\end{proof}

\begin{lemmap}{Canonical forms}{canonical-forms}
 Suppose that $\Gamma  \vdash  \ottnt{v}  \ottsym{:}  \ottnt{A}$.
 \begin{enumerate}
  \item If $ \mathit{unqualify}  (  \ottnt{A}  )  \,  =  \, \iota$,
        then $\ottnt{v} \,  =  \, \ottnt{c}$ for some $\ottnt{c}$.

  \item If $ \mathit{unqualify}  (  \ottnt{A}  )  \,  =  \, \ottnt{B}  \rightarrow  \ottnt{C}$,
        then $\ottnt{v} \,  =  \, \ottnt{c}$ for some $\ottnt{c}$,
        or $\ottnt{v} \,  =  \,  \lambda\!  \, \mathit{x}  \ottsym{.}  \ottnt{M}$ for some $\mathit{x}$ and $\ottnt{M}$.

  \item If $ \mathit{unqualify}  (  \ottnt{A}  )  \,  =  \,  \ottnt{B}  \times  \ottnt{C} $,
        then $\ottnt{v} \,  =  \, \ottsym{(}  \ottnt{v_{{\mathrm{1}}}}  \ottsym{,}  \ottnt{v_{{\mathrm{2}}}}  \ottsym{)}$ for some $\ottnt{v_{{\mathrm{1}}}}$ and $\ottnt{v_{{\mathrm{2}}}}$.

  \item If $ \mathit{unqualify}  (  \ottnt{A}  )  \,  =  \,  \ottnt{B}  +  \ottnt{C} $,
        then $\ottnt{v} \,  =  \, \mathsf{inl} \, \ottnt{v'}$ or $\ottnt{v} \,  =  \, \mathsf{inr} \, \ottnt{v'}$ for some $\ottnt{v'}$.

  \item If $ \mathit{unqualify}  (  \ottnt{A}  )  \,  =  \,  \ottnt{B}  \, \mathsf{list} $,
        then $\ottnt{v} \,  =  \, \mathsf{nil}$ or $\ottnt{v} \,  =  \, \mathsf{cons} \, \ottnt{v'}$ for some $\ottnt{v'}$.
 \end{enumerate}
\end{lemmap}
\begin{proof}
 Straightforward by induction on the typing derivation for $\ottnt{v}$.
 Only the interesting case is for \T{Inst}.  In the case, we are given, by
 inversion, $\Gamma  \vdash  \ottnt{v}  \ottsym{:}  \ottnt{B}$ and $\Gamma  \vdash  \ottnt{B}  \sqsubseteq  \ottnt{A}$ and $\Gamma  \vdash  \ottnt{A}$ for some
 $\ottnt{B}$.  By \reflem{canonical-forms-unqualify-no-tyvar}, $ \mathit{unqualify}  (  \ottnt{B}  ) $ is
 not a type variable.  Thus, by \reflem{subtyping-unqualify} and the IH,
 we finish.
\end{proof}

\begin{defn}
 We use the metavariable $\Delta$ for ranging over typing contexts that consist of
 only type variables.  Formally, they are defined by the following syntax.
 \[
  \Delta ::=  \emptyset  \mid \Delta  \ottsym{,}  \alpha
 \]
\end{defn}

\ifrestate
\lemmSubtypingInvFun*
\else
\begin{lemmap}{Type containment inversion: function types}{subtyping-inv-fun}
 If $\Gamma  \vdash   \text{\unboldmath$\forall$}  \,  \algeffseqoverindex{ \alpha_{{\mathrm{1}}} }{ \text{\unboldmath$\mathit{I_{{\mathrm{1}}}}$} }   \ottsym{.} \, \ottnt{A_{{\mathrm{1}}}}  \rightarrow  \ottnt{A_{{\mathrm{2}}}}  \sqsubseteq   \text{\unboldmath$\forall$}  \,  \algeffseqoverindex{ \alpha_{{\mathrm{2}}} }{ \text{\unboldmath$\mathit{I_{{\mathrm{2}}}}$} }   \ottsym{.} \, \ottnt{B_{{\mathrm{1}}}}  \rightarrow  \ottnt{B_{{\mathrm{2}}}}$,
 then there exist $ \algeffseqoverindex{ \alpha_{{\mathrm{11}}} }{ \text{\unboldmath$\mathit{I_{{\mathrm{11}}}}$} } $, $ \algeffseqoverindex{ \alpha_{{\mathrm{12}}} }{ \text{\unboldmath$\mathit{I_{{\mathrm{12}}}}$} } $, $ \algeffseqoverindex{ \beta }{ \text{\unboldmath$\mathit{J}$} } $, and $ \algeffseqoverindex{ \ottnt{C} }{ \text{\unboldmath$\mathit{I_{{\mathrm{11}}}}$} } $
 such that
 \begin{itemize}
  \item $\ottsym{\{}   \algeffseqoverindex{ \alpha_{{\mathrm{1}}} }{ \text{\unboldmath$\mathit{I_{{\mathrm{1}}}}$} }   \ottsym{\}} \,  =  \, \ottsym{\{}   \algeffseqoverindex{ \alpha_{{\mathrm{11}}} }{ \text{\unboldmath$\mathit{I_{{\mathrm{11}}}}$} }   \ottsym{\}} \,  \mathbin{\uplus}  \, \ottsym{\{}   \algeffseqoverindex{ \alpha_{{\mathrm{12}}} }{ \text{\unboldmath$\mathit{I_{{\mathrm{12}}}}$} }   \ottsym{\}}$,
  \item $\Gamma  \ottsym{,}   \algeffseqoverindex{ \alpha_{{\mathrm{2}}} }{ \text{\unboldmath$\mathit{I_{{\mathrm{2}}}}$} }   \ottsym{,}   \algeffseqoverindex{ \beta }{ \text{\unboldmath$\mathit{J}$} }   \vdash   \algeffseqoverindex{ \ottnt{C} }{ \text{\unboldmath$\mathit{I_{{\mathrm{11}}}}$} } $,
  \item $\Gamma  \ottsym{,}   \algeffseqoverindex{ \alpha_{{\mathrm{2}}} }{ \text{\unboldmath$\mathit{I_{{\mathrm{2}}}}$} }   \vdash  \ottnt{B_{{\mathrm{1}}}}  \sqsubseteq    \text{\unboldmath$\forall$}  \,  \algeffseqoverindex{ \beta }{ \text{\unboldmath$\mathit{J}$} }   \ottsym{.} \, \ottnt{A_{{\mathrm{1}}}}    [   \algeffseqoverindex{ \ottnt{C} }{ \text{\unboldmath$\mathit{I_{{\mathrm{11}}}}$} }   \ottsym{/}   \algeffseqoverindex{ \alpha_{{\mathrm{11}}} }{ \text{\unboldmath$\mathit{I_{{\mathrm{11}}}}$} }   ]  $,
  \item $\Gamma  \ottsym{,}   \algeffseqoverindex{ \alpha_{{\mathrm{2}}} }{ \text{\unboldmath$\mathit{I_{{\mathrm{2}}}}$} }   \vdash    \text{\unboldmath$\forall$}  \,  \algeffseqoverindex{ \alpha_{{\mathrm{12}}} }{ \text{\unboldmath$\mathit{I_{{\mathrm{12}}}}$} }   \ottsym{.} \,  \text{\unboldmath$\forall$}  \,  \algeffseqoverindex{ \beta }{ \text{\unboldmath$\mathit{J}$} }   \ottsym{.} \, \ottnt{A_{{\mathrm{2}}}}    [   \algeffseqoverindex{ \ottnt{C} }{ \text{\unboldmath$\mathit{I_{{\mathrm{11}}}}$} }   \ottsym{/}   \algeffseqoverindex{ \alpha_{{\mathrm{11}}} }{ \text{\unboldmath$\mathit{I_{{\mathrm{11}}}}$} }   ]    \sqsubseteq  \ottnt{B_{{\mathrm{2}}}}$, and
  \item type variables in $\ottsym{\{}   \algeffseqoverindex{ \beta }{ \text{\unboldmath$\mathit{J}$} }   \ottsym{\}}$ do not appear free in $\ottnt{A_{{\mathrm{1}}}}$ and $\ottnt{A_{{\mathrm{2}}}}$.
 \end{itemize}
\end{lemmap}
\fi
\begin{proof}
 \ifrestate
 We show the following statement (which just uses different metavariables):
 \begin{quote}
  If $\Gamma  \vdash   \text{\unboldmath$\forall$}  \,  \algeffseqoverindex{ \alpha_{{\mathrm{1}}} }{ \text{\unboldmath$\mathit{I_{{\mathrm{1}}}}$} }   \ottsym{.} \, \ottnt{A_{{\mathrm{1}}}}  \rightarrow  \ottnt{A_{{\mathrm{2}}}}  \sqsubseteq   \text{\unboldmath$\forall$}  \,  \algeffseqoverindex{ \alpha_{{\mathrm{2}}} }{ \text{\unboldmath$\mathit{I_{{\mathrm{2}}}}$} }   \ottsym{.} \, \ottnt{B_{{\mathrm{1}}}}  \rightarrow  \ottnt{B_{{\mathrm{2}}}}$,
  then there exist $ \algeffseqoverindex{ \alpha_{{\mathrm{11}}} }{ \text{\unboldmath$\mathit{I_{{\mathrm{11}}}}$} } $, $ \algeffseqoverindex{ \alpha_{{\mathrm{12}}} }{ \text{\unboldmath$\mathit{I_{{\mathrm{12}}}}$} } $, $ \algeffseqoverindex{ \beta }{ \text{\unboldmath$\mathit{J}$} } $, and $ \algeffseqoverindex{ \ottnt{C} }{ \text{\unboldmath$\mathit{I_{{\mathrm{11}}}}$} } $
  such that
  \begin{itemize}
   \item $\ottsym{\{}   \algeffseqoverindex{ \alpha_{{\mathrm{1}}} }{ \text{\unboldmath$\mathit{I_{{\mathrm{1}}}}$} }   \ottsym{\}} \,  =  \, \ottsym{\{}   \algeffseqoverindex{ \alpha_{{\mathrm{11}}} }{ \text{\unboldmath$\mathit{I_{{\mathrm{11}}}}$} }   \ottsym{\}} \,  \mathbin{\uplus}  \, \ottsym{\{}   \algeffseqoverindex{ \alpha_{{\mathrm{12}}} }{ \text{\unboldmath$\mathit{I_{{\mathrm{12}}}}$} }   \ottsym{\}}$,
   \item $\Gamma  \ottsym{,}   \algeffseqoverindex{ \alpha_{{\mathrm{2}}} }{ \text{\unboldmath$\mathit{I_{{\mathrm{2}}}}$} }   \ottsym{,}   \algeffseqoverindex{ \beta }{ \text{\unboldmath$\mathit{J}$} }   \vdash   \algeffseqoverindex{ \ottnt{C} }{ \text{\unboldmath$\mathit{I_{{\mathrm{11}}}}$} } $,
   \item $\Gamma  \ottsym{,}   \algeffseqoverindex{ \alpha_{{\mathrm{2}}} }{ \text{\unboldmath$\mathit{I_{{\mathrm{2}}}}$} }   \vdash  \ottnt{B_{{\mathrm{1}}}}  \sqsubseteq    \text{\unboldmath$\forall$}  \,  \algeffseqoverindex{ \beta }{ \text{\unboldmath$\mathit{J}$} }   \ottsym{.} \, \ottnt{A_{{\mathrm{1}}}}    [   \algeffseqoverindex{ \ottnt{C} }{ \text{\unboldmath$\mathit{I_{{\mathrm{11}}}}$} }   \ottsym{/}   \algeffseqoverindex{ \alpha_{{\mathrm{11}}} }{ \text{\unboldmath$\mathit{I_{{\mathrm{11}}}}$} }   ]  $,
   \item $\Gamma  \ottsym{,}   \algeffseqoverindex{ \alpha_{{\mathrm{2}}} }{ \text{\unboldmath$\mathit{I_{{\mathrm{2}}}}$} }   \vdash    \text{\unboldmath$\forall$}  \,  \algeffseqoverindex{ \alpha_{{\mathrm{12}}} }{ \text{\unboldmath$\mathit{I_{{\mathrm{12}}}}$} }   \ottsym{.} \,  \text{\unboldmath$\forall$}  \,  \algeffseqoverindex{ \beta }{ \text{\unboldmath$\mathit{J}$} }   \ottsym{.} \, \ottnt{A_{{\mathrm{2}}}}    [   \algeffseqoverindex{ \ottnt{C} }{ \text{\unboldmath$\mathit{I_{{\mathrm{11}}}}$} }   \ottsym{/}   \algeffseqoverindex{ \alpha_{{\mathrm{11}}} }{ \text{\unboldmath$\mathit{I_{{\mathrm{11}}}}$} }   ]    \sqsubseteq  \ottnt{B_{{\mathrm{2}}}}$, and
   \item type variables in $\ottsym{\{}   \algeffseqoverindex{ \beta }{ \text{\unboldmath$\mathit{J}$} }   \ottsym{\}}$ do not appear free in $\ottnt{A_{{\mathrm{1}}}}$ and $\ottnt{A_{{\mathrm{2}}}}$.
  \end{itemize}
 \end{quote}
 \fi

 By induction on the type containment derivation.  Throughout the proof, we use the
 fact of $\vdash  \Gamma$ for applying \Srule{Refl}; it is shown easily by induction
 on the type containment derivation.
 \begin{caseanalysis}
  \case \Srule{Refl}: We have $ \algeffseqoverindex{ \alpha_{{\mathrm{1}}} }{ \text{\unboldmath$\mathit{I_{{\mathrm{1}}}}$} }  \,  =  \,  \algeffseqoverindex{ \alpha_{{\mathrm{2}}} }{ \text{\unboldmath$\mathit{I_{{\mathrm{2}}}}$} } $ and $\ottnt{A_{{\mathrm{1}}}} \,  =  \, \ottnt{B_{{\mathrm{1}}}}$ and
  $\ottnt{A_{{\mathrm{2}}}} \,  =  \, \ottnt{B_{{\mathrm{2}}}}$.
  Let $ \algeffseqoverindex{ \alpha_{{\mathrm{12}}} }{ \text{\unboldmath$\mathit{I_{{\mathrm{12}}}}$} } $ and $ \algeffseqoverindex{ \beta }{ \text{\unboldmath$\mathit{J}$} } $ be the empty sequence,
  $ \algeffseqoverindex{ \alpha_{{\mathrm{11}}} }{ \text{\unboldmath$\mathit{I_{{\mathrm{11}}}}$} }  \,  =  \,  \algeffseqoverindex{ \alpha_{{\mathrm{1}}} }{ \text{\unboldmath$\mathit{I_{{\mathrm{1}}}}$} } $, and $ \algeffseqoverindex{ \ottnt{C} }{ \text{\unboldmath$\mathit{I_{{\mathrm{11}}}}$} }  \,  =  \,  \algeffseqoverindex{ \alpha_{{\mathrm{1}}} }{ \text{\unboldmath$\mathit{I_{{\mathrm{1}}}}$} } $.
   We have to show that
   \begin{itemize}
    \item $\Gamma  \ottsym{,}   \algeffseqoverindex{ \alpha_{{\mathrm{2}}} }{ \text{\unboldmath$\mathit{I_{{\mathrm{2}}}}$} }   \vdash  \ottnt{B_{{\mathrm{1}}}}  \sqsubseteq  \ottnt{A_{{\mathrm{1}}}}$ and
    \item $\Gamma  \ottsym{,}   \algeffseqoverindex{ \alpha_{{\mathrm{2}}} }{ \text{\unboldmath$\mathit{I_{{\mathrm{2}}}}$} }   \vdash  \ottnt{A_{{\mathrm{2}}}}  \sqsubseteq  \ottnt{B_{{\mathrm{2}}}}$.
   \end{itemize}
   They are derived by \Srule{Refl}.

  \case \Srule{Trans}:
   By inversion, we have
   $\Gamma  \vdash   \text{\unboldmath$\forall$}  \,  \algeffseqoverindex{ \alpha_{{\mathrm{1}}} }{ \text{\unboldmath$\mathit{I_{{\mathrm{1}}}}$} }   \ottsym{.} \, \ottnt{A_{{\mathrm{1}}}}  \rightarrow  \ottnt{A_{{\mathrm{2}}}}  \sqsubseteq  \ottnt{D}$ and
   $\Gamma  \vdash  \ottnt{D}  \sqsubseteq   \text{\unboldmath$\forall$}  \,  \algeffseqoverindex{ \alpha_{{\mathrm{2}}} }{ \text{\unboldmath$\mathit{I_{{\mathrm{2}}}}$} }   \ottsym{.} \, \ottnt{B_{{\mathrm{1}}}}  \rightarrow  \ottnt{B_{{\mathrm{2}}}}$ for some $\ottnt{D}$.
   By \reflem{subtyping-unqualify},
   $\ottnt{D} \,  =  \,  \text{\unboldmath$\forall$}  \,  \algeffseqoverindex{ \alpha_{{\mathrm{3}}} }{ \text{\unboldmath$\mathit{I_{{\mathrm{3}}}}$} }   \ottsym{.} \, \ottnt{D_{{\mathrm{1}}}}  \rightarrow  \ottnt{D_{{\mathrm{2}}}}$ for some $ \algeffseqoverindex{ \alpha_{{\mathrm{3}}} }{ \text{\unboldmath$\mathit{I_{{\mathrm{3}}}}$} } $, $\ottnt{D_{{\mathrm{1}}}}$, and $\ottnt{D_{{\mathrm{2}}}}$.
   By the IH on $\Gamma  \vdash   \text{\unboldmath$\forall$}  \,  \algeffseqoverindex{ \alpha_{{\mathrm{1}}} }{ \text{\unboldmath$\mathit{I_{{\mathrm{1}}}}$} }   \ottsym{.} \, \ottnt{A_{{\mathrm{1}}}}  \rightarrow  \ottnt{A_{{\mathrm{2}}}}  \sqsubseteq   \text{\unboldmath$\forall$}  \,  \algeffseqoverindex{ \alpha_{{\mathrm{3}}} }{ \text{\unboldmath$\mathit{I_{{\mathrm{3}}}}$} }   \ottsym{.} \, \ottnt{D_{{\mathrm{1}}}}  \rightarrow  \ottnt{D_{{\mathrm{2}}}}$,
   there exist $ \algeffseqoverindex{ \alpha_{{\mathrm{11}}} }{ \text{\unboldmath$\mathit{I_{{\mathrm{11}}}}$} } $, $ \algeffseqoverindex{ \alpha_{{\mathrm{12}}} }{ \text{\unboldmath$\mathit{I_{{\mathrm{12}}}}$} } $, $ \algeffseqoverindex{ \ottnt{C_{{\mathrm{1}}}} }{ \text{\unboldmath$\mathit{I_{{\mathrm{11}}}}$} } $, and $ \algeffseqoverindex{ \beta_{{\mathrm{1}}} }{ \text{\unboldmath$\mathit{J_{{\mathrm{1}}}}$} } $
   such that
   \begin{itemize}
    \item $\ottsym{\{}   \algeffseqoverindex{ \alpha_{{\mathrm{1}}} }{ \text{\unboldmath$\mathit{I_{{\mathrm{1}}}}$} }   \ottsym{\}} \,  =  \, \ottsym{\{}   \algeffseqoverindex{ \alpha_{{\mathrm{11}}} }{ \text{\unboldmath$\mathit{I_{{\mathrm{11}}}}$} }   \ottsym{\}} \,  \mathbin{\uplus}  \, \ottsym{\{}   \algeffseqoverindex{ \alpha_{{\mathrm{12}}} }{ \text{\unboldmath$\mathit{I_{{\mathrm{12}}}}$} }   \ottsym{\}}$,
    \item $\Gamma  \ottsym{,}   \algeffseqoverindex{ \alpha_{{\mathrm{3}}} }{ \text{\unboldmath$\mathit{I_{{\mathrm{3}}}}$} }   \ottsym{,}   \algeffseqoverindex{ \beta_{{\mathrm{1}}} }{ \text{\unboldmath$\mathit{J_{{\mathrm{1}}}}$} }   \vdash   \algeffseqoverindex{ \ottnt{C_{{\mathrm{1}}}} }{ \text{\unboldmath$\mathit{I_{{\mathrm{11}}}}$} } $,
    \item $\Gamma  \ottsym{,}   \algeffseqoverindex{ \alpha_{{\mathrm{3}}} }{ \text{\unboldmath$\mathit{I_{{\mathrm{3}}}}$} }   \vdash  \ottnt{D_{{\mathrm{1}}}}  \sqsubseteq    \text{\unboldmath$\forall$}  \,  \algeffseqoverindex{ \beta_{{\mathrm{1}}} }{ \text{\unboldmath$\mathit{J_{{\mathrm{1}}}}$} }   \ottsym{.} \, \ottnt{A_{{\mathrm{1}}}}    [   \algeffseqoverindex{ \ottnt{C_{{\mathrm{1}}}} }{ \text{\unboldmath$\mathit{I_{{\mathrm{11}}}}$} }   \ottsym{/}   \algeffseqoverindex{ \alpha_{{\mathrm{11}}} }{ \text{\unboldmath$\mathit{I_{{\mathrm{11}}}}$} }   ]  $,
    \item $\Gamma  \ottsym{,}   \algeffseqoverindex{ \alpha_{{\mathrm{3}}} }{ \text{\unboldmath$\mathit{I_{{\mathrm{3}}}}$} }   \vdash    \text{\unboldmath$\forall$}  \,  \algeffseqoverindex{ \alpha_{{\mathrm{12}}} }{ \text{\unboldmath$\mathit{I_{{\mathrm{12}}}}$} }   \ottsym{.} \,  \text{\unboldmath$\forall$}  \,  \algeffseqoverindex{ \beta_{{\mathrm{1}}} }{ \text{\unboldmath$\mathit{J_{{\mathrm{1}}}}$} }   \ottsym{.} \, \ottnt{A_{{\mathrm{2}}}}    [   \algeffseqoverindex{ \ottnt{C_{{\mathrm{1}}}} }{ \text{\unboldmath$\mathit{I_{{\mathrm{11}}}}$} }   \ottsym{/}   \algeffseqoverindex{ \alpha_{{\mathrm{11}}} }{ \text{\unboldmath$\mathit{I_{{\mathrm{11}}}}$} }   ]    \sqsubseteq  \ottnt{D_{{\mathrm{2}}}}$, and
    \item type variables in $ \algeffseqoverindex{ \beta_{{\mathrm{1}}} }{ \text{\unboldmath$\mathit{J_{{\mathrm{1}}}}$} } $ do not appear free in $\ottnt{A_{{\mathrm{1}}}}$ and $\ottnt{A_{{\mathrm{2}}}}$.
   \end{itemize}
   By the IH on $\Gamma  \vdash   \text{\unboldmath$\forall$}  \,  \algeffseqoverindex{ \alpha_{{\mathrm{3}}} }{ \text{\unboldmath$\mathit{I_{{\mathrm{3}}}}$} }   \ottsym{.} \, \ottnt{D_{{\mathrm{1}}}}  \rightarrow  \ottnt{D_{{\mathrm{2}}}}  \sqsubseteq   \text{\unboldmath$\forall$}  \,  \algeffseqoverindex{ \alpha_{{\mathrm{2}}} }{ \text{\unboldmath$\mathit{I_{{\mathrm{2}}}}$} }   \ottsym{.} \, \ottnt{B_{{\mathrm{1}}}}  \rightarrow  \ottnt{B_{{\mathrm{2}}}}$,
   there exist $ \algeffseqoverindex{ \alpha_{{\mathrm{31}}} }{ \text{\unboldmath$\mathit{I_{{\mathrm{31}}}}$} } $, $ \algeffseqoverindex{ \alpha_{{\mathrm{32}}} }{ \text{\unboldmath$\mathit{I_{{\mathrm{32}}}}$} } $, $ \algeffseqoverindex{ \ottnt{C_{{\mathrm{3}}}} }{ \text{\unboldmath$\mathit{I_{{\mathrm{31}}}}$} } $, and $ \algeffseqoverindex{ \beta_{{\mathrm{3}}} }{ \text{\unboldmath$\mathit{J_{{\mathrm{3}}}}$} } $
   such that
   \begin{itemize}
    \item $\ottsym{\{}   \algeffseqoverindex{ \alpha_{{\mathrm{3}}} }{ \text{\unboldmath$\mathit{I_{{\mathrm{3}}}}$} }   \ottsym{\}} \,  =  \, \ottsym{\{}   \algeffseqoverindex{ \alpha_{{\mathrm{31}}} }{ \text{\unboldmath$\mathit{I_{{\mathrm{31}}}}$} }   \ottsym{\}} \,  \mathbin{\uplus}  \, \ottsym{\{}   \algeffseqoverindex{ \alpha_{{\mathrm{32}}} }{ \text{\unboldmath$\mathit{I_{{\mathrm{32}}}}$} }   \ottsym{\}}$,
    \item $\Gamma  \ottsym{,}   \algeffseqoverindex{ \alpha_{{\mathrm{2}}} }{ \text{\unboldmath$\mathit{I_{{\mathrm{2}}}}$} }   \ottsym{,}   \algeffseqoverindex{ \beta_{{\mathrm{3}}} }{ \text{\unboldmath$\mathit{J_{{\mathrm{3}}}}$} }   \vdash   \algeffseqoverindex{ \ottnt{C_{{\mathrm{3}}}} }{ \text{\unboldmath$\mathit{I_{{\mathrm{31}}}}$} } $,
    \item $\Gamma  \ottsym{,}   \algeffseqoverindex{ \alpha_{{\mathrm{2}}} }{ \text{\unboldmath$\mathit{I_{{\mathrm{2}}}}$} }   \vdash  \ottnt{B_{{\mathrm{1}}}}  \sqsubseteq    \text{\unboldmath$\forall$}  \,  \algeffseqoverindex{ \beta_{{\mathrm{3}}} }{ \text{\unboldmath$\mathit{J_{{\mathrm{3}}}}$} }   \ottsym{.} \, \ottnt{D_{{\mathrm{1}}}}    [   \algeffseqoverindex{ \ottnt{C_{{\mathrm{3}}}} }{ \text{\unboldmath$\mathit{I_{{\mathrm{31}}}}$} }   \ottsym{/}   \algeffseqoverindex{ \alpha_{{\mathrm{31}}} }{ \text{\unboldmath$\mathit{I_{{\mathrm{31}}}}$} }   ]  $,
    \item $\Gamma  \ottsym{,}   \algeffseqoverindex{ \alpha_{{\mathrm{2}}} }{ \text{\unboldmath$\mathit{I_{{\mathrm{2}}}}$} }   \vdash    \text{\unboldmath$\forall$}  \,  \algeffseqoverindex{ \alpha_{{\mathrm{32}}} }{ \text{\unboldmath$\mathit{I_{{\mathrm{32}}}}$} }   \ottsym{.} \,  \text{\unboldmath$\forall$}  \,  \algeffseqoverindex{ \beta_{{\mathrm{3}}} }{ \text{\unboldmath$\mathit{J_{{\mathrm{3}}}}$} }   \ottsym{.} \, \ottnt{D_{{\mathrm{2}}}}    [   \algeffseqoverindex{ \ottnt{C_{{\mathrm{3}}}} }{ \text{\unboldmath$\mathit{I_{{\mathrm{31}}}}$} }   \ottsym{/}   \algeffseqoverindex{ \alpha_{{\mathrm{31}}} }{ \text{\unboldmath$\mathit{I_{{\mathrm{31}}}}$} }   ]    \sqsubseteq  \ottnt{B_{{\mathrm{2}}}}$, and
    \item type variables in $ \algeffseqoverindex{ \beta_{{\mathrm{3}}} }{ \text{\unboldmath$\mathit{J_{{\mathrm{3}}}}$} } $ do not appear free in $\ottnt{D_{{\mathrm{1}}}}$ and $\ottnt{D_{{\mathrm{2}}}}$.
   \end{itemize}

   We show the conclusion by letting
   $ \algeffseqoverindex{ \ottnt{C} }{ \text{\unboldmath$\mathit{I_{{\mathrm{11}}}}$} }  \,  =  \,  \algeffseqoverindex{  \ottnt{C_{{\mathrm{1}}}}    [   \algeffseqoverindex{ \ottnt{C_{{\mathrm{3}}}} }{ \text{\unboldmath$\mathit{I_{{\mathrm{31}}}}$} }   \ottsym{/}   \algeffseqoverindex{ \alpha_{{\mathrm{31}}} }{ \text{\unboldmath$\mathit{I_{{\mathrm{31}}}}$} }   ]   }{ \text{\unboldmath$\mathit{I_{{\mathrm{11}}}}$} } $ and
   $ \algeffseqoverindex{ \beta }{ \text{\unboldmath$\mathit{J}$} }  \,  =  \,  \algeffseqoverindex{ \alpha_{{\mathrm{32}}} }{ \text{\unboldmath$\mathit{I_{{\mathrm{32}}}}$} }   \ottsym{,}   \algeffseqoverindex{ \beta_{{\mathrm{3}}} }{ \text{\unboldmath$\mathit{J_{{\mathrm{3}}}}$} }   \ottsym{,}   \algeffseqoverindex{ \beta_{{\mathrm{1}}} }{ \text{\unboldmath$\mathit{J_{{\mathrm{1}}}}$} } $.
   We have to show that
   \begin{itemize}
    \item $\Gamma  \ottsym{,}   \algeffseqoverindex{ \alpha_{{\mathrm{2}}} }{ \text{\unboldmath$\mathit{I_{{\mathrm{2}}}}$} }   \ottsym{,}   \algeffseqoverindex{ \alpha_{{\mathrm{32}}} }{ \text{\unboldmath$\mathit{I_{{\mathrm{32}}}}$} }   \ottsym{,}   \algeffseqoverindex{ \beta_{{\mathrm{3}}} }{ \text{\unboldmath$\mathit{J_{{\mathrm{3}}}}$} }   \ottsym{,}   \algeffseqoverindex{ \beta_{{\mathrm{1}}} }{ \text{\unboldmath$\mathit{J_{{\mathrm{1}}}}$} }   \vdash   \algeffseqoverindex{  \ottnt{C_{{\mathrm{1}}}}    [   \algeffseqoverindex{ \ottnt{C_{{\mathrm{3}}}} }{ \text{\unboldmath$\mathit{I_{{\mathrm{31}}}}$} }   \ottsym{/}   \algeffseqoverindex{ \alpha_{{\mathrm{31}}} }{ \text{\unboldmath$\mathit{I_{{\mathrm{31}}}}$} }   ]   }{ \text{\unboldmath$\mathit{I_{{\mathrm{11}}}}$} } $,
    \item $\Gamma  \ottsym{,}   \algeffseqoverindex{ \alpha_{{\mathrm{2}}} }{ \text{\unboldmath$\mathit{I_{{\mathrm{2}}}}$} }   \vdash  \ottnt{B_{{\mathrm{1}}}}  \sqsubseteq    \text{\unboldmath$\forall$}  \,  \algeffseqoverindex{ \alpha_{{\mathrm{32}}} }{ \text{\unboldmath$\mathit{I_{{\mathrm{32}}}}$} }   \ottsym{.} \,  \text{\unboldmath$\forall$}  \,  \algeffseqoverindex{ \beta_{{\mathrm{3}}} }{ \text{\unboldmath$\mathit{J_{{\mathrm{3}}}}$} }   \ottsym{.} \,  \text{\unboldmath$\forall$}  \,  \algeffseqoverindex{ \beta_{{\mathrm{1}}} }{ \text{\unboldmath$\mathit{J_{{\mathrm{1}}}}$} }   \ottsym{.} \, \ottnt{A_{{\mathrm{1}}}}    [   \algeffseqoverindex{ \ottnt{C} }{ \text{\unboldmath$\mathit{I_{{\mathrm{11}}}}$} }   \ottsym{/}   \algeffseqoverindex{ \alpha_{{\mathrm{11}}} }{ \text{\unboldmath$\mathit{I_{{\mathrm{11}}}}$} }   ]  $, and
    \item $\Gamma  \ottsym{,}   \algeffseqoverindex{ \alpha_{{\mathrm{2}}} }{ \text{\unboldmath$\mathit{I_{{\mathrm{2}}}}$} }   \vdash    \text{\unboldmath$\forall$}  \,  \algeffseqoverindex{ \alpha_{{\mathrm{12}}} }{ \text{\unboldmath$\mathit{I_{{\mathrm{12}}}}$} }   \ottsym{.} \,  \text{\unboldmath$\forall$}  \,  \algeffseqoverindex{ \alpha_{{\mathrm{32}}} }{ \text{\unboldmath$\mathit{I_{{\mathrm{32}}}}$} }   \ottsym{.} \,  \text{\unboldmath$\forall$}  \,  \algeffseqoverindex{ \beta_{{\mathrm{3}}} }{ \text{\unboldmath$\mathit{J_{{\mathrm{3}}}}$} }   \ottsym{.} \,  \text{\unboldmath$\forall$}  \,  \algeffseqoverindex{ \beta_{{\mathrm{1}}} }{ \text{\unboldmath$\mathit{J_{{\mathrm{1}}}}$} }   \ottsym{.} \, \ottnt{A_{{\mathrm{2}}}}    [   \algeffseqoverindex{ \ottnt{C} }{ \text{\unboldmath$\mathit{I_{{\mathrm{11}}}}$} }   \ottsym{/}   \algeffseqoverindex{ \alpha_{{\mathrm{11}}} }{ \text{\unboldmath$\mathit{I_{{\mathrm{11}}}}$} }   ]    \sqsubseteq  \ottnt{B_{{\mathrm{2}}}}$.
   \end{itemize}

   The first requirement is shown by
   $\Gamma  \ottsym{,}   \algeffseqoverindex{ \alpha_{{\mathrm{3}}} }{ \text{\unboldmath$\mathit{I_{{\mathrm{3}}}}$} }   \ottsym{,}   \algeffseqoverindex{ \beta_{{\mathrm{1}}} }{ \text{\unboldmath$\mathit{J_{{\mathrm{1}}}}$} }   \vdash   \algeffseqoverindex{ \ottnt{C_{{\mathrm{1}}}} }{ \text{\unboldmath$\mathit{I_{{\mathrm{11}}}}$} } $ and
   $\Gamma  \ottsym{,}   \algeffseqoverindex{ \alpha_{{\mathrm{2}}} }{ \text{\unboldmath$\mathit{I_{{\mathrm{2}}}}$} }   \ottsym{,}   \algeffseqoverindex{ \beta_{{\mathrm{3}}} }{ \text{\unboldmath$\mathit{J_{{\mathrm{3}}}}$} }   \vdash   \algeffseqoverindex{ \ottnt{C_{{\mathrm{3}}}} }{ \text{\unboldmath$\mathit{I_{{\mathrm{31}}}}$} } $ and
   \reflem{weakening} (\ref{lem:weakening:type}) and
   \reflem{ty-subst} (\ref{lem:ty-subst:type}).

   Next, we show the second requirement.
   Since $\Gamma  \ottsym{,}   \algeffseqoverindex{ \alpha_{{\mathrm{3}}} }{ \text{\unboldmath$\mathit{I_{{\mathrm{3}}}}$} }   \vdash  \ottnt{D_{{\mathrm{1}}}}  \sqsubseteq    \text{\unboldmath$\forall$}  \,  \algeffseqoverindex{ \beta_{{\mathrm{1}}} }{ \text{\unboldmath$\mathit{J_{{\mathrm{1}}}}$} }   \ottsym{.} \, \ottnt{A_{{\mathrm{1}}}}    [   \algeffseqoverindex{ \ottnt{C_{{\mathrm{1}}}} }{ \text{\unboldmath$\mathit{I_{{\mathrm{11}}}}$} }   \ottsym{/}   \algeffseqoverindex{ \alpha_{{\mathrm{11}}} }{ \text{\unboldmath$\mathit{I_{{\mathrm{11}}}}$} }   ]  $ and $\Gamma  \ottsym{,}   \algeffseqoverindex{ \alpha_{{\mathrm{2}}} }{ \text{\unboldmath$\mathit{I_{{\mathrm{2}}}}$} }   \ottsym{,}   \algeffseqoverindex{ \beta_{{\mathrm{3}}} }{ \text{\unboldmath$\mathit{J_{{\mathrm{3}}}}$} }   \vdash   \algeffseqoverindex{ \ottnt{C_{{\mathrm{3}}}} }{ \text{\unboldmath$\mathit{I_{{\mathrm{31}}}}$} } $,
   we have
   $\Gamma  \ottsym{,}   \algeffseqoverindex{ \alpha_{{\mathrm{2}}} }{ \text{\unboldmath$\mathit{I_{{\mathrm{2}}}}$} }   \ottsym{,}   \algeffseqoverindex{ \alpha_{{\mathrm{3}}} }{ \text{\unboldmath$\mathit{I_{{\mathrm{3}}}}$} }   \ottsym{,}   \algeffseqoverindex{ \beta_{{\mathrm{3}}} }{ \text{\unboldmath$\mathit{J_{{\mathrm{3}}}}$} }   \vdash  \ottnt{D_{{\mathrm{1}}}}  \sqsubseteq    \text{\unboldmath$\forall$}  \,  \algeffseqoverindex{ \beta_{{\mathrm{1}}} }{ \text{\unboldmath$\mathit{J_{{\mathrm{1}}}}$} }   \ottsym{.} \, \ottnt{A_{{\mathrm{1}}}}    [   \algeffseqoverindex{ \ottnt{C_{{\mathrm{1}}}} }{ \text{\unboldmath$\mathit{I_{{\mathrm{11}}}}$} }   \ottsym{/}   \algeffseqoverindex{ \alpha_{{\mathrm{11}}} }{ \text{\unboldmath$\mathit{I_{{\mathrm{11}}}}$} }   ]  $ and
   $\Gamma  \ottsym{,}   \algeffseqoverindex{ \alpha_{{\mathrm{2}}} }{ \text{\unboldmath$\mathit{I_{{\mathrm{2}}}}$} }   \ottsym{,}   \algeffseqoverindex{ \alpha_{{\mathrm{32}}} }{ \text{\unboldmath$\mathit{I_{{\mathrm{32}}}}$} }   \ottsym{,}   \algeffseqoverindex{ \beta_{{\mathrm{3}}} }{ \text{\unboldmath$\mathit{J_{{\mathrm{3}}}}$} }   \vdash   \algeffseqoverindex{ \ottnt{C_{{\mathrm{3}}}} }{ \text{\unboldmath$\mathit{I_{{\mathrm{31}}}}$} } $
   by \reflem{weakening} (\ref{lem:weakening:sub}) and (\ref{lem:weakening:type}), respectively.
   Thus, by \reflem{ty-subst} (\ref{lem:ty-subst:sub}),
   \[
    \Gamma  \ottsym{,}   \algeffseqoverindex{ \alpha_{{\mathrm{2}}} }{ \text{\unboldmath$\mathit{I_{{\mathrm{2}}}}$} }   \ottsym{,}   \algeffseqoverindex{ \alpha_{{\mathrm{32}}} }{ \text{\unboldmath$\mathit{I_{{\mathrm{32}}}}$} }   \ottsym{,}   \algeffseqoverindex{ \beta_{{\mathrm{3}}} }{ \text{\unboldmath$\mathit{J_{{\mathrm{3}}}}$} }   \vdash   \ottnt{D_{{\mathrm{1}}}}    [   \algeffseqoverindex{ \ottnt{C_{{\mathrm{3}}}} }{ \text{\unboldmath$\mathit{I_{{\mathrm{31}}}}$} }   \ottsym{/}   \algeffseqoverindex{ \alpha_{{\mathrm{31}}} }{ \text{\unboldmath$\mathit{I_{{\mathrm{31}}}}$} }   ]    \sqsubseteq    \text{\unboldmath$\forall$}  \,  \algeffseqoverindex{ \beta_{{\mathrm{1}}} }{ \text{\unboldmath$\mathit{J_{{\mathrm{1}}}}$} }   \ottsym{.} \, \ottnt{A_{{\mathrm{1}}}}    [   \algeffseqoverindex{ \ottnt{C} }{ \text{\unboldmath$\mathit{I_{{\mathrm{11}}}}$} }   \ottsym{/}   \algeffseqoverindex{ \alpha_{{\mathrm{11}}} }{ \text{\unboldmath$\mathit{I_{{\mathrm{11}}}}$} }   ]  
   \]
   (note that we can suppose that $ \algeffseqoverindex{ \alpha_{{\mathrm{31}}} }{ \text{\unboldmath$\mathit{I_{{\mathrm{31}}}}$} } $ do not appear free in $\ottnt{A_{{\mathrm{1}}}}$).
   By \Srule{Poly}, 
   \[
    \Gamma  \ottsym{,}   \algeffseqoverindex{ \alpha_{{\mathrm{2}}} }{ \text{\unboldmath$\mathit{I_{{\mathrm{2}}}}$} }   \ottsym{,}   \algeffseqoverindex{ \alpha_{{\mathrm{32}}} }{ \text{\unboldmath$\mathit{I_{{\mathrm{32}}}}$} }   \vdash    \text{\unboldmath$\forall$}  \,  \algeffseqoverindex{ \beta_{{\mathrm{3}}} }{ \text{\unboldmath$\mathit{J_{{\mathrm{3}}}}$} }   \ottsym{.} \, \ottnt{D_{{\mathrm{1}}}}    [   \algeffseqoverindex{ \ottnt{C_{{\mathrm{3}}}} }{ \text{\unboldmath$\mathit{I_{{\mathrm{31}}}}$} }   \ottsym{/}   \algeffseqoverindex{ \alpha_{{\mathrm{31}}} }{ \text{\unboldmath$\mathit{I_{{\mathrm{31}}}}$} }   ]    \sqsubseteq    \text{\unboldmath$\forall$}  \,  \algeffseqoverindex{ \beta_{{\mathrm{3}}} }{ \text{\unboldmath$\mathit{J_{{\mathrm{3}}}}$} }   \ottsym{.} \,  \text{\unboldmath$\forall$}  \,  \algeffseqoverindex{ \beta_{{\mathrm{1}}} }{ \text{\unboldmath$\mathit{J_{{\mathrm{1}}}}$} }   \ottsym{.} \, \ottnt{A_{{\mathrm{1}}}}    [   \algeffseqoverindex{ \ottnt{C} }{ \text{\unboldmath$\mathit{I_{{\mathrm{11}}}}$} }   \ottsym{/}   \algeffseqoverindex{ \alpha_{{\mathrm{11}}} }{ \text{\unboldmath$\mathit{I_{{\mathrm{11}}}}$} }   ]  .
   \]
   Since $\Gamma  \ottsym{,}   \algeffseqoverindex{ \alpha_{{\mathrm{2}}} }{ \text{\unboldmath$\mathit{I_{{\mathrm{2}}}}$} }   \vdash  \ottnt{B_{{\mathrm{1}}}}  \sqsubseteq    \text{\unboldmath$\forall$}  \,  \algeffseqoverindex{ \beta_{{\mathrm{3}}} }{ \text{\unboldmath$\mathit{J_{{\mathrm{3}}}}$} }   \ottsym{.} \, \ottnt{D_{{\mathrm{1}}}}    [   \algeffseqoverindex{ \ottnt{C_{{\mathrm{3}}}} }{ \text{\unboldmath$\mathit{I_{{\mathrm{31}}}}$} }   \ottsym{/}   \algeffseqoverindex{ \alpha_{{\mathrm{31}}} }{ \text{\unboldmath$\mathit{I_{{\mathrm{31}}}}$} }   ]  $,
   we have
   \[
    \Gamma  \ottsym{,}   \algeffseqoverindex{ \alpha_{{\mathrm{2}}} }{ \text{\unboldmath$\mathit{I_{{\mathrm{2}}}}$} }   \ottsym{,}   \algeffseqoverindex{ \alpha_{{\mathrm{32}}} }{ \text{\unboldmath$\mathit{I_{{\mathrm{32}}}}$} }   \vdash  \ottnt{B_{{\mathrm{1}}}}  \sqsubseteq    \text{\unboldmath$\forall$}  \,  \algeffseqoverindex{ \beta_{{\mathrm{3}}} }{ \text{\unboldmath$\mathit{J_{{\mathrm{3}}}}$} }   \ottsym{.} \,  \text{\unboldmath$\forall$}  \,  \algeffseqoverindex{ \beta_{{\mathrm{1}}} }{ \text{\unboldmath$\mathit{J_{{\mathrm{1}}}}$} }   \ottsym{.} \, \ottnt{A_{{\mathrm{1}}}}    [   \algeffseqoverindex{ \ottnt{C_{{\mathrm{01}}}} }{ \text{\unboldmath$\mathit{I_{{\mathrm{11}}}}$} }   \ottsym{/}   \algeffseqoverindex{ \alpha_{{\mathrm{11}}} }{ \text{\unboldmath$\mathit{I_{{\mathrm{11}}}}$} }   ]  
   \]
   by \reflem{weakening} (\ref{lem:weakening:sub}) and \Srule{Trans}.
   Since we can suppose that $ \algeffseqoverindex{ \alpha_{{\mathrm{32}}} }{ \text{\unboldmath$\mathit{I_{{\mathrm{32}}}}$} } $ do not appear free in $\ottnt{B_{{\mathrm{1}}}}$,
   we have
   \[
    \Gamma  \ottsym{,}   \algeffseqoverindex{ \alpha_{{\mathrm{2}}} }{ \text{\unboldmath$\mathit{I_{{\mathrm{2}}}}$} }   \vdash  \ottnt{B_{{\mathrm{1}}}}  \sqsubseteq    \text{\unboldmath$\forall$}  \,  \algeffseqoverindex{ \alpha_{{\mathrm{32}}} }{ \text{\unboldmath$\mathit{I_{{\mathrm{32}}}}$} }   \ottsym{.} \,  \text{\unboldmath$\forall$}  \,  \algeffseqoverindex{ \beta_{{\mathrm{3}}} }{ \text{\unboldmath$\mathit{J_{{\mathrm{3}}}}$} }   \ottsym{.} \,  \text{\unboldmath$\forall$}  \,  \algeffseqoverindex{ \beta_{{\mathrm{1}}} }{ \text{\unboldmath$\mathit{J_{{\mathrm{1}}}}$} }   \ottsym{.} \, \ottnt{A_{{\mathrm{1}}}}    [   \algeffseqoverindex{ \ottnt{C} }{ \text{\unboldmath$\mathit{I_{{\mathrm{11}}}}$} }   \ottsym{/}   \algeffseqoverindex{ \alpha_{{\mathrm{11}}} }{ \text{\unboldmath$\mathit{I_{{\mathrm{11}}}}$} }   ]  
   \]
   by \Srule{Gen}, \Srule{Poly}, and \Srule{Trans}.

   Finally, we show the third requirement.
   Since $\Gamma  \ottsym{,}   \algeffseqoverindex{ \alpha_{{\mathrm{3}}} }{ \text{\unboldmath$\mathit{I_{{\mathrm{3}}}}$} }   \vdash    \text{\unboldmath$\forall$}  \,  \algeffseqoverindex{ \alpha_{{\mathrm{12}}} }{ \text{\unboldmath$\mathit{I_{{\mathrm{12}}}}$} }   \ottsym{.} \,  \text{\unboldmath$\forall$}  \,  \algeffseqoverindex{ \beta_{{\mathrm{1}}} }{ \text{\unboldmath$\mathit{J_{{\mathrm{1}}}}$} }   \ottsym{.} \, \ottnt{A_{{\mathrm{2}}}}    [   \algeffseqoverindex{ \ottnt{C_{{\mathrm{1}}}} }{ \text{\unboldmath$\mathit{I_{{\mathrm{11}}}}$} }   \ottsym{/}   \algeffseqoverindex{ \alpha_{{\mathrm{11}}} }{ \text{\unboldmath$\mathit{I_{{\mathrm{11}}}}$} }   ]    \sqsubseteq  \ottnt{D_{{\mathrm{2}}}}$ and
   $\Gamma  \ottsym{,}   \algeffseqoverindex{ \alpha_{{\mathrm{2}}} }{ \text{\unboldmath$\mathit{I_{{\mathrm{2}}}}$} }   \ottsym{,}   \algeffseqoverindex{ \beta_{{\mathrm{3}}} }{ \text{\unboldmath$\mathit{J_{{\mathrm{3}}}}$} }   \vdash   \algeffseqoverindex{ \ottnt{C_{{\mathrm{3}}}} }{ \text{\unboldmath$\mathit{I_{{\mathrm{31}}}}$} } $,
   we have
   $\Gamma  \ottsym{,}   \algeffseqoverindex{ \alpha_{{\mathrm{2}}} }{ \text{\unboldmath$\mathit{I_{{\mathrm{2}}}}$} }   \ottsym{,}   \algeffseqoverindex{ \alpha_{{\mathrm{3}}} }{ \text{\unboldmath$\mathit{I_{{\mathrm{3}}}}$} }   \ottsym{,}   \algeffseqoverindex{ \beta_{{\mathrm{3}}} }{ \text{\unboldmath$\mathit{J_{{\mathrm{3}}}}$} }   \vdash    \text{\unboldmath$\forall$}  \,  \algeffseqoverindex{ \alpha_{{\mathrm{12}}} }{ \text{\unboldmath$\mathit{I_{{\mathrm{12}}}}$} }   \ottsym{.} \,  \text{\unboldmath$\forall$}  \,  \algeffseqoverindex{ \beta_{{\mathrm{1}}} }{ \text{\unboldmath$\mathit{J_{{\mathrm{1}}}}$} }   \ottsym{.} \, \ottnt{A_{{\mathrm{2}}}}    [   \algeffseqoverindex{ \ottnt{C_{{\mathrm{1}}}} }{ \text{\unboldmath$\mathit{I_{{\mathrm{11}}}}$} }   \ottsym{/}   \algeffseqoverindex{ \alpha_{{\mathrm{11}}} }{ \text{\unboldmath$\mathit{I_{{\mathrm{11}}}}$} }   ]    \sqsubseteq  \ottnt{D_{{\mathrm{2}}}}$ and
   $\Gamma  \ottsym{,}   \algeffseqoverindex{ \alpha_{{\mathrm{2}}} }{ \text{\unboldmath$\mathit{I_{{\mathrm{2}}}}$} }   \ottsym{,}   \algeffseqoverindex{ \alpha_{{\mathrm{32}}} }{ \text{\unboldmath$\mathit{I_{{\mathrm{32}}}}$} }   \ottsym{,}   \algeffseqoverindex{ \beta_{{\mathrm{3}}} }{ \text{\unboldmath$\mathit{J_{{\mathrm{3}}}}$} }   \vdash   \algeffseqoverindex{ \ottnt{C_{{\mathrm{3}}}} }{ \text{\unboldmath$\mathit{I_{{\mathrm{31}}}}$} } $
   by \reflem{weakening} (\ref{lem:weakening:sub}) and (\ref{lem:weakening:type}), respectively.
   Thus, by \reflem{ty-subst} (\ref{lem:ty-subst:sub}),
   \[
    \Gamma  \ottsym{,}   \algeffseqoverindex{ \alpha_{{\mathrm{2}}} }{ \text{\unboldmath$\mathit{I_{{\mathrm{2}}}}$} }   \ottsym{,}   \algeffseqoverindex{ \alpha_{{\mathrm{32}}} }{ \text{\unboldmath$\mathit{I_{{\mathrm{32}}}}$} }   \ottsym{,}   \algeffseqoverindex{ \beta_{{\mathrm{3}}} }{ \text{\unboldmath$\mathit{J_{{\mathrm{3}}}}$} }   \vdash    \text{\unboldmath$\forall$}  \,  \algeffseqoverindex{ \alpha_{{\mathrm{12}}} }{ \text{\unboldmath$\mathit{I_{{\mathrm{12}}}}$} }   \ottsym{.} \,  \text{\unboldmath$\forall$}  \,  \algeffseqoverindex{ \beta_{{\mathrm{1}}} }{ \text{\unboldmath$\mathit{J_{{\mathrm{1}}}}$} }   \ottsym{.} \, \ottnt{A_{{\mathrm{2}}}}    [   \algeffseqoverindex{ \ottnt{C} }{ \text{\unboldmath$\mathit{I_{{\mathrm{11}}}}$} }   \ottsym{/}   \algeffseqoverindex{ \alpha_{{\mathrm{11}}} }{ \text{\unboldmath$\mathit{I_{{\mathrm{11}}}}$} }   ]    \sqsubseteq   \ottnt{D_{{\mathrm{2}}}}    [   \algeffseqoverindex{ \ottnt{C_{{\mathrm{3}}}} }{ \text{\unboldmath$\mathit{I_{{\mathrm{31}}}}$} }   \ottsym{/}   \algeffseqoverindex{ \alpha_{{\mathrm{31}}} }{ \text{\unboldmath$\mathit{I_{{\mathrm{31}}}}$} }   ]  
   \]
   (note that we can suppose that $ \algeffseqoverindex{ \alpha_{{\mathrm{31}}} }{ \text{\unboldmath$\mathit{I_{{\mathrm{31}}}}$} } $ do not appear free in $\ottnt{A_{{\mathrm{2}}}}$).
   By \Srule{Poly},
   \[
    \Gamma  \ottsym{,}   \algeffseqoverindex{ \alpha_{{\mathrm{2}}} }{ \text{\unboldmath$\mathit{I_{{\mathrm{2}}}}$} }   \vdash    \text{\unboldmath$\forall$}  \,  \algeffseqoverindex{ \alpha_{{\mathrm{32}}} }{ \text{\unboldmath$\mathit{I_{{\mathrm{32}}}}$} }   \ottsym{.} \,  \text{\unboldmath$\forall$}  \,  \algeffseqoverindex{ \beta_{{\mathrm{3}}} }{ \text{\unboldmath$\mathit{J_{{\mathrm{3}}}}$} }   \ottsym{.} \,  \text{\unboldmath$\forall$}  \,  \algeffseqoverindex{ \alpha_{{\mathrm{12}}} }{ \text{\unboldmath$\mathit{I_{{\mathrm{12}}}}$} }   \ottsym{.} \,  \text{\unboldmath$\forall$}  \,  \algeffseqoverindex{ \beta_{{\mathrm{1}}} }{ \text{\unboldmath$\mathit{J_{{\mathrm{1}}}}$} }   \ottsym{.} \, \ottnt{A_{{\mathrm{2}}}}    [   \algeffseqoverindex{ \ottnt{C} }{ \text{\unboldmath$\mathit{I_{{\mathrm{11}}}}$} }   \ottsym{/}   \algeffseqoverindex{ \alpha_{{\mathrm{11}}} }{ \text{\unboldmath$\mathit{I_{{\mathrm{11}}}}$} }   ]    \sqsubseteq    \text{\unboldmath$\forall$}  \,  \algeffseqoverindex{ \alpha_{{\mathrm{32}}} }{ \text{\unboldmath$\mathit{I_{{\mathrm{32}}}}$} }   \ottsym{.} \,  \text{\unboldmath$\forall$}  \,  \algeffseqoverindex{ \beta_{{\mathrm{3}}} }{ \text{\unboldmath$\mathit{J_{{\mathrm{3}}}}$} }   \ottsym{.} \, \ottnt{D_{{\mathrm{2}}}}    [   \algeffseqoverindex{ \ottnt{C_{{\mathrm{3}}}} }{ \text{\unboldmath$\mathit{I_{{\mathrm{31}}}}$} }   \ottsym{/}   \algeffseqoverindex{ \alpha_{{\mathrm{31}}} }{ \text{\unboldmath$\mathit{I_{{\mathrm{31}}}}$} }   ]  
   \]
   Since $\Gamma  \ottsym{,}   \algeffseqoverindex{ \alpha_{{\mathrm{2}}} }{ \text{\unboldmath$\mathit{I_{{\mathrm{2}}}}$} }   \vdash    \text{\unboldmath$\forall$}  \,  \algeffseqoverindex{ \alpha_{{\mathrm{32}}} }{ \text{\unboldmath$\mathit{I_{{\mathrm{32}}}}$} }   \ottsym{.} \,  \text{\unboldmath$\forall$}  \,  \algeffseqoverindex{ \beta_{{\mathrm{3}}} }{ \text{\unboldmath$\mathit{J_{{\mathrm{3}}}}$} }   \ottsym{.} \, \ottnt{D_{{\mathrm{2}}}}    [   \algeffseqoverindex{ \ottnt{C_{{\mathrm{3}}}} }{ \text{\unboldmath$\mathit{I_{{\mathrm{31}}}}$} }   \ottsym{/}   \algeffseqoverindex{ \alpha_{{\mathrm{31}}} }{ \text{\unboldmath$\mathit{I_{{\mathrm{31}}}}$} }   ]    \sqsubseteq  \ottnt{B_{{\mathrm{2}}}}$,
   we have
   \[
    \Gamma  \ottsym{,}   \algeffseqoverindex{ \alpha_{{\mathrm{2}}} }{ \text{\unboldmath$\mathit{I_{{\mathrm{2}}}}$} }   \vdash    \text{\unboldmath$\forall$}  \,  \algeffseqoverindex{ \alpha_{{\mathrm{32}}} }{ \text{\unboldmath$\mathit{I_{{\mathrm{32}}}}$} }   \ottsym{.} \,  \text{\unboldmath$\forall$}  \,  \algeffseqoverindex{ \beta_{{\mathrm{3}}} }{ \text{\unboldmath$\mathit{J_{{\mathrm{3}}}}$} }   \ottsym{.} \,  \text{\unboldmath$\forall$}  \,  \algeffseqoverindex{ \alpha_{{\mathrm{12}}} }{ \text{\unboldmath$\mathit{I_{{\mathrm{12}}}}$} }   \ottsym{.} \,  \text{\unboldmath$\forall$}  \,  \algeffseqoverindex{ \beta_{{\mathrm{1}}} }{ \text{\unboldmath$\mathit{J_{{\mathrm{1}}}}$} }   \ottsym{.} \, \ottnt{A_{{\mathrm{2}}}}    [   \algeffseqoverindex{ \ottnt{C} }{ \text{\unboldmath$\mathit{I_{{\mathrm{11}}}}$} }   \ottsym{/}   \algeffseqoverindex{ \alpha_{{\mathrm{11}}} }{ \text{\unboldmath$\mathit{I_{{\mathrm{11}}}}$} }   ]    \sqsubseteq  \ottnt{B_{{\mathrm{2}}}}
   \]
   by \Srule{Trans}.  Thus, by permutating $\forall$s on the left-hand side,
   \[
    \Gamma  \ottsym{,}   \algeffseqoverindex{ \alpha_{{\mathrm{2}}} }{ \text{\unboldmath$\mathit{I_{{\mathrm{2}}}}$} }   \vdash    \text{\unboldmath$\forall$}  \,  \algeffseqoverindex{ \alpha_{{\mathrm{12}}} }{ \text{\unboldmath$\mathit{I_{{\mathrm{12}}}}$} }   \ottsym{.} \,  \text{\unboldmath$\forall$}  \,  \algeffseqoverindex{ \alpha_{{\mathrm{32}}} }{ \text{\unboldmath$\mathit{I_{{\mathrm{32}}}}$} }   \ottsym{.} \,  \text{\unboldmath$\forall$}  \,  \algeffseqoverindex{ \beta_{{\mathrm{3}}} }{ \text{\unboldmath$\mathit{J_{{\mathrm{3}}}}$} }   \ottsym{.} \,  \text{\unboldmath$\forall$}  \,  \algeffseqoverindex{ \beta_{{\mathrm{1}}} }{ \text{\unboldmath$\mathit{J_{{\mathrm{1}}}}$} }   \ottsym{.} \, \ottnt{A_{{\mathrm{2}}}}    [   \algeffseqoverindex{ \ottnt{C} }{ \text{\unboldmath$\mathit{I_{{\mathrm{11}}}}$} }   \ottsym{/}   \algeffseqoverindex{ \alpha_{{\mathrm{11}}} }{ \text{\unboldmath$\mathit{I_{{\mathrm{11}}}}$} }   ]    \sqsubseteq  \ottnt{B_{{\mathrm{2}}}}.
   \]

  \case \Srule{Fun}: Obvious by inversion.
  \case \Srule{Inst}:
   We have $ \algeffseqoverindex{ \alpha_{{\mathrm{1}}} }{ \text{\unboldmath$\mathit{I_{{\mathrm{1}}}}$} }  \,  =  \, \alpha  \ottsym{,}   \algeffseqoverindex{ \alpha_{{\mathrm{2}}} }{ \text{\unboldmath$\mathit{I_{{\mathrm{2}}}}$} } $ and $\ottnt{B_{{\mathrm{1}}}} \,  =  \,  \ottnt{A_{{\mathrm{1}}}}    [  \ottnt{C}  \ottsym{/}  \alpha  ]  $ and $\ottnt{B_{{\mathrm{2}}}} \,  =  \,  \ottnt{A_{{\mathrm{2}}}}    [  \ottnt{C}  \ottsym{/}  \alpha  ]  $
   for some $\ottnt{C}$ such that $\Gamma  \vdash  \ottnt{C}$.
   We show the conclusion by letting
   $ \algeffseqoverindex{ \alpha_{{\mathrm{11}}} }{ \text{\unboldmath$\mathit{I_{{\mathrm{11}}}}$} }  \,  =  \, \alpha  \ottsym{,}   \algeffseqoverindex{ \alpha_{{\mathrm{2}}} }{ \text{\unboldmath$\mathit{I_{{\mathrm{2}}}}$} } $, $ \algeffseqoverindex{ \ottnt{C} }{ \text{\unboldmath$\mathit{I_{{\mathrm{11}}}}$} }  \,  =  \, \ottnt{C}  \ottsym{,}   \algeffseqoverindex{ \alpha_{{\mathrm{2}}} }{ \text{\unboldmath$\mathit{I_{{\mathrm{2}}}}$} } $, and
   $ \algeffseqoverindex{ \alpha_{{\mathrm{12}}} }{ \text{\unboldmath$\mathit{I_{{\mathrm{12}}}}$} } $ and $ \algeffseqoverindex{ \beta }{ \text{\unboldmath$\mathit{J}$} } $ be the empty sequence.
   We have to show that
   \begin{itemize}
    \item $\Gamma  \ottsym{,}   \algeffseqoverindex{ \alpha_{{\mathrm{2}}} }{ \text{\unboldmath$\mathit{I_{{\mathrm{2}}}}$} }   \vdash  \ottnt{C}$,
    \item $\Gamma  \ottsym{,}   \algeffseqoverindex{ \alpha_{{\mathrm{2}}} }{ \text{\unboldmath$\mathit{I_{{\mathrm{2}}}}$} }   \vdash   \ottnt{A_{{\mathrm{1}}}}    [  \ottnt{C}  \ottsym{/}  \alpha  ]    \sqsubseteq   \ottnt{A_{{\mathrm{1}}}}    [  \ottnt{C}  \ottsym{/}  \alpha  ]  $,
    \item $\Gamma  \ottsym{,}   \algeffseqoverindex{ \alpha_{{\mathrm{2}}} }{ \text{\unboldmath$\mathit{I_{{\mathrm{2}}}}$} }   \vdash   \ottnt{A_{{\mathrm{2}}}}    [  \ottnt{C}  \ottsym{/}  \alpha  ]    \sqsubseteq   \ottnt{A_{{\mathrm{2}}}}    [  \ottnt{C}  \ottsym{/}  \alpha  ]  $.
   \end{itemize}
   The first is shown by \reflem{weakening} (\ref{lem:weakening:typing-context}).
   The second is by \Srule{Refl}.
   The third is by \Srule{Refl}.

  \case \Srule{Gen}:
   We have $ \algeffseqoverindex{ \alpha_{{\mathrm{2}}} }{ \text{\unboldmath$\mathit{I_{{\mathrm{2}}}}$} }  \,  =  \, \alpha  \ottsym{,}   \algeffseqoverindex{ \alpha_{{\mathrm{1}}} }{ \text{\unboldmath$\mathit{I_{{\mathrm{1}}}}$} } $ and $\ottnt{A_{{\mathrm{1}}}} \,  =  \, \ottnt{B_{{\mathrm{1}}}}$ and $\ottnt{A_{{\mathrm{2}}}} \,  =  \, \ottnt{B_{{\mathrm{2}}}}$ and
   $\alpha \,  \not\in  \,  \mathit{ftv}  (   \text{\unboldmath$\forall$}  \,  \algeffseqoverindex{ \alpha_{{\mathrm{1}}} }{ \text{\unboldmath$\mathit{I_{{\mathrm{1}}}}$} }   \ottsym{.} \, \ottnt{A_{{\mathrm{1}}}}  \rightarrow  \ottnt{A_{{\mathrm{2}}}}  ) $.
   We show the conclusion by letting
   $ \algeffseqoverindex{ \alpha_{{\mathrm{11}}} }{ \text{\unboldmath$\mathit{I_{{\mathrm{11}}}}$} }  \,  =  \,  \algeffseqoverindex{ \alpha_{{\mathrm{1}}} }{ \text{\unboldmath$\mathit{I_{{\mathrm{1}}}}$} } $, $ \algeffseqoverindex{ \ottnt{C} }{ \text{\unboldmath$\mathit{I_{{\mathrm{11}}}}$} }  \,  =  \,  \algeffseqoverindex{ \alpha_{{\mathrm{1}}} }{ \text{\unboldmath$\mathit{I_{{\mathrm{1}}}}$} } $, and
   $ \algeffseqoverindex{ \alpha_{{\mathrm{12}}} }{ \text{\unboldmath$\mathit{I_{{\mathrm{12}}}}$} } $ and $ \algeffseqoverindex{ \beta }{ \text{\unboldmath$\mathit{J}$} } $ be the empty sequence.
   We have to show that
   \begin{itemize}
    \item $\Gamma  \ottsym{,}  \alpha  \ottsym{,}   \algeffseqoverindex{ \alpha_{{\mathrm{1}}} }{ \text{\unboldmath$\mathit{I_{{\mathrm{1}}}}$} }   \vdash  \ottnt{A_{{\mathrm{1}}}}  \sqsubseteq  \ottnt{A_{{\mathrm{1}}}}$ and
    \item $\Gamma  \ottsym{,}  \alpha  \ottsym{,}   \algeffseqoverindex{ \alpha_{{\mathrm{1}}} }{ \text{\unboldmath$\mathit{I_{{\mathrm{1}}}}$} }   \vdash  \ottnt{A_{{\mathrm{2}}}}  \sqsubseteq  \ottnt{A_{{\mathrm{2}}}}$.
   \end{itemize}
   They are derived by \Srule{Refl}.

  \case \Srule{Poly}:
   We have $ \algeffseqoverindex{ \alpha_{{\mathrm{1}}} }{ \text{\unboldmath$\mathit{I_{{\mathrm{1}}}}$} }  \,  =  \, \alpha  \ottsym{,}   \algeffseqoverindex{ \alpha_{{\mathrm{01}}} }{ \text{\unboldmath$\mathit{I_{{\mathrm{01}}}}$} } $ and $ \algeffseqoverindex{ \alpha_{{\mathrm{2}}} }{ \text{\unboldmath$\mathit{I_{{\mathrm{2}}}}$} }  \,  =  \, \alpha  \ottsym{,}   \algeffseqoverindex{ \alpha_{{\mathrm{02}}} }{ \text{\unboldmath$\mathit{I_{{\mathrm{02}}}}$} } $ and, by inversion,
   $\Gamma  \ottsym{,}  \alpha  \vdash   \text{\unboldmath$\forall$}  \,  \algeffseqoverindex{ \alpha_{{\mathrm{01}}} }{ \text{\unboldmath$\mathit{I_{{\mathrm{01}}}}$} }   \ottsym{.} \, \ottnt{A_{{\mathrm{1}}}}  \rightarrow  \ottnt{A_{{\mathrm{2}}}}  \sqsubseteq   \text{\unboldmath$\forall$}  \,  \algeffseqoverindex{ \alpha_{{\mathrm{02}}} }{ \text{\unboldmath$\mathit{I_{{\mathrm{02}}}}$} }   \ottsym{.} \, \ottnt{B_{{\mathrm{1}}}}  \rightarrow  \ottnt{B_{{\mathrm{2}}}}$.
   By the IH, there exist some $ \algeffseqoverindex{ \alpha_{{\mathrm{011}}} }{ \text{\unboldmath$\mathit{I_{{\mathrm{011}}}}$} } $, $ \algeffseqoverindex{ \alpha_{{\mathrm{12}}} }{ \text{\unboldmath$\mathit{I_{{\mathrm{12}}}}$} } $, $ \algeffseqoverindex{ \beta }{ \text{\unboldmath$\mathit{J}$} } $, and
   $ \algeffseqoverindex{ \ottnt{C_{{\mathrm{0}}}} }{ \text{\unboldmath$\mathit{I_{{\mathrm{011}}}}$} } $ such that
   \begin{itemize}
    \item $\ottsym{\{}   \algeffseqoverindex{ \alpha_{{\mathrm{01}}} }{ \text{\unboldmath$\mathit{I_{{\mathrm{01}}}}$} }   \ottsym{\}} \,  =  \, \ottsym{\{}   \algeffseqoverindex{ \alpha_{{\mathrm{011}}} }{ \text{\unboldmath$\mathit{I_{{\mathrm{011}}}}$} }   \ottsym{\}} \,  \mathbin{\uplus}  \, \ottsym{\{}   \algeffseqoverindex{ \alpha_{{\mathrm{12}}} }{ \text{\unboldmath$\mathit{I_{{\mathrm{12}}}}$} }   \ottsym{\}}$,
    \item $\Gamma  \ottsym{,}  \alpha  \ottsym{,}   \algeffseqoverindex{ \alpha_{{\mathrm{02}}} }{ \text{\unboldmath$\mathit{I_{{\mathrm{02}}}}$} }   \ottsym{,}   \algeffseqoverindex{ \beta }{ \text{\unboldmath$\mathit{J}$} }   \vdash   \algeffseqoverindex{ \ottnt{C_{{\mathrm{0}}}} }{ \text{\unboldmath$\mathit{I_{{\mathrm{011}}}}$} } $,
    \item $\Gamma  \ottsym{,}  \alpha  \ottsym{,}   \algeffseqoverindex{ \alpha_{{\mathrm{02}}} }{ \text{\unboldmath$\mathit{I_{{\mathrm{02}}}}$} }   \vdash  \ottnt{B_{{\mathrm{1}}}}  \sqsubseteq    \text{\unboldmath$\forall$}  \,  \algeffseqoverindex{ \beta }{ \text{\unboldmath$\mathit{J}$} }   \ottsym{.} \, \ottnt{A_{{\mathrm{1}}}}    [   \algeffseqoverindex{ \ottnt{C_{{\mathrm{0}}}} }{ \text{\unboldmath$\mathit{I_{{\mathrm{011}}}}$} }   \ottsym{/}   \algeffseqoverindex{ \alpha_{{\mathrm{011}}} }{ \text{\unboldmath$\mathit{I_{{\mathrm{011}}}}$} }   ]  $,
    \item $\Gamma  \ottsym{,}  \alpha  \ottsym{,}   \algeffseqoverindex{ \alpha_{{\mathrm{02}}} }{ \text{\unboldmath$\mathit{I_{{\mathrm{02}}}}$} }   \vdash    \text{\unboldmath$\forall$}  \,  \algeffseqoverindex{ \alpha_{{\mathrm{12}}} }{ \text{\unboldmath$\mathit{I_{{\mathrm{12}}}}$} }   \ottsym{.} \,  \text{\unboldmath$\forall$}  \,  \algeffseqoverindex{ \beta }{ \text{\unboldmath$\mathit{J}$} }   \ottsym{.} \, \ottnt{A_{{\mathrm{2}}}}    [   \algeffseqoverindex{ \ottnt{C_{{\mathrm{0}}}} }{ \text{\unboldmath$\mathit{I_{{\mathrm{011}}}}$} }   \ottsym{/}   \algeffseqoverindex{ \alpha_{{\mathrm{011}}} }{ \text{\unboldmath$\mathit{I_{{\mathrm{011}}}}$} }   ]    \sqsubseteq  \ottnt{B_{{\mathrm{2}}}}$, and
    \item type variables in $ \algeffseqoverindex{ \beta }{ \text{\unboldmath$\mathit{J}$} } $ do not appear free in $\ottnt{A_{{\mathrm{1}}}}$ and $\ottnt{B_{{\mathrm{1}}}}$.
   \end{itemize}
   We can prove the conclusion by letting $ \algeffseqoverindex{ \alpha_{{\mathrm{11}}} }{ \text{\unboldmath$\mathit{I_{{\mathrm{11}}}}$} }  \,  =  \, \alpha  \ottsym{,}   \algeffseqoverindex{ \alpha_{{\mathrm{011}}} }{ \text{\unboldmath$\mathit{I_{{\mathrm{011}}}}$} } $ and
   $ \algeffseqoverindex{ \ottnt{C} }{ \text{\unboldmath$\mathit{I_{{\mathrm{11}}}}$} }  \,  =  \, \alpha  \ottsym{,}   \algeffseqoverindex{ \ottnt{C_{{\mathrm{0}}}} }{ \text{\unboldmath$\mathit{I_{{\mathrm{011}}}}$} } $.

  \case \Srule{DFun}:
   It is found that, for some $\alpha$, $ \algeffseqoverindex{ \alpha_{{\mathrm{1}}} }{ \text{\unboldmath$\mathit{I_{{\mathrm{1}}}}$} }  \,  =  \, \alpha$ and $ \algeffseqoverindex{ \alpha_{{\mathrm{2}}} }{ \text{\unboldmath$\mathit{I_{{\mathrm{2}}}}$} } $ is the empty sequence and
   $\ottnt{B_{{\mathrm{1}}}} \,  =  \, \ottnt{A_{{\mathrm{1}}}}$ and $\ottnt{B_{{\mathrm{2}}}} \,  =  \,  \text{\unboldmath$\forall$}  \, \alpha  \ottsym{.} \, \ottnt{A_{{\mathrm{2}}}}$.
   We show the conclusion by letting
   $ \algeffseqoverindex{ \alpha_{{\mathrm{12}}} }{ \text{\unboldmath$\mathit{I_{{\mathrm{12}}}}$} }  \,  =  \, \alpha$ and $ \algeffseqoverindex{ \alpha_{{\mathrm{11}}} }{ \text{\unboldmath$\mathit{I_{{\mathrm{11}}}}$} } $, $ \algeffseqoverindex{ \ottnt{C} }{ \text{\unboldmath$\mathit{I_{{\mathrm{11}}}}$} } $, and $ \algeffseqoverindex{ \beta }{ \text{\unboldmath$\mathit{J}$} } $ be the empty sequence.
   It suffices to show that $\Gamma  \vdash  \ottnt{A_{{\mathrm{1}}}}  \sqsubseteq  \ottnt{A_{{\mathrm{1}}}}$ and $\Gamma  \vdash   \text{\unboldmath$\forall$}  \, \alpha  \ottsym{.} \, \ottnt{A_{{\mathrm{2}}}}  \sqsubseteq   \text{\unboldmath$\forall$}  \, \alpha  \ottsym{.} \, \ottnt{A_{{\mathrm{2}}}}$,
   which are derived by \Srule{Refl}.

  \case \Srule{Prod}, \Srule{Sum}, \Srule{List}, \Srule{DProd}, \Srule{DSum}, and \Srule{DList}:
   Contradictory.
 \end{caseanalysis}
\end{proof}

\ifrestate
\lemmSubtypingInvFunMono*
\else
\begin{lemma}{subtyping-inv-fun-mono}
 If $\Gamma  \vdash  \ottnt{A_{{\mathrm{1}}}}  \rightarrow  \ottnt{A_{{\mathrm{2}}}}  \sqsubseteq  \ottnt{B_{{\mathrm{1}}}}  \rightarrow  \ottnt{B_{{\mathrm{2}}}}$,
 then $\Gamma  \vdash  \ottnt{B_{{\mathrm{1}}}}  \sqsubseteq  \ottnt{A_{{\mathrm{1}}}}$ and $\Gamma  \vdash  \ottnt{A_{{\mathrm{2}}}}  \sqsubseteq  \ottnt{B_{{\mathrm{2}}}}$.
\end{lemma}
\fi
\begin{proof}
 By \reflem{subtyping-inv-fun},
 $\Gamma  \vdash  \ottnt{B_{{\mathrm{1}}}}  \sqsubseteq   \text{\unboldmath$\forall$}  \,  \algeffseqover{ \alpha }   \ottsym{.} \, \ottnt{A_{{\mathrm{1}}}}$ and
 $\Gamma  \vdash   \text{\unboldmath$\forall$}  \,  \algeffseqover{ \alpha }   \ottsym{.} \, \ottnt{A_{{\mathrm{2}}}}  \sqsubseteq  \ottnt{B_{{\mathrm{2}}}}$ for some ${[<X>]}$ such that
 type variables in $ \algeffseqover{ \alpha } $ do not appear free in $\ottnt{A_{{\mathrm{1}}}}$ and $\ottnt{A_{{\mathrm{2}}}}$.
 Since $\Gamma  \vdash   \text{\unboldmath$\forall$}  \,  \algeffseqover{ \alpha }   \ottsym{.} \, \ottnt{A_{{\mathrm{1}}}}  \sqsubseteq  \ottnt{A_{{\mathrm{1}}}}$ by \Srule{Inst} (we can substitute any type,
 e.g., $ \text{\unboldmath$\forall$}  \, \beta  \ottsym{.} \, \beta$, for $ \algeffseqover{ \alpha } $), we have $\Gamma  \vdash  \ottnt{B_{{\mathrm{1}}}}  \sqsubseteq  \ottnt{A_{{\mathrm{1}}}}$ by \Srule{Trans}.
 Since $\Gamma  \vdash  \ottnt{A_{{\mathrm{2}}}}  \sqsubseteq   \text{\unboldmath$\forall$}  \,  \algeffseqover{ \alpha }   \ottsym{.} \, \ottnt{A_{{\mathrm{2}}}}$ by \Srule{Gen},
 we have $\Gamma  \vdash  \ottnt{A_{{\mathrm{2}}}}  \sqsubseteq  \ottnt{B_{{\mathrm{2}}}}$.
\end{proof}

\begin{lemmap}{Value inversion: constants}{val-inv-const}
 If $\Gamma  \vdash  \ottnt{c}  \ottsym{:}  \ottnt{A}$, then $\Gamma  \vdash   \mathit{ty}  (  \ottnt{c}  )   \sqsubseteq  \ottnt{A}$.
\end{lemmap}
\begin{proof}
 By induction on the typing derivation for $\ottnt{c}$.
 There are only three typing rules that can be applied to $\ottnt{c}$.
 \begin{caseanalysis}
  \case \T{Const}:  By \Srule{Refl}.
  \case \T{Gen}: We are given
   $\Gamma  \vdash  \ottnt{c}  \ottsym{:}   \text{\unboldmath$\forall$}  \, \alpha  \ottsym{.} \, \ottnt{B}$ (i.e., $\ottnt{A} \,  =  \,  \text{\unboldmath$\forall$}  \, \alpha  \ottsym{.} \, \ottnt{B}$) and, by inversion,
   $\Gamma  \ottsym{,}  \alpha  \vdash  \ottnt{c}  \ottsym{:}  \ottnt{B}$.
   By the IH, $\Gamma  \ottsym{,}  \alpha  \vdash   \mathit{ty}  (  \ottnt{c}  )   \sqsubseteq  \ottnt{B}$.
   By \Srule{Poly}, $\Gamma  \vdash   \text{\unboldmath$\forall$}  \, \alpha  \ottsym{.} \,  \mathit{ty}  (  \ottnt{c}  )   \sqsubseteq   \text{\unboldmath$\forall$}  \, \alpha  \ottsym{.} \, \ottnt{B}$.
   Since $ \mathit{ty}  (  \ottnt{c}  ) $ is closed, we have $\Gamma  \vdash   \mathit{ty}  (  \ottnt{c}  )   \sqsubseteq   \text{\unboldmath$\forall$}  \, \alpha  \ottsym{.} \,  \mathit{ty}  (  \ottnt{c}  ) $ by
   \Srule{Gen}.
   Thus, by \Srule{Trans}, we have the conclusion.

  \case \T{Inst}: By the IH and \Srule{Trans}.
 \end{caseanalysis}
\end{proof}

\ifrestate
\lemmProgress*
\else
\begin{lemmap}{Progress}{progress}
 If $\Delta  \vdash  \ottnt{M}  \ottsym{:}  \ottnt{A}$, then:
 \begin{itemize}
  \item $\ottnt{M}  \longrightarrow  \ottnt{M'}$ for some $\ottnt{M'}$;
  \item $\ottnt{M}$ is a value; or
  \item $\ottnt{M} \,  =  \,  \ottnt{E}  [   \textup{\texttt{\#}\relax}  \mathsf{op}   \ottsym{(}   \ottnt{v}   \ottsym{)}   ] $ for some $\ottnt{E}$, $\mathsf{op}$, and $\ottnt{v}$
 such that $\mathsf{op} \,  \not\in  \, \ottnt{E}$.
 \end{itemize}
\end{lemmap}
\fi
\begin{proof}
 By induction on the typing derivation for $\ottnt{M}$.  We proceed by case analysis
 on the typing rule applied last to derive $\Delta  \vdash  \ottnt{M}  \ottsym{:}  \ottnt{A}$.
 \begin{caseanalysis}
  \case \T{Var}: Contradictory.
  \case \T{Const}, \T{Abs}, and \T{Nil}: Obvious.
  \case \T{Abs}: Obvious.
  \case \T{App}:
  We are given
  \begin{itemize}
   \item $\ottnt{M} \,  =  \, \ottnt{M_{{\mathrm{1}}}} \, \ottnt{M_{{\mathrm{2}}}}$,
   \item $\Delta  \vdash  \ottnt{M_{{\mathrm{1}}}} \, \ottnt{M_{{\mathrm{2}}}}  \ottsym{:}  \ottnt{A}$,
   \item $\Delta  \vdash  \ottnt{M_{{\mathrm{1}}}}  \ottsym{:}  \ottnt{B}  \rightarrow  \ottnt{A}$, and
   \item $\Delta  \vdash  \ottnt{M_{{\mathrm{2}}}}  \ottsym{:}  \ottnt{B}$
  \end{itemize}
  for some $\ottnt{M_{{\mathrm{1}}}}$, $\ottnt{M_{{\mathrm{2}}}}$, and $\ottnt{B}$.
  By case analysis on the behavior of $\ottnt{M_{{\mathrm{1}}}}$.
  We have three cases to consider by the IH.
  \begin{caseanalysis}
   \case $\ottnt{M_{{\mathrm{1}}}}  \longrightarrow  \ottnt{M'_{{\mathrm{1}}}}$ for some $\ottnt{M'_{{\mathrm{1}}}}$: We have $\ottnt{M}  \longrightarrow  \ottnt{M'_{{\mathrm{1}}}} \, \ottnt{M_{{\mathrm{2}}}}$.
   \case $\ottnt{M_{{\mathrm{1}}}} \,  =  \,  \ottnt{E_{{\mathrm{1}}}}  [   \textup{\texttt{\#}\relax}  \mathsf{op}   \ottsym{(}   \ottnt{v}   \ottsym{)}   ] $ for some $\ottnt{E_{{\mathrm{1}}}}$, $\mathsf{op}$, and $\ottnt{v}$
         such that $\mathsf{op} \,  \not\in  \, \ottnt{E_{{\mathrm{1}}}}$:
     We have the third case in the conclusion by letting $\ottnt{E} \,  =  \, \ottnt{E_{{\mathrm{1}}}} \, \ottnt{M_{{\mathrm{2}}}}$.
   \case $\ottnt{M_{{\mathrm{1}}}} \,  =  \, \ottnt{v_{{\mathrm{1}}}}$ for some $\ottnt{v_{{\mathrm{1}}}}$:
    By case analysis on the behavior of $\ottnt{M_{{\mathrm{2}}}}$ with the IH.
    \begin{caseanalysis}
    \case $\ottnt{M_{{\mathrm{2}}}}  \longrightarrow  \ottnt{M'_{{\mathrm{2}}}}$ for some $\ottnt{M'_{{\mathrm{2}}}}$: We have $\ottnt{M}  \longrightarrow  \ottnt{v_{{\mathrm{1}}}} \, \ottnt{M'_{{\mathrm{2}}}}$.
    \case $\ottnt{M_{{\mathrm{2}}}} \,  =  \,  \ottnt{E_{{\mathrm{2}}}}  [   \textup{\texttt{\#}\relax}  \mathsf{op}   \ottsym{(}   \ottnt{v}   \ottsym{)}   ] $ for some $\ottnt{E_{{\mathrm{2}}}}$, $\mathsf{op}$, and $\ottnt{v}$
          such that $\mathsf{op} \,  \not\in  \, \ottnt{E_{{\mathrm{2}}}}$:
       We have the third case in the conclusion by letting $\ottnt{E} \,  =  \, \ottnt{v_{{\mathrm{1}}}} \, \ottnt{E_{{\mathrm{2}}}}$.
    \case $\ottnt{M_{{\mathrm{2}}}} \,  =  \, \ottnt{v_{{\mathrm{2}}}}$ for some $\ottnt{v_{{\mathrm{2}}}}$:
     By \reflem{canonical-forms} on $\ottnt{v_{{\mathrm{1}}}}$, we have two cases to consider.
     \begin{caseanalysis}
      \case $\ottnt{v_{{\mathrm{1}}}} \,  =  \, \ottnt{c_{{\mathrm{1}}}}$:
       Since $\Delta  \vdash  \ottnt{c_{{\mathrm{1}}}}  \ottsym{:}  \ottnt{B}  \rightarrow  \ottnt{A}$,
       we have $\Delta  \vdash   \mathit{ty}  (  \ottnt{c_{{\mathrm{1}}}}  )   \sqsubseteq  \ottnt{B}  \rightarrow  \ottnt{A}$
       by \reflem{val-inv-const}.
       By \reflem{subtyping-unqualify} (\ref{lem:subtyping-unqualify:fun}),
       it is found that $ \mathit{ty}  (  \ottnt{c_{{\mathrm{1}}}}  )  \,  =  \, \iota  \rightarrow  \ottnt{C}$ for some $\iota$ and $\ottnt{C}$.
       Since $\Delta  \vdash  \iota  \rightarrow  \ottnt{C}  \sqsubseteq  \ottnt{B}  \rightarrow  \ottnt{A}$,
       we have $\Delta  \vdash  \ottnt{B}  \sqsubseteq  \iota$
       for some $ \algeffseqoverindex{ \gamma }{ \text{\unboldmath$\mathit{I_{{\mathrm{0}}}}$} } $ by \reflem{subtyping-inv-fun-mono}.
       Since $\Delta  \vdash  \ottnt{v_{{\mathrm{2}}}}  \ottsym{:}  \ottnt{B}$,
       $ \mathit{unqualify}  (  \ottnt{B}  ) $ is not a type variable
       by \reflem{canonical-forms-unqualify-no-tyvar}.
       Thus, since $\Delta  \vdash  \ottnt{B}  \sqsubseteq  \iota$,
       it is found that $ \mathit{unqualify}  (  \ottnt{B}  )  \,  =  \, \iota$
       by \reflem{subtyping-unqualify}.
       Since $\Delta  \vdash  \ottnt{v_{{\mathrm{2}}}}  \ottsym{:}  \ottnt{B}$, we have $\ottnt{v_{{\mathrm{2}}}} \,  =  \, \ottnt{c_{{\mathrm{2}}}}$ for some $\ottnt{c_{{\mathrm{2}}}}$
       by \reflem{canonical-forms}.
       Since $\Delta  \vdash  \ottnt{c_{{\mathrm{2}}}}  \ottsym{:}  \ottnt{B}$,
       we have $\Delta  \vdash   \mathit{ty}  (  \ottnt{c_{{\mathrm{2}}}}  )   \sqsubseteq  \ottnt{B}$ by \reflem{val-inv-const}.
       Since $ \mathit{unqualify}  (  \ottnt{B}  )  \,  =  \, \iota$, we have $ \mathit{ty}  (  \ottnt{c_{{\mathrm{2}}}}  )  \,  =  \, \iota$
       by \reflem{subtyping-unqualify}.
       Thus, $ \zeta  (  \ottnt{c_{{\mathrm{1}}}}  ,  \ottnt{c_{{\mathrm{2}}}}  ) $ is defined, and
       $\ottnt{M} = \ottnt{c_{{\mathrm{1}}}} \, \ottnt{c_{{\mathrm{2}}}}  \longrightarrow   \zeta  (  \ottnt{c_{{\mathrm{1}}}}  ,  \ottnt{c_{{\mathrm{2}}}}  ) $ by \R{Const}/\E{Eval}.

      \case $\ottnt{v_{{\mathrm{1}}}} \,  =  \,  \lambda\!  \, \mathit{x}  \ottsym{.}  \ottnt{M'}$:
       By \R{Beta}/\E{Eval}, $\ottnt{M} = \ottsym{(}   \lambda\!  \, \mathit{x}  \ottsym{.}  \ottnt{M'}  \ottsym{)} \, \ottnt{v_{{\mathrm{2}}}}  \longrightarrow   \ottnt{M'}    [  \ottnt{v_{{\mathrm{2}}}}  /  \mathit{x}  ]  $.
     \end{caseanalysis}
   \end{caseanalysis}
  \end{caseanalysis}

  \case \T{Gen}: By the IH.
  \case \T{Inst}: By the IH.

  \case \T{Op}:
   We are given
   \begin{itemize}
    \item $\ottnt{M} \,  =  \,  \textup{\texttt{\#}\relax}  \mathsf{op}   \ottsym{(}   \ottnt{M'}   \ottsym{)} $,
    \item $\mathit{ty} \, \ottsym{(}  \mathsf{op}  \ottsym{)} \,  =  \,   \text{\unboldmath$\forall$}  \,  \algeffseqover{ \alpha }   \ottsym{.} \,  \ottnt{A'}  \hookrightarrow  \ottnt{B'} $,
    \item $\Delta  \vdash   \textup{\texttt{\#}\relax}  \mathsf{op}   \ottsym{(}   \ottnt{M'}   \ottsym{)}   \ottsym{:}   \ottnt{B'}    [   \algeffseqover{ \ottnt{C} }   \ottsym{/}   \algeffseqover{ \alpha }   ]  $, and
    \item $\Delta  \vdash  \ottnt{M'}  \ottsym{:}   \ottnt{A'}    [   \algeffseqover{ \ottnt{C} }   \ottsym{/}   \algeffseqover{ \alpha }   ]  $
  \end{itemize}
   for some $\mathsf{op}$, $\ottnt{M'}$, $ \algeffseqover{ \alpha } $, $\ottnt{A'}$, $\ottnt{B'}$, and $ \algeffseqover{ \ottnt{C} } $.
   By case analysis on the behavior of $\ottnt{M'}$ with the IH.
   \begin{caseanalysis}
    \case $\ottnt{M'}  \longrightarrow  \ottnt{M''}$ for some $\ottnt{M''}$: We have $\ottnt{M}  \longrightarrow   \textup{\texttt{\#}\relax}  \mathsf{op}   \ottsym{(}   \ottnt{M''}   \ottsym{)} $.
    \case $\ottnt{M'} \,  =  \,  \ottnt{E'}  [   \textup{\texttt{\#}\relax}  \mathsf{op}'   \ottsym{(}   \ottnt{v}   \ottsym{)}   ] $ for some $\ottnt{E'}$, $\mathsf{op}'$, and $\ottnt{v}$
          such that $\mathsf{op}' \,  \not\in  \, \ottnt{E'}$:
      We have the third case in the conclusion by letting $\ottnt{E} \,  =  \,  \textup{\texttt{\#}\relax}  \mathsf{op}   \ottsym{(}   \ottnt{E'}   \ottsym{)} $.
   \case $\ottnt{M'} \,  =  \, \ottnt{v}$ for some $\ottnt{v}$:
     We have the third case in the conclusion by letting $\ottnt{E} \,  =  \,  [] $.
  \end{caseanalysis}

  \case \T{Handle}:
  We are given
  \begin{itemize}
   \item $\ottnt{M} \,  =  \, \mathsf{handle} \, \ottnt{M'} \, \mathsf{with} \, \ottnt{H}$,
   \item $\Delta  \vdash  \ottnt{M'}  \ottsym{:}  \ottnt{B}$, and
   \item $\Delta  \vdash  \ottnt{H}  \ottsym{:}  \ottnt{B}  \Rightarrow  \ottnt{A}$
  \end{itemize}
  for some $\ottnt{M'}$, $\ottnt{H}$, and $\ottnt{B}$.
  By case analysis on the behavior of $\ottnt{M'}$ with the IH.
  \begin{caseanalysis}
   \case $\ottnt{M'}  \longrightarrow  \ottnt{M''}$ for some $\ottnt{M''}$:
    We have $\ottnt{M}  \longrightarrow  \mathsf{handle} \, \ottnt{M''} \, \mathsf{with} \, \ottnt{H}$.
   \case $\ottnt{M'} \,  =  \,  \ottnt{E'}  [   \textup{\texttt{\#}\relax}  \mathsf{op}   \ottsym{(}   \ottnt{v}   \ottsym{)}   ] $ for some $\ottnt{E'}$, $\mathsf{op}$, and $\ottnt{v}$
         such that $\mathsf{op} \,  \not\in  \, \ottnt{E'}$:
    If handler $\ottnt{H}$ contains an operation clause $\mathsf{op}  \ottsym{(}  \mathit{x}  \ottsym{,}  \mathit{k}  \ottsym{)}  \rightarrow  \ottnt{M''}$,
    then we have $\ottnt{M}  \longrightarrow    \ottnt{M''}    [  \ottnt{v}  /  \mathit{x}  ]      [   \lambda\!  \, \mathit{y}  \ottsym{.}  \mathsf{handle} \,  \ottnt{E'}  [  \mathit{y}  ]  \, \mathsf{with} \, \ottnt{H}  /  \mathit{k}  ]  $
    by \R{Handle}/\E{Eval}.

    Otherwise, if $\ottnt{H}$ contains no operation clause for $\mathsf{op}$,
    we have the third case in the conclusion by letting
    $\ottnt{E} \,  =  \, \mathsf{handle} \, \ottnt{E'} \, \mathsf{with} \, \ottnt{H}$.

   \case $\ottnt{M'} \,  =  \, \ottnt{v}$ for some $\ottnt{v}$: By \R{Return}/\E{Eval}.
  \end{caseanalysis}

  \case \T{Pair}:
   We are given
   \begin{itemize}
    \item $\ottnt{M} \,  =  \, \ottsym{(}  \ottnt{M_{{\mathrm{1}}}}  \ottsym{,}  \ottnt{M_{{\mathrm{2}}}}  \ottsym{)}$,
    \item $\Delta  \vdash  \ottnt{M_{{\mathrm{1}}}}  \ottsym{:}  \ottnt{B_{{\mathrm{1}}}}$, and
    \item $\Delta  \vdash  \ottnt{M_{{\mathrm{2}}}}  \ottsym{:}  \ottnt{B_{{\mathrm{2}}}}$
   \end{itemize}
   for some $\ottnt{M_{{\mathrm{1}}}}$, $\ottnt{M_{{\mathrm{2}}}}$, $\ottnt{B_{{\mathrm{1}}}}$, and $\ottnt{B_{{\mathrm{2}}}}$.
   By case analysis on the behavior of $\ottnt{M_{{\mathrm{1}}}}$ with the IH.
   \begin{caseanalysis}
    \case $\ottnt{M_{{\mathrm{1}}}}  \longrightarrow  \ottnt{M'_{{\mathrm{1}}}}$ for some $\ottnt{M'_{{\mathrm{1}}}}$:
     We have $\ottnt{M} \,  =  \, \ottsym{(}  \ottnt{M'_{{\mathrm{1}}}}  \ottsym{,}  \ottnt{M_{{\mathrm{2}}}}  \ottsym{)}$.

    \case $\ottnt{M_{{\mathrm{1}}}} \,  =  \,  \ottnt{E_{{\mathrm{1}}}}  [   \textup{\texttt{\#}\relax}  \mathsf{op}   \ottsym{(}   \ottnt{v}   \ottsym{)}   ] $ for some $\ottnt{E_{{\mathrm{1}}}}$, $\mathsf{op}$, and $\ottnt{v}$
          such that $\mathsf{op} \,  \not\in  \, \ottnt{E_{{\mathrm{1}}}}$:
     We have the third case in the conclusion by letting
     $\ottnt{E} \,  =  \, \ottsym{(}  \ottnt{E_{{\mathrm{1}}}}  \ottsym{,}  \ottnt{M_{{\mathrm{2}}}}  \ottsym{)}$.

    \case $\ottnt{M_{{\mathrm{1}}}} \,  =  \, \ottnt{v_{{\mathrm{1}}}}$ for some $\ottnt{v_{{\mathrm{1}}}}$:
     By case analysis on the behavior of $\ottnt{M_{{\mathrm{2}}}}$ with the IH.
     \begin{caseanalysis}
      \case $\ottnt{M_{{\mathrm{2}}}}  \longrightarrow  \ottnt{M'_{{\mathrm{2}}}}$:
       We have $\ottnt{M_{{\mathrm{2}}}}  \longrightarrow  \ottsym{(}  \ottnt{v_{{\mathrm{1}}}}  \ottsym{,}  \ottnt{M'_{{\mathrm{2}}}}  \ottsym{)}$.

      \case $\ottnt{M_{{\mathrm{2}}}} \,  =  \,  \ottnt{E_{{\mathrm{2}}}}  [   \textup{\texttt{\#}\relax}  \mathsf{op}   \ottsym{(}   \ottnt{v}   \ottsym{)}   ] $ for some $\ottnt{E_{{\mathrm{2}}}}$, $\mathsf{op}$, and $\ottnt{v}$
            such that $\mathsf{op} \,  \not\in  \, \ottnt{E_{{\mathrm{2}}}}$:
       We have the third case in the conclusion by letting
       $\ottnt{E} \,  =  \, \ottsym{(}  \ottnt{v_{{\mathrm{1}}}}  \ottsym{,}  \ottnt{E_{{\mathrm{2}}}}  \ottsym{)}$.

      \case $\ottnt{M_{{\mathrm{2}}}} \,  =  \, \ottnt{v_{{\mathrm{2}}}}$:
       We have the second case in the conclusion since $\ottnt{M} \,  =  \, \ottsym{(}  \ottnt{v_{{\mathrm{1}}}}  \ottsym{,}  \ottnt{v_{{\mathrm{2}}}}  \ottsym{)}$.
     \end{caseanalysis}
   \end{caseanalysis}

  \case \T{Proj1}:
   We are given
   \begin{itemize}
    \item $\ottnt{M} \,  =  \, \pi_1  \ottnt{M'}$ and
    \item $\Delta  \vdash  \ottnt{M'}  \ottsym{:}   \ottnt{A}  \times  \ottnt{B} $
   \end{itemize}
   for some $\ottnt{M'}$ and $\ottnt{B}$.
   By case analysis on the behavior of $\ottnt{M'}$ with the IH.
   \begin{caseanalysis}
    \case $\ottnt{M'}  \longrightarrow  \ottnt{M''}$ for some $\ottnt{M''}$:
     We have $\ottnt{M}  \longrightarrow  \pi_1  \ottnt{M''}$.

    \case $\ottnt{M'} \,  =  \,  \ottnt{E'}  [   \textup{\texttt{\#}\relax}  \mathsf{op}   \ottsym{(}   \ottnt{v}   \ottsym{)}   ] $ for some $\ottnt{E'}$, $\mathsf{op}$, and $\ottnt{v}$
          such that $\mathsf{op} \,  \not\in  \, \ottnt{E'}$:
     We have the third case in the conclusion by letting $\ottnt{E} \,  =  \, \pi_1  \ottnt{E'}$.

    \case $\ottnt{M'} \,  =  \, \ottnt{v'}$ for some $\ottnt{v'}$:
     Since $\Delta  \vdash  \ottnt{M'}  \ottsym{:}   \ottnt{A}  \times  \ottnt{B} $ (i.e., $\Delta  \vdash  \ottnt{v'}  \ottsym{:}   \ottnt{A}  \times  \ottnt{B} $),
     we have $\ottnt{v'} \,  =  \, \ottsym{(}  \ottnt{v_{{\mathrm{1}}}}  \ottsym{,}  \ottnt{v_{{\mathrm{2}}}}  \ottsym{)}$ for some $\ottnt{v_{{\mathrm{1}}}}$ and $\ottnt{v_{{\mathrm{2}}}}$
     by \reflem{canonical-forms}.
     By \R{Proj1}/\E{Eval}, we finish.
   \end{caseanalysis}

  \case \T{Proj2}: Similarly to the case for \T{Proj1}.
  \case \T{InL}, \T{InR}, and \T{Cons}: Similarly to the case for \T{Pair}.
  \case \T{Case}:
   We are given
   \begin{itemize}
    \item $\ottnt{M} \,  =  \, \mathsf{case} \, \ottnt{M'} \, \mathsf{of} \, \mathsf{inl} \, \mathit{x}  \rightarrow  \ottnt{M_{{\mathrm{1}}}}  \ottsym{;} \, \mathsf{inr} \, \mathit{y}  \rightarrow  \ottnt{M_{{\mathrm{2}}}}$ and
    \item $\Delta  \vdash  \ottnt{M'}  \ottsym{:}   \ottnt{B}  +  \ottnt{C} $
   \end{itemize}
   for some $\ottnt{M'}$, $\ottnt{M_{{\mathrm{1}}}}$, $\ottnt{M_{{\mathrm{2}}}}$, $\mathit{x}$, $\mathit{y}$, $\ottnt{B}$, and $\ottnt{C}$.
   By case analysis on the behavior of $\ottnt{M'}$ wit the IH.
   \begin{caseanalysis}
    \case $\ottnt{M'}  \longrightarrow  \ottnt{M''}$ for some $\ottnt{M''}$:
     We have $\ottnt{M}  \longrightarrow  \mathsf{case} \, \ottnt{M''} \, \mathsf{of} \, \mathsf{inl} \, \mathit{x}  \rightarrow  \ottnt{M_{{\mathrm{1}}}}  \ottsym{;} \, \mathsf{inr} \, \mathit{y}  \rightarrow  \ottnt{M_{{\mathrm{2}}}}$.

    \case $\ottnt{M'} \,  =  \,  \ottnt{E'}  [   \textup{\texttt{\#}\relax}  \mathsf{op}   \ottsym{(}   \ottnt{v}   \ottsym{)}   ] $ for some $\ottnt{E'}$, $\mathsf{op}$, and $\ottnt{v}$
          such that $\mathsf{op} \,  \not\in  \, \ottnt{E'}$:
     We have the third case in the conclusion by letting
     $\ottnt{E} \,  =  \, \mathsf{case} \, \ottnt{E'} \, \mathsf{of} \, \mathsf{inl} \, \mathit{x}  \rightarrow  \ottnt{M_{{\mathrm{1}}}}  \ottsym{;} \, \mathsf{inr} \, \mathit{y}  \rightarrow  \ottnt{M_{{\mathrm{2}}}}$.

    \case $\ottnt{M'} \,  =  \, \ottnt{v}$ for some $\ottnt{v}$:
     By \reflem{canonical-forms},
     $\ottnt{v} \,  =  \, \mathsf{inl} \, \ottnt{v'}$ or $\ottnt{v} \,  =  \, \mathsf{inr} \, \ottnt{v'}$ for some $\ottnt{v'}$.
     We finish by \R{CaseL}/\E{Eval} or \R{CaseR}/\E{Eval}.
   \end{caseanalysis}

  \case \T{CaseList}: Similar to the case for \T{Case}.
  \case \T{Fix}: By \R{Fix}/\E{Eval}.
 \end{caseanalysis}
\end{proof}


\begin{lemma}{type-wf}
 \begin{enumerate}
  \item If $\Gamma  \vdash  \ottnt{M}  \ottsym{:}  \ottnt{A}$, then $\Gamma  \vdash  \ottnt{A}$.
  \item If $\Gamma  \vdash  \ottnt{H}  \ottsym{:}  \ottnt{A}  \Rightarrow  \ottnt{B}$, then $\Gamma  \vdash  \ottnt{B}$.
 \end{enumerate}
 
\end{lemma}
\begin{proof}
 Straightforward by mutual induction on the typing derivations.  The case for
 \T{Op} depends on \reflem{ty-subst} and \refdef{eff}, which states that,
 for $\mathsf{op}$ such that $\mathit{ty} \, \ottsym{(}  \mathsf{op}  \ottsym{)} \,  =  \,   \text{\unboldmath$\forall$}  \,  \algeffseqover{ \alpha }   \ottsym{.} \,  \ottnt{A}  \hookrightarrow  \ottnt{B} $, $ \mathit{ftv}  (  \ottnt{B}  )  \,  \subseteq  \, \ottsym{\{}   \algeffseqover{ \alpha }   \ottsym{\}}$.
\end{proof}

\begin{lemmap}{Value inversion: lambda abstractions}{val-inv-abs}
 If $\Gamma  \vdash   \lambda\!  \, \mathit{x}  \ottsym{.}  \ottnt{M}  \ottsym{:}  \ottnt{A}$,
 then $\Gamma  \ottsym{,}   \algeffseqover{ \alpha }   \ottsym{,}  \mathit{x} \,  \mathord{:}  \, \ottnt{B}  \vdash  \ottnt{M}  \ottsym{:}  \ottnt{C}$ and $\Gamma  \vdash   \text{\unboldmath$\forall$}  \,  \algeffseqover{ \alpha }   \ottsym{.} \, \ottnt{B}  \rightarrow  \ottnt{C}  \sqsubseteq  \ottnt{A}$
 for some $ \algeffseqover{ \alpha } $, $\ottnt{B}$, and $\ottnt{C}$.
\end{lemmap}
\begin{proof}
 By induction on the typing derivation for $ \lambda\!  \, \mathit{x}  \ottsym{.}  \ottnt{M}$.
 There are only three typing rules that can be applied to $ \lambda\!  \, \mathit{x}  \ottsym{.}  \ottnt{M}$.
 \begin{caseanalysis}
  \case \T{Abs}:  We have $\ottnt{A} \,  =  \, \ottnt{B}  \rightarrow  \ottnt{C}$ and let $ \algeffseqover{ \alpha } $ be the empty sequence.
   We have the conclusion by inversion and \Srule{Refl}.

  \case \T{Gen}: We are given
   $\Gamma  \vdash   \lambda\!  \, \mathit{x}  \ottsym{.}  \ottnt{M}  \ottsym{:}   \text{\unboldmath$\forall$}  \, \beta  \ottsym{.} \, \ottnt{D}$ (i.e., $\ottnt{A} \,  =  \,  \text{\unboldmath$\forall$}  \, \beta  \ottsym{.} \, \ottnt{D}$) and, by inversion,
   $\Gamma  \ottsym{,}  \beta  \vdash   \lambda\!  \, \mathit{x}  \ottsym{.}  \ottnt{M}  \ottsym{:}  \ottnt{D}$.
   By the IH, $\Gamma  \ottsym{,}  \beta  \ottsym{,}   \algeffseqoverindex{ \gamma }{ \text{\unboldmath$\mathit{I}$} }   \ottsym{,}  \mathit{x} \,  \mathord{:}  \, \ottnt{B}  \vdash  \ottnt{M}  \ottsym{:}  \ottnt{C}$ and
   $\Gamma  \ottsym{,}  \beta  \vdash   \text{\unboldmath$\forall$}  \,  \algeffseqoverindex{ \gamma }{ \text{\unboldmath$\mathit{I}$} }   \ottsym{.} \, \ottnt{B}  \rightarrow  \ottnt{C}  \sqsubseteq  \ottnt{D}$
   for some $ \algeffseqoverindex{ \gamma }{ \text{\unboldmath$\mathit{I}$} } $, $\ottnt{B}$, and $\ottnt{C}$.
   We show the conclusion by letting $ \algeffseqover{ \alpha }  \,  =  \, \beta  \ottsym{,}   \algeffseqoverindex{ \gamma }{ \text{\unboldmath$\mathit{I}$} } $.
   It suffices to show that $\Gamma  \vdash   \text{\unboldmath$\forall$}  \, \beta  \ottsym{.} \,  \text{\unboldmath$\forall$}  \,  \algeffseqoverindex{ \gamma }{ \text{\unboldmath$\mathit{I}$} }   \ottsym{.} \, \ottnt{B}  \rightarrow  \ottnt{C}  \sqsubseteq   \text{\unboldmath$\forall$}  \, \beta  \ottsym{.} \, \ottnt{D}$,
   which is derived from $\Gamma  \ottsym{,}  \beta  \vdash   \text{\unboldmath$\forall$}  \,  \algeffseqoverindex{ \gamma }{ \text{\unboldmath$\mathit{I}$} }   \ottsym{.} \, \ottnt{B}  \rightarrow  \ottnt{C}  \sqsubseteq  \ottnt{D}$ with \Srule{Poly}.

  \case \T{Inst}: By the IH and \Srule{Trans}.
 \end{caseanalysis}
\end{proof}

\begin{lemmap}{Value inversion: pairs}{val-inv-pair}
 If $\Gamma  \vdash  \ottsym{(}  \ottnt{M_{{\mathrm{1}}}}  \ottsym{,}  \ottnt{M_{{\mathrm{2}}}}  \ottsym{)}  \ottsym{:}  \ottnt{A}$,
 then $\Gamma  \ottsym{,}   \algeffseqover{ \alpha }   \vdash  \ottnt{M_{{\mathrm{1}}}}  \ottsym{:}  \ottnt{B_{{\mathrm{1}}}}$ and $\Gamma  \ottsym{,}   \algeffseqover{ \alpha }   \vdash  \ottnt{M_{{\mathrm{2}}}}  \ottsym{:}  \ottnt{B_{{\mathrm{2}}}}$ and
 $\Gamma  \vdash    \text{\unboldmath$\forall$}  \,  \algeffseqover{ \alpha }   \ottsym{.} \, \ottnt{B_{{\mathrm{1}}}}  \times  \ottnt{B_{{\mathrm{2}}}}   \sqsubseteq  \ottnt{A}$
 for some $ \algeffseqover{ \alpha } $, $\ottnt{B_{{\mathrm{1}}}}$, and $\ottnt{B_{{\mathrm{2}}}}$.
\end{lemmap}
\begin{proof}
 By induction on the typing derivation for $\ottsym{(}  \ottnt{M_{{\mathrm{1}}}}  \ottsym{,}  \ottnt{M_{{\mathrm{2}}}}  \ottsym{)}$.
 There are only three typing rules that can be applied to $\ottsym{(}  \ottnt{M_{{\mathrm{1}}}}  \ottsym{,}  \ottnt{M_{{\mathrm{2}}}}  \ottsym{)}$.
 \begin{caseanalysis}
  \case \T{Pair}:  Obvious by \Srule{Refl}.

  \case \T{Gen}: We are given
   $\Gamma  \vdash  \ottsym{(}  \ottnt{M_{{\mathrm{1}}}}  \ottsym{,}  \ottnt{M_{{\mathrm{2}}}}  \ottsym{)}  \ottsym{:}   \text{\unboldmath$\forall$}  \, \beta  \ottsym{.} \, \ottnt{C}$ (i.e., $\ottnt{A} \,  =  \,  \text{\unboldmath$\forall$}  \, \beta  \ottsym{.} \, \ottnt{C}$) and, by inversion,
   $\Gamma  \ottsym{,}  \beta  \vdash  \ottsym{(}  \ottnt{M_{{\mathrm{1}}}}  \ottsym{,}  \ottnt{M_{{\mathrm{2}}}}  \ottsym{)}  \ottsym{:}  \ottnt{C}$.
   By the IH,
   $\Gamma  \ottsym{,}  \beta  \ottsym{,}   \algeffseqoverindex{ \gamma }{ \text{\unboldmath$\mathit{I}$} }   \vdash  \ottnt{M_{{\mathrm{1}}}}  \ottsym{:}  \ottnt{B_{{\mathrm{1}}}}$ and
   $\Gamma  \ottsym{,}  \beta  \ottsym{,}   \algeffseqoverindex{ \gamma }{ \text{\unboldmath$\mathit{I}$} }   \vdash  \ottnt{M_{{\mathrm{2}}}}  \ottsym{:}  \ottnt{B_{{\mathrm{2}}}}$
   $\Gamma  \ottsym{,}  \beta  \vdash    \text{\unboldmath$\forall$}  \,  \algeffseqoverindex{ \gamma }{ \text{\unboldmath$\mathit{I}$} }   \ottsym{.} \, \ottnt{B_{{\mathrm{1}}}}  \times  \ottnt{B_{{\mathrm{2}}}}   \sqsubseteq  \ottnt{C}$
   for some $ \algeffseqoverindex{ \gamma }{ \text{\unboldmath$\mathit{I}$} } $, $\ottnt{B_{{\mathrm{1}}}}$, and $\ottnt{B_{{\mathrm{2}}}}$.
   We show the conclusion by letting $ \algeffseqover{ \alpha }  \,  =  \, \beta  \ottsym{,}   \algeffseqoverindex{ \gamma }{ \text{\unboldmath$\mathit{I}$} } $.
   It suffices to show that $\Gamma  \vdash    \text{\unboldmath$\forall$}  \, \beta  \ottsym{.} \,  \text{\unboldmath$\forall$}  \,  \algeffseqoverindex{ \gamma }{ \text{\unboldmath$\mathit{I}$} }   \ottsym{.} \, \ottnt{B_{{\mathrm{1}}}}  \times  \ottnt{B_{{\mathrm{2}}}}   \sqsubseteq   \text{\unboldmath$\forall$}  \, \beta  \ottsym{.} \, \ottnt{C}$,
   which is derived from $\Gamma  \ottsym{,}  \beta  \vdash    \text{\unboldmath$\forall$}  \,  \algeffseqoverindex{ \gamma }{ \text{\unboldmath$\mathit{I}$} }   \ottsym{.} \, \ottnt{B_{{\mathrm{1}}}}  \times  \ottnt{B_{{\mathrm{2}}}}   \sqsubseteq  \ottnt{C}$ with \Srule{Poly}.

  \case \T{Inst}: By the IH and \Srule{Trans}.
 \end{caseanalysis}
\end{proof}

\begin{lemmap}{Value inversion: left injections}{val-inv-inl}
 If $\Gamma  \vdash  \mathsf{inl} \, \ottnt{M}  \ottsym{:}  \ottnt{A}$, then
 $\Gamma  \ottsym{,}   \algeffseqover{ \alpha }   \vdash  \ottnt{M}  \ottsym{:}  \ottnt{B}$ and
 $\Gamma  \vdash    \text{\unboldmath$\forall$}  \,  \algeffseqover{ \alpha }   \ottsym{.} \, \ottnt{B}  +  \ottnt{C}   \sqsubseteq  \ottnt{A}$
 for some $ \algeffseqover{ \alpha } $, $\ottnt{B}$, and $\ottnt{C}$.
\end{lemmap}
\begin{proof}
 By induction on the typing derivation for $\mathsf{inl} \, \ottnt{M}$.
 There are only three typing rules that can be applied to $\mathsf{inl} \, \ottnt{M}$.
 \begin{caseanalysis}
  \case \T{InL}:  Obvious by \Srule{Refl}.

  \case \T{Gen}: We are given
   $\Gamma  \vdash  \mathsf{inl} \, \ottnt{M}  \ottsym{:}   \text{\unboldmath$\forall$}  \, \beta  \ottsym{.} \, \ottnt{D}$ (i.e., $\ottnt{A} \,  =  \,  \text{\unboldmath$\forall$}  \, \beta  \ottsym{.} \, \ottnt{D}$) and, by inversion,
   $\Gamma  \ottsym{,}  \beta  \vdash  \mathsf{inl} \, \ottnt{M}  \ottsym{:}  \ottnt{D}$.
   By the IH,
   $\Gamma  \ottsym{,}  \beta  \ottsym{,}   \algeffseqoverindex{ \gamma }{ \text{\unboldmath$\mathit{I}$} }   \vdash  \ottnt{M}  \ottsym{:}  \ottnt{B}$ and
   $\Gamma  \ottsym{,}  \beta  \vdash    \text{\unboldmath$\forall$}  \,  \algeffseqoverindex{ \gamma }{ \text{\unboldmath$\mathit{I}$} }   \ottsym{.} \, \ottnt{B}  +  \ottnt{C}   \sqsubseteq  \ottnt{D}$
   for some $ \algeffseqoverindex{ \gamma }{ \text{\unboldmath$\mathit{I}$} } $, $\ottnt{B}$, and $\ottnt{C}$.
   We show the conclusion by letting $ \algeffseqover{ \alpha }  \,  =  \, \beta  \ottsym{,}   \algeffseqoverindex{ \gamma }{ \text{\unboldmath$\mathit{I}$} } $.
   It suffices to show that $\Gamma  \vdash    \text{\unboldmath$\forall$}  \, \beta  \ottsym{.} \,  \text{\unboldmath$\forall$}  \,  \algeffseqoverindex{ \gamma }{ \text{\unboldmath$\mathit{I}$} }   \ottsym{.} \, \ottnt{B}  +  \ottnt{C}   \sqsubseteq   \text{\unboldmath$\forall$}  \, \beta  \ottsym{.} \, \ottnt{D}$,
   which is derived from $\Gamma  \ottsym{,}  \beta  \vdash    \text{\unboldmath$\forall$}  \,  \algeffseqoverindex{ \gamma }{ \text{\unboldmath$\mathit{I}$} }   \ottsym{.} \, \ottnt{B}  +  \ottnt{C}   \sqsubseteq  \ottnt{D}$ with \Srule{Poly}.

  \case \T{Inst}: By the IH and \Srule{Trans}.
 \end{caseanalysis}
\end{proof}

\begin{lemmap}{Value inversion: right injections}{val-inv-inr}
 If $\Gamma  \vdash  \mathsf{inr} \, \ottnt{M}  \ottsym{:}  \ottnt{A}$, then
 $\Gamma  \ottsym{,}   \algeffseqover{ \alpha }   \vdash  \ottnt{M}  \ottsym{:}  \ottnt{C}$ and
 $\Gamma  \vdash    \text{\unboldmath$\forall$}  \,  \algeffseqover{ \alpha }   \ottsym{.} \, \ottnt{B}  +  \ottnt{C}   \sqsubseteq  \ottnt{A}$
 for some $ \algeffseqover{ \alpha } $, $\ottnt{B}$, and $\ottnt{C}$.
\end{lemmap}
\begin{proof}
 Similarly to the proof of \reflem{val-inv-inl}.
\end{proof}

\begin{lemmap}{Value inversion: cons}{val-inv-cons}
 If $\Gamma  \vdash  \mathsf{cons} \, \ottnt{M}  \ottsym{:}  \ottnt{A}$, then
 $\Gamma  \ottsym{,}   \algeffseqover{ \alpha }   \vdash  \ottnt{M}  \ottsym{:}    \ottnt{B}  \times  \ottnt{B}   \, \mathsf{list} $ and
 $\Gamma  \vdash    \text{\unboldmath$\forall$}  \,  \algeffseqover{ \alpha }   \ottsym{.} \, \ottnt{B}  \, \mathsf{list}   \sqsubseteq  \ottnt{A}$
 for some $ \algeffseqover{ \alpha } $ and $\ottnt{B}$.
\end{lemmap}
\begin{proof}
 By induction on the typing derivations for $\mathsf{cons} \, \ottnt{M}$.
 There are only three typing rules that can be applied to $\mathsf{cons} \, \ottnt{M}$.
 \begin{caseanalysis}
  \case \T{Cons}:  Obvious by \Srule{Refl}.

  \case \T{Gen}: We are given
   $\Gamma  \vdash  \mathsf{cons} \, \ottnt{M}  \ottsym{:}   \text{\unboldmath$\forall$}  \, \beta  \ottsym{.} \, \ottnt{C}$ (i.e., $\ottnt{A} \,  =  \,  \text{\unboldmath$\forall$}  \, \beta  \ottsym{.} \, \ottnt{C}$) and, by inversion,
   $\Gamma  \ottsym{,}  \beta  \vdash  \mathsf{cons} \, \ottnt{M}  \ottsym{:}  \ottnt{C}$.
   By the IH,
   $\Gamma  \ottsym{,}  \beta  \ottsym{,}   \algeffseqoverindex{ \gamma }{ \text{\unboldmath$\mathit{I}$} }   \vdash  \ottnt{M}  \ottsym{:}    \ottnt{B}  \times  \ottnt{B}   \, \mathsf{list} $ and
   $\Gamma  \ottsym{,}  \beta  \vdash    \text{\unboldmath$\forall$}  \,  \algeffseqoverindex{ \gamma }{ \text{\unboldmath$\mathit{I}$} }   \ottsym{.} \, \ottnt{B}  \, \mathsf{list}   \sqsubseteq  \ottnt{C}$
   for some $ \algeffseqoverindex{ \gamma }{ \text{\unboldmath$\mathit{I}$} } $ and $\ottnt{B}$.
   We show the conclusion by letting $ \algeffseqover{ \alpha }  \,  =  \, \beta  \ottsym{,}   \algeffseqoverindex{ \gamma }{ \text{\unboldmath$\mathit{I}$} } $.
   It suffices to show that $\Gamma  \vdash    \text{\unboldmath$\forall$}  \, \beta  \ottsym{.} \,  \text{\unboldmath$\forall$}  \,  \algeffseqoverindex{ \gamma }{ \text{\unboldmath$\mathit{I}$} }   \ottsym{.} \, \ottnt{B}  \, \mathsf{list}   \sqsubseteq   \text{\unboldmath$\forall$}  \, \beta  \ottsym{.} \, \ottnt{C}$,
   which is derived from $\Gamma  \ottsym{,}  \beta  \vdash    \text{\unboldmath$\forall$}  \,  \algeffseqoverindex{ \gamma }{ \text{\unboldmath$\mathit{I}$} }   \ottsym{.} \, \ottnt{B}  \, \mathsf{list}   \sqsubseteq  \ottnt{C}$ with \Srule{Poly}.

  \case \T{Inst}: By the IH and \Srule{Trans}.
 \end{caseanalysis}
\end{proof}

\begin{lemma}{term-inv-op}
 If $\mathit{ty} \, \ottsym{(}  \mathsf{op}  \ottsym{)} \,  =  \,   \text{\unboldmath$\forall$}  \,  \algeffseqoverindex{ \alpha }{ \text{\unboldmath$\mathit{I}$} }   \ottsym{.} \,  \ottnt{A}  \hookrightarrow  \ottnt{B} $ and $\Gamma  \vdash   \textup{\texttt{\#}\relax}  \mathsf{op}   \ottsym{(}   \ottnt{v}   \ottsym{)}   \ottsym{:}  \ottnt{C}$, then
 \begin{itemize}
  \item $\Gamma  \ottsym{,}   \algeffseqoverindex{ \beta }{ \text{\unboldmath$\mathit{J}$} }   \vdash   \algeffseqoverindex{ \ottnt{D} }{ \text{\unboldmath$\mathit{I}$} } $,
  \item $\Gamma  \ottsym{,}   \algeffseqoverindex{ \beta }{ \text{\unboldmath$\mathit{J}$} }   \vdash  \ottnt{v}  \ottsym{:}   \ottnt{A}    [   \algeffseqoverindex{ \ottnt{D} }{ \text{\unboldmath$\mathit{I}$} }   \ottsym{/}   \algeffseqoverindex{ \alpha }{ \text{\unboldmath$\mathit{I}$} }   ]  $, and
  \item $\Gamma  \vdash    \text{\unboldmath$\forall$}  \,  \algeffseqoverindex{ \beta }{ \text{\unboldmath$\mathit{J}$} }   \ottsym{.} \, \ottnt{B}    [   \algeffseqoverindex{ \ottnt{D} }{ \text{\unboldmath$\mathit{I}$} }   \ottsym{/}   \algeffseqoverindex{ \alpha }{ \text{\unboldmath$\mathit{I}$} }   ]    \sqsubseteq  \ottnt{C}$
 \end{itemize}
 for some $ \algeffseqoverindex{ \beta }{ \text{\unboldmath$\mathit{J}$} } $ and $ \algeffseqoverindex{ \ottnt{D} }{ \text{\unboldmath$\mathit{I}$} } $.
\end{lemma}
\begin{proof}
 By induction on the typing derivation for $ \textup{\texttt{\#}\relax}  \mathsf{op}   \ottsym{(}   \ottnt{v}   \ottsym{)} $.
 There are only three typing rules that can be applied to $ \textup{\texttt{\#}\relax}  \mathsf{op}   \ottsym{(}   \ottnt{v}   \ottsym{)} $.
 \begin{caseanalysis}
  \case \T{Op}:
   We have $\ottnt{C} \,  =  \,  \ottnt{B}    [   \algeffseqoverindex{ \ottnt{D} }{ \text{\unboldmath$\mathit{I}$} }   \ottsym{/}   \algeffseqoverindex{ \alpha }{ \text{\unboldmath$\mathit{I}$} }   ]  $ and
   $\Gamma  \vdash   \algeffseqoverindex{ \ottnt{D} }{ \text{\unboldmath$\mathit{I}$} } $ and
   $\Gamma  \vdash  \ottnt{v}  \ottsym{:}   \ottnt{A}    [   \algeffseqoverindex{ \ottnt{D} }{ \text{\unboldmath$\mathit{I}$} }   \ottsym{/}   \algeffseqoverindex{ \alpha }{ \text{\unboldmath$\mathit{I}$} }   ]  $
   for some $ \algeffseqoverindex{ \ottnt{D} }{ \text{\unboldmath$\mathit{I}$} } $.
   We have the conclusion by letting $ \algeffseqoverindex{ \beta }{ \text{\unboldmath$\mathit{J}$} } $ be the empty sequence;
   note that $\Gamma  \vdash   \ottnt{B}    [   \algeffseqoverindex{ \ottnt{D} }{ \text{\unboldmath$\mathit{I}$} }   \ottsym{/}   \algeffseqoverindex{ \alpha }{ \text{\unboldmath$\mathit{I}$} }   ]    \sqsubseteq   \ottnt{B}    [   \algeffseqoverindex{ \ottnt{D} }{ \text{\unboldmath$\mathit{I}$} }   \ottsym{/}   \algeffseqoverindex{ \alpha }{ \text{\unboldmath$\mathit{I}$} }   ]  $
   by \Srule{Refl}.

  \case \T{Gen}:
   We are given $\ottnt{C} \,  =  \,  \text{\unboldmath$\forall$}  \, \beta  \ottsym{.} \, \ottnt{C_{{\mathrm{0}}}}$ and, by inversion,
   $\Gamma  \ottsym{,}  \beta  \vdash   \textup{\texttt{\#}\relax}  \mathsf{op}   \ottsym{(}   \ottnt{v}   \ottsym{)}   \ottsym{:}  \ottnt{C_{{\mathrm{0}}}}$ for some $\beta$ and $\ottnt{C_{{\mathrm{0}}}}$.
   By the IH, there exist some $ \algeffseqoverindex{ \beta_{{\mathrm{0}}} }{ \text{\unboldmath$\mathit{J_{{\mathrm{0}}}}$} } $ and $ \algeffseqoverindex{ \ottnt{D} }{ \text{\unboldmath$\mathit{I}$} } $
   such that
   \begin{itemize}
    \item $\Gamma  \ottsym{,}  \beta  \ottsym{,}   \algeffseqoverindex{ \beta_{{\mathrm{0}}} }{ \text{\unboldmath$\mathit{J_{{\mathrm{0}}}}$} }   \vdash   \algeffseqoverindex{ \ottnt{D} }{ \text{\unboldmath$\mathit{I}$} } $,
    \item $\Gamma  \ottsym{,}  \beta  \ottsym{,}   \algeffseqoverindex{ \beta_{{\mathrm{0}}} }{ \text{\unboldmath$\mathit{J_{{\mathrm{0}}}}$} }   \vdash  \ottnt{v}  \ottsym{:}   \ottnt{A}    [   \algeffseqoverindex{ \ottnt{D} }{ \text{\unboldmath$\mathit{I}$} }   \ottsym{/}   \algeffseqoverindex{ \alpha }{ \text{\unboldmath$\mathit{I}$} }   ]  $ and
    \item $\Gamma  \ottsym{,}  \beta  \vdash    \text{\unboldmath$\forall$}  \,  \algeffseqoverindex{ \beta_{{\mathrm{0}}} }{ \text{\unboldmath$\mathit{J_{{\mathrm{0}}}}$} }   \ottsym{.} \, \ottnt{B}    [   \algeffseqoverindex{ \ottnt{D} }{ \text{\unboldmath$\mathit{I}$} }   \ottsym{/}   \algeffseqoverindex{ \alpha }{ \text{\unboldmath$\mathit{I}$} }   ]    \sqsubseteq  \ottnt{C_{{\mathrm{0}}}}$.
   \end{itemize}
   We show the conclusion by letting $ \algeffseqoverindex{ \beta }{ \text{\unboldmath$\mathit{J}$} }  \,  =  \, \beta  \ottsym{,}   \algeffseqoverindex{ \beta_{{\mathrm{0}}} }{ \text{\unboldmath$\mathit{J_{{\mathrm{0}}}}$} } $.
   It suffices to show $\Gamma  \vdash    \text{\unboldmath$\forall$}  \, \beta  \ottsym{.} \,  \text{\unboldmath$\forall$}  \,  \algeffseqoverindex{ \beta_{{\mathrm{0}}} }{ \text{\unboldmath$\mathit{J_{{\mathrm{0}}}}$} }   \ottsym{.} \, \ottnt{B}    [   \algeffseqoverindex{ \ottnt{D} }{ \text{\unboldmath$\mathit{I}$} }   \ottsym{/}   \algeffseqoverindex{ \alpha }{ \text{\unboldmath$\mathit{I}$} }   ]    \sqsubseteq   \text{\unboldmath$\forall$}  \, \beta  \ottsym{.} \, \ottnt{C_{{\mathrm{0}}}}$,
   which is proven from $\Gamma  \ottsym{,}  \beta  \vdash    \text{\unboldmath$\forall$}  \,  \algeffseqoverindex{ \beta_{{\mathrm{0}}} }{ \text{\unboldmath$\mathit{J_{{\mathrm{0}}}}$} }   \ottsym{.} \, \ottnt{B}    [   \algeffseqoverindex{ \ottnt{D} }{ \text{\unboldmath$\mathit{I}$} }   \ottsym{/}   \algeffseqoverindex{ \alpha }{ \text{\unboldmath$\mathit{I}$} }   ]    \sqsubseteq  \ottnt{C_{{\mathrm{0}}}}$ with \Srule{Poly}.

  \case \T{Inst}: By the IH and \Srule{Trans}.
 \end{caseanalysis}
\end{proof}

\begin{lemma}{ectx-typing}
 If $\Gamma  \ottsym{,}   \algeffseqoverindex{ \alpha }{ \text{\unboldmath$\mathit{I}$} }   \vdash   \ottnt{E}  [  \ottnt{M}  ]   \ottsym{:}  \ottnt{A}$, then
 \begin{itemize}
  \item $\Gamma  \ottsym{,}   \algeffseqoverindex{ \alpha }{ \text{\unboldmath$\mathit{I}$} }   \ottsym{,}   \algeffseqoverindex{ \beta }{ \text{\unboldmath$\mathit{J}$} }   \vdash  \ottnt{M}  \ottsym{:}  \ottnt{B}$ and
  \item $\Gamma  \ottsym{,}  \mathit{y} \,  \mathord{:}  \,  \text{\unboldmath$\forall$}  \,  \algeffseqoverindex{ \alpha }{ \text{\unboldmath$\mathit{I}$} }   \ottsym{.} \,  \text{\unboldmath$\forall$}  \,  \algeffseqoverindex{ \beta }{ \text{\unboldmath$\mathit{J}$} }   \ottsym{.} \, \ottnt{B}  \ottsym{,}   \algeffseqoverindex{ \alpha }{ \text{\unboldmath$\mathit{I}$} }   \vdash   \ottnt{E}  [  \mathit{y}  ]   \ottsym{:}  \ottnt{A}$ for any $\mathit{y} \,  \not\in  \,  \mathit{dom}  (  \Gamma  ) $
 \end{itemize}
 for some $ \algeffseqoverindex{ \beta }{ \text{\unboldmath$\mathit{J}$} } $ and $\ottnt{B}$.
\end{lemma}
\begin{proof}
 By induction on the typing derivation of $\Gamma  \ottsym{,}   \algeffseqoverindex{ \alpha }{ \text{\unboldmath$\mathit{I}$} }   \vdash   \ottnt{E}  [  \ottnt{M}  ]   \ottsym{:}  \ottnt{A}$.

 Suppose that $\ottnt{E} \,  =  \,  [] $.  Since $\Gamma  \ottsym{,}   \algeffseqoverindex{ \alpha }{ \text{\unboldmath$\mathit{I}$} }   \vdash   \ottnt{E}  [  \ottnt{M}  ]   \ottsym{:}  \ottnt{A}$, we have
 $\Gamma  \ottsym{,}   \algeffseqoverindex{ \alpha }{ \text{\unboldmath$\mathit{I}$} }   \vdash  \ottnt{M}  \ottsym{:}  \ottnt{A}$.
 We let $ \algeffseqoverindex{ \beta }{ \text{\unboldmath$\mathit{J}$} } $ be the empty sequence and $\ottnt{B} \,  =  \, \ottnt{A}$.  It is then trivial
 that $\Gamma  \ottsym{,}  \mathit{y} \,  \mathord{:}  \,  \text{\unboldmath$\forall$}  \,  \algeffseqoverindex{ \alpha }{ \text{\unboldmath$\mathit{I}$} }   \ottsym{.} \, \ottnt{B}  \ottsym{,}   \algeffseqoverindex{ \alpha }{ \text{\unboldmath$\mathit{I}$} }   \vdash   \ottnt{E}  [  \mathit{y}  ]   \ottsym{:}  \ottnt{A}$ by \T{Inst}.
 Note that $\vdash  \Gamma$ and $\Gamma  \vdash   \text{\unboldmath$\forall$}  \,  \algeffseqover{ \alpha }   \ottsym{.} \, \ottnt{B}$ by \reflem{type-wf}.

 In what follows, we suppose that $\ottnt{E} \,  \not=  \,  [] $.  We proceed by case analysis on
 the typing rule applied last to derive $\Gamma  \ottsym{,}   \algeffseqoverindex{ \alpha }{ \text{\unboldmath$\mathit{I}$} }   \vdash   \ottnt{E}  [  \ottnt{M}  ]   \ottsym{:}  \ottnt{A}$.
 \begin{caseanalysis}
  \case \T{Var}, \T{Const}, \T{Abs}, \T{Nil}, and \T{Fix}:
   Contradictory with the assumption that $\ottnt{E} \,  \not=  \,  [] $.

  \case \T{App}:
   By case analysis on $\ottnt{E}$.
   \begin{caseanalysis}
    \case $\ottnt{E} \,  =  \, \ottnt{E'} \, \ottnt{M_{{\mathrm{2}}}}$:
     By inversion of the typing derivation, we have
     $\Gamma  \ottsym{,}   \algeffseqoverindex{ \alpha }{ \text{\unboldmath$\mathit{I}$} }   \vdash   \ottnt{E'}  [  \ottnt{M}  ]   \ottsym{:}  \ottnt{C}  \rightarrow  \ottnt{A}$ and
     $\Gamma  \ottsym{,}   \algeffseqoverindex{ \alpha }{ \text{\unboldmath$\mathit{I}$} }   \vdash  \ottnt{M_{{\mathrm{2}}}}  \ottsym{:}  \ottnt{C}$
     for some $\ottnt{C}$.
     By the IH,
     (1) $\Gamma  \ottsym{,}   \algeffseqoverindex{ \alpha }{ \text{\unboldmath$\mathit{I}$} }   \ottsym{,}   \algeffseqoverindex{ \beta }{ \text{\unboldmath$\mathit{J}$} }   \vdash  \ottnt{M}  \ottsym{:}  \ottnt{B}$ for some $ \algeffseqoverindex{ \beta }{ \text{\unboldmath$\mathit{J}$} } $ and $\ottnt{B}$ and
     (2) for any $\mathit{y} \,  \not\in  \,  \mathit{dom}  (  \Gamma  ) $,
     $\Gamma  \ottsym{,}  \mathit{y} \,  \mathord{:}  \,  \text{\unboldmath$\forall$}  \,  \algeffseqoverindex{ \alpha }{ \text{\unboldmath$\mathit{I}$} }   \ottsym{.} \,  \text{\unboldmath$\forall$}  \,  \algeffseqoverindex{ \beta }{ \text{\unboldmath$\mathit{J}$} }   \ottsym{.} \, \ottnt{B}  \ottsym{,}   \algeffseqoverindex{ \alpha }{ \text{\unboldmath$\mathit{I}$} }   \vdash   \ottnt{E'}  [  \mathit{y}  ]   \ottsym{:}  \ottnt{C}  \rightarrow  \ottnt{A}$.
     By \reflem{weakening} (\ref{lem:weakening:term}) and \T{App},
     $\Gamma  \ottsym{,}  \mathit{y} \,  \mathord{:}  \,  \text{\unboldmath$\forall$}  \,  \algeffseqoverindex{ \alpha }{ \text{\unboldmath$\mathit{I}$} }   \ottsym{.} \,  \text{\unboldmath$\forall$}  \,  \algeffseqoverindex{ \beta }{ \text{\unboldmath$\mathit{J}$} }   \ottsym{.} \, \ottnt{B}  \ottsym{,}   \algeffseqoverindex{ \alpha }{ \text{\unboldmath$\mathit{I}$} }   \vdash   \ottnt{E'}  [  \mathit{y}  ]  \, \ottnt{M_{{\mathrm{2}}}}  \ottsym{:}  \ottnt{A}$, i.e.,
     $\Gamma  \ottsym{,}  \mathit{y} \,  \mathord{:}  \,  \text{\unboldmath$\forall$}  \,  \algeffseqoverindex{ \alpha }{ \text{\unboldmath$\mathit{I}$} }   \ottsym{.} \,  \text{\unboldmath$\forall$}  \,  \algeffseqoverindex{ \beta }{ \text{\unboldmath$\mathit{J}$} }   \ottsym{.} \, \ottnt{B}  \ottsym{,}   \algeffseqoverindex{ \alpha }{ \text{\unboldmath$\mathit{I}$} }   \vdash   \ottnt{E}  [  \mathit{y}  ]   \ottsym{:}  \ottnt{A}$.

    \case $\ottnt{E} \,  =  \, \ottnt{v_{{\mathrm{1}}}} \, \ottnt{E'}$:  Similarly to the above case.
   \end{caseanalysis}

  \case \T{Gen}:
   We have $\Gamma  \ottsym{,}   \algeffseqoverindex{ \alpha }{ \text{\unboldmath$\mathit{I}$} }   \vdash   \ottnt{E}  [  \ottnt{M}  ]   \ottsym{:}   \text{\unboldmath$\forall$}  \, \gamma  \ottsym{.} \, \ottnt{A'}$ and, by inversion,
   $\Gamma  \ottsym{,}   \algeffseqoverindex{ \alpha }{ \text{\unboldmath$\mathit{I}$} }   \ottsym{,}  \gamma  \vdash   \ottnt{E}  [  \ottnt{M}  ]   \ottsym{:}  \ottnt{A'}$ for some $\gamma$ and $\ottnt{A'}$
   (note $\ottnt{A} \,  =  \,  \text{\unboldmath$\forall$}  \, \gamma  \ottsym{.} \, \ottnt{A'}$).
   By the IH, (1) $\Gamma  \ottsym{,}   \algeffseqoverindex{ \alpha }{ \text{\unboldmath$\mathit{I}$} }   \ottsym{,}  \gamma  \ottsym{,}   \algeffseqoverindex{ \beta }{ \text{\unboldmath$\mathit{J}$} }   \vdash  \ottnt{M}  \ottsym{:}  \ottnt{B}$ for some $ \algeffseqoverindex{ \beta }{ \text{\unboldmath$\mathit{J}$} } $ and $\ottnt{B}$ and
   (2) for any $\mathit{y} \,  \not\in  \,  \mathit{dom}  (  \Gamma  ) $,
   $\Gamma  \ottsym{,}  \mathit{y} \,  \mathord{:}  \,  \text{\unboldmath$\forall$}  \,  \algeffseqoverindex{ \alpha }{ \text{\unboldmath$\mathit{I}$} }   \ottsym{.} \,  \text{\unboldmath$\forall$}  \, \gamma  \ottsym{.} \,  \text{\unboldmath$\forall$}  \,  \algeffseqoverindex{ \beta }{ \text{\unboldmath$\mathit{J}$} }   \ottsym{.} \, \ottnt{B}  \ottsym{,}   \algeffseqoverindex{ \alpha }{ \text{\unboldmath$\mathit{I}$} }   \ottsym{,}  \gamma  \vdash   \ottnt{E}  [  \mathit{y}  ]   \ottsym{:}  \ottnt{A'}$.

   By \T{Gen}, $\Gamma  \ottsym{,}  \mathit{y} \,  \mathord{:}  \,  \text{\unboldmath$\forall$}  \,  \algeffseqoverindex{ \alpha }{ \text{\unboldmath$\mathit{I}$} }   \ottsym{.} \,  \text{\unboldmath$\forall$}  \, \gamma  \ottsym{.} \,  \text{\unboldmath$\forall$}  \,  \algeffseqoverindex{ \beta }{ \text{\unboldmath$\mathit{J}$} }   \ottsym{.} \, \ottnt{B}  \ottsym{,}   \algeffseqoverindex{ \alpha }{ \text{\unboldmath$\mathit{I}$} }   \vdash   \ottnt{E}  [  \mathit{y}  ]   \ottsym{:}   \text{\unboldmath$\forall$}  \, \gamma  \ottsym{.} \, \ottnt{A'}$.
   Since $\ottnt{A} \,  =  \,  \text{\unboldmath$\forall$}  \, \gamma  \ottsym{.} \, \ottnt{A'}$, we finish.

  \otherwise: By the IH(s) and the corresponding typing rule, as the case for \T{App}.
 \end{caseanalysis}
\end{proof}

\begin{lemma}{var-subtype}
 Suppose that $\Gamma_{{\mathrm{1}}}  \vdash  \ottnt{A}  \sqsubseteq  \ottnt{B}$ and $\Gamma_{{\mathrm{1}}}  \vdash  \ottnt{A}$.
 \begin{enumerate}
  \item If $\Gamma_{{\mathrm{1}}}  \ottsym{,}  \mathit{x} \,  \mathord{:}  \, \ottnt{B}  \ottsym{,}  \Gamma_{{\mathrm{2}}}  \vdash  \ottnt{M}  \ottsym{:}  \ottnt{C}$, then $\Gamma_{{\mathrm{1}}}  \ottsym{,}  \mathit{x} \,  \mathord{:}  \, \ottnt{A}  \ottsym{,}  \Gamma_{{\mathrm{2}}}  \vdash  \ottnt{M}  \ottsym{:}  \ottnt{C}$.
  \item If $\Gamma_{{\mathrm{1}}}  \ottsym{,}  \mathit{x} \,  \mathord{:}  \, \ottnt{B}  \ottsym{,}  \Gamma_{{\mathrm{2}}}  \vdash  \ottnt{H}  \ottsym{:}  \ottnt{C}  \Rightarrow  \ottnt{D}$, then $\Gamma_{{\mathrm{1}}}  \ottsym{,}  \mathit{x} \,  \mathord{:}  \, \ottnt{A}  \ottsym{,}  \Gamma_{{\mathrm{2}}}  \vdash  \ottnt{H}  \ottsym{:}  \ottnt{C}  \Rightarrow  \ottnt{D}$.
 \end{enumerate}
\end{lemma}
\begin{proof}
 By mutual induction on the typing derivations.
\end{proof}

\begin{lemma}{ectx-op-typing}
 If $\mathit{ty} \, \ottsym{(}  \mathsf{op}  \ottsym{)} \,  =  \,   \text{\unboldmath$\forall$}  \,  \algeffseqoverindex{ \alpha }{ \text{\unboldmath$\mathit{I}$} }   \ottsym{.} \,  \ottnt{A}  \hookrightarrow  \ottnt{B} $ and $\Gamma  \vdash   \ottnt{E}  [   \textup{\texttt{\#}\relax}  \mathsf{op}   \ottsym{(}   \ottnt{v}   \ottsym{)}   ]   \ottsym{:}  \ottnt{C}$, then
 \begin{itemize}
  \item $\Gamma  \ottsym{,}   \algeffseqoverindex{ \beta }{ \text{\unboldmath$\mathit{J}$} }   \vdash   \algeffseqoverindex{ \ottnt{D} }{ \text{\unboldmath$\mathit{I}$} } $,
  \item $\Gamma  \ottsym{,}   \algeffseqoverindex{ \beta }{ \text{\unboldmath$\mathit{J}$} }   \vdash  \ottnt{v}  \ottsym{:}   \ottnt{A}    [   \algeffseqoverindex{ \ottnt{D} }{ \text{\unboldmath$\mathit{I}$} }   \ottsym{/}   \algeffseqoverindex{ \alpha }{ \text{\unboldmath$\mathit{I}$} }   ]  $, and
  \item for any $\mathit{y} \,  \not\in  \,  \mathit{dom}  (  \Gamma  ) $,
        $\Gamma  \ottsym{,}  \mathit{y} \,  \mathord{:}  \,  \text{\unboldmath$\forall$}  \,  \algeffseqoverindex{ \beta }{ \text{\unboldmath$\mathit{J}$} }   \ottsym{.} \, \ottnt{B} \,  [   \algeffseqoverindex{ \ottnt{D} }{ \text{\unboldmath$\mathit{I}$} }   \ottsym{/}   \algeffseqoverindex{ \alpha }{ \text{\unboldmath$\mathit{I}$} }   ]   \vdash   \ottnt{E}  [  \mathit{y}  ]   \ottsym{:}  \ottnt{C}$
 \end{itemize}
 for some $ \algeffseqoverindex{ \beta }{ \text{\unboldmath$\mathit{J}$} } $ and $ \algeffseqoverindex{ \ottnt{D} }{ \text{\unboldmath$\mathit{I}$} } $.
\end{lemma}
\begin{proof}
 By \reflem{ectx-typing},
 \begin{itemize}
  \item $\Gamma  \ottsym{,}   \algeffseqoverindex{ \beta_{{\mathrm{1}}} }{ \text{\unboldmath$\mathit{J_{{\mathrm{1}}}}$} }   \vdash   \textup{\texttt{\#}\relax}  \mathsf{op}   \ottsym{(}   \ottnt{v}   \ottsym{)}   \ottsym{:}  \ottnt{C'}$ and
  \item $\Gamma  \ottsym{,}  \mathit{y} \,  \mathord{:}  \,  \text{\unboldmath$\forall$}  \,  \algeffseqoverindex{ \beta_{{\mathrm{1}}} }{ \text{\unboldmath$\mathit{J_{{\mathrm{1}}}}$} }   \ottsym{.} \, \ottnt{C'}  \vdash   \ottnt{E}  [  \mathit{y}  ]   \ottsym{:}  \ottnt{C}$ for any $\mathit{y} \,  \not\in  \,  \mathit{dom}  (  \Gamma  ) $
 \end{itemize}
 for some $ \algeffseqoverindex{ \beta_{{\mathrm{1}}} }{ \text{\unboldmath$\mathit{J_{{\mathrm{1}}}}$} } $ and $\ottnt{C'}$.
 By \reflem{term-inv-op},
 \begin{itemize}
  \item $\Gamma  \ottsym{,}   \algeffseqoverindex{ \beta_{{\mathrm{1}}} }{ \text{\unboldmath$\mathit{J_{{\mathrm{1}}}}$} }   \ottsym{,}   \algeffseqoverindex{ \beta_{{\mathrm{2}}} }{ \text{\unboldmath$\mathit{J_{{\mathrm{2}}}}$} }   \vdash   \algeffseqoverindex{ \ottnt{D} }{ \text{\unboldmath$\mathit{I}$} } $,
  \item $\Gamma  \ottsym{,}   \algeffseqoverindex{ \beta_{{\mathrm{1}}} }{ \text{\unboldmath$\mathit{J_{{\mathrm{1}}}}$} }   \ottsym{,}   \algeffseqoverindex{ \beta_{{\mathrm{2}}} }{ \text{\unboldmath$\mathit{J_{{\mathrm{2}}}}$} }   \vdash  \ottnt{v}  \ottsym{:}   \ottnt{A}    [   \algeffseqoverindex{ \ottnt{D} }{ \text{\unboldmath$\mathit{I}$} }   \ottsym{/}   \algeffseqoverindex{ \alpha }{ \text{\unboldmath$\mathit{I}$} }   ]  $, and
  \item $\Gamma  \ottsym{,}   \algeffseqoverindex{ \beta_{{\mathrm{1}}} }{ \text{\unboldmath$\mathit{J_{{\mathrm{1}}}}$} }   \vdash    \text{\unboldmath$\forall$}  \,  \algeffseqoverindex{ \beta_{{\mathrm{2}}} }{ \text{\unboldmath$\mathit{J_{{\mathrm{2}}}}$} }   \ottsym{.} \, \ottnt{B}    [   \algeffseqoverindex{ \ottnt{D} }{ \text{\unboldmath$\mathit{I}$} }   \ottsym{/}   \algeffseqoverindex{ \alpha }{ \text{\unboldmath$\mathit{I}$} }   ]    \sqsubseteq  \ottnt{C'}$
 \end{itemize}
 for some $ \algeffseqoverindex{ \beta_{{\mathrm{2}}} }{ \text{\unboldmath$\mathit{J_{{\mathrm{2}}}}$} } $ and $ \algeffseqoverindex{ \ottnt{D} }{ \text{\unboldmath$\mathit{I}$} } $.

 We show the conclusion by letting $ \algeffseqoverindex{ \beta }{ \text{\unboldmath$\mathit{J}$} }  \,  =  \,  \algeffseqoverindex{ \beta_{{\mathrm{1}}} }{ \text{\unboldmath$\mathit{J_{{\mathrm{1}}}}$} }   \ottsym{,}   \algeffseqoverindex{ \beta_{{\mathrm{2}}} }{ \text{\unboldmath$\mathit{J_{{\mathrm{2}}}}$} } $.
 It suffices to show that, for any $\mathit{y} \,  \not\in  \,  \mathit{dom}  (  \Gamma  ) $,
 \[
  \Gamma  \ottsym{,}  \mathit{y} \,  \mathord{:}  \,  \text{\unboldmath$\forall$}  \,  \algeffseqoverindex{ \beta_{{\mathrm{1}}} }{ \text{\unboldmath$\mathit{J_{{\mathrm{1}}}}$} }   \ottsym{.} \,  \text{\unboldmath$\forall$}  \,  \algeffseqoverindex{ \beta_{{\mathrm{2}}} }{ \text{\unboldmath$\mathit{J_{{\mathrm{2}}}}$} }   \ottsym{.} \, \ottnt{B} \,  [   \algeffseqoverindex{ \ottnt{D} }{ \text{\unboldmath$\mathit{I}$} }   \ottsym{/}   \algeffseqoverindex{ \alpha }{ \text{\unboldmath$\mathit{I}$} }   ]   \vdash   \ottnt{E}  [  \mathit{y}  ]   \ottsym{:}  \ottnt{C}.
 \]
 Since $\Gamma  \ottsym{,}   \algeffseqoverindex{ \beta_{{\mathrm{1}}} }{ \text{\unboldmath$\mathit{J_{{\mathrm{1}}}}$} }   \vdash    \text{\unboldmath$\forall$}  \,  \algeffseqoverindex{ \beta_{{\mathrm{2}}} }{ \text{\unboldmath$\mathit{J_{{\mathrm{2}}}}$} }   \ottsym{.} \, \ottnt{B}    [   \algeffseqoverindex{ \ottnt{D} }{ \text{\unboldmath$\mathit{I}$} }   \ottsym{/}   \algeffseqoverindex{ \alpha }{ \text{\unboldmath$\mathit{I}$} }   ]    \sqsubseteq  \ottnt{C'}$,
 we have
 \[
  \Gamma  \vdash    \text{\unboldmath$\forall$}  \,  \algeffseqoverindex{ \beta_{{\mathrm{1}}} }{ \text{\unboldmath$\mathit{J_{{\mathrm{1}}}}$} }   \ottsym{.} \,  \text{\unboldmath$\forall$}  \,  \algeffseqoverindex{ \beta_{{\mathrm{2}}} }{ \text{\unboldmath$\mathit{J_{{\mathrm{2}}}}$} }   \ottsym{.} \, \ottnt{B}    [   \algeffseqoverindex{ \ottnt{D} }{ \text{\unboldmath$\mathit{I}$} }   \ottsym{/}   \algeffseqoverindex{ \alpha }{ \text{\unboldmath$\mathit{I}$} }   ]    \sqsubseteq   \text{\unboldmath$\forall$}  \,  \algeffseqoverindex{ \beta_{{\mathrm{1}}} }{ \text{\unboldmath$\mathit{J_{{\mathrm{1}}}}$} }   \ottsym{.} \, \ottnt{C'}
 \]
 by \Srule{Poly}.
 Since $\Gamma  \ottsym{,}  \mathit{y} \,  \mathord{:}  \,  \text{\unboldmath$\forall$}  \,  \algeffseqoverindex{ \beta_{{\mathrm{1}}} }{ \text{\unboldmath$\mathit{J_{{\mathrm{1}}}}$} }   \ottsym{.} \, \ottnt{C'}  \vdash   \ottnt{E}  [  \mathit{y}  ]   \ottsym{:}  \ottnt{C}$,
 we have
 \[
  \Gamma  \ottsym{,}  \mathit{y} \,  \mathord{:}  \,  \text{\unboldmath$\forall$}  \,  \algeffseqoverindex{ \beta_{{\mathrm{1}}} }{ \text{\unboldmath$\mathit{J_{{\mathrm{1}}}}$} }   \ottsym{.} \,  \text{\unboldmath$\forall$}  \,  \algeffseqoverindex{ \beta_{{\mathrm{2}}} }{ \text{\unboldmath$\mathit{J_{{\mathrm{2}}}}$} }   \ottsym{.} \, \ottnt{B} \,  [   \algeffseqoverindex{ \ottnt{D} }{ \text{\unboldmath$\mathit{I}$} }   \ottsym{/}   \algeffseqoverindex{ \alpha }{ \text{\unboldmath$\mathit{I}$} }   ]   \vdash   \ottnt{E}  [  \mathit{y}  ]   \ottsym{:}  \ottnt{C}.
 \]
 by \reflem{var-subtype}.
\end{proof}

\begin{lemmap}{Type containment inversion: product types}{subtyping-inv-prod}
 If $\Gamma  \vdash    \text{\unboldmath$\forall$}  \,  \algeffseqoverindex{ \alpha_{{\mathrm{1}}} }{ \text{\unboldmath$\mathit{I_{{\mathrm{1}}}}$} }   \ottsym{.} \, \ottnt{A_{{\mathrm{1}}}}  \times  \ottnt{A_{{\mathrm{2}}}}   \sqsubseteq    \text{\unboldmath$\forall$}  \,  \algeffseqoverindex{ \alpha_{{\mathrm{2}}} }{ \text{\unboldmath$\mathit{I_{{\mathrm{2}}}}$} }   \ottsym{.} \, \ottnt{B_{{\mathrm{1}}}}  \times  \ottnt{B_{{\mathrm{2}}}} $,
 then there exist $ \algeffseqoverindex{ \alpha_{{\mathrm{11}}} }{ \text{\unboldmath$\mathit{I_{{\mathrm{11}}}}$} } $, $ \algeffseqoverindex{ \alpha_{{\mathrm{12}}} }{ \text{\unboldmath$\mathit{I_{{\mathrm{12}}}}$} } $, $ \algeffseqoverindex{ \beta }{ \text{\unboldmath$\mathit{J}$} } $, and $ \algeffseqoverindex{ \ottnt{C} }{ \text{\unboldmath$\mathit{I_{{\mathrm{11}}}}$} } $
 such that
 \begin{itemize}
  \item $\ottsym{\{}   \algeffseqoverindex{ \alpha_{{\mathrm{1}}} }{ \text{\unboldmath$\mathit{I_{{\mathrm{1}}}}$} }   \ottsym{\}} \,  =  \, \ottsym{\{}   \algeffseqoverindex{ \alpha_{{\mathrm{11}}} }{ \text{\unboldmath$\mathit{I_{{\mathrm{11}}}}$} }   \ottsym{\}} \,  \mathbin{\uplus}  \, \ottsym{\{}   \algeffseqoverindex{ \alpha_{{\mathrm{12}}} }{ \text{\unboldmath$\mathit{I_{{\mathrm{12}}}}$} }   \ottsym{\}}$,
  \item $\Gamma  \ottsym{,}   \algeffseqoverindex{ \alpha_{{\mathrm{2}}} }{ \text{\unboldmath$\mathit{I_{{\mathrm{2}}}}$} }   \ottsym{,}   \algeffseqoverindex{ \beta }{ \text{\unboldmath$\mathit{J}$} }   \vdash   \algeffseqoverindex{ \ottnt{C} }{ \text{\unboldmath$\mathit{I_{{\mathrm{11}}}}$} } $,
  \item $\Gamma  \ottsym{,}   \algeffseqoverindex{ \alpha_{{\mathrm{2}}} }{ \text{\unboldmath$\mathit{I_{{\mathrm{2}}}}$} }   \vdash    \text{\unboldmath$\forall$}  \,  \algeffseqoverindex{ \alpha_{{\mathrm{12}}} }{ \text{\unboldmath$\mathit{I_{{\mathrm{12}}}}$} }   \ottsym{.} \,  \text{\unboldmath$\forall$}  \,  \algeffseqoverindex{ \beta }{ \text{\unboldmath$\mathit{J}$} }   \ottsym{.} \, \ottnt{A_{{\mathrm{1}}}}    [   \algeffseqoverindex{ \ottnt{C} }{ \text{\unboldmath$\mathit{I_{{\mathrm{11}}}}$} }   \ottsym{/}   \algeffseqoverindex{ \alpha_{{\mathrm{11}}} }{ \text{\unboldmath$\mathit{I_{{\mathrm{11}}}}$} }   ]    \sqsubseteq  \ottnt{B_{{\mathrm{1}}}}$,
  \item $\Gamma  \ottsym{,}   \algeffseqoverindex{ \alpha_{{\mathrm{2}}} }{ \text{\unboldmath$\mathit{I_{{\mathrm{2}}}}$} }   \vdash    \text{\unboldmath$\forall$}  \,  \algeffseqoverindex{ \alpha_{{\mathrm{12}}} }{ \text{\unboldmath$\mathit{I_{{\mathrm{12}}}}$} }   \ottsym{.} \,  \text{\unboldmath$\forall$}  \,  \algeffseqoverindex{ \beta }{ \text{\unboldmath$\mathit{J}$} }   \ottsym{.} \, \ottnt{A_{{\mathrm{2}}}}    [   \algeffseqoverindex{ \ottnt{C} }{ \text{\unboldmath$\mathit{I_{{\mathrm{11}}}}$} }   \ottsym{/}   \algeffseqoverindex{ \alpha_{{\mathrm{11}}} }{ \text{\unboldmath$\mathit{I_{{\mathrm{11}}}}$} }   ]    \sqsubseteq  \ottnt{B_{{\mathrm{2}}}}$, and
  \item type variables in $\ottsym{\{}   \algeffseqoverindex{ \beta }{ \text{\unboldmath$\mathit{J}$} }   \ottsym{\}}$ do not appear free in $\ottnt{A_{{\mathrm{1}}}}$ and $\ottnt{A_{{\mathrm{2}}}}$.
 \end{itemize}
\end{lemmap}
\begin{proof}
 \iffullproof
\else 
 By induction on the type containment derivation.
 The proof is similar to that of \reflem{subtyping-inv-fun}.
 \fi
\end{proof}

\begin{lemmap}{Type containment inversion: sum types}{subtyping-inv-sum}
 If $\Gamma  \vdash    \text{\unboldmath$\forall$}  \,  \algeffseqoverindex{ \alpha_{{\mathrm{1}}} }{ \text{\unboldmath$\mathit{I_{{\mathrm{1}}}}$} }   \ottsym{.} \, \ottnt{A_{{\mathrm{1}}}}  +  \ottnt{A_{{\mathrm{2}}}}   \sqsubseteq    \text{\unboldmath$\forall$}  \,  \algeffseqoverindex{ \alpha_{{\mathrm{2}}} }{ \text{\unboldmath$\mathit{I_{{\mathrm{2}}}}$} }   \ottsym{.} \, \ottnt{B_{{\mathrm{1}}}}  +  \ottnt{B_{{\mathrm{2}}}} $,
 then there exist $ \algeffseqoverindex{ \alpha_{{\mathrm{11}}} }{ \text{\unboldmath$\mathit{I_{{\mathrm{11}}}}$} } $, $ \algeffseqoverindex{ \alpha_{{\mathrm{12}}} }{ \text{\unboldmath$\mathit{I_{{\mathrm{12}}}}$} } $, $ \algeffseqoverindex{ \beta }{ \text{\unboldmath$\mathit{J}$} } $, and $ \algeffseqoverindex{ \ottnt{C} }{ \text{\unboldmath$\mathit{I_{{\mathrm{11}}}}$} } $
 such that
 \begin{itemize}
  \item $\ottsym{\{}   \algeffseqoverindex{ \alpha_{{\mathrm{1}}} }{ \text{\unboldmath$\mathit{I_{{\mathrm{1}}}}$} }   \ottsym{\}} \,  =  \, \ottsym{\{}   \algeffseqoverindex{ \alpha_{{\mathrm{11}}} }{ \text{\unboldmath$\mathit{I_{{\mathrm{11}}}}$} }   \ottsym{\}} \,  \mathbin{\uplus}  \, \ottsym{\{}   \algeffseqoverindex{ \alpha_{{\mathrm{12}}} }{ \text{\unboldmath$\mathit{I_{{\mathrm{12}}}}$} }   \ottsym{\}}$,
  \item $\Gamma  \ottsym{,}   \algeffseqoverindex{ \alpha_{{\mathrm{2}}} }{ \text{\unboldmath$\mathit{I_{{\mathrm{2}}}}$} }   \ottsym{,}   \algeffseqoverindex{ \beta }{ \text{\unboldmath$\mathit{J}$} }   \vdash   \algeffseqoverindex{ \ottnt{C} }{ \text{\unboldmath$\mathit{I_{{\mathrm{11}}}}$} } $,
  \item $\Gamma  \ottsym{,}   \algeffseqoverindex{ \alpha_{{\mathrm{2}}} }{ \text{\unboldmath$\mathit{I_{{\mathrm{2}}}}$} }   \vdash    \text{\unboldmath$\forall$}  \,  \algeffseqoverindex{ \alpha_{{\mathrm{12}}} }{ \text{\unboldmath$\mathit{I_{{\mathrm{12}}}}$} }   \ottsym{.} \,  \text{\unboldmath$\forall$}  \,  \algeffseqoverindex{ \beta }{ \text{\unboldmath$\mathit{J}$} }   \ottsym{.} \, \ottnt{A_{{\mathrm{1}}}}    [   \algeffseqoverindex{ \ottnt{C} }{ \text{\unboldmath$\mathit{I_{{\mathrm{11}}}}$} }   \ottsym{/}   \algeffseqoverindex{ \alpha_{{\mathrm{11}}} }{ \text{\unboldmath$\mathit{I_{{\mathrm{11}}}}$} }   ]    \sqsubseteq  \ottnt{B_{{\mathrm{1}}}}$,
  \item $\Gamma  \ottsym{,}   \algeffseqoverindex{ \alpha_{{\mathrm{2}}} }{ \text{\unboldmath$\mathit{I_{{\mathrm{2}}}}$} }   \vdash    \text{\unboldmath$\forall$}  \,  \algeffseqoverindex{ \alpha_{{\mathrm{12}}} }{ \text{\unboldmath$\mathit{I_{{\mathrm{12}}}}$} }   \ottsym{.} \,  \text{\unboldmath$\forall$}  \,  \algeffseqoverindex{ \beta }{ \text{\unboldmath$\mathit{J}$} }   \ottsym{.} \, \ottnt{A_{{\mathrm{2}}}}    [   \algeffseqoverindex{ \ottnt{C} }{ \text{\unboldmath$\mathit{I_{{\mathrm{11}}}}$} }   \ottsym{/}   \algeffseqoverindex{ \alpha_{{\mathrm{11}}} }{ \text{\unboldmath$\mathit{I_{{\mathrm{11}}}}$} }   ]    \sqsubseteq  \ottnt{B_{{\mathrm{2}}}}$, and
  \item type variables in $\ottsym{\{}   \algeffseqoverindex{ \beta }{ \text{\unboldmath$\mathit{J}$} }   \ottsym{\}}$ do not appear free in $\ottnt{A_{{\mathrm{1}}}}$ and $\ottnt{A_{{\mathrm{2}}}}$.
 \end{itemize}
\end{lemmap}
\begin{proof}
 By induction on the type containment derivation.
 The proof is similar to that of \reflem{subtyping-inv-fun}.
\end{proof}

\begin{lemmap}{Type containment inversion: list types}{subtyping-inv-list}
 If $\Gamma  \vdash    \text{\unboldmath$\forall$}  \,  \algeffseqoverindex{ \alpha_{{\mathrm{1}}} }{ \text{\unboldmath$\mathit{I_{{\mathrm{1}}}}$} }   \ottsym{.} \, \ottnt{A}  \, \mathsf{list}   \sqsubseteq    \text{\unboldmath$\forall$}  \,  \algeffseqoverindex{ \alpha_{{\mathrm{2}}} }{ \text{\unboldmath$\mathit{I_{{\mathrm{2}}}}$} }   \ottsym{.} \, \ottnt{B}  \, \mathsf{list} $,
 then there exist $ \algeffseqoverindex{ \alpha_{{\mathrm{11}}} }{ \text{\unboldmath$\mathit{I_{{\mathrm{11}}}}$} } $, $ \algeffseqoverindex{ \alpha_{{\mathrm{12}}} }{ \text{\unboldmath$\mathit{I_{{\mathrm{12}}}}$} } $, $ \algeffseqoverindex{ \beta }{ \text{\unboldmath$\mathit{J}$} } $, and $ \algeffseqoverindex{ \ottnt{C} }{ \text{\unboldmath$\mathit{I_{{\mathrm{11}}}}$} } $
 such that
 \begin{itemize}
  \item $\ottsym{\{}   \algeffseqoverindex{ \alpha_{{\mathrm{1}}} }{ \text{\unboldmath$\mathit{I_{{\mathrm{1}}}}$} }   \ottsym{\}} \,  =  \, \ottsym{\{}   \algeffseqoverindex{ \alpha_{{\mathrm{11}}} }{ \text{\unboldmath$\mathit{I_{{\mathrm{11}}}}$} }   \ottsym{\}} \,  \mathbin{\uplus}  \, \ottsym{\{}   \algeffseqoverindex{ \alpha_{{\mathrm{12}}} }{ \text{\unboldmath$\mathit{I_{{\mathrm{12}}}}$} }   \ottsym{\}}$,
  \item $\Gamma  \ottsym{,}   \algeffseqoverindex{ \alpha_{{\mathrm{2}}} }{ \text{\unboldmath$\mathit{I_{{\mathrm{2}}}}$} }   \ottsym{,}   \algeffseqoverindex{ \beta }{ \text{\unboldmath$\mathit{J}$} }   \vdash   \algeffseqoverindex{ \ottnt{C} }{ \text{\unboldmath$\mathit{I_{{\mathrm{11}}}}$} } $,
  \item $\Gamma  \ottsym{,}   \algeffseqoverindex{ \alpha_{{\mathrm{2}}} }{ \text{\unboldmath$\mathit{I_{{\mathrm{2}}}}$} }   \vdash    \text{\unboldmath$\forall$}  \,  \algeffseqoverindex{ \alpha_{{\mathrm{12}}} }{ \text{\unboldmath$\mathit{I_{{\mathrm{12}}}}$} }   \ottsym{.} \,  \text{\unboldmath$\forall$}  \,  \algeffseqoverindex{ \beta }{ \text{\unboldmath$\mathit{J}$} }   \ottsym{.} \, \ottnt{A}    [   \algeffseqoverindex{ \ottnt{C} }{ \text{\unboldmath$\mathit{I_{{\mathrm{11}}}}$} }   \ottsym{/}   \algeffseqoverindex{ \alpha_{{\mathrm{11}}} }{ \text{\unboldmath$\mathit{I_{{\mathrm{11}}}}$} }   ]    \sqsubseteq  \ottnt{B}$, and
  \item type variables in $\ottsym{\{}   \algeffseqoverindex{ \beta }{ \text{\unboldmath$\mathit{J}$} }   \ottsym{\}}$ do not appear free in $\ottnt{A}$.
 \end{itemize}
\end{lemmap}
\begin{proof}
 By induction on the type containment derivation.
 The proof is similar to that of \reflem{subtyping-inv-fun}.
\end{proof}

\ifrestate
lemmSubtypingForallRemove
\else
\begin{lemma}{subtyping-forall-remove}
 Suppose that $\alpha$ does not appear free in $\ottnt{A}$.
 \begin{enumerate}
  \item \label{lem:subtyping-forall-remove:neg}
        If the occurrences of $\beta$ in $\ottnt{A}$ are only negative,
        then $\Gamma_{{\mathrm{1}}}  \ottsym{,}  \alpha  \ottsym{,}  \Gamma_{{\mathrm{2}}}  \vdash   \ottnt{A}    [  \ottnt{B}  \ottsym{/}  \beta  ]    \sqsubseteq   \ottnt{A}    [   \text{\unboldmath$\forall$}  \, \alpha  \ottsym{.} \, \ottnt{B}  \ottsym{/}  \beta  ]  $.
  \item \label{lem:subtyping-forall-remove:pos}
        If the occurrences of $\beta$ in $\ottnt{A}$ are only positive,
        then $\Gamma_{{\mathrm{1}}}  \ottsym{,}  \alpha  \ottsym{,}  \Gamma_{{\mathrm{2}}}  \vdash   \ottnt{A}    [   \text{\unboldmath$\forall$}  \, \alpha  \ottsym{.} \, \ottnt{B}  \ottsym{/}  \beta  ]    \sqsubseteq   \ottnt{A}    [  \ottnt{B}  \ottsym{/}  \beta  ]  $.
 \end{enumerate}
\end{lemma}
\fi
\begin{proof}
 By structural induction on $\ottnt{A}$.
 \begin{caseanalysis}
  \case $\ottnt{A} \,  =  \, \gamma$:
   If $\gamma \,  =  \, \beta$, then we have to show that
   $\Gamma_{{\mathrm{1}}}  \ottsym{,}  \alpha  \ottsym{,}  \Gamma_{{\mathrm{2}}}  \vdash   \text{\unboldmath$\forall$}  \, \alpha  \ottsym{.} \, \ottnt{B}  \sqsubseteq  \ottnt{B}$, which is derived
   by \Srule{Refl}, \Srule{Inst}, and \Srule{Trans}.
   Note that we do not need to consider the negative case, i.e., to show
   $\Gamma_{{\mathrm{1}}}  \ottsym{,}  \alpha  \ottsym{,}  \Gamma_{{\mathrm{2}}}  \vdash  \ottnt{B}  \sqsubseteq   \text{\unboldmath$\forall$}  \, \alpha  \ottsym{.} \, \ottnt{B}$, because the occurrence $\beta$ in $\beta$
   is not negative.

  \case $\ottnt{A} \,  =  \, \iota$: By \Srule{Refl}.

  \case $\ottnt{A} \,  =  \,  \text{\unboldmath$\forall$}  \, \gamma  \ottsym{.} \, \ottnt{C}$: By the IH and \Srule{Poly} for each case.
  \case $\ottnt{A} \,  =  \, \ottnt{C}  \rightarrow  \ottnt{D}$: By the IHs and \Srule{Fun} for each case.
  \case $\ottnt{A} \,  =  \,  \ottnt{C}  \times  \ottnt{D} $: By the IH and \Srule{Prod} for each case.
  \case $\ottnt{A} \,  =  \,  \ottnt{C}  +  \ottnt{D} $: By the IH and \Srule{Sum} for each case.
  \case $\ottnt{A} \,  =  \,  \ottnt{C}  \, \mathsf{list} $: By the IH and \Srule{List} for each case.
 \end{caseanalysis}
\end{proof}

\ifrestate
\lemmSubtypingForallMove*
\else
\begin{lemma}{subtyping-forall-move}
 Suppose that $\alpha$ does not appear free in $\ottnt{A}$.
 \begin{enumerate}
  \item \label{lem:subtyping-forall-move:neg}
        If the occurrences of $\beta$ in $\ottnt{A}$ are only negative or strictly
        positive, then $\Gamma  \vdash    \text{\unboldmath$\forall$}  \, \alpha  \ottsym{.} \, \ottnt{A}    [  \ottnt{B}  \ottsym{/}  \beta  ]    \sqsubseteq   \ottnt{A}    [   \text{\unboldmath$\forall$}  \, \alpha  \ottsym{.} \, \ottnt{B}  \ottsym{/}  \beta  ]  $.
  \item \label{lem:subtyping-forall-move:pos}
        If the occurrences of $\beta$ in $\ottnt{A}$ are only positive,
        then $\Gamma  \vdash   \ottnt{A}    [   \text{\unboldmath$\forall$}  \, \alpha  \ottsym{.} \, \ottnt{B}  \ottsym{/}  \beta  ]    \sqsubseteq    \text{\unboldmath$\forall$}  \, \alpha  \ottsym{.} \, \ottnt{A}    [  \ottnt{B}  \ottsym{/}  \beta  ]  $.
 \end{enumerate}
\end{lemma}
\fi
\begin{proof}
 By induction on $\ottnt{A}$.
 \begin{caseanalysis}
  \case $\ottnt{A} \,  =  \, \gamma$:
   If $\gamma \,  =  \, \beta$, then we have to show that
   $\Gamma  \vdash   \text{\unboldmath$\forall$}  \, \alpha  \ottsym{.} \, \ottnt{B}  \sqsubseteq   \text{\unboldmath$\forall$}  \, \alpha  \ottsym{.} \, \ottnt{B}$ in the both cases, which
   is shown by \Srule{Refl}.
   Otherwise, if $\gamma \,  \not=  \, \beta$, then
   we have to show that
   $\Gamma  \vdash   \text{\unboldmath$\forall$}  \, \alpha  \ottsym{.} \, \gamma  \sqsubseteq  \gamma$ and
   $\Gamma  \vdash  \gamma  \sqsubseteq   \text{\unboldmath$\forall$}  \, \alpha  \ottsym{.} \, \gamma$.
   By the assumption, $\alpha \,  \not=  \, \gamma$.
   Thus, by \Srule{Gen}, $\Gamma  \vdash  \gamma  \sqsubseteq   \text{\unboldmath$\forall$}  \, \alpha  \ottsym{.} \, \gamma$.
   We also have $\Gamma  \vdash   \text{\unboldmath$\forall$}  \, \alpha  \ottsym{.} \, \gamma  \sqsubseteq  \gamma$
   by \Srule{Inst} (the type used for instantiation can be any, e.g., $ \mathsf{int} $).

  \case $\ottnt{A} \,  =  \, \iota$:
   Similar for the case that $\ottnt{A} \,  =  \, \gamma$ and $\gamma \,  \not=  \, \beta$.

  \case $\ottnt{A} \,  =  \, \ottnt{C}  \rightarrow  \ottnt{D}$:
   We prove the first case.
   The occurrences of $\beta$ in $\ottnt{C}  \rightarrow  \ottnt{D}$ are only negative or strictly positive.
   By definition, the occurrences of $\beta$ in $\ottnt{C}$
   are only positive.  Thus, by the IH, $\Gamma  \vdash   \ottnt{C}    [   \text{\unboldmath$\forall$}  \, \alpha  \ottsym{.} \, \ottnt{B}  \ottsym{/}  \beta  ]    \sqsubseteq    \text{\unboldmath$\forall$}  \, \alpha  \ottsym{.} \, \ottnt{C}    [  \ottnt{B}  \ottsym{/}  \beta  ]  $.
   By definition, the occurrences of $\beta$ in $\ottnt{D}$
   are only negative or strictly positive.  Thus, by the IH,
   $\Gamma  \vdash    \text{\unboldmath$\forall$}  \, \alpha  \ottsym{.} \, \ottnt{D}    [  \ottnt{B}  \ottsym{/}  \beta  ]    \sqsubseteq   \ottnt{D}    [   \text{\unboldmath$\forall$}  \, \alpha  \ottsym{.} \, \ottnt{B}  \ottsym{/}  \beta  ]  $.
   By \Srule{Fun},
   \[
    \Gamma  \vdash   \ottsym{(}    \text{\unboldmath$\forall$}  \, \alpha  \ottsym{.} \, \ottnt{C}    [  \ottnt{B}  \ottsym{/}  \beta  ]    \ottsym{)}  \rightarrow   \text{\unboldmath$\forall$}  \, \alpha  \ottsym{.} \, \ottnt{D}    [  \ottnt{B}  \ottsym{/}  \beta  ]    \sqsubseteq    \ottnt{C}    [   \text{\unboldmath$\forall$}  \, \alpha  \ottsym{.} \, \ottnt{B}  \ottsym{/}  \beta  ]    \rightarrow  \ottnt{D}    [   \text{\unboldmath$\forall$}  \, \alpha  \ottsym{.} \, \ottnt{B}  \ottsym{/}  \beta  ]  .
   \]
   By \Srule{DFun} and \Srule{Trans},
   \begin{equation}
    \Gamma  \vdash    \text{\unboldmath$\forall$}  \, \alpha  \ottsym{.} \, \ottsym{(}    \text{\unboldmath$\forall$}  \, \alpha  \ottsym{.} \, \ottnt{C}    [  \ottnt{B}  \ottsym{/}  \beta  ]    \ottsym{)}  \rightarrow  \ottnt{D}    [  \ottnt{B}  \ottsym{/}  \beta  ]    \sqsubseteq    \ottnt{C}    [   \text{\unboldmath$\forall$}  \, \alpha  \ottsym{.} \, \ottnt{B}  \ottsym{/}  \beta  ]    \rightarrow  \ottnt{D}    [   \text{\unboldmath$\forall$}  \, \alpha  \ottsym{.} \, \ottnt{B}  \ottsym{/}  \beta  ]  .
     \label{eqn:subtyping-forall-move:fun:one}
   \end{equation}
   By \Srule{Inst},
   \begin{equation}
    \Gamma  \ottsym{,}  \alpha  \vdash    \text{\unboldmath$\forall$}  \, \alpha  \ottsym{.} \, \ottnt{C}    [  \ottnt{B}  \ottsym{/}  \beta  ]    \sqsubseteq   \ottnt{C}    [  \ottnt{B}  \ottsym{/}  \beta  ]  .
     \label{eqn:subtyping-forall-move:fun:two}
   \end{equation}
   By \Srule{Fun} and \Srule{Poly} with (\ref{eqn:subtyping-forall-move:fun:two}),
   \[
    \Gamma  \vdash     \text{\unboldmath$\forall$}  \, \alpha  \ottsym{.} \, \ottnt{C}    [  \ottnt{B}  \ottsym{/}  \beta  ]    \rightarrow  \ottnt{D}    [  \ottnt{B}  \ottsym{/}  \beta  ]    \sqsubseteq    \text{\unboldmath$\forall$}  \, \alpha  \ottsym{.} \, \ottsym{(}    \text{\unboldmath$\forall$}  \, \alpha  \ottsym{.} \, \ottnt{C}    [  \ottnt{B}  \ottsym{/}  \beta  ]    \ottsym{)}  \rightarrow  \ottnt{D}    [  \ottnt{B}  \ottsym{/}  \beta  ]  .
   \]
   Thus, by \Srule{Trans} with (\ref{eqn:subtyping-forall-move:fun:one}),
   \[
    \Gamma  \vdash     \text{\unboldmath$\forall$}  \, \alpha  \ottsym{.} \, \ottnt{C}    [  \ottnt{B}  \ottsym{/}  \beta  ]    \rightarrow  \ottnt{D}    [  \ottnt{B}  \ottsym{/}  \beta  ]    \sqsubseteq    \ottnt{C}    [   \text{\unboldmath$\forall$}  \, \alpha  \ottsym{.} \, \ottnt{B}  \ottsym{/}  \beta  ]    \rightarrow  \ottnt{D}    [   \text{\unboldmath$\forall$}  \, \alpha  \ottsym{.} \, \ottnt{B}  \ottsym{/}  \beta  ]  .
   \]

   Next, we prove the second case.
   The occurrences of $\beta$ in $\ottnt{C}  \rightarrow  \ottnt{D}$ are only positive.
   By definition, the occurrences of $\beta$ in $\ottnt{C}$
   are only negative.
   Thus, by \reflem{subtyping-forall-remove} (\ref{lem:subtyping-forall-remove:neg}),
   $\Gamma  \ottsym{,}  \alpha  \vdash   \ottnt{C}    [  \ottnt{B}  \ottsym{/}  \beta  ]    \sqsubseteq   \ottnt{C}    [   \text{\unboldmath$\forall$}  \, \alpha  \ottsym{.} \, \ottnt{B}  \ottsym{/}  \beta  ]  $.
   By definition, the occurrences of $\beta$ in $\ottnt{D}$
   are only positive.
   Thus, by \reflem{subtyping-forall-remove} (\ref{lem:subtyping-forall-remove:pos}),
   $\Gamma  \ottsym{,}  \alpha  \vdash   \ottnt{D}    [   \text{\unboldmath$\forall$}  \, \alpha  \ottsym{.} \, \ottnt{B}  \ottsym{/}  \beta  ]    \sqsubseteq   \ottnt{D}    [  \ottnt{B}  \ottsym{/}  \beta  ]  $.
   By \Srule{Fun}, \Srule{Poly}, and \Srule{Trans},
   \[
    \Gamma  \vdash     \text{\unboldmath$\forall$}  \, \alpha  \ottsym{.} \, \ottnt{C}    [   \text{\unboldmath$\forall$}  \, \alpha  \ottsym{.} \, \ottnt{B}  \ottsym{/}  \beta  ]    \rightarrow  \ottnt{D}    [   \text{\unboldmath$\forall$}  \, \alpha  \ottsym{.} \, \ottnt{B}  \ottsym{/}  \beta  ]    \sqsubseteq     \text{\unboldmath$\forall$}  \, \alpha  \ottsym{.} \, \ottnt{C}    [  \ottnt{B}  \ottsym{/}  \beta  ]    \rightarrow  \ottnt{D}    [  \ottnt{B}  \ottsym{/}  \beta  ]  .
   \]
   Since $\alpha$ does not appear free in $\ottnt{A} \,  =  \, \ottnt{C}  \rightarrow  \ottnt{D}$,
   we have $\Gamma  \vdash    \ottnt{C}    [   \text{\unboldmath$\forall$}  \, \alpha  \ottsym{.} \, \ottnt{B}  \ottsym{/}  \beta  ]    \rightarrow  \ottnt{D}    [   \text{\unboldmath$\forall$}  \, \alpha  \ottsym{.} \, \ottnt{B}  \ottsym{/}  \beta  ]    \sqsubseteq     \text{\unboldmath$\forall$}  \, \alpha  \ottsym{.} \, \ottnt{C}    [   \text{\unboldmath$\forall$}  \, \alpha  \ottsym{.} \, \ottnt{B}  \ottsym{/}  \beta  ]    \rightarrow  \ottnt{D}    [   \text{\unboldmath$\forall$}  \, \alpha  \ottsym{.} \, \ottnt{B}  \ottsym{/}  \beta  ]  $
   by \Srule{Gen}.
   Thus, by \Srule{Trans},
   \[
    \Gamma  \vdash    \ottnt{C}    [   \text{\unboldmath$\forall$}  \, \alpha  \ottsym{.} \, \ottnt{B}  \ottsym{/}  \beta  ]    \rightarrow  \ottnt{D}    [   \text{\unboldmath$\forall$}  \, \alpha  \ottsym{.} \, \ottnt{B}  \ottsym{/}  \beta  ]    \sqsubseteq     \text{\unboldmath$\forall$}  \, \alpha  \ottsym{.} \, \ottnt{C}    [  \ottnt{B}  \ottsym{/}  \beta  ]    \rightarrow  \ottnt{D}    [  \ottnt{B}  \ottsym{/}  \beta  ]  .
   \]

  \case $\ottnt{A} \,  =  \,  \text{\unboldmath$\forall$}  \, \gamma  \ottsym{.} \, \ottnt{C}$: By the IH, \Srule{Poly}, and permutation of the
   top-level $\forall$s for each case.

  \case $\ottnt{A} \,  =  \,  \ottnt{C}  \times  \ottnt{D} $:
   We prove the first case.
   The occurrences of $\beta$ in $ \ottnt{C}  \times  \ottnt{D} $ are only negative or strictly positive.
   By definition, the occurrences of $\beta$ in $\ottnt{C}$
   are only negative or strictly positive.  Thus, by the IH,
   $\Gamma  \vdash    \text{\unboldmath$\forall$}  \, \alpha  \ottsym{.} \, \ottnt{C}    [  \ottnt{B}  \ottsym{/}  \beta  ]    \sqsubseteq   \ottnt{C}    [   \text{\unboldmath$\forall$}  \, \alpha  \ottsym{.} \, \ottnt{B}  \ottsym{/}  \beta  ]  $.
   Similarly, we also have
   $\Gamma  \vdash    \text{\unboldmath$\forall$}  \, \alpha  \ottsym{.} \, \ottnt{D}    [  \ottnt{B}  \ottsym{/}  \beta  ]    \sqsubseteq   \ottnt{D}    [   \text{\unboldmath$\forall$}  \, \alpha  \ottsym{.} \, \ottnt{B}  \ottsym{/}  \beta  ]  $.
   By \Srule{Prod},
   \[
    \Gamma  \vdash    \ottsym{(}    \text{\unboldmath$\forall$}  \, \alpha  \ottsym{.} \, \ottnt{C}    [  \ottnt{B}  \ottsym{/}  \beta  ]    \ottsym{)}  \times   \text{\unboldmath$\forall$}  \, \alpha  \ottsym{.} \, \ottnt{D}     [  \ottnt{B}  \ottsym{/}  \beta  ]    \sqsubseteq     \ottnt{C}    [   \text{\unboldmath$\forall$}  \, \alpha  \ottsym{.} \, \ottnt{B}  \ottsym{/}  \beta  ]    \times  \ottnt{D}     [   \text{\unboldmath$\forall$}  \, \alpha  \ottsym{.} \, \ottnt{B}  \ottsym{/}  \beta  ]  .
   \]
   By \Srule{DProd} and \Srule{Trans},
   \[
    \Gamma  \vdash   \text{\unboldmath$\forall$}  \, \alpha  \ottsym{.} \, \ottsym{(}     \ottnt{C}    [  \ottnt{B}  \ottsym{/}  \beta  ]    \times  \ottnt{D}     [  \ottnt{B}  \ottsym{/}  \beta  ]    \ottsym{)}  \sqsubseteq     \ottnt{C}    [   \text{\unboldmath$\forall$}  \, \alpha  \ottsym{.} \, \ottnt{B}  \ottsym{/}  \beta  ]    \times  \ottnt{D}     [   \text{\unboldmath$\forall$}  \, \alpha  \ottsym{.} \, \ottnt{B}  \ottsym{/}  \beta  ]  .
   \]

   We prove the second case.
   The occurrences of $\beta$ in $ \ottnt{C}  \times  \ottnt{D} $ are only positive.
   By definition, the occurrences of $\beta$ in $\ottnt{C}$
   are only positive.  Thus, by the IH,
   $\Gamma  \vdash   \ottnt{C}    [   \text{\unboldmath$\forall$}  \, \alpha  \ottsym{.} \, \ottnt{B}  \ottsym{/}  \beta  ]    \sqsubseteq    \text{\unboldmath$\forall$}  \, \alpha  \ottsym{.} \, \ottnt{C}    [  \ottnt{B}  \ottsym{/}  \beta  ]  $.
   Similarly, we also have
   $\Gamma  \vdash   \ottnt{D}    [   \text{\unboldmath$\forall$}  \, \alpha  \ottsym{.} \, \ottnt{B}  \ottsym{/}  \beta  ]    \sqsubseteq    \text{\unboldmath$\forall$}  \, \alpha  \ottsym{.} \, \ottnt{D}    [  \ottnt{B}  \ottsym{/}  \beta  ]  $.
   By \Srule{Prod},
   \[
    \Gamma  \vdash     \ottnt{C}    [   \text{\unboldmath$\forall$}  \, \alpha  \ottsym{.} \, \ottnt{B}  \ottsym{/}  \beta  ]    \times  \ottnt{D}     [   \text{\unboldmath$\forall$}  \, \alpha  \ottsym{.} \, \ottnt{B}  \ottsym{/}  \beta  ]    \sqsubseteq    \ottsym{(}   \text{\unboldmath$\forall$}  \, \alpha  \ottsym{.} \,  \ottnt{C}    [  \ottnt{B}  \ottsym{/}  \beta  ]    \ottsym{)}  \times   \text{\unboldmath$\forall$}  \, \alpha  \ottsym{.} \, \ottnt{D}     [  \ottnt{B}  \ottsym{/}  \beta  ]  .
   \]
   By \Srule{Gen}, \Srule{Poly}, \Srule{Inst}, \Srule{Prod}, and \Srule{Trans},
   we have
   $\Gamma  \vdash    \ottsym{(}    \text{\unboldmath$\forall$}  \, \alpha  \ottsym{.} \, \ottnt{C}    [  \ottnt{B}  \ottsym{/}  \beta  ]    \ottsym{)}  \times   \text{\unboldmath$\forall$}  \, \alpha  \ottsym{.} \, \ottnt{D}     [  \ottnt{B}  \ottsym{/}  \beta  ]    \sqsubseteq   \text{\unboldmath$\forall$}  \, \alpha  \ottsym{.} \, \ottsym{(}     \ottnt{C}    [  \ottnt{B}  \ottsym{/}  \beta  ]    \times  \ottnt{D}     [  \ottnt{B}  \ottsym{/}  \beta  ]    \ottsym{)}$.
   Thus, by \Srule{Trans},
   \[
    \Gamma  \vdash     \ottnt{C}    [   \text{\unboldmath$\forall$}  \, \alpha  \ottsym{.} \, \ottnt{B}  \ottsym{/}  \beta  ]    \times  \ottnt{D}     [   \text{\unboldmath$\forall$}  \, \alpha  \ottsym{.} \, \ottnt{B}  \ottsym{/}  \beta  ]    \sqsubseteq   \text{\unboldmath$\forall$}  \, \alpha  \ottsym{.} \, \ottsym{(}     \ottnt{C}    [  \ottnt{B}  \ottsym{/}  \beta  ]    \times  \ottnt{D}     [  \ottnt{B}  \ottsym{/}  \beta  ]    \ottsym{)}.
   \]

  \case $\ottnt{A} \,  =  \,  \ottnt{C}  +  \ottnt{D} $: Similarly to the case that $\ottnt{A}$ is a product type;
   this case uses \Srule{Sum} and \Srule{DSum}
   instead of \Srule{Prod} and \Srule{DProd}.

  \case $\ottnt{A} \,  =  \,  \ottnt{C}  \, \mathsf{list} $: Similarly to the case that $\ottnt{A}$ is a product type;
   this case uses \Srule{List} and \Srule{DList}
   instead of \Srule{Prod} and \Srule{DProd}.
 \end{caseanalysis}
\end{proof}

\ifrestate
\lemmSubjectRed*
\else
\begin{lemmap}{Subject reduction}{subject-red}
 Suppose that all operations satisfy the signature restriction.
 \begin{enumerate}
  \item If $\Delta  \vdash  \ottnt{M_{{\mathrm{1}}}}  \ottsym{:}  \ottnt{A}$ and $\ottnt{M_{{\mathrm{1}}}}  \rightsquigarrow  \ottnt{M_{{\mathrm{2}}}}$,
        then $\Delta  \vdash  \ottnt{M_{{\mathrm{2}}}}  \ottsym{:}  \ottnt{A}$.
  \item If $\Delta  \vdash  \ottnt{M_{{\mathrm{1}}}}  \ottsym{:}  \ottnt{A}$ and $\ottnt{M_{{\mathrm{1}}}}  \longrightarrow  \ottnt{M_{{\mathrm{2}}}}$,
        then $\Delta  \vdash  \ottnt{M_{{\mathrm{2}}}}  \ottsym{:}  \ottnt{A}$.
 \end{enumerate}
\end{lemmap}
\fi
\begin{proof}
 \begin{enumerate}
  \item Suppose that $\Delta  \vdash  \ottnt{M_{{\mathrm{1}}}}  \ottsym{:}  \ottnt{A}$ and $\ottnt{M_{{\mathrm{1}}}}  \rightsquigarrow  \ottnt{M_{{\mathrm{2}}}}$.
        By induction on the typing derivation for $\ottnt{M_{{\mathrm{1}}}}$.
        \begin{caseanalysis}
         \case \T{Var}, \T{Op}, \T{Pair}, \T{InL}, \T{InR}, and \T{Cons}:
          Contradictory because there are no reduction rules
          that can be applied to $\ottnt{M_{{\mathrm{1}}}}$.

         \case \T{Const}, \T{Abs}, and \T{Nil}: Contradictory
          since $\ottnt{M_{{\mathrm{1}}}}$ is a value and no reduction rules can be applied to values.

         \case \T{App}: We have two reduction rules which can be applied to
          function applications.
          \begin{caseanalysis}
           \case \R{Const}:
            We are given
            \begin{itemize}
             \item $\ottnt{M_{{\mathrm{1}}}} \,  =  \, \ottnt{c_{{\mathrm{1}}}} \, \ottnt{c_{{\mathrm{2}}}}$,
             \item $\ottnt{M_{{\mathrm{2}}}} \,  =  \,  \zeta  (  \ottnt{c_{{\mathrm{1}}}}  ,  \ottnt{c_{{\mathrm{2}}}}  ) $,
             \item $\Delta  \vdash  \ottnt{c_{{\mathrm{1}}}} \, \ottnt{c_{{\mathrm{2}}}}  \ottsym{:}  \ottnt{A}$,
             \item $\Delta  \vdash  \ottnt{c_{{\mathrm{1}}}}  \ottsym{:}  \ottnt{B}  \rightarrow  \ottnt{A}$, and
             \item $\Delta  \vdash  \ottnt{c_{{\mathrm{2}}}}  \ottsym{:}  \ottnt{B}$
            \end{itemize}
            for some $\ottnt{c_{{\mathrm{1}}}}$, $\ottnt{c_{{\mathrm{2}}}}$, and $\ottnt{B}$.
            By \reflem{val-inv-const},
            $\Delta  \vdash   \mathit{ty}  (  \ottnt{c_{{\mathrm{1}}}}  )   \sqsubseteq  \ottnt{B}  \rightarrow  \ottnt{A}$.
            By \reflem{subtyping-unqualify} and \refasm{const},
            $ \mathit{ty}  (  \ottnt{c_{{\mathrm{1}}}}  )  \,  =  \, \iota  \rightarrow  \ottnt{C}$ for some $\iota$ and $\ottnt{C}$.
            Since $ \zeta  (  \ottnt{c_{{\mathrm{1}}}}  ,  \ottnt{c_{{\mathrm{2}}}}  ) $ is defined, it is found that
            $ \mathit{ty}  (  \ottnt{c_{{\mathrm{2}}}}  )  \,  =  \, \iota$ and $ \mathit{ty}  (   \zeta  (  \ottnt{c_{{\mathrm{1}}}}  ,  \ottnt{c_{{\mathrm{2}}}}  )   )  \,  =  \, \ottnt{C}$.
            Since $\vdash  \Delta$ by \reflem{type-wf},
            we have $\Delta  \vdash   \zeta  (  \ottnt{c_{{\mathrm{1}}}}  ,  \ottnt{c_{{\mathrm{2}}}}  )   \ottsym{:}   \mathit{ty}  (   \zeta  (  \ottnt{c_{{\mathrm{1}}}}  ,  \ottnt{c_{{\mathrm{2}}}}  )   ) $.
            Since $\Delta  \vdash  \iota  \rightarrow   \mathit{ty}  (   \zeta  (  \ottnt{c_{{\mathrm{1}}}}  ,  \ottnt{c_{{\mathrm{2}}}}  )   )   \sqsubseteq  \ottnt{B}  \rightarrow  \ottnt{A}$ (recall that
            $\ottnt{C} \,  =  \,  \mathit{ty}  (   \zeta  (  \ottnt{c_{{\mathrm{1}}}}  ,  \ottnt{c_{{\mathrm{2}}}}  )   ) $), we have $\Delta  \vdash   \mathit{ty}  (   \zeta  (  \ottnt{c_{{\mathrm{1}}}}  ,  \ottnt{c_{{\mathrm{2}}}}  )   )   \sqsubseteq  \ottnt{A}$
            by \reflem{subtyping-inv-fun-mono}.
            By \T{Inst}, we have $\Delta  \vdash   \zeta  (  \ottnt{c_{{\mathrm{1}}}}  ,  \ottnt{c_{{\mathrm{2}}}}  )   \ottsym{:}  \ottnt{A}$.

           \case \R{Beta}:
            We are given
            \begin{itemize}
             \item $\ottnt{M_{{\mathrm{1}}}} \,  =  \, \ottsym{(}   \lambda\!  \, \mathit{x}  \ottsym{.}  \ottnt{M}  \ottsym{)} \, \ottnt{v}$,
             \item $\ottnt{M_{{\mathrm{2}}}} \,  =  \,  \ottnt{M}    [  \ottnt{v}  /  \mathit{x}  ]  $,
             \item $\Delta  \vdash  \ottsym{(}   \lambda\!  \, \mathit{x}  \ottsym{.}  \ottnt{M}  \ottsym{)} \, \ottnt{v}  \ottsym{:}  \ottnt{A}$,
             \item $\Delta  \vdash   \lambda\!  \, \mathit{x}  \ottsym{.}  \ottnt{M}  \ottsym{:}  \ottnt{B}  \rightarrow  \ottnt{A}$, and
             \item $\Delta  \vdash  \ottnt{v}  \ottsym{:}  \ottnt{B}$
            \end{itemize}
            for some $\mathit{x}$, $\ottnt{M}$, $\ottnt{v}$, and $\ottnt{B}$.
            By \reflem{val-inv-abs}
            $\Delta  \ottsym{,}   \algeffseqoverindex{ \alpha }{ \text{\unboldmath$\mathit{I}$} }   \ottsym{,}  \mathit{x} \,  \mathord{:}  \, \ottnt{B'}  \vdash  \ottnt{M}  \ottsym{:}  \ottnt{A'}$ and
            $\Delta  \vdash   \text{\unboldmath$\forall$}  \,  \algeffseqoverindex{ \alpha }{ \text{\unboldmath$\mathit{I}$} }   \ottsym{.} \, \ottnt{B'}  \rightarrow  \ottnt{A'}  \sqsubseteq  \ottnt{B}  \rightarrow  \ottnt{A}$
            for some $ \algeffseqoverindex{ \alpha }{ \text{\unboldmath$\mathit{I}$} } $, $\ottnt{A'}$, and $\ottnt{B'}$.
            By \reflem{subtyping-inv-fun},
            there exist $ \algeffseqoverindex{ \alpha_{{\mathrm{1}}} }{ \text{\unboldmath$\mathit{I_{{\mathrm{1}}}}$} } $, $ \algeffseqoverindex{ \alpha_{{\mathrm{2}}} }{ \text{\unboldmath$\mathit{I_{{\mathrm{2}}}}$} } $, $ \algeffseqoverindex{ \beta }{ \text{\unboldmath$\mathit{J}$} } $, and $ \algeffseqoverindex{ \ottnt{C} }{ \text{\unboldmath$\mathit{I_{{\mathrm{1}}}}$} } $
            such that
            \begin{itemize}
             \item $\ottsym{\{}   \algeffseqoverindex{ \alpha }{ \text{\unboldmath$\mathit{I}$} }   \ottsym{\}} \,  =  \, \ottsym{\{}   \algeffseqoverindex{ \alpha_{{\mathrm{1}}} }{ \text{\unboldmath$\mathit{I_{{\mathrm{1}}}}$} }   \ottsym{\}} \,  \mathbin{\uplus}  \, \ottsym{\{}   \algeffseqoverindex{ \alpha_{{\mathrm{2}}} }{ \text{\unboldmath$\mathit{I_{{\mathrm{2}}}}$} }   \ottsym{\}}$,
             \item $\Delta  \ottsym{,}   \algeffseqoverindex{ \beta }{ \text{\unboldmath$\mathit{J}$} }   \vdash   \algeffseqoverindex{ \ottnt{C} }{ \text{\unboldmath$\mathit{I_{{\mathrm{1}}}}$} } $,
             \item $\Delta  \vdash  \ottnt{B}  \sqsubseteq    \text{\unboldmath$\forall$}  \,  \algeffseqoverindex{ \beta }{ \text{\unboldmath$\mathit{J}$} }   \ottsym{.} \, \ottnt{B'}    [   \algeffseqoverindex{ \ottnt{C} }{ \text{\unboldmath$\mathit{I_{{\mathrm{1}}}}$} }   \ottsym{/}   \algeffseqoverindex{ \alpha_{{\mathrm{1}}} }{ \text{\unboldmath$\mathit{I_{{\mathrm{1}}}}$} }   ]  $,
             \item $\Delta  \vdash    \text{\unboldmath$\forall$}  \,  \algeffseqoverindex{ \alpha_{{\mathrm{2}}} }{ \text{\unboldmath$\mathit{I_{{\mathrm{2}}}}$} }   \ottsym{.} \,  \text{\unboldmath$\forall$}  \,  \algeffseqoverindex{ \beta }{ \text{\unboldmath$\mathit{J}$} }   \ottsym{.} \, \ottnt{A'}    [   \algeffseqoverindex{ \ottnt{C} }{ \text{\unboldmath$\mathit{I_{{\mathrm{1}}}}$} }   \ottsym{/}   \algeffseqoverindex{ \alpha_{{\mathrm{1}}} }{ \text{\unboldmath$\mathit{I_{{\mathrm{1}}}}$} }   ]    \sqsubseteq  \ottnt{A}$, and
             \item type variables in $ \algeffseqoverindex{ \beta }{ \text{\unboldmath$\mathit{J}$} } $ do not appear free in
                   $\ottnt{A'}$ and $\ottnt{B'}$.
            \end{itemize}
            By \reflem{weakening},
            $\Delta  \ottsym{,}   \algeffseqoverindex{ \beta }{ \text{\unboldmath$\mathit{J}$} }   \ottsym{,}   \algeffseqoverindex{ \alpha }{ \text{\unboldmath$\mathit{I}$} }   \ottsym{,}  \mathit{x} \,  \mathord{:}  \, \ottnt{B'}  \vdash  \ottnt{M}  \ottsym{:}  \ottnt{A'}$ and
            $\Delta  \ottsym{,}   \algeffseqoverindex{ \beta }{ \text{\unboldmath$\mathit{J}$} }   \ottsym{,}   \algeffseqoverindex{ \alpha_{{\mathrm{2}}} }{ \text{\unboldmath$\mathit{I_{{\mathrm{2}}}}$} }   \vdash   \algeffseqoverindex{ \ottnt{C} }{ \text{\unboldmath$\mathit{I_{{\mathrm{1}}}}$} } $.
            Thus, by \reflem{ty-subst} (\ref{lem:ty-subst:term}),
            \begin{equation}
             \Delta  \ottsym{,}   \algeffseqoverindex{ \beta }{ \text{\unboldmath$\mathit{J}$} }   \ottsym{,}   \algeffseqoverindex{ \alpha_{{\mathrm{2}}} }{ \text{\unboldmath$\mathit{I_{{\mathrm{2}}}}$} }   \ottsym{,}  \mathit{x} \,  \mathord{:}  \, \ottnt{B'} \,  [   \algeffseqoverindex{ \ottnt{C} }{ \text{\unboldmath$\mathit{I_{{\mathrm{1}}}}$} }   \ottsym{/}   \algeffseqoverindex{ \alpha }{ \text{\unboldmath$\mathit{I_{{\mathrm{1}}}}$} }   ]   \vdash  \ottnt{M}  \ottsym{:}   \ottnt{A'}    [   \algeffseqoverindex{ \ottnt{C} }{ \text{\unboldmath$\mathit{I_{{\mathrm{1}}}}$} }   \ottsym{/}   \algeffseqoverindex{ \alpha_{{\mathrm{1}}} }{ \text{\unboldmath$\mathit{I_{{\mathrm{1}}}}$} }   ]  
              \label{eqn:subject-red:app:beta:body}
            \end{equation}

            Since $\Delta  \vdash  \ottnt{v}  \ottsym{:}  \ottnt{B}$ and $\Delta  \vdash  \ottnt{B}  \sqsubseteq    \text{\unboldmath$\forall$}  \,  \algeffseqoverindex{ \beta }{ \text{\unboldmath$\mathit{J}$} }   \ottsym{.} \, \ottnt{B'}    [   \algeffseqoverindex{ \ottnt{C} }{ \text{\unboldmath$\mathit{I_{{\mathrm{1}}}}$} }   \ottsym{/}   \algeffseqoverindex{ \alpha_{{\mathrm{1}}} }{ \text{\unboldmath$\mathit{I_{{\mathrm{1}}}}$} }   ]  $,
            we have
            \[
             \Delta  \vdash  \ottnt{v}  \ottsym{:}    \text{\unboldmath$\forall$}  \,  \algeffseqoverindex{ \beta }{ \text{\unboldmath$\mathit{J}$} }   \ottsym{.} \, \ottnt{B'}    [   \algeffseqoverindex{ \ottnt{C} }{ \text{\unboldmath$\mathit{I_{{\mathrm{1}}}}$} }   \ottsym{/}   \algeffseqoverindex{ \alpha_{{\mathrm{1}}} }{ \text{\unboldmath$\mathit{I_{{\mathrm{1}}}}$} }   ]  
            \]
            by \T{Inst} (note that $\Delta  \vdash    \text{\unboldmath$\forall$}  \,  \algeffseqoverindex{ \beta }{ \text{\unboldmath$\mathit{J}$} }   \ottsym{.} \, \ottnt{B'}    [   \algeffseqoverindex{ \ottnt{C} }{ \text{\unboldmath$\mathit{I_{{\mathrm{1}}}}$} }   \ottsym{/}   \algeffseqoverindex{ \alpha_{{\mathrm{1}}} }{ \text{\unboldmath$\mathit{I_{{\mathrm{1}}}}$} }   ]  $ is
            shown easily with \reflem{type-wf}).
            By \reflem{weakening} (\ref{lem:weakening:term}),
            \Srule{Inst}, and \T{Inst},
            we have
            \[
             \Delta  \ottsym{,}   \algeffseqoverindex{ \beta }{ \text{\unboldmath$\mathit{J}$} }   \ottsym{,}   \algeffseqoverindex{ \alpha_{{\mathrm{2}}} }{ \text{\unboldmath$\mathit{I_{{\mathrm{2}}}}$} }   \vdash  \ottnt{v}  \ottsym{:}   \ottnt{B'}    [   \algeffseqoverindex{ \ottnt{C} }{ \text{\unboldmath$\mathit{I_{{\mathrm{1}}}}$} }   \ottsym{/}   \algeffseqoverindex{ \alpha }{ \text{\unboldmath$\mathit{I_{{\mathrm{1}}}}$} }   ]  .
            \]
            By \reflem{term-subst} (\ref{lem:term-subst:term})
            with (\ref{eqn:subject-red:app:beta:body}),
            \[
             \Delta  \ottsym{,}   \algeffseqoverindex{ \beta }{ \text{\unboldmath$\mathit{J}$} }   \ottsym{,}   \algeffseqoverindex{ \alpha_{{\mathrm{2}}} }{ \text{\unboldmath$\mathit{I_{{\mathrm{2}}}}$} }   \vdash   \ottnt{M}    [  \ottnt{v}  /  \mathit{x}  ]    \ottsym{:}   \ottnt{A'}    [   \algeffseqoverindex{ \ottnt{C} }{ \text{\unboldmath$\mathit{I_{{\mathrm{1}}}}$} }   \ottsym{/}   \algeffseqoverindex{ \alpha_{{\mathrm{1}}} }{ \text{\unboldmath$\mathit{I_{{\mathrm{1}}}}$} }   ]  .
            \]
            By \T{Gen} (with permutation of the bindings in the typing context),
            \[
             \Delta  \vdash   \ottnt{M}    [  \ottnt{v}  /  \mathit{x}  ]    \ottsym{:}    \text{\unboldmath$\forall$}  \,  \algeffseqoverindex{ \alpha_{{\mathrm{2}}} }{ \text{\unboldmath$\mathit{I_{{\mathrm{2}}}}$} }   \ottsym{.} \,  \text{\unboldmath$\forall$}  \,  \algeffseqoverindex{ \beta }{ \text{\unboldmath$\mathit{J}$} }   \ottsym{.} \, \ottnt{A'}    [   \algeffseqoverindex{ \ottnt{C} }{ \text{\unboldmath$\mathit{I_{{\mathrm{1}}}}$} }   \ottsym{/}   \algeffseqoverindex{ \alpha_{{\mathrm{1}}} }{ \text{\unboldmath$\mathit{I_{{\mathrm{1}}}}$} }   ]  .
            \]
            Since $\Delta  \vdash    \text{\unboldmath$\forall$}  \,  \algeffseqoverindex{ \alpha_{{\mathrm{2}}} }{ \text{\unboldmath$\mathit{I_{{\mathrm{2}}}}$} }   \ottsym{.} \,  \text{\unboldmath$\forall$}  \,  \algeffseqoverindex{ \beta }{ \text{\unboldmath$\mathit{J}$} }   \ottsym{.} \, \ottnt{A'}    [   \algeffseqoverindex{ \ottnt{C} }{ \text{\unboldmath$\mathit{I_{{\mathrm{1}}}}$} }   \ottsym{/}   \algeffseqoverindex{ \alpha_{{\mathrm{1}}} }{ \text{\unboldmath$\mathit{I_{{\mathrm{1}}}}$} }   ]    \sqsubseteq  \ottnt{A}$,
            we have $\Delta  \vdash   \ottnt{M}    [  \ottnt{v}  /  \mathit{x}  ]    \ottsym{:}  \ottnt{A}$ by \T{Inst}.
          \end{caseanalysis}

         \case \T{Gen}: By the IH and \T{Gen}.
         \case \T{Inst}: By the IH and \T{Inst}.

         \case \T{Handle}:
          We have two reduction rules which can be applied to
          $ \mathsf{handle} $--$ \mathsf{with} $ expressions.
          \begin{caseanalysis}
           \case \R{Return}:
            We are given
            \begin{itemize}
             \item $\ottnt{M_{{\mathrm{1}}}} \,  =  \, \mathsf{handle} \, \ottnt{v} \, \mathsf{with} \, \ottnt{H}$,
             \item $ \ottnt{H} ^\mathsf{return}  \,  =  \, \mathsf{return} \, \mathit{x}  \rightarrow  \ottnt{M}$,
             \item $\ottnt{M_{{\mathrm{2}}}} \,  =  \,  \ottnt{M}    [  \ottnt{v}  /  \mathit{x}  ]  $,
             \item $\Delta  \vdash  \mathsf{handle} \, \ottnt{v} \, \mathsf{with} \, \ottnt{H}  \ottsym{:}  \ottnt{A}$,
             \item $\Delta  \vdash  \ottnt{v}  \ottsym{:}  \ottnt{B}$,
             \item $\Delta  \vdash  \ottnt{H}  \ottsym{:}  \ottnt{B}  \Rightarrow  \ottnt{A}$
            \end{itemize}
            for some $\ottnt{v}$, $\ottnt{H}$, $\mathit{x}$, $\ottnt{M}$, and $\ottnt{B}$.
            By inversion of the derivation of
            $\Delta  \vdash  \ottnt{H}  \ottsym{:}  \ottnt{B}  \Rightarrow  \ottnt{A}$, we have $\Delta  \ottsym{,}  \mathit{x} \,  \mathord{:}  \, \ottnt{B}  \vdash  \ottnt{M}  \ottsym{:}  \ottnt{A}$.
            By \reflem{term-subst} (\ref{lem:term-subst:term}),
            $\Delta  \vdash   \ottnt{M}    [  \ottnt{v}  /  \mathit{x}  ]    \ottsym{:}  \ottnt{A}$, which is the conclusion we have to show.

           \case \R{Handle}:
            We are given
            \begin{itemize}
             \item $\ottnt{M_{{\mathrm{1}}}} \,  =  \, \mathsf{handle} \,  \ottnt{E}  [   \textup{\texttt{\#}\relax}  \mathsf{op}   \ottsym{(}   \ottnt{v}   \ottsym{)}   ]  \, \mathsf{with} \, \ottnt{H}$,
             \item $\mathsf{op} \,  \not\in  \, \ottnt{E}$,
             \item $\ottnt{H}  \ottsym{(}  \mathsf{op}  \ottsym{)} \,  =  \, \mathsf{op}  \ottsym{(}  \mathit{x}  \ottsym{,}  \mathit{k}  \ottsym{)}  \rightarrow  \ottnt{M}$,
             \item $\ottnt{M_{{\mathrm{2}}}} \,  =  \,   \ottnt{M}    [  \ottnt{v}  /  \mathit{x}  ]      [   \lambda\!  \, \mathit{y}  \ottsym{.}  \mathsf{handle} \,  \ottnt{E}  [  \mathit{y}  ]  \, \mathsf{with} \, \ottnt{H}  /  \mathit{k}  ]  $,
             \item $\Delta  \vdash  \mathsf{handle} \,  \ottnt{E}  [   \textup{\texttt{\#}\relax}  \mathsf{op}   \ottsym{(}   \ottnt{v}   \ottsym{)}   ]  \, \mathsf{with} \, \ottnt{H}  \ottsym{:}  \ottnt{A}$,
             \item $\Delta  \vdash   \ottnt{E}  [   \textup{\texttt{\#}\relax}  \mathsf{op}   \ottsym{(}   \ottnt{v}   \ottsym{)}   ]   \ottsym{:}  \ottnt{B}$,
             \item $\Delta  \vdash  \ottnt{H}  \ottsym{:}  \ottnt{B}  \Rightarrow  \ottnt{A}$
            \end{itemize}
            for some $\ottnt{E}$, $\mathsf{op}$, $\ottnt{v}$, $\ottnt{H}$, $\mathit{x}$, $\mathit{y}$,
            $\mathit{k}$, $\ottnt{M}$, and $\ottnt{B}$.
            Suppose that $\mathit{ty} \, \ottsym{(}  \mathsf{op}  \ottsym{)} \,  =  \,   \text{\unboldmath$\forall$}  \,  \algeffseqover{ \alpha }   \ottsym{.} \,  \ottnt{C}  \hookrightarrow  \ottnt{D} $.
            By inversion of the derivation of
            $\Delta  \vdash  \ottnt{H}  \ottsym{:}  \ottnt{B}  \Rightarrow  \ottnt{A}$,
            we have $\Delta  \ottsym{,}   \algeffseqover{ \alpha }   \ottsym{,}  \mathit{x} \,  \mathord{:}  \, \ottnt{C}  \ottsym{,}  \mathit{k} \,  \mathord{:}  \, \ottnt{D}  \rightarrow  \ottnt{A}  \vdash  \ottnt{M}  \ottsym{:}  \ottnt{A}$.

            By \reflem{ectx-op-typing},
            $\Delta  \ottsym{,}   \algeffseqoverindex{ \beta }{ \text{\unboldmath$\mathit{J}$} }   \vdash   \algeffseqover{ \ottnt{C_{{\mathrm{0}}}} } $ and
            $\Delta  \ottsym{,}   \algeffseqoverindex{ \beta }{ \text{\unboldmath$\mathit{J}$} }   \vdash  \ottnt{v}  \ottsym{:}   \ottnt{C}    [   \algeffseqover{ \ottnt{C_{{\mathrm{0}}}} }   \ottsym{/}   \algeffseqover{ \alpha }   ]  $
            for some $ \algeffseqoverindex{ \beta }{ \text{\unboldmath$\mathit{J}$} } $ and $ \algeffseqover{ \ottnt{C_{{\mathrm{0}}}} } $.
            Since $\Delta  \vdash   \text{\unboldmath$\forall$}  \,  \algeffseqoverindex{ \beta }{ \text{\unboldmath$\mathit{J}$} }   \ottsym{.} \,  \algeffseqover{ \ottnt{C_{{\mathrm{0}}}} } $,
            \begin{equation}
             \Delta  \ottsym{,}  \mathit{x} \,  \mathord{:}  \, \ottnt{C} \,  [   \text{\unboldmath$\forall$}  \,  \algeffseqoverindex{ \beta }{ \text{\unboldmath$\mathit{J}$} }   \ottsym{.} \,  \algeffseqover{ \ottnt{C_{{\mathrm{0}}}} }   \ottsym{/}   \algeffseqover{ \alpha }   ]   \ottsym{,}  \mathit{k} \,  \mathord{:}  \,  \ottnt{D}    [   \text{\unboldmath$\forall$}  \,  \algeffseqoverindex{ \beta }{ \text{\unboldmath$\mathit{J}$} }   \ottsym{.} \,  \algeffseqover{ \ottnt{C_{{\mathrm{0}}}} }   \ottsym{/}   \algeffseqover{ \alpha }   ]    \rightarrow  \ottnt{A}  \vdash  \ottnt{M}  \ottsym{:}  \ottnt{A}
             \label{eqn:subject-red:handle:handle:one}
            \end{equation}
            by \reflem{ty-subst} (\ref{lem:ty-subst:term}) (note that
            type variables in $ \algeffseqover{ \alpha } $ do not appear free in $\ottnt{A}$).

            Since $\Delta  \ottsym{,}   \algeffseqoverindex{ \beta }{ \text{\unboldmath$\mathit{J}$} }   \vdash  \ottnt{v}  \ottsym{:}   \ottnt{C}    [   \algeffseqover{ \ottnt{C_{{\mathrm{0}}}} }   \ottsym{/}   \algeffseqover{ \alpha }   ]  $, we have
            $\Delta  \vdash  \ottnt{v}  \ottsym{:}    \text{\unboldmath$\forall$}  \,  \algeffseqoverindex{ \beta }{ \text{\unboldmath$\mathit{J}$} }   \ottsym{.} \, \ottnt{C}    [   \algeffseqover{ \ottnt{C_{{\mathrm{0}}}} }   \ottsym{/}   \algeffseqover{ \alpha }   ]  $ by \T{Gen}.
            By \refdef{signature-restriction}, the occurrences of $ \algeffseqover{ \alpha } $ in the domain type
            $\ottnt{C}$ of the type signature of $\mathsf{op}$ are only negative or
            strictly positive.
            Thus, we have $\Delta  \vdash  \ottnt{v}  \ottsym{:}   \ottnt{C}    [   \algeffseqover{  \text{\unboldmath$\forall$}  \,  \algeffseqoverindex{ \beta }{ \text{\unboldmath$\mathit{J}$} }   \ottsym{.} \, \ottnt{C_{{\mathrm{0}}}} }   \ottsym{/}   \algeffseqover{ \alpha }   ]  $
            by \reflem{subtyping-forall-move} (\ref{lem:subtyping-forall-move:neg})
            and \T{Inst}
            (note that we can suppose that $ \algeffseqoverindex{ \beta }{ \text{\unboldmath$\mathit{J}$} } $ do not appear
            free in $\ottnt{C}$).
            Thus, by applying \reflem{term-subst} (\ref{lem:term-subst:term})
            to (\ref{eqn:subject-red:handle:handle:one}), we have
            \begin{equation}
             \Delta  \ottsym{,}  \mathit{k} \,  \mathord{:}  \,  \ottnt{D}    [   \text{\unboldmath$\forall$}  \,  \algeffseqoverindex{ \beta }{ \text{\unboldmath$\mathit{J}$} }   \ottsym{.} \,  \algeffseqover{ \ottnt{C_{{\mathrm{0}}}} }   \ottsym{/}   \algeffseqover{ \alpha }   ]    \rightarrow  \ottnt{A}  \vdash   \ottnt{M}    [  \ottnt{v}  /  \mathit{x}  ]    \ottsym{:}  \ottnt{A}.
              \label{eqn:subject-red:handle:handle:two}
            \end{equation}

            We show that
            \[
             \Delta  \vdash   \lambda\!  \, \mathit{y}  \ottsym{.}  \mathsf{handle} \,  \ottnt{E}  [  \mathit{y}  ]  \, \mathsf{with} \, \ottnt{H}  \ottsym{:}   \ottnt{D}    [   \text{\unboldmath$\forall$}  \,  \algeffseqoverindex{ \beta }{ \text{\unboldmath$\mathit{J}$} }   \ottsym{.} \,  \algeffseqover{ \ottnt{C_{{\mathrm{0}}}} }   \ottsym{/}   \algeffseqover{ \alpha }   ]    \rightarrow  \ottnt{A}.
            \]
            By \refdef{signature-restriction}, the occurrences of $ \algeffseqover{ \alpha } $ in the codomain type
            $\ottnt{D}$ of the type signature of $\mathsf{op}$ are only positive.
            Thus, we have $\Delta  \vdash   \ottnt{D}    [   \text{\unboldmath$\forall$}  \,  \algeffseqoverindex{ \beta }{ \text{\unboldmath$\mathit{J}$} }   \ottsym{.} \,  \algeffseqover{ \ottnt{C_{{\mathrm{0}}}} }   \ottsym{/}   \algeffseqover{ \alpha }   ]    \sqsubseteq    \text{\unboldmath$\forall$}  \,  \algeffseqoverindex{ \beta }{ \text{\unboldmath$\mathit{J}$} }   \ottsym{.} \, \ottnt{D}    [   \algeffseqover{ \ottnt{C_{{\mathrm{0}}}} }   \ottsym{/}   \algeffseqover{ \alpha }   ]  $
            by \reflem{subtyping-forall-move} (\ref{lem:subtyping-forall-move:pos})
            (note that we can suppose that
            $ \algeffseqoverindex{ \beta }{ \text{\unboldmath$\mathit{J}$} } $ do not appear free in $\ottnt{D}$).
            Thus,
            \[
             \Delta  \ottsym{,}  \mathit{y} \,  \mathord{:}  \, \ottnt{D} \,  [   \text{\unboldmath$\forall$}  \,  \algeffseqoverindex{ \beta }{ \text{\unboldmath$\mathit{J}$} }   \ottsym{.} \,  \algeffseqover{ \ottnt{C_{{\mathrm{0}}}} }   \ottsym{/}   \algeffseqover{ \alpha }   ]   \vdash  \mathit{y}  \ottsym{:}    \text{\unboldmath$\forall$}  \,  \algeffseqoverindex{ \beta }{ \text{\unboldmath$\mathit{J}$} }   \ottsym{.} \, \ottnt{D}    [   \algeffseqover{ \ottnt{C_{{\mathrm{0}}}} }   \ottsym{/}   \algeffseqover{ \alpha }   ]  
            \]
            by \T{Inst}.
            By \reflem{weakening} (\ref{lem:weakening:term}) and \Srule{Inst},
            \[
             \Delta  \ottsym{,}  \mathit{y} \,  \mathord{:}  \, \ottnt{D} \,  [   \text{\unboldmath$\forall$}  \,  \algeffseqoverindex{ \beta }{ \text{\unboldmath$\mathit{J}$} }   \ottsym{.} \,  \algeffseqover{ \ottnt{C_{{\mathrm{0}}}} }   \ottsym{/}   \algeffseqover{ \alpha }   ]   \ottsym{,}   \algeffseqoverindex{ \beta }{ \text{\unboldmath$\mathit{J}$} }   \vdash  \mathit{y}  \ottsym{:}   \ottnt{D}    [   \algeffseqover{ \ottnt{C_{{\mathrm{0}}}} }   \ottsym{/}   \algeffseqover{ \alpha }   ]  .
            \]
            By \reflem{ectx-op-typing},
            \[
             \Delta  \ottsym{,}  \mathit{y} \,  \mathord{:}  \,  \text{\unboldmath$\forall$}  \,  \algeffseqoverindex{ \beta }{ \text{\unboldmath$\mathit{J}$} }   \ottsym{.} \, \ottnt{D} \,  [   \algeffseqover{ \ottnt{C_{{\mathrm{0}}}} }   \ottsym{/}   \algeffseqover{ \alpha }   ]   \vdash   \ottnt{E}  [  \mathit{y}  ]   \ottsym{:}  \ottnt{B}.
            \]
            By \reflem{var-subtype},
            \[
             \Delta  \ottsym{,}  \mathit{y} \,  \mathord{:}  \, \ottnt{D} \,  [   \text{\unboldmath$\forall$}  \,  \algeffseqoverindex{ \beta }{ \text{\unboldmath$\mathit{J}$} }   \ottsym{.} \,  \algeffseqover{ \ottnt{C_{{\mathrm{0}}}} }   \ottsym{/}   \algeffseqover{ \alpha }   ]   \vdash   \ottnt{E}  [  \mathit{y}  ]   \ottsym{:}  \ottnt{B}.
            \]
            Thus, we have
            \[
             \Delta  \ottsym{,}  \mathit{y} \,  \mathord{:}  \, \ottnt{D} \,  [   \text{\unboldmath$\forall$}  \,  \algeffseqoverindex{ \beta }{ \text{\unboldmath$\mathit{J}$} }   \ottsym{.} \,  \algeffseqover{ \ottnt{C_{{\mathrm{0}}}} }   \ottsym{/}   \algeffseqover{ \alpha }   ]   \vdash  \mathsf{handle} \,  \ottnt{E}  [  \mathit{y}  ]  \, \mathsf{with} \, \ottnt{H}  \ottsym{:}  \ottnt{A}
            \]
            by \reflem{weakening} (\ref{lem:weakening:handler}) and \T{Handle}.
            By \T{Abs},
            \[
             \Delta  \vdash   \lambda\!  \, \mathit{y}  \ottsym{.}  \mathsf{handle} \,  \ottnt{E}  [  \mathit{y}  ]  \, \mathsf{with} \, \ottnt{H}  \ottsym{:}   \ottnt{D}    [   \text{\unboldmath$\forall$}  \,  \algeffseqoverindex{ \beta }{ \text{\unboldmath$\mathit{J}$} }   \ottsym{.} \,  \algeffseqover{ \ottnt{C_{{\mathrm{0}}}} }   \ottsym{/}   \algeffseqover{ \alpha }   ]    \rightarrow  \ottnt{A}.
            \]

            By applying \reflem{term-subst} (\ref{lem:term-subst:term}) to
            (\ref{eqn:subject-red:handle:handle:two}), we have
            \[
             \Delta  \vdash    \ottnt{M}    [  \ottnt{v}  /  \mathit{x}  ]      [   \lambda\!  \, \mathit{y}  \ottsym{.}  \mathsf{handle} \,  \ottnt{E}  [  \mathit{y}  ]  \, \mathsf{with} \, \ottnt{H}  /  \mathit{k}  ]    \ottsym{:}  \ottnt{A},
            \]
            which is what we have to show.
          \end{caseanalysis}

         \case \T{Proj1}: We have one reduction rule \R{Proj1} which can be
          applied to projection $ \pi_1 $.  Thus, we are given
          \begin{itemize}
           \item $\ottnt{M_{{\mathrm{1}}}} \,  =  \, \pi_1  \ottsym{(}  \ottnt{v_{{\mathrm{1}}}}  \ottsym{,}  \ottnt{v_{{\mathrm{2}}}}  \ottsym{)}$,
           \item $\ottnt{M_{{\mathrm{2}}}} \,  =  \, \ottnt{v_{{\mathrm{1}}}}$,
           \item $\Delta  \vdash  \pi_1  \ottsym{(}  \ottnt{v_{{\mathrm{1}}}}  \ottsym{,}  \ottnt{v_{{\mathrm{2}}}}  \ottsym{)}  \ottsym{:}  \ottnt{A}$,
           \item $\Delta  \vdash  \ottsym{(}  \ottnt{v_{{\mathrm{1}}}}  \ottsym{,}  \ottnt{v_{{\mathrm{2}}}}  \ottsym{)}  \ottsym{:}   \ottnt{A}  \times  \ottnt{B} $
          \end{itemize}
          for some $\ottnt{v_{{\mathrm{1}}}}$, $\ottnt{v_{{\mathrm{2}}}}$, and $\ottnt{B}$.
          By \reflem{val-inv-pair},
          $\Delta  \ottsym{,}   \algeffseqoverindex{ \alpha }{ \text{\unboldmath$\mathit{I}$} }   \vdash  \ottnt{v_{{\mathrm{1}}}}  \ottsym{:}  \ottnt{C_{{\mathrm{1}}}}$ and $\Delta  \ottsym{,}   \algeffseqoverindex{ \alpha }{ \text{\unboldmath$\mathit{I}$} }   \vdash  \ottnt{v_{{\mathrm{2}}}}  \ottsym{:}  \ottnt{C_{{\mathrm{2}}}}$ and
          $\Delta  \vdash    \text{\unboldmath$\forall$}  \,  \algeffseqoverindex{ \alpha }{ \text{\unboldmath$\mathit{I}$} }   \ottsym{.} \, \ottnt{C_{{\mathrm{1}}}}  \times  \ottnt{C_{{\mathrm{2}}}}   \sqsubseteq   \ottnt{A}  \times  \ottnt{B} $
          for some $ \algeffseqoverindex{ \alpha }{ \text{\unboldmath$\mathit{I}$} } $, $\ottnt{C_{{\mathrm{1}}}}$, and $\ottnt{C_{{\mathrm{2}}}}$.
          By \reflem{subtyping-inv-prod},
          there exist $ \algeffseqoverindex{ \alpha_{{\mathrm{1}}} }{ \text{\unboldmath$\mathit{I_{{\mathrm{1}}}}$} } $, $ \algeffseqoverindex{ \alpha_{{\mathrm{2}}} }{ \text{\unboldmath$\mathit{I_{{\mathrm{2}}}}$} } $, $ \algeffseqoverindex{ \beta }{ \text{\unboldmath$\mathit{J}$} } $, and $ \algeffseqoverindex{ \ottnt{D} }{ \text{\unboldmath$\mathit{I_{{\mathrm{1}}}}$} } $
          such that
          \begin{itemize}
           \item $\ottsym{\{}   \algeffseqoverindex{ \alpha }{ \text{\unboldmath$\mathit{I}$} }   \ottsym{\}} \,  =  \, \ottsym{\{}   \algeffseqoverindex{ \alpha_{{\mathrm{1}}} }{ \text{\unboldmath$\mathit{I_{{\mathrm{1}}}}$} }   \ottsym{\}} \,  \mathbin{\uplus}  \, \ottsym{\{}   \algeffseqoverindex{ \alpha_{{\mathrm{2}}} }{ \text{\unboldmath$\mathit{I_{{\mathrm{2}}}}$} }   \ottsym{\}}$,
           \item $\Delta  \ottsym{,}   \algeffseqoverindex{ \beta }{ \text{\unboldmath$\mathit{J}$} }   \vdash   \algeffseqoverindex{ \ottnt{D} }{ \text{\unboldmath$\mathit{I_{{\mathrm{1}}}}$} } $,
           \item $\Delta  \vdash    \text{\unboldmath$\forall$}  \,  \algeffseqoverindex{ \alpha_{{\mathrm{2}}} }{ \text{\unboldmath$\mathit{I_{{\mathrm{2}}}}$} }   \ottsym{.} \,  \text{\unboldmath$\forall$}  \,  \algeffseqoverindex{ \beta }{ \text{\unboldmath$\mathit{J}$} }   \ottsym{.} \, \ottnt{C_{{\mathrm{1}}}}    [   \algeffseqoverindex{ \ottnt{D} }{ \text{\unboldmath$\mathit{I_{{\mathrm{1}}}}$} }   \ottsym{/}   \algeffseqoverindex{ \alpha_{{\mathrm{1}}} }{ \text{\unboldmath$\mathit{I_{{\mathrm{1}}}}$} }   ]    \sqsubseteq  \ottnt{A}$,
           \item $\Delta  \vdash    \text{\unboldmath$\forall$}  \,  \algeffseqoverindex{ \alpha_{{\mathrm{2}}} }{ \text{\unboldmath$\mathit{I_{{\mathrm{2}}}}$} }   \ottsym{.} \,  \text{\unboldmath$\forall$}  \,  \algeffseqoverindex{ \beta }{ \text{\unboldmath$\mathit{J}$} }   \ottsym{.} \, \ottnt{C_{{\mathrm{2}}}}    [   \algeffseqoverindex{ \ottnt{D} }{ \text{\unboldmath$\mathit{I_{{\mathrm{1}}}}$} }   \ottsym{/}   \algeffseqoverindex{ \alpha_{{\mathrm{1}}} }{ \text{\unboldmath$\mathit{I_{{\mathrm{1}}}}$} }   ]    \sqsubseteq  \ottnt{B}$, and
           \item type variables in $ \algeffseqoverindex{ \beta }{ \text{\unboldmath$\mathit{J}$} } $ do not appear in $\ottnt{C_{{\mathrm{1}}}}$ and $\ottnt{C_{{\mathrm{2}}}}$.
          \end{itemize}

          We have to show that
          \[
           \Delta  \vdash  \ottnt{v_{{\mathrm{1}}}}  \ottsym{:}  \ottnt{A}.
          \]
          Since $\Delta  \vdash    \text{\unboldmath$\forall$}  \,  \algeffseqoverindex{ \alpha_{{\mathrm{2}}} }{ \text{\unboldmath$\mathit{I_{{\mathrm{2}}}}$} }   \ottsym{.} \,  \text{\unboldmath$\forall$}  \,  \algeffseqoverindex{ \beta }{ \text{\unboldmath$\mathit{J}$} }   \ottsym{.} \, \ottnt{C_{{\mathrm{1}}}}    [   \algeffseqoverindex{ \ottnt{D} }{ \text{\unboldmath$\mathit{I_{{\mathrm{1}}}}$} }   \ottsym{/}   \algeffseqoverindex{ \alpha_{{\mathrm{1}}} }{ \text{\unboldmath$\mathit{I_{{\mathrm{1}}}}$} }   ]    \sqsubseteq  \ottnt{A}$,
          it suffices to show that
          \[
           \Delta  \vdash  \ottnt{v_{{\mathrm{1}}}}  \ottsym{:}    \text{\unboldmath$\forall$}  \,  \algeffseqoverindex{ \alpha_{{\mathrm{2}}} }{ \text{\unboldmath$\mathit{I_{{\mathrm{2}}}}$} }   \ottsym{.} \,  \text{\unboldmath$\forall$}  \,  \algeffseqoverindex{ \beta }{ \text{\unboldmath$\mathit{J}$} }   \ottsym{.} \, \ottnt{C_{{\mathrm{1}}}}    [   \algeffseqoverindex{ \ottnt{D} }{ \text{\unboldmath$\mathit{I_{{\mathrm{1}}}}$} }   \ottsym{/}   \algeffseqoverindex{ \alpha_{{\mathrm{1}}} }{ \text{\unboldmath$\mathit{I_{{\mathrm{1}}}}$} }   ]  
          \]
          by \T{Inst}.
          We have $\Delta  \ottsym{,}   \algeffseqoverindex{ \beta }{ \text{\unboldmath$\mathit{J}$} }   \ottsym{,}   \algeffseqoverindex{ \alpha }{ \text{\unboldmath$\mathit{I}$} }   \vdash  \ottnt{v_{{\mathrm{1}}}}  \ottsym{:}  \ottnt{C_{{\mathrm{1}}}}$
          by \reflem{weakening} (\ref{lem:weakening:term}).
          By \reflem{ty-subst} (\ref{lem:ty-subst:term}),
          we have $\Delta  \ottsym{,}   \algeffseqoverindex{ \beta }{ \text{\unboldmath$\mathit{J}$} }   \ottsym{,}   \algeffseqoverindex{ \alpha_{{\mathrm{2}}} }{ \text{\unboldmath$\mathit{I_{{\mathrm{2}}}}$} }   \vdash  \ottnt{v_{{\mathrm{1}}}}  \ottsym{:}   \ottnt{C_{{\mathrm{1}}}}    [   \algeffseqoverindex{ \ottnt{D} }{ \text{\unboldmath$\mathit{I_{{\mathrm{1}}}}$} }   \ottsym{/}   \algeffseqoverindex{ \alpha_{{\mathrm{1}}} }{ \text{\unboldmath$\mathit{I_{{\mathrm{1}}}}$} }   ]  $.
          By \T{Gen} (and swapping $ \algeffseqoverindex{ \beta }{ \text{\unboldmath$\mathit{J}$} } $ and $ \algeffseqoverindex{ \alpha_{{\mathrm{2}}} }{ \text{\unboldmath$\mathit{I_{{\mathrm{2}}}}$} } $
          in the typing context $\Delta  \ottsym{,}   \algeffseqoverindex{ \beta }{ \text{\unboldmath$\mathit{J}$} }   \ottsym{,}   \algeffseqoverindex{ \alpha_{{\mathrm{2}}} }{ \text{\unboldmath$\mathit{I_{{\mathrm{2}}}}$} } $),
          we have
          \[
           \Delta  \vdash  \ottnt{v_{{\mathrm{1}}}}  \ottsym{:}    \text{\unboldmath$\forall$}  \,  \algeffseqoverindex{ \alpha_{{\mathrm{2}}} }{ \text{\unboldmath$\mathit{I_{{\mathrm{2}}}}$} }   \ottsym{.} \,  \text{\unboldmath$\forall$}  \,  \algeffseqoverindex{ \beta }{ \text{\unboldmath$\mathit{J}$} }   \ottsym{.} \, \ottnt{C_{{\mathrm{1}}}}    [   \algeffseqoverindex{ \ottnt{D} }{ \text{\unboldmath$\mathit{I_{{\mathrm{1}}}}$} }   \ottsym{/}   \algeffseqoverindex{ \alpha_{{\mathrm{1}}} }{ \text{\unboldmath$\mathit{I_{{\mathrm{1}}}}$} }   ]  .
          \]

         \case \T{Proj2}: Similar to the case for \T{Proj1}.

         \case \T{Case}:
          We have two reduction rules which can be applied to
          case expressions.
          \begin{caseanalysis}
           \case \R{CaseL}:
            We are given
            \begin{itemize}
             \item $\ottnt{M_{{\mathrm{1}}}} \,  =  \, \mathsf{case} \, \ottsym{(}  \mathsf{inl} \, \ottnt{v}  \ottsym{)} \, \mathsf{of} \, \mathsf{inl} \, \mathit{x}  \rightarrow  \ottnt{M'_{{\mathrm{1}}}}  \ottsym{;} \, \mathsf{inr} \, \mathit{y}  \rightarrow  \ottnt{M'_{{\mathrm{2}}}}$,
             \item $\ottnt{M_{{\mathrm{2}}}} \,  =  \,  \ottnt{M'_{{\mathrm{1}}}}    [  \ottnt{v}  /  \mathit{x}  ]  $,
             \item $\Delta  \vdash  \mathsf{case} \, \ottsym{(}  \mathsf{inl} \, \ottnt{v}  \ottsym{)} \, \mathsf{of} \, \mathsf{inl} \, \mathit{x}  \rightarrow  \ottnt{M'_{{\mathrm{1}}}}  \ottsym{;} \, \mathsf{inr} \, \mathit{y}  \rightarrow  \ottnt{M'_{{\mathrm{2}}}}  \ottsym{:}  \ottnt{A}$,
             \item $\Delta  \vdash  \mathsf{inl} \, \ottnt{v}  \ottsym{:}   \ottnt{B_{{\mathrm{1}}}}  +  \ottnt{B_{{\mathrm{2}}}} $,
             \item $\Delta  \ottsym{,}  \mathit{x} \,  \mathord{:}  \, \ottnt{B_{{\mathrm{1}}}}  \vdash  \ottnt{M'_{{\mathrm{1}}}}  \ottsym{:}  \ottnt{A}$, and
             \item $\Delta  \ottsym{,}  \mathit{x} \,  \mathord{:}  \, \ottnt{B_{{\mathrm{2}}}}  \vdash  \ottnt{M'_{{\mathrm{2}}}}  \ottsym{:}  \ottnt{A}$
            \end{itemize}
            for some $\ottnt{v}$, $\mathit{x}$, $\mathit{y}$, $\ottnt{M'_{{\mathrm{1}}}}$,, $\ottnt{M'_{{\mathrm{2}}}}$,
            $\ottnt{B_{{\mathrm{1}}}}$, and $\ottnt{B_{{\mathrm{2}}}}$.
            By \reflem{val-inv-inl},
            $\Delta  \ottsym{,}   \algeffseqoverindex{ \alpha }{ \text{\unboldmath$\mathit{I}$} }   \vdash  \ottnt{v}  \ottsym{:}  \ottnt{C_{{\mathrm{1}}}}$ and
            $\Delta  \vdash    \text{\unboldmath$\forall$}  \,  \algeffseqoverindex{ \alpha }{ \text{\unboldmath$\mathit{I}$} }   \ottsym{.} \, \ottnt{C_{{\mathrm{1}}}}  +  \ottnt{C_{{\mathrm{2}}}}   \sqsubseteq   \ottnt{B_{{\mathrm{1}}}}  +  \ottnt{B_{{\mathrm{2}}}} $
            for some $ \algeffseqoverindex{ \alpha }{ \text{\unboldmath$\mathit{I}$} } $, $\ottnt{C_{{\mathrm{1}}}}$, and $\ottnt{C_{{\mathrm{2}}}}$.
            By \reflem{subtyping-inv-sum},
            there exist $ \algeffseqoverindex{ \alpha_{{\mathrm{1}}} }{ \text{\unboldmath$\mathit{I_{{\mathrm{1}}}}$} } $, $ \algeffseqoverindex{ \alpha_{{\mathrm{2}}} }{ \text{\unboldmath$\mathit{I_{{\mathrm{2}}}}$} } $, $ \algeffseqoverindex{ \beta }{ \text{\unboldmath$\mathit{J}$} } $, and $ \algeffseqoverindex{ \ottnt{D} }{ \text{\unboldmath$\mathit{I_{{\mathrm{1}}}}$} } $
            such that
            \begin{itemize}
             \item $\ottsym{\{}   \algeffseqoverindex{ \alpha }{ \text{\unboldmath$\mathit{I}$} }   \ottsym{\}} \,  =  \, \ottsym{\{}   \algeffseqoverindex{ \alpha_{{\mathrm{1}}} }{ \text{\unboldmath$\mathit{I_{{\mathrm{1}}}}$} }   \ottsym{\}} \,  \mathbin{\uplus}  \, \ottsym{\{}   \algeffseqoverindex{ \alpha_{{\mathrm{2}}} }{ \text{\unboldmath$\mathit{I_{{\mathrm{2}}}}$} }   \ottsym{\}}$,
             \item $\Delta  \ottsym{,}   \algeffseqoverindex{ \beta }{ \text{\unboldmath$\mathit{J}$} }   \vdash   \algeffseqoverindex{ \ottnt{D} }{ \text{\unboldmath$\mathit{I_{{\mathrm{1}}}}$} } $,
             \item $\Delta  \vdash    \text{\unboldmath$\forall$}  \,  \algeffseqoverindex{ \alpha_{{\mathrm{2}}} }{ \text{\unboldmath$\mathit{I_{{\mathrm{2}}}}$} }   \ottsym{.} \,  \text{\unboldmath$\forall$}  \,  \algeffseqoverindex{ \beta }{ \text{\unboldmath$\mathit{J}$} }   \ottsym{.} \, \ottnt{C_{{\mathrm{1}}}}    [   \algeffseqoverindex{ \ottnt{D} }{ \text{\unboldmath$\mathit{I_{{\mathrm{1}}}}$} }   \ottsym{/}   \algeffseqoverindex{ \alpha_{{\mathrm{1}}} }{ \text{\unboldmath$\mathit{I_{{\mathrm{1}}}}$} }   ]    \sqsubseteq  \ottnt{B_{{\mathrm{1}}}}$,
             \item $\Delta  \vdash    \text{\unboldmath$\forall$}  \,  \algeffseqoverindex{ \alpha_{{\mathrm{2}}} }{ \text{\unboldmath$\mathit{I_{{\mathrm{2}}}}$} }   \ottsym{.} \,  \text{\unboldmath$\forall$}  \,  \algeffseqoverindex{ \beta }{ \text{\unboldmath$\mathit{J}$} }   \ottsym{.} \, \ottnt{C_{{\mathrm{2}}}}    [   \algeffseqoverindex{ \ottnt{D} }{ \text{\unboldmath$\mathit{I_{{\mathrm{1}}}}$} }   \ottsym{/}   \algeffseqoverindex{ \alpha_{{\mathrm{1}}} }{ \text{\unboldmath$\mathit{I_{{\mathrm{1}}}}$} }   ]    \sqsubseteq  \ottnt{B_{{\mathrm{2}}}}$, and
             \item type variables in $ \algeffseqoverindex{ \beta }{ \text{\unboldmath$\mathit{J}$} } $ do not appear in $\ottnt{C_{{\mathrm{1}}}}$ and $\ottnt{C_{{\mathrm{2}}}}$.
            \end{itemize}

            We first show that
            \[
             \Delta  \vdash  \ottnt{v}  \ottsym{:}  \ottnt{B_{{\mathrm{1}}}}.
            \]
            Since $\Delta  \vdash    \text{\unboldmath$\forall$}  \,  \algeffseqoverindex{ \alpha_{{\mathrm{2}}} }{ \text{\unboldmath$\mathit{I_{{\mathrm{2}}}}$} }   \ottsym{.} \,  \text{\unboldmath$\forall$}  \,  \algeffseqoverindex{ \beta }{ \text{\unboldmath$\mathit{J}$} }   \ottsym{.} \, \ottnt{C_{{\mathrm{1}}}}    [   \algeffseqoverindex{ \ottnt{D} }{ \text{\unboldmath$\mathit{I_{{\mathrm{1}}}}$} }   \ottsym{/}   \algeffseqoverindex{ \alpha_{{\mathrm{1}}} }{ \text{\unboldmath$\mathit{I_{{\mathrm{1}}}}$} }   ]    \sqsubseteq  \ottnt{B_{{\mathrm{1}}}}$,
            it suffices to show that
            \[
             \Delta  \vdash  \ottnt{v}  \ottsym{:}    \text{\unboldmath$\forall$}  \,  \algeffseqoverindex{ \alpha_{{\mathrm{2}}} }{ \text{\unboldmath$\mathit{I_{{\mathrm{2}}}}$} }   \ottsym{.} \,  \text{\unboldmath$\forall$}  \,  \algeffseqoverindex{ \beta }{ \text{\unboldmath$\mathit{J}$} }   \ottsym{.} \, \ottnt{C_{{\mathrm{1}}}}    [   \algeffseqoverindex{ \ottnt{D} }{ \text{\unboldmath$\mathit{I_{{\mathrm{1}}}}$} }   \ottsym{/}   \algeffseqoverindex{ \alpha_{{\mathrm{1}}} }{ \text{\unboldmath$\mathit{I_{{\mathrm{1}}}}$} }   ]  
            \]
            by \T{Inst}.
            We have $\Delta  \ottsym{,}   \algeffseqoverindex{ \beta }{ \text{\unboldmath$\mathit{J}$} }   \ottsym{,}   \algeffseqoverindex{ \alpha }{ \text{\unboldmath$\mathit{I}$} }   \vdash  \ottnt{v_{{\mathrm{1}}}}  \ottsym{:}  \ottnt{C_{{\mathrm{1}}}}$
            by \reflem{weakening} (\ref{lem:weakening:term}).
            By \reflem{ty-subst} (\ref{lem:ty-subst:term}),
            we have $\Delta  \ottsym{,}   \algeffseqoverindex{ \beta }{ \text{\unboldmath$\mathit{J}$} }   \ottsym{,}   \algeffseqoverindex{ \alpha_{{\mathrm{2}}} }{ \text{\unboldmath$\mathit{I_{{\mathrm{2}}}}$} }   \vdash  \ottnt{v_{{\mathrm{1}}}}  \ottsym{:}   \ottnt{C_{{\mathrm{1}}}}    [   \algeffseqoverindex{ \ottnt{D} }{ \text{\unboldmath$\mathit{I_{{\mathrm{1}}}}$} }   \ottsym{/}   \algeffseqoverindex{ \alpha_{{\mathrm{1}}} }{ \text{\unboldmath$\mathit{I_{{\mathrm{1}}}}$} }   ]  $.
            By \T{Gen} (and swapping $ \algeffseqoverindex{ \beta }{ \text{\unboldmath$\mathit{J}$} } $ and $ \algeffseqoverindex{ \alpha_{{\mathrm{2}}} }{ \text{\unboldmath$\mathit{I_{{\mathrm{2}}}}$} } $
            in the typing context $\Delta  \ottsym{,}   \algeffseqoverindex{ \beta }{ \text{\unboldmath$\mathit{J}$} }   \ottsym{,}   \algeffseqoverindex{ \alpha_{{\mathrm{2}}} }{ \text{\unboldmath$\mathit{I_{{\mathrm{2}}}}$} } $),
            we have
            \[
             \Delta  \vdash  \ottnt{v_{{\mathrm{1}}}}  \ottsym{:}    \text{\unboldmath$\forall$}  \,  \algeffseqoverindex{ \alpha_{{\mathrm{2}}} }{ \text{\unboldmath$\mathit{I_{{\mathrm{2}}}}$} }   \ottsym{.} \,  \text{\unboldmath$\forall$}  \,  \algeffseqoverindex{ \beta }{ \text{\unboldmath$\mathit{J}$} }   \ottsym{.} \, \ottnt{C_{{\mathrm{1}}}}    [   \algeffseqoverindex{ \ottnt{D} }{ \text{\unboldmath$\mathit{I_{{\mathrm{1}}}}$} }   \ottsym{/}   \algeffseqoverindex{ \alpha_{{\mathrm{1}}} }{ \text{\unboldmath$\mathit{I_{{\mathrm{1}}}}$} }   ]  .
            \]

            Since $\Delta  \ottsym{,}  \mathit{x} \,  \mathord{:}  \, \ottnt{B_{{\mathrm{1}}}}  \vdash  \ottnt{M'_{{\mathrm{1}}}}  \ottsym{:}  \ottnt{A}$,
            we have
            \[
             \Delta  \vdash   \ottnt{M'_{{\mathrm{1}}}}    [  \ottnt{v}  /  \mathit{x}  ]    \ottsym{:}  \ottnt{A}
            \]
            by \reflem{term-subst} (\ref{lem:term-subst:term}).

           \case \R{CaseR}: Similar to the case for \R{CaseR}, using
            \reflem{val-inv-inr} instead of \reflem{val-inv-inl}.
          \end{caseanalysis}

         \case \T{CaseList}:
          We have two reduction rules which can be applied to
          case expressions for lists.
          \begin{caseanalysis}
           \case \R{Nil}: Obvious.
           \case \R{Cons}:
            We are given
            \begin{itemize}
             \item $\ottnt{M_{{\mathrm{1}}}} \,  =  \, \mathsf{case} \, \ottsym{(}  \mathsf{cons} \, \ottnt{v}  \ottsym{)} \, \mathsf{of} \, \mathsf{nil} \, \rightarrow  \ottnt{M'_{{\mathrm{1}}}}  \ottsym{;} \, \mathsf{cons} \, \mathit{x}  \rightarrow  \ottnt{M'_{{\mathrm{2}}}}$,
             \item $\ottnt{M_{{\mathrm{2}}}} \,  =  \,  \ottnt{M'_{{\mathrm{2}}}}    [  \ottnt{v}  /  \mathit{x}  ]  $,
             \item $\Delta  \vdash  \mathsf{case} \, \ottsym{(}  \mathsf{cons} \, \ottnt{v}  \ottsym{)} \, \mathsf{of} \, \mathsf{nil} \, \rightarrow  \ottnt{M'_{{\mathrm{1}}}}  \ottsym{;} \, \mathsf{cons} \, \mathit{y}  \rightarrow  \ottnt{M'_{{\mathrm{2}}}}  \ottsym{:}  \ottnt{A}$,
             \item $\Delta  \vdash  \mathsf{cons} \, \ottnt{v}  \ottsym{:}   \ottnt{B}  \, \mathsf{list} $, and
             \item $\Delta  \ottsym{,}  \mathit{x} \,  \mathord{:}  \,   \ottnt{B}  \times  \ottnt{B}   \, \mathsf{list}   \vdash  \ottnt{M'_{{\mathrm{2}}}}  \ottsym{:}  \ottnt{A}$
            \end{itemize}
            for some $\ottnt{v}$, $\mathit{x}$, $\ottnt{M'_{{\mathrm{1}}}}$,, $\ottnt{M'_{{\mathrm{2}}}}$, and $\ottnt{B}$.
            By \reflem{val-inv-cons},
            $\Delta  \ottsym{,}   \algeffseqoverindex{ \alpha }{ \text{\unboldmath$\mathit{I}$} }   \vdash  \ottnt{v}  \ottsym{:}    \ottnt{C}  \times  \ottnt{C}   \, \mathsf{list} $ and
            $\Delta  \vdash    \text{\unboldmath$\forall$}  \,  \algeffseqoverindex{ \alpha }{ \text{\unboldmath$\mathit{I}$} }   \ottsym{.} \, \ottnt{C}  \, \mathsf{list}   \sqsubseteq   \ottnt{B}  \, \mathsf{list} $
            for some $ \algeffseqoverindex{ \alpha }{ \text{\unboldmath$\mathit{I}$} } $ and $\ottnt{C}$.
            By \reflem{subtyping-inv-list},
            there exist $ \algeffseqoverindex{ \alpha_{{\mathrm{1}}} }{ \text{\unboldmath$\mathit{I_{{\mathrm{1}}}}$} } $, $ \algeffseqoverindex{ \alpha_{{\mathrm{2}}} }{ \text{\unboldmath$\mathit{I_{{\mathrm{2}}}}$} } $, $ \algeffseqoverindex{ \beta }{ \text{\unboldmath$\mathit{J}$} } $, and $ \algeffseqoverindex{ \ottnt{D} }{ \text{\unboldmath$\mathit{I_{{\mathrm{1}}}}$} } $
            such that
            \begin{itemize}
             \item $\ottsym{\{}   \algeffseqoverindex{ \alpha }{ \text{\unboldmath$\mathit{I}$} }   \ottsym{\}} \,  =  \, \ottsym{\{}   \algeffseqoverindex{ \alpha_{{\mathrm{1}}} }{ \text{\unboldmath$\mathit{I_{{\mathrm{1}}}}$} }   \ottsym{\}} \,  \mathbin{\uplus}  \, \ottsym{\{}   \algeffseqoverindex{ \alpha_{{\mathrm{2}}} }{ \text{\unboldmath$\mathit{I_{{\mathrm{2}}}}$} }   \ottsym{\}}$,
             \item $\Delta  \ottsym{,}   \algeffseqoverindex{ \beta }{ \text{\unboldmath$\mathit{J}$} }   \vdash   \algeffseqoverindex{ \ottnt{D} }{ \text{\unboldmath$\mathit{I_{{\mathrm{1}}}}$} } $,
             \item $\Delta  \vdash    \text{\unboldmath$\forall$}  \,  \algeffseqoverindex{ \alpha_{{\mathrm{2}}} }{ \text{\unboldmath$\mathit{I_{{\mathrm{2}}}}$} }   \ottsym{.} \,  \text{\unboldmath$\forall$}  \,  \algeffseqoverindex{ \beta }{ \text{\unboldmath$\mathit{J}$} }   \ottsym{.} \, \ottnt{C}    [   \algeffseqoverindex{ \ottnt{D} }{ \text{\unboldmath$\mathit{I_{{\mathrm{1}}}}$} }   \ottsym{/}   \algeffseqoverindex{ \alpha_{{\mathrm{1}}} }{ \text{\unboldmath$\mathit{I_{{\mathrm{1}}}}$} }   ]    \sqsubseteq  \ottnt{B}$, and
             \item type variables in $ \algeffseqoverindex{ \beta }{ \text{\unboldmath$\mathit{J}$} } $ do not appear in $\ottnt{C}$.
            \end{itemize}

            We first show that
            \[
             \Delta  \vdash     \text{\unboldmath$\forall$}  \,  \algeffseqoverindex{ \alpha_{{\mathrm{2}}} }{ \text{\unboldmath$\mathit{I_{{\mathrm{2}}}}$} }   \ottsym{.} \,  \text{\unboldmath$\forall$}  \,  \algeffseqoverindex{ \beta }{ \text{\unboldmath$\mathit{J}$} }   \ottsym{.} \,    \ottnt{C}    [   \algeffseqoverindex{ \ottnt{D} }{ \text{\unboldmath$\mathit{I_{{\mathrm{1}}}}$} }   \ottsym{/}   \algeffseqoverindex{ \alpha_{{\mathrm{1}}} }{ \text{\unboldmath$\mathit{I_{{\mathrm{1}}}}$} }   ]      \times     \ottnt{C}    [   \algeffseqoverindex{ \ottnt{D} }{ \text{\unboldmath$\mathit{I_{{\mathrm{1}}}}$} }   \ottsym{/}   \algeffseqoverindex{ \alpha_{{\mathrm{1}}} }{ \text{\unboldmath$\mathit{I_{{\mathrm{1}}}}$} }   ]       \, \mathsf{list}   \sqsubseteq    \ottnt{B}  \times  \ottnt{B}   \, \mathsf{list} .
            \]
            Since $\Delta  \vdash    \text{\unboldmath$\forall$}  \,  \algeffseqoverindex{ \alpha_{{\mathrm{2}}} }{ \text{\unboldmath$\mathit{I_{{\mathrm{2}}}}$} }   \ottsym{.} \,  \text{\unboldmath$\forall$}  \,  \algeffseqoverindex{ \beta }{ \text{\unboldmath$\mathit{J}$} }   \ottsym{.} \, \ottnt{C}    [   \algeffseqoverindex{ \ottnt{D} }{ \text{\unboldmath$\mathit{I_{{\mathrm{1}}}}$} }   \ottsym{/}   \algeffseqoverindex{ \alpha_{{\mathrm{1}}} }{ \text{\unboldmath$\mathit{I_{{\mathrm{1}}}}$} }   ]    \sqsubseteq  \ottnt{B}$,
            we have
            \[
             \Delta  \vdash   \ottsym{(}    \text{\unboldmath$\forall$}  \,  \algeffseqoverindex{ \alpha_{{\mathrm{2}}} }{ \text{\unboldmath$\mathit{I_{{\mathrm{2}}}}$} }   \ottsym{.} \,  \text{\unboldmath$\forall$}  \,  \algeffseqoverindex{ \beta }{ \text{\unboldmath$\mathit{J}$} }   \ottsym{.} \, \ottnt{C}    [   \algeffseqoverindex{ \ottnt{D} }{ \text{\unboldmath$\mathit{I_{{\mathrm{1}}}}$} }   \ottsym{/}   \algeffseqoverindex{ \alpha_{{\mathrm{1}}} }{ \text{\unboldmath$\mathit{I_{{\mathrm{1}}}}$} }   ]    \ottsym{)}  \, \mathsf{list}   \sqsubseteq   \ottnt{B}  \, \mathsf{list} 
            \]
            by \Srule{List}.  We also have
            \[
             \Delta  \vdash     \text{\unboldmath$\forall$}  \,  \algeffseqoverindex{ \alpha_{{\mathrm{2}}} }{ \text{\unboldmath$\mathit{I_{{\mathrm{2}}}}$} }   \ottsym{.} \,  \text{\unboldmath$\forall$}  \,  \algeffseqoverindex{ \beta }{ \text{\unboldmath$\mathit{J}$} }   \ottsym{.} \, \ottnt{C}    [   \algeffseqoverindex{ \ottnt{D} }{ \text{\unboldmath$\mathit{I_{{\mathrm{1}}}}$} }   \ottsym{/}   \algeffseqoverindex{ \alpha_{{\mathrm{1}}} }{ \text{\unboldmath$\mathit{I_{{\mathrm{1}}}}$} }   ]    \, \mathsf{list}   \sqsubseteq   \ottsym{(}    \text{\unboldmath$\forall$}  \,  \algeffseqoverindex{ \alpha_{{\mathrm{2}}} }{ \text{\unboldmath$\mathit{I_{{\mathrm{2}}}}$} }   \ottsym{.} \,  \text{\unboldmath$\forall$}  \,  \algeffseqoverindex{ \beta }{ \text{\unboldmath$\mathit{J}$} }   \ottsym{.} \, \ottnt{C}    [   \algeffseqoverindex{ \ottnt{D} }{ \text{\unboldmath$\mathit{I_{{\mathrm{1}}}}$} }   \ottsym{/}   \algeffseqoverindex{ \alpha_{{\mathrm{1}}} }{ \text{\unboldmath$\mathit{I_{{\mathrm{1}}}}$} }   ]    \ottsym{)}  \, \mathsf{list} 
            \]
            by \Srule{DList}.  Thus, by \Srule{Trans}, we have
            \[
             \Delta  \vdash     \text{\unboldmath$\forall$}  \,  \algeffseqoverindex{ \alpha_{{\mathrm{2}}} }{ \text{\unboldmath$\mathit{I_{{\mathrm{2}}}}$} }   \ottsym{.} \,  \text{\unboldmath$\forall$}  \,  \algeffseqoverindex{ \beta }{ \text{\unboldmath$\mathit{J}$} }   \ottsym{.} \, \ottnt{C}    [   \algeffseqoverindex{ \ottnt{D} }{ \text{\unboldmath$\mathit{I_{{\mathrm{1}}}}$} }   \ottsym{/}   \algeffseqoverindex{ \alpha_{{\mathrm{1}}} }{ \text{\unboldmath$\mathit{I_{{\mathrm{1}}}}$} }   ]    \, \mathsf{list}   \sqsubseteq   \ottnt{B}  \, \mathsf{list} .
            \]
            By \Srule{Prod},
            \[
             \Delta  \vdash   \ottsym{(}    \text{\unboldmath$\forall$}  \,  \algeffseqoverindex{ \alpha_{{\mathrm{2}}} }{ \text{\unboldmath$\mathit{I_{{\mathrm{2}}}}$} }   \ottsym{.} \,  \text{\unboldmath$\forall$}  \,  \algeffseqoverindex{ \beta }{ \text{\unboldmath$\mathit{J}$} }   \ottsym{.} \, \ottnt{C}    [   \algeffseqoverindex{ \ottnt{D} }{ \text{\unboldmath$\mathit{I_{{\mathrm{1}}}}$} }   \ottsym{/}   \algeffseqoverindex{ \alpha_{{\mathrm{1}}} }{ \text{\unboldmath$\mathit{I_{{\mathrm{1}}}}$} }   ]    \ottsym{)}  \times  \ottsym{(}   \text{\unboldmath$\forall$}  \,  \algeffseqoverindex{ \alpha_{{\mathrm{2}}} }{ \text{\unboldmath$\mathit{I_{{\mathrm{2}}}}$} }   \ottsym{.} \,  \text{\unboldmath$\forall$}  \,  \algeffseqoverindex{ \beta }{ \text{\unboldmath$\mathit{J}$} }   \ottsym{.} \,   \ottnt{C}    [   \algeffseqoverindex{ \ottnt{D} }{ \text{\unboldmath$\mathit{I_{{\mathrm{1}}}}$} }   \ottsym{/}   \algeffseqoverindex{ \alpha_{{\mathrm{1}}} }{ \text{\unboldmath$\mathit{I_{{\mathrm{1}}}}$} }   ]    \, \mathsf{list}   \ottsym{)}   \sqsubseteq    \ottnt{B}  \times  \ottnt{B}   \, \mathsf{list} .
            \]
            By \Srule{DProd} and \Srule{Trans}, we have
            \begin{equation}
             \Delta  \vdash     \text{\unboldmath$\forall$}  \,  \algeffseqoverindex{ \alpha_{{\mathrm{2}}} }{ \text{\unboldmath$\mathit{I_{{\mathrm{2}}}}$} }   \ottsym{.} \,  \text{\unboldmath$\forall$}  \,  \algeffseqoverindex{ \beta }{ \text{\unboldmath$\mathit{J}$} }   \ottsym{.} \,    \ottnt{C}    [   \algeffseqoverindex{ \ottnt{D} }{ \text{\unboldmath$\mathit{I_{{\mathrm{1}}}}$} }   \ottsym{/}   \algeffseqoverindex{ \alpha_{{\mathrm{1}}} }{ \text{\unboldmath$\mathit{I_{{\mathrm{1}}}}$} }   ]      \times     \ottnt{C}    [   \algeffseqoverindex{ \ottnt{D} }{ \text{\unboldmath$\mathit{I_{{\mathrm{1}}}}$} }   \ottsym{/}   \algeffseqoverindex{ \alpha_{{\mathrm{1}}} }{ \text{\unboldmath$\mathit{I_{{\mathrm{1}}}}$} }   ]       \, \mathsf{list}   \sqsubseteq    \ottnt{B}  \times  \ottnt{B}   \, \mathsf{list} 
              \label{eqn:subject-red:caselist:one}
            \end{equation}

            Next, we show that
            \[
             \Delta  \vdash  \ottnt{v}  \ottsym{:}    \ottnt{B}  \times  \ottnt{B}   \, \mathsf{list} .
            \]
            By \T{Inst} with (\ref{eqn:subject-red:caselist:one}),
            it suffices to show that
            \[
             \Delta  \vdash  \ottnt{v}  \ottsym{:}     \text{\unboldmath$\forall$}  \,  \algeffseqoverindex{ \alpha_{{\mathrm{2}}} }{ \text{\unboldmath$\mathit{I_{{\mathrm{2}}}}$} }   \ottsym{.} \,  \text{\unboldmath$\forall$}  \,  \algeffseqoverindex{ \beta }{ \text{\unboldmath$\mathit{J}$} }   \ottsym{.} \,    \ottnt{C}    [   \algeffseqoverindex{ \ottnt{D} }{ \text{\unboldmath$\mathit{I_{{\mathrm{1}}}}$} }   \ottsym{/}   \algeffseqoverindex{ \alpha_{{\mathrm{1}}} }{ \text{\unboldmath$\mathit{I_{{\mathrm{1}}}}$} }   ]      \times     \ottnt{C}    [   \algeffseqoverindex{ \ottnt{D} }{ \text{\unboldmath$\mathit{I_{{\mathrm{1}}}}$} }   \ottsym{/}   \algeffseqoverindex{ \alpha_{{\mathrm{1}}} }{ \text{\unboldmath$\mathit{I_{{\mathrm{1}}}}$} }   ]       \, \mathsf{list} .
            \]
            We have $\Delta  \ottsym{,}   \algeffseqoverindex{ \beta }{ \text{\unboldmath$\mathit{J}$} }   \ottsym{,}   \algeffseqoverindex{ \alpha }{ \text{\unboldmath$\mathit{I}$} }   \vdash  \ottnt{v}  \ottsym{:}    \ottnt{C}  \times  \ottnt{C}   \, \mathsf{list} $
            by \reflem{weakening} (\ref{lem:weakening:term}).
            By \reflem{ty-subst} (\ref{lem:ty-subst:term}),
            we have $\Delta  \ottsym{,}   \algeffseqoverindex{ \beta }{ \text{\unboldmath$\mathit{J}$} }   \ottsym{,}   \algeffseqoverindex{ \alpha_{{\mathrm{2}}} }{ \text{\unboldmath$\mathit{I_{{\mathrm{2}}}}$} }   \vdash  \ottnt{v}  \ottsym{:}      \ottnt{C}    [   \algeffseqoverindex{ \ottnt{D} }{ \text{\unboldmath$\mathit{I_{{\mathrm{1}}}}$} }   \ottsym{/}   \algeffseqoverindex{ \alpha_{{\mathrm{1}}} }{ \text{\unboldmath$\mathit{I_{{\mathrm{1}}}}$} }   ]    \times  \ottnt{C}     [   \algeffseqoverindex{ \ottnt{D} }{ \text{\unboldmath$\mathit{I_{{\mathrm{1}}}}$} }   \ottsym{/}   \algeffseqoverindex{ \alpha_{{\mathrm{1}}} }{ \text{\unboldmath$\mathit{I_{{\mathrm{1}}}}$} }   ]    \, \mathsf{list} $.
            By \T{Gen} (and swapping $ \algeffseqoverindex{ \beta }{ \text{\unboldmath$\mathit{J}$} } $ and $ \algeffseqoverindex{ \alpha_{{\mathrm{2}}} }{ \text{\unboldmath$\mathit{I_{{\mathrm{2}}}}$} } $
            in the typing context $\Delta  \ottsym{,}   \algeffseqoverindex{ \beta }{ \text{\unboldmath$\mathit{J}$} }   \ottsym{,}   \algeffseqoverindex{ \alpha_{{\mathrm{2}}} }{ \text{\unboldmath$\mathit{I_{{\mathrm{2}}}}$} } $),
            we have
            \[
             \Delta  \vdash  \ottnt{v}  \ottsym{:}     \text{\unboldmath$\forall$}  \,  \algeffseqoverindex{ \alpha_{{\mathrm{2}}} }{ \text{\unboldmath$\mathit{I_{{\mathrm{2}}}}$} }   \ottsym{.} \,  \text{\unboldmath$\forall$}  \,  \algeffseqoverindex{ \beta }{ \text{\unboldmath$\mathit{J}$} }   \ottsym{.} \,    \ottnt{C}    [   \algeffseqoverindex{ \ottnt{D} }{ \text{\unboldmath$\mathit{I_{{\mathrm{1}}}}$} }   \ottsym{/}   \algeffseqoverindex{ \alpha_{{\mathrm{1}}} }{ \text{\unboldmath$\mathit{I_{{\mathrm{1}}}}$} }   ]      \times     \ottnt{C}    [   \algeffseqoverindex{ \ottnt{D} }{ \text{\unboldmath$\mathit{I_{{\mathrm{1}}}}$} }   \ottsym{/}   \algeffseqoverindex{ \alpha_{{\mathrm{1}}} }{ \text{\unboldmath$\mathit{I_{{\mathrm{1}}}}$} }   ]       \, \mathsf{list} .
            \]

            Since $\Delta  \ottsym{,}  \mathit{x} \,  \mathord{:}  \,   \ottnt{B}  \times  \ottnt{B}   \, \mathsf{list}   \vdash  \ottnt{M'_{{\mathrm{2}}}}  \ottsym{:}  \ottnt{A}$,
            we have
            \[
             \Delta  \vdash   \ottnt{M'_{{\mathrm{2}}}}    [  \ottnt{v}  /  \mathit{x}  ]    \ottsym{:}  \ottnt{A}
            \]
            by \reflem{term-subst} (\ref{lem:term-subst:term}).
          \end{caseanalysis}

         \case \T{Fix}: We have one reduction rule \R{Fix} which can be
          applied to the fixed-point operator.  The proof is straightforward
          with \reflem{term-subst} (\ref{lem:term-subst:term}) and \T{Abs}.

        \end{caseanalysis}

  \item Suppose that $\Delta  \vdash  \ottnt{M_{{\mathrm{1}}}}  \ottsym{:}  \ottnt{A}$ and $\ottnt{M_{{\mathrm{1}}}}  \longrightarrow  \ottnt{M_{{\mathrm{2}}}}$.
        By definition, there exist some $\ottnt{E}$, $\ottnt{M'_{{\mathrm{1}}}}$, and $\ottnt{M'_{{\mathrm{2}}}}$ such that
        $\ottnt{M_{{\mathrm{1}}}} \,  =  \,  \ottnt{E}  [  \ottnt{M'_{{\mathrm{1}}}}  ] $, $\ottnt{M_{{\mathrm{2}}}} \,  =  \,  \ottnt{E}  [  \ottnt{M'_{{\mathrm{2}}}}  ] $, and $\ottnt{M'_{{\mathrm{1}}}}  \rightsquigarrow  \ottnt{M'_{{\mathrm{2}}}}$.
        The proof proceeds by induction on the typing derivation of
        for $\ottnt{M_{{\mathrm{1}}}} =  \ottnt{E}  [  \ottnt{M'_{{\mathrm{1}}}}  ] $.
        If $\ottnt{E} \,  =  \,  [] $, then we have the conclusion by the first case.
        In what follows, we suppose that $\ottnt{E} \,  \not=  \,  [] $.
        By case analysis on the typing rule applied last to derive
        $\Delta  \vdash   \ottnt{E}  [  \ottnt{M'_{{\mathrm{1}}}}  ]   \ottsym{:}  \ottnt{A}$.
        \begin{caseanalysis}
         \case \T{Var}, \T{Const}, \T{Abs}, \T{Nil}, and \T{Fix}:
          Contradictory because $\ottnt{E}$ has to be $ [] $.

         \case \T{App}:
          By case analysis on $\ottnt{E}$.
          \begin{caseanalysis}
           \case $\ottnt{E} \,  =  \, \ottnt{E'} \, \ottnt{M}$:
           We are given
           \begin{itemize}
            \item $\Delta  \vdash   \ottnt{E'}  [  \ottnt{M'_{{\mathrm{1}}}}  ]   \ottsym{:}  \ottnt{B}  \rightarrow  \ottnt{A}$ and
            \item $\Delta  \vdash  \ottnt{M}  \ottsym{:}  \ottnt{B}$
           \end{itemize}
           for some $\ottnt{B}$.
           By the IH, $\Delta  \vdash   \ottnt{E'}  [  \ottnt{M'_{{\mathrm{2}}}}  ]   \ottsym{:}  \ottnt{B}  \rightarrow  \ottnt{A}$.
           Since $\ottnt{M_{{\mathrm{2}}}} \,  =  \,  \ottnt{E'}  [  \ottnt{M'_{{\mathrm{2}}}}  ]  \, \ottnt{M}$,
           we have the conclusion by \T{App}.

          \case $\ottnt{E} \,  =  \, \ottnt{v} \, \ottnt{E'}$: By the IH.
         \end{caseanalysis}

         \case \T{Gen}: By the IH.
         \case \T{Inst}: By the IH.
         \case \T{Op}: By the IH.
         \case \T{Handle}: By the IH.
         \case \T{Pair}: By the IH.
         \case \T{Proj1}: By the IH.
         \case \T{Proj2}: By the IH.
         \case \T{InL}: By the IH.
         \case \T{InR}: By the IH.
         \case \T{Case}: By the IH.
         \case \T{Cons}: By the IH.
         \case \T{CaseList}: By the IH.
        \end{caseanalysis}
 \end{enumerate}
\end{proof}

\ifrestate
\thmTypeSoundness*
\else
\begin{thm}[Type Soundness]
 \label{thm:type-sound}
 Suppose that all operations satisfy the signature restriction.
 If $\Delta  \vdash  \ottnt{M}  \ottsym{:}  \ottnt{A}$ and $\ottnt{M}  \longrightarrow^{*}  \ottnt{M'}$ and $\ottnt{M'}  \centernot\longrightarrow$, then:
 \begin{itemize}
  \item $\ottnt{M'}$ is a value; or
  \item $\ottnt{M'} \,  =  \,  \ottnt{E}  [   \textup{\texttt{\#}\relax}  \mathsf{op}   \ottsym{(}   \ottnt{v}   \ottsym{)}   ] $ for some $\ottnt{E}$, $\mathsf{op}$, and $\ottnt{v}$
        such that $\mathsf{op} \,  \not\in  \, \ottnt{E}$.
 \end{itemize}
\end{thm}
\fi
\begin{proof}
 By Lemmas~\ref{lem:subject-red} and \ref{lem:progress}.
\end{proof}

\subsection{Soundness of the Type-and-effect System}
\label{sec:app:proof:eff-sound-typing}

This section show soundness of the type-and-effect system.  We may reuse the
lemmas proven in \refsec{app:proof:sound-typing} if their statements and proofs
do not need change.


\begin{lemmap}{Weakening}{eff-weakening}
 Suppose that $\vdash  \Gamma_{{\mathrm{1}}}  \ottsym{,}  \Gamma_{{\mathrm{2}}}$.
 Let $\Gamma_{{\mathrm{3}}}$ be a typing context such that
 $ \mathit{dom}  (  \Gamma_{{\mathrm{2}}}  )  \,  \mathbin{\cap}  \,  \mathit{dom}  (  \Gamma_{{\mathrm{3}}}  )  \,  =  \,  \emptyset $.
 \begin{enumerate}
  \item \label{lem:eff-weakening:typing-context}
        If $\vdash  \Gamma_{{\mathrm{1}}}  \ottsym{,}  \Gamma_{{\mathrm{3}}}$, then $\vdash  \Gamma_{{\mathrm{1}}}  \ottsym{,}  \Gamma_{{\mathrm{2}}}  \ottsym{,}  \Gamma_{{\mathrm{3}}}$.

  \item \label{lem:eff-weakening:type}
        If $\Gamma_{{\mathrm{1}}}  \ottsym{,}  \Gamma_{{\mathrm{3}}}  \vdash  \ottnt{A}$, then $\Gamma_{{\mathrm{1}}}  \ottsym{,}  \Gamma_{{\mathrm{2}}}  \ottsym{,}  \Gamma_{{\mathrm{3}}}  \vdash  \ottnt{A}$.

  \item \label{lem:eff-weakening:sub}
        If $\Gamma_{{\mathrm{1}}}  \ottsym{,}  \Gamma_{{\mathrm{3}}}  \vdash  \ottnt{A}  \sqsubseteq  \ottnt{B}$,
        then $\Gamma_{{\mathrm{1}}}  \ottsym{,}  \Gamma_{{\mathrm{2}}}  \ottsym{,}  \Gamma_{{\mathrm{3}}}  \vdash  \ottnt{A}  \sqsubseteq  \ottnt{B}$.

  \item \label{lem:eff-weakening:term}
        If $\Gamma_{{\mathrm{1}}}  \ottsym{,}  \Gamma_{{\mathrm{3}}}  \vdash  \ottnt{M}  \ottsym{:}  \ottnt{A} \,  |  \, \epsilon$,
        then $\Gamma_{{\mathrm{1}}}  \ottsym{,}  \Gamma_{{\mathrm{2}}}  \ottsym{,}  \Gamma_{{\mathrm{3}}}  \vdash  \ottnt{M}  \ottsym{:}  \ottnt{A} \,  |  \, \epsilon$.

  \item \label{lem:eff-weakening:handler}
        If $\Gamma_{{\mathrm{1}}}  \ottsym{,}  \Gamma_{{\mathrm{3}}}  \vdash  \ottnt{H}  \ottsym{:}  \ottnt{A} \,  |  \, \epsilon  \Rightarrow  \ottnt{B} \,  |  \, \epsilon'$,
        then $\Gamma_{{\mathrm{1}}}  \ottsym{,}  \Gamma_{{\mathrm{2}}}  \ottsym{,}  \Gamma_{{\mathrm{3}}}  \vdash  \ottnt{H}  \ottsym{:}  \ottnt{A} \,  |  \, \epsilon  \Rightarrow  \ottnt{B} \,  |  \, \epsilon'$.
 \end{enumerate}
\end{lemmap}
\begin{proof}
 By (mutual) induction on the derivations of the judgments.
\end{proof}

\begin{lemmap}{Type substitution}{eff-ty-subst}
 Suppose that $\Gamma_{{\mathrm{1}}}  \vdash  \ottnt{A}$.
 \begin{enumerate}
  \item \label{lem:eff-ty-subst:typing-context}
        If $\vdash  \Gamma_{{\mathrm{1}}}  \ottsym{,}  \alpha  \ottsym{,}  \Gamma_{{\mathrm{2}}}$, then $\vdash  \Gamma_{{\mathrm{1}}}  \ottsym{,}  \Gamma_{{\mathrm{2}}} \,  [  \ottnt{A}  \ottsym{/}  \alpha  ] $.

  \item \label{lem:eff-ty-subst:type}
        If $\Gamma_{{\mathrm{1}}}  \ottsym{,}  \alpha  \ottsym{,}  \Gamma_{{\mathrm{2}}}  \vdash  \ottnt{B}$,
        then $\Gamma_{{\mathrm{1}}}  \ottsym{,}  \Gamma_{{\mathrm{2}}} \,  [  \ottnt{A}  \ottsym{/}  \alpha  ]   \vdash   \ottnt{B}    [  \ottnt{A}  \ottsym{/}  \alpha  ]  $.

  \item \label{lem:eff-ty-subst:sub}
        If $\Gamma_{{\mathrm{1}}}  \ottsym{,}  \alpha  \ottsym{,}  \Gamma_{{\mathrm{2}}}  \vdash  \ottnt{B}  \sqsubseteq  \ottnt{C}$,
        then $\Gamma_{{\mathrm{1}}}  \ottsym{,}  \Gamma_{{\mathrm{2}}} \,  [  \ottnt{A}  \ottsym{/}  \alpha  ]   \vdash   \ottnt{B}    [  \ottnt{A}  \ottsym{/}  \alpha  ]    \sqsubseteq   \ottnt{C}    [  \ottnt{A}  \ottsym{/}  \alpha  ]  $.

  \item \label{lem:eff-ty-subst:term}
        If $\Gamma_{{\mathrm{1}}}  \ottsym{,}  \alpha  \ottsym{,}  \Gamma_{{\mathrm{2}}}  \vdash  \ottnt{M}  \ottsym{:}  \ottnt{B} \,  |  \, \epsilon$,
        then $\Gamma_{{\mathrm{1}}}  \ottsym{,}  \Gamma_{{\mathrm{2}}} \,  [  \ottnt{A}  \ottsym{/}  \alpha  ]   \vdash  \ottnt{M}  \ottsym{:}   \ottnt{B}    [  \ottnt{A}  \ottsym{/}  \alpha  ]   \,  |  \, \epsilon$.

  \item If $\Gamma_{{\mathrm{1}}}  \ottsym{,}  \alpha  \ottsym{,}  \Gamma_{{\mathrm{2}}}  \vdash  \ottnt{H}  \ottsym{:}  \ottnt{B} \,  |  \, \epsilon  \Rightarrow  \ottnt{C} \,  |  \, \epsilon'$,
        then $\Gamma_{{\mathrm{1}}}  \ottsym{,}  \Gamma_{{\mathrm{2}}} \,  [  \ottnt{A}  \ottsym{/}  \alpha  ]   \vdash  \ottnt{H}  \ottsym{:}   \ottnt{B}    [  \ottnt{A}  \ottsym{/}  \alpha  ]   \,  |  \, \epsilon  \Rightarrow   \ottnt{C}    [  \ottnt{A}  \ottsym{/}  \alpha  ]   \,  |  \, \epsilon'$.

 \end{enumerate}
\end{lemmap}
\begin{proof}
 Straightforward by (mutual) induction on the derivations of the judgments,
 as in \reflem{ty-subst}.
\end{proof}

\begin{lemmap}{Term substitution}{eff-term-subst}
 Suppose that $\Gamma_{{\mathrm{1}}}  \vdash  \ottnt{M}  \ottsym{:}  \ottnt{A} \,  |  \, \epsilon$ for any $\epsilon$.
 \begin{enumerate}
  \item \label{lem:eff-term-subst:term}
        If $\Gamma_{{\mathrm{1}}}  \ottsym{,}  \mathit{x} \,  \mathord{:}  \, \ottnt{A}  \ottsym{,}  \Gamma_{{\mathrm{2}}}  \vdash  \ottnt{M'}  \ottsym{:}  \ottnt{B} \,  |  \, \epsilon$,
        then $\Gamma_{{\mathrm{1}}}  \ottsym{,}  \Gamma_{{\mathrm{2}}}  \vdash   \ottnt{M'}    [  \ottnt{M}  /  \mathit{x}  ]    \ottsym{:}  \ottnt{B} \,  |  \, \epsilon$.

  \item If $\Gamma_{{\mathrm{1}}}  \ottsym{,}  \mathit{x} \,  \mathord{:}  \, \ottnt{A}  \ottsym{,}  \Gamma_{{\mathrm{2}}}  \vdash  \ottnt{H}  \ottsym{:}  \ottnt{B} \,  |  \, \epsilon  \Rightarrow  \ottnt{C} \,  |  \, \epsilon'$,
        then $\Gamma_{{\mathrm{1}}}  \ottsym{,}  \Gamma_{{\mathrm{2}}}  \vdash   \ottnt{H}    [  \ottnt{M}  /  \mathit{x}  ]    \ottsym{:}  \ottnt{B} \,  |  \, \epsilon  \Rightarrow  \ottnt{C} \,  |  \, \epsilon'$.
 \end{enumerate}
\end{lemmap}
\begin{proof}
 By mutual induction on the typing derivations as in \reflem{term-subst}.
\end{proof}

\begin{lemmap}{Canonical forms}{eff-canonical-forms}
 Suppose that $\Gamma  \vdash  \ottnt{v}  \ottsym{:}  \ottnt{A} \,  |  \, \epsilon$.
 \begin{enumerate}
  \item If $ \mathit{unqualify}  (  \ottnt{A}  )  \,  =  \, \iota$,
        then $\ottnt{v} \,  =  \, \ottnt{c}$ for some $\ottnt{c}$.

  \item If $ \mathit{unqualify}  (  \ottnt{A}  )  \,  =  \,  \ottnt{B}   \rightarrow ^{ \epsilon' }  \ottnt{C} $,
        then $\ottnt{v} \,  =  \, \ottnt{c}$ for some $\ottnt{c}$,
        or $\ottnt{v} \,  =  \,  \lambda\!  \, \mathit{x}  \ottsym{.}  \ottnt{M}$ for some $\mathit{x}$ and $\ottnt{M}$.

  \item If $ \mathit{unqualify}  (  \ottnt{A}  )  \,  =  \,  \ottnt{B}  \times  \ottnt{C} $,
        then $\ottnt{v} \,  =  \, \ottsym{(}  \ottnt{v_{{\mathrm{1}}}}  \ottsym{,}  \ottnt{v_{{\mathrm{2}}}}  \ottsym{)}$ for some $\ottnt{v_{{\mathrm{1}}}}$ and $\ottnt{v_{{\mathrm{2}}}}$.

  \item If $ \mathit{unqualify}  (  \ottnt{A}  )  \,  =  \,  \ottnt{B}  +  \ottnt{C} $,
        then $\ottnt{v} \,  =  \, \mathsf{inl} \, \ottnt{v'}$ or $\ottnt{v} \,  =  \, \mathsf{inr} \, \ottnt{v'}$ for some $\ottnt{v'}$.

  \item If $ \mathit{unqualify}  (  \ottnt{A}  )  \,  =  \,  \ottnt{B}  \, \mathsf{list} $,
        then $\ottnt{v} \,  =  \, \mathsf{nil}$ or $\ottnt{v} \,  =  \, \mathsf{cons} \, \ottnt{v'}$ for some $\ottnt{v'}$.
 \end{enumerate}
\end{lemmap}
\begin{proof}
 Similarly to \reflem{canonical-forms}.
\end{proof}

\begin{lemmap}{Type containment inversion: function types}{eff-subtyping-inv-fun}
 If $\Gamma  \vdash    \text{\unboldmath$\forall$}  \,  \algeffseqoverindex{ \alpha_{{\mathrm{1}}} }{ \text{\unboldmath$\mathit{I_{{\mathrm{1}}}}$} }   \ottsym{.} \, \ottnt{A_{{\mathrm{1}}}}   \rightarrow ^{ \epsilon_{{\mathrm{1}}} }  \ottnt{A_{{\mathrm{2}}}}   \sqsubseteq    \text{\unboldmath$\forall$}  \,  \algeffseqoverindex{ \alpha_{{\mathrm{2}}} }{ \text{\unboldmath$\mathit{I_{{\mathrm{2}}}}$} }   \ottsym{.} \, \ottnt{B_{{\mathrm{1}}}}   \rightarrow ^{ \epsilon_{{\mathrm{2}}} }  \ottnt{B_{{\mathrm{2}}}} $,
 then $\epsilon_{{\mathrm{1}}} \,  =  \, \epsilon_{{\mathrm{2}}}$ and
 there exist $ \algeffseqoverindex{ \alpha_{{\mathrm{11}}} }{ \text{\unboldmath$\mathit{I_{{\mathrm{11}}}}$} } $, $ \algeffseqoverindex{ \alpha_{{\mathrm{12}}} }{ \text{\unboldmath$\mathit{I_{{\mathrm{12}}}}$} } $, $ \algeffseqoverindex{ \beta }{ \text{\unboldmath$\mathit{J}$} } $, and $ \algeffseqoverindex{ \ottnt{C} }{ \text{\unboldmath$\mathit{I_{{\mathrm{11}}}}$} } $
 such that
 \begin{itemize}
  \item $\ottsym{\{}   \algeffseqoverindex{ \alpha_{{\mathrm{1}}} }{ \text{\unboldmath$\mathit{I_{{\mathrm{1}}}}$} }   \ottsym{\}} \,  =  \, \ottsym{\{}   \algeffseqoverindex{ \alpha_{{\mathrm{11}}} }{ \text{\unboldmath$\mathit{I_{{\mathrm{11}}}}$} }   \ottsym{\}} \,  \mathbin{\uplus}  \, \ottsym{\{}   \algeffseqoverindex{ \alpha_{{\mathrm{12}}} }{ \text{\unboldmath$\mathit{I_{{\mathrm{12}}}}$} }   \ottsym{\}}$,
  \item $\Gamma  \ottsym{,}   \algeffseqoverindex{ \alpha_{{\mathrm{2}}} }{ \text{\unboldmath$\mathit{I_{{\mathrm{2}}}}$} }   \ottsym{,}   \algeffseqoverindex{ \beta }{ \text{\unboldmath$\mathit{J}$} }   \vdash   \algeffseqoverindex{ \ottnt{C} }{ \text{\unboldmath$\mathit{I_{{\mathrm{11}}}}$} } $,
  \item $\Gamma  \ottsym{,}   \algeffseqoverindex{ \alpha_{{\mathrm{2}}} }{ \text{\unboldmath$\mathit{I_{{\mathrm{2}}}}$} }   \vdash  \ottnt{B_{{\mathrm{1}}}}  \sqsubseteq    \text{\unboldmath$\forall$}  \,  \algeffseqoverindex{ \beta }{ \text{\unboldmath$\mathit{J}$} }   \ottsym{.} \, \ottnt{A_{{\mathrm{1}}}}    [   \algeffseqoverindex{ \ottnt{C} }{ \text{\unboldmath$\mathit{I_{{\mathrm{11}}}}$} }   \ottsym{/}   \algeffseqoverindex{ \alpha_{{\mathrm{11}}} }{ \text{\unboldmath$\mathit{I_{{\mathrm{11}}}}$} }   ]  $,
  \item $\Gamma  \ottsym{,}   \algeffseqoverindex{ \alpha_{{\mathrm{2}}} }{ \text{\unboldmath$\mathit{I_{{\mathrm{2}}}}$} }   \vdash    \text{\unboldmath$\forall$}  \,  \algeffseqoverindex{ \alpha_{{\mathrm{12}}} }{ \text{\unboldmath$\mathit{I_{{\mathrm{12}}}}$} }   \ottsym{.} \,  \text{\unboldmath$\forall$}  \,  \algeffseqoverindex{ \beta }{ \text{\unboldmath$\mathit{J}$} }   \ottsym{.} \, \ottnt{A_{{\mathrm{2}}}}    [   \algeffseqoverindex{ \ottnt{C} }{ \text{\unboldmath$\mathit{I_{{\mathrm{11}}}}$} }   \ottsym{/}   \algeffseqoverindex{ \alpha_{{\mathrm{11}}} }{ \text{\unboldmath$\mathit{I_{{\mathrm{11}}}}$} }   ]    \sqsubseteq  \ottnt{B_{{\mathrm{2}}}}$,
  \item type variables in $\ottsym{\{}   \algeffseqoverindex{ \beta }{ \text{\unboldmath$\mathit{J}$} }   \ottsym{\}}$ do not appear free in $\ottnt{A_{{\mathrm{1}}}}$ and $\ottnt{A_{{\mathrm{2}}}}$, and
  \item if $ \algeffseqoverindex{ \alpha_{{\mathrm{12}}} }{ \text{\unboldmath$\mathit{I_{{\mathrm{12}}}}$} } $ or $ \algeffseqoverindex{ \beta }{ \text{\unboldmath$\mathit{J}$} } $ is not the empty sequence,
        $\mathit{SR} \, \ottsym{(}  \epsilon_{{\mathrm{1}}}  \ottsym{)}$.
 \end{itemize}
\end{lemmap}
\begin{proof}
 Similarly to \reflem{subtyping-inv-fun}.
\end{proof}

\begin{lemma}{eff-subtyping-inv-fun-mono}
 If $\Gamma  \vdash   \ottnt{A_{{\mathrm{1}}}}   \rightarrow ^{ \epsilon_{{\mathrm{1}}} }  \ottnt{A_{{\mathrm{2}}}}   \sqsubseteq   \ottnt{B_{{\mathrm{1}}}}   \rightarrow ^{ \epsilon_{{\mathrm{2}}} }  \ottnt{B_{{\mathrm{2}}}} $,
 then $\epsilon_{{\mathrm{1}}} \,  =  \, \epsilon_{{\mathrm{2}}}$ and $\Gamma  \vdash  \ottnt{B_{{\mathrm{1}}}}  \sqsubseteq  \ottnt{A_{{\mathrm{1}}}}$ and $\Gamma  \vdash  \ottnt{A_{{\mathrm{2}}}}  \sqsubseteq  \ottnt{B_{{\mathrm{2}}}}$.
\end{lemma}
\begin{proof}
 Similarly to \reflem{subtyping-inv-fun-mono} with
 \reflem{eff-subtyping-inv-fun}.
\end{proof}

\begin{lemmap}{Value inversion: constants}{eff-val-inv-const}
 If $\Gamma  \vdash  \ottnt{c}  \ottsym{:}  \ottnt{A} \,  |  \, \epsilon$, then $\Gamma  \vdash   \mathit{ty}  (  \ottnt{c}  )   \sqsubseteq  \ottnt{A}$.
\end{lemmap}
\begin{proof}
 Similarly to \reflem{val-inv-const}.
\end{proof}

\begin{lemmap}{Progress}{eff-progress}
 If $\Delta  \vdash  \ottnt{M}  \ottsym{:}  \ottnt{A} \,  |  \, \epsilon$, then:
 \begin{itemize}
  \item $\ottnt{M}  \longrightarrow  \ottnt{M'}$ for some $\ottnt{M'}$;
  \item $\ottnt{M}$ is a value; or
  \item $\ottnt{M} \,  =  \,  \ottnt{E}  [   \textup{\texttt{\#}\relax}  \mathsf{op}   \ottsym{(}   \ottnt{v}   \ottsym{)}   ] $ for some $\ottnt{E}$, $\mathsf{op}$, and $\ottnt{v}$
        such that $\mathsf{op} \,  \not\in  \, \ottnt{E}$ and $\mathsf{op} \,  \in  \, \epsilon$.
 \end{itemize}
\end{lemmap}
\begin{proof}
 Similarly to \reflem{progress} with the lemmas proven in this section.
 The case for \Te{Weak} is also straightforward.
\end{proof}

\begin{lemmap}{Value inversion: lambda abstractions}{eff-val-inv-abs}
 If $\Gamma  \vdash   \lambda\!  \, \mathit{x}  \ottsym{.}  \ottnt{M}  \ottsym{:}  \ottnt{A} \,  |  \, \epsilon$,
 then $\Gamma  \ottsym{,}   \algeffseqover{ \alpha }   \ottsym{,}  \mathit{x} \,  \mathord{:}  \, \ottnt{B}  \vdash  \ottnt{M}  \ottsym{:}  \ottnt{C} \,  |  \, \epsilon'$ and $\Gamma  \vdash    \text{\unboldmath$\forall$}  \,  \algeffseqover{ \alpha }   \ottsym{.} \, \ottnt{B}   \rightarrow ^{ \epsilon' }  \ottnt{C}   \sqsubseteq  \ottnt{A}$
 for some $ \algeffseqover{ \alpha } $, $\ottnt{B}$, $\ottnt{C}$, and $\epsilon'$.
\end{lemmap}
\begin{proof}
 Similarly to \reflem{val-inv-abs}.
\end{proof}

\begin{lemmap}{Value inversion: pairs}{eff-val-inv-pair}
 If $\Gamma  \vdash  \ottsym{(}  \ottnt{M_{{\mathrm{1}}}}  \ottsym{,}  \ottnt{M_{{\mathrm{2}}}}  \ottsym{)}  \ottsym{:}  \ottnt{A} \,  |  \, \epsilon$,
 then $\Gamma  \ottsym{,}   \algeffseqover{ \alpha }   \vdash  \ottnt{M_{{\mathrm{1}}}}  \ottsym{:}  \ottnt{B_{{\mathrm{1}}}} \,  |  \, \epsilon$ and $\Gamma  \ottsym{,}   \algeffseqover{ \alpha }   \vdash  \ottnt{M_{{\mathrm{2}}}}  \ottsym{:}  \ottnt{B_{{\mathrm{2}}}} \,  |  \, \epsilon$ and
 $\Gamma  \vdash    \text{\unboldmath$\forall$}  \,  \algeffseqover{ \alpha }   \ottsym{.} \, \ottnt{B_{{\mathrm{1}}}}  \times  \ottnt{B_{{\mathrm{2}}}}   \sqsubseteq  \ottnt{A}$
 for some $ \algeffseqover{ \alpha } $, $\ottnt{B_{{\mathrm{1}}}}$, and $\ottnt{B_{{\mathrm{2}}}}$.
\end{lemmap}
\begin{proof}
 Similarly to \reflem{val-inv-pair}.
\end{proof}

\begin{lemmap}{Value inversion: left injections}{eff-val-inv-inl}
 If $\Gamma  \vdash  \mathsf{inl} \, \ottnt{M}  \ottsym{:}  \ottnt{A} \,  |  \, \epsilon$, then
 $\Gamma  \ottsym{,}   \algeffseqover{ \alpha }   \vdash  \ottnt{M}  \ottsym{:}  \ottnt{B} \,  |  \, \epsilon$ and
 $\Gamma  \vdash    \text{\unboldmath$\forall$}  \,  \algeffseqover{ \alpha }   \ottsym{.} \, \ottnt{B}  +  \ottnt{C}   \sqsubseteq  \ottnt{A}$
 for some $ \algeffseqover{ \alpha } $, $\ottnt{B}$, and $\ottnt{C}$.
\end{lemmap}
\begin{proof}
 Similarly to \reflem{val-inv-inl}.
\end{proof}

\begin{lemmap}{Value inversion: right injections}{eff-val-inv-inr}
 If $\Gamma  \vdash  \mathsf{inr} \, \ottnt{M}  \ottsym{:}  \ottnt{A} \,  |  \, \epsilon$, then
 $\Gamma  \ottsym{,}   \algeffseqover{ \alpha }   \vdash  \ottnt{M}  \ottsym{:}  \ottnt{C} \,  |  \, \epsilon$ and
 $\Gamma  \vdash    \text{\unboldmath$\forall$}  \,  \algeffseqover{ \alpha }   \ottsym{.} \, \ottnt{B}  +  \ottnt{C}   \sqsubseteq  \ottnt{A}$
 for some $ \algeffseqover{ \alpha } $, $\ottnt{B}$, and $\ottnt{C}$.
\end{lemmap}
\begin{proof}
 Similarly to the proof of \reflem{val-inv-inr}.
\end{proof}

\begin{lemmap}{Value inversion: cons}{eff-val-inv-cons}
 If $\Gamma  \vdash  \mathsf{cons} \, \ottnt{M}  \ottsym{:}  \ottnt{A} \,  |  \, \epsilon$, then
 $\Gamma  \ottsym{,}   \algeffseqover{ \alpha }   \vdash  \ottnt{M}  \ottsym{:}    \ottnt{B}  \times  \ottnt{B}   \, \mathsf{list}  \,  |  \, \epsilon$ and
 $\Gamma  \vdash    \text{\unboldmath$\forall$}  \,  \algeffseqover{ \alpha }   \ottsym{.} \, \ottnt{B}  \, \mathsf{list}   \sqsubseteq  \ottnt{A}$
 for some $ \algeffseqover{ \alpha } $ and $\ottnt{B}$.
\end{lemmap}
\begin{proof}
 Similarly to \reflem{val-inv-cons}.
\end{proof}

\begin{lemma}{eff-term-inv-op}
 If $\mathit{ty} \, \ottsym{(}  \mathsf{op}  \ottsym{)} \,  =  \,   \text{\unboldmath$\forall$}  \,  \algeffseqoverindex{ \alpha }{ \text{\unboldmath$\mathit{I}$} }   \ottsym{.} \,  \ottnt{A}  \hookrightarrow  \ottnt{B} $ and $\Gamma  \vdash   \textup{\texttt{\#}\relax}  \mathsf{op}   \ottsym{(}   \ottnt{v}   \ottsym{)}   \ottsym{:}  \ottnt{C} \,  |  \, \epsilon$, then
 \begin{itemize}
  \item $\Gamma  \ottsym{,}   \algeffseqoverindex{ \beta }{ \text{\unboldmath$\mathit{J}$} }   \vdash   \algeffseqoverindex{ \ottnt{D} }{ \text{\unboldmath$\mathit{I}$} } $,
  \item $\Gamma  \ottsym{,}   \algeffseqoverindex{ \beta }{ \text{\unboldmath$\mathit{J}$} }   \vdash  \ottnt{v}  \ottsym{:}   \ottnt{A}    [   \algeffseqoverindex{ \ottnt{D} }{ \text{\unboldmath$\mathit{I}$} }   \ottsym{/}   \algeffseqoverindex{ \alpha }{ \text{\unboldmath$\mathit{I}$} }   ]   \,  |  \, \epsilon'$,
  \item $\epsilon' \,  \subseteq  \, \epsilon$,
  \item $\mathsf{op} \,  \in  \, \epsilon'$, and
  \item $\Gamma  \vdash    \text{\unboldmath$\forall$}  \,  \algeffseqoverindex{ \beta }{ \text{\unboldmath$\mathit{J}$} }   \ottsym{.} \, \ottnt{B}    [   \algeffseqoverindex{ \ottnt{D} }{ \text{\unboldmath$\mathit{I}$} }   \ottsym{/}   \algeffseqoverindex{ \alpha }{ \text{\unboldmath$\mathit{I}$} }   ]    \sqsubseteq  \ottnt{C}$; or
 \end{itemize}
 for some $ \algeffseqoverindex{ \beta }{ \text{\unboldmath$\mathit{J}$} } $, $ \algeffseqoverindex{ \ottnt{D} }{ \text{\unboldmath$\mathit{I}$} } $, and $\epsilon'$.
 Furthermore, if $ \algeffseqoverindex{ \beta }{ \text{\unboldmath$\mathit{J}$} } $ is not the empty sequence, $\mathit{SR} \, \ottsym{(}  \epsilon'  \ottsym{)}$ holds.
\end{lemma}
\begin{proof}
 By induction on the typing derivation.
 There are only five typing rules that can be applied to $ \textup{\texttt{\#}\relax}  \mathsf{op}   \ottsym{(}   \ottnt{v}   \ottsym{)} $.
 \begin{caseanalysis}
  \case \Te{Gen}: Straightforward by the IH.
   Note that $\mathit{SR} \, \ottsym{(}  \epsilon  \ottsym{)}$ by inversion.


  \case \Te{Inst}: Straightforward by the IH and \Srule{Trans}.
  \case \Te{Op}: Trivial.
  \case \Te{Weak}: By the IH.
 \end{caseanalysis}
\end{proof}

\begin{lemma}{eff-ectx-typing}
 If $\Gamma  \ottsym{,}   \algeffseqoverindex{ \alpha }{ \text{\unboldmath$\mathit{I}$} }   \vdash   \ottnt{E}  [   \textup{\texttt{\#}\relax}  \mathsf{op}   \ottsym{(}   \ottnt{v}   \ottsym{)}   ]   \ottsym{:}  \ottnt{A} \,  |  \, \epsilon$ and $\mathsf{op} \,  \not\in  \, \ottnt{E}$, then
 \begin{itemize}
  \item $\Gamma  \ottsym{,}   \algeffseqoverindex{ \alpha }{ \text{\unboldmath$\mathit{I}$} }   \ottsym{,}   \algeffseqoverindex{ \beta }{ \text{\unboldmath$\mathit{J}$} }   \vdash   \textup{\texttt{\#}\relax}  \mathsf{op}   \ottsym{(}   \ottnt{v}   \ottsym{)}   \ottsym{:}  \ottnt{B} \,  |  \, \epsilon'$ and
  \item $\Gamma  \ottsym{,}  \mathit{y} \,  \mathord{:}  \,  \text{\unboldmath$\forall$}  \,  \algeffseqoverindex{ \alpha }{ \text{\unboldmath$\mathit{I}$} }   \ottsym{.} \,  \text{\unboldmath$\forall$}  \,  \algeffseqoverindex{ \beta }{ \text{\unboldmath$\mathit{J}$} }   \ottsym{.} \, \ottnt{B}  \ottsym{,}   \algeffseqoverindex{ \alpha }{ \text{\unboldmath$\mathit{I}$} }   \vdash   \ottnt{E}  [  \mathit{y}  ]   \ottsym{:}  \ottnt{A} \,  |  \, \epsilon$ for any $\mathit{y} \,  \not\in  \,  \mathit{dom}  (  \Gamma  ) $, and
  \item $\mathsf{op} \,  \in  \, \epsilon$
 \end{itemize}
 for some $ \algeffseqoverindex{ \beta }{ \text{\unboldmath$\mathit{J}$} } $, $\ottnt{B}$, and $\epsilon'$.
 Furthermore, if $ \algeffseqoverindex{ \beta }{ \text{\unboldmath$\mathit{J}$} } $ is not the empty sequence,
 then $\mathit{SR} \, \ottsym{(}  \ottsym{\{}  \mathsf{op}  \ottsym{\}}  \ottsym{)}$ holds.
\end{lemma}
\begin{proof}
 By induction on the typing derivation.
 \begin{caseanalysis}
  \case \Te{Var}, \Te{Const}, \Te{Abs}, \Te{Nil}, and \Te{Fix}:
   Contradictory.

  \case \Te{App}:
   By case analysis on $\ottnt{E}$.
   \begin{caseanalysis}
    \case $\ottnt{E} \,  =  \, \ottnt{E'} \, \ottnt{M_{{\mathrm{2}}}}$:
     By inversion of the typing derivation, we have
     $\Gamma  \ottsym{,}   \algeffseqoverindex{ \alpha }{ \text{\unboldmath$\mathit{I}$} }   \vdash   \ottnt{E'}  [   \textup{\texttt{\#}\relax}  \mathsf{op}   \ottsym{(}   \ottnt{v}   \ottsym{)}   ]   \ottsym{:}   \ottnt{C}   \rightarrow ^{ \epsilon'' }  \ottnt{A}  \,  |  \, \epsilon$ and
     $\Gamma  \ottsym{,}   \algeffseqoverindex{ \alpha }{ \text{\unboldmath$\mathit{I}$} }   \vdash  \ottnt{M_{{\mathrm{2}}}}  \ottsym{:}  \ottnt{C} \,  |  \, \epsilon$ and
     $\epsilon'' \,  \subseteq  \, \epsilon$
     for some $\ottnt{C}$ and $\epsilon''$.
     By the IH,
     \begin{itemize}
      \item $\Gamma  \ottsym{,}   \algeffseqoverindex{ \alpha }{ \text{\unboldmath$\mathit{I}$} }   \ottsym{,}   \algeffseqoverindex{ \beta }{ \text{\unboldmath$\mathit{J}$} }   \vdash   \textup{\texttt{\#}\relax}  \mathsf{op}   \ottsym{(}   \ottnt{v}   \ottsym{)}   \ottsym{:}  \ottnt{B} \,  |  \, \epsilon'$,
      \item $\Gamma  \ottsym{,}  \mathit{y} \,  \mathord{:}  \,  \text{\unboldmath$\forall$}  \,  \algeffseqoverindex{ \alpha }{ \text{\unboldmath$\mathit{I}$} }   \ottsym{.} \,  \text{\unboldmath$\forall$}  \,  \algeffseqoverindex{ \beta }{ \text{\unboldmath$\mathit{J}$} }   \ottsym{.} \, \ottnt{B}  \ottsym{,}   \algeffseqoverindex{ \alpha }{ \text{\unboldmath$\mathit{I}$} }   \vdash   \ottnt{E'}  [  \mathit{y}  ]   \ottsym{:}   \ottnt{C}   \rightarrow ^{ \epsilon'' }  \ottnt{A}  \,  |  \, \epsilon$
            for any $\mathit{y} \,  \not\in  \,  \mathit{dom}  (  \Gamma  ) $, and
      \item $\mathsf{op} \,  \in  \, \epsilon$,
      \item If $ \algeffseqoverindex{ \beta }{ \text{\unboldmath$\mathit{J}$} } $ is not the empty sequence,
            then $\mathit{SR} \, \ottsym{(}  \ottsym{\{}  \mathsf{op}  \ottsym{\}}  \ottsym{)}$ holds.
     \end{itemize}
     for some $ \algeffseqoverindex{ \beta }{ \text{\unboldmath$\mathit{J}$} } $, $\ottnt{B}$, and $\epsilon'$.
     By \reflem{eff-weakening} (\ref{lem:eff-weakening:term}) and \Te{App},
     $\Gamma  \ottsym{,}  \mathit{y} \,  \mathord{:}  \,  \text{\unboldmath$\forall$}  \,  \algeffseqoverindex{ \alpha }{ \text{\unboldmath$\mathit{I}$} }   \ottsym{.} \,  \text{\unboldmath$\forall$}  \,  \algeffseqoverindex{ \beta }{ \text{\unboldmath$\mathit{J}$} }   \ottsym{.} \, \ottnt{B}  \ottsym{,}   \algeffseqoverindex{ \alpha }{ \text{\unboldmath$\mathit{I}$} }   \vdash   \ottnt{E'}  [  \mathit{y}  ]  \, \ottnt{M_{{\mathrm{2}}}}  \ottsym{:}  \ottnt{A} \,  |  \, \epsilon$, i.e.,
     $\Gamma  \ottsym{,}  \mathit{y} \,  \mathord{:}  \,  \text{\unboldmath$\forall$}  \,  \algeffseqoverindex{ \alpha }{ \text{\unboldmath$\mathit{I}$} }   \ottsym{.} \,  \text{\unboldmath$\forall$}  \,  \algeffseqoverindex{ \beta }{ \text{\unboldmath$\mathit{J}$} }   \ottsym{.} \, \ottnt{B}  \ottsym{,}   \algeffseqoverindex{ \alpha }{ \text{\unboldmath$\mathit{I}$} }   \vdash   \ottnt{E}  [  \mathit{y}  ]   \ottsym{:}  \ottnt{A} \,  |  \, \epsilon$.

    \case $\ottnt{E} \,  =  \, \ottnt{v_{{\mathrm{1}}}} \, \ottnt{E'}$:
     Similarly to the above case.
   \end{caseanalysis}

  \case \Te{Gen}:
   By the IH.  We find $\mathit{SR} \, \ottsym{(}  \ottsym{\{}  \mathsf{op}  \ottsym{\}}  \ottsym{)}$ by $\mathsf{op} \,  \in  \, \epsilon$ and $\mathit{SR} \, \ottsym{(}  \epsilon  \ottsym{)}$.


  \case \Te{Inst}: By the IH.

  \case \Te{Op}:
   If $\ottnt{E} \,  =  \,  [] $, the proof is straightforward by letting
   $ \algeffseqoverindex{ \beta }{ \text{\unboldmath$\mathit{J}$} } $ be the empty sequence,
   $\ottnt{B} \,  =  \, \ottnt{A}$, and
   $\epsilon' \,  =  \, \epsilon$;
   $\mathsf{op} \,  \in  \, \epsilon$ is found by \reflem{eff-term-inv-op}.

   Otherwise, the proof is similar to the case for \Te{App}.

  \case \Te{Handle}: By the IH.  We find $\mathsf{op} \,  \in  \, \epsilon$ because
   the handler does not have an operation clause for $\mathsf{op}$ ($\mathsf{op} \,  \not\in  \, \ottnt{E}$).

  \case \Te{Weak}: By the IH.

  \otherwise: Similarly to the case for \Te{App}.

 \end{caseanalysis}
\end{proof}

\begin{lemma}{eff-var-subtype}
 Suppose that $\Gamma_{{\mathrm{1}}}  \vdash  \ottnt{A}  \sqsubseteq  \ottnt{B}$ and $\Gamma_{{\mathrm{1}}}  \vdash  \ottnt{A}$.
 \begin{enumerate}
  \item If $\Gamma_{{\mathrm{1}}}  \ottsym{,}  \mathit{x} \,  \mathord{:}  \, \ottnt{B}  \ottsym{,}  \Gamma_{{\mathrm{2}}}  \vdash  \ottnt{M}  \ottsym{:}  \ottnt{C} \,  |  \, \epsilon$,
        then $\Gamma_{{\mathrm{1}}}  \ottsym{,}  \mathit{x} \,  \mathord{:}  \, \ottnt{A}  \ottsym{,}  \Gamma_{{\mathrm{2}}}  \vdash  \ottnt{M}  \ottsym{:}  \ottnt{C} \,  |  \, \epsilon$.
  \item If $\Gamma_{{\mathrm{1}}}  \ottsym{,}  \mathit{x} \,  \mathord{:}  \, \ottnt{B}  \ottsym{,}  \Gamma_{{\mathrm{2}}}  \vdash  \ottnt{H}  \ottsym{:}  \ottnt{C} \,  |  \, \epsilon  \Rightarrow  \ottnt{D} \,  |  \, \epsilon'$,
        then $\Gamma_{{\mathrm{1}}}  \ottsym{,}  \mathit{x} \,  \mathord{:}  \, \ottnt{A}  \ottsym{,}  \Gamma_{{\mathrm{2}}}  \vdash  \ottnt{H}  \ottsym{:}  \ottnt{C} \,  |  \, \epsilon  \Rightarrow  \ottnt{D} \,  |  \, \epsilon'$.
 \end{enumerate}
\end{lemma}
\begin{proof}
 By mutual induction on the typing derivations.
\end{proof}

\begin{lemma}{eff-ectx-op-typing}
 If $\mathit{ty} \, \ottsym{(}  \mathsf{op}  \ottsym{)} \,  =  \,   \text{\unboldmath$\forall$}  \,  \algeffseqoverindex{ \alpha }{ \text{\unboldmath$\mathit{I}$} }   \ottsym{.} \,  \ottnt{A}  \hookrightarrow  \ottnt{B} $ and $\Gamma  \vdash   \ottnt{E}  [   \textup{\texttt{\#}\relax}  \mathsf{op}   \ottsym{(}   \ottnt{v}   \ottsym{)}   ]   \ottsym{:}  \ottnt{C} \,  |  \, \epsilon$ and
 $\mathsf{op} \,  \not\in  \, \ottnt{E}$, then
 \begin{itemize}
  \item $\Gamma  \ottsym{,}   \algeffseqoverindex{ \beta }{ \text{\unboldmath$\mathit{J}$} }   \vdash   \algeffseqoverindex{ \ottnt{D} }{ \text{\unboldmath$\mathit{I}$} } $,
  \item $\Gamma  \ottsym{,}   \algeffseqoverindex{ \beta }{ \text{\unboldmath$\mathit{J}$} }   \vdash  \ottnt{v}  \ottsym{:}   \ottnt{A}    [   \algeffseqoverindex{ \ottnt{D} }{ \text{\unboldmath$\mathit{I}$} }   \ottsym{/}   \algeffseqoverindex{ \alpha }{ \text{\unboldmath$\mathit{I}$} }   ]   \,  |  \, \epsilon'$, and
  \item for any $\mathit{y} \,  \not\in  \,  \mathit{dom}  (  \Gamma  ) $,
        $\Gamma  \ottsym{,}  \mathit{y} \,  \mathord{:}  \,  \text{\unboldmath$\forall$}  \,  \algeffseqoverindex{ \beta }{ \text{\unboldmath$\mathit{J}$} }   \ottsym{.} \, \ottnt{B} \,  [   \algeffseqoverindex{ \ottnt{D} }{ \text{\unboldmath$\mathit{I}$} }   \ottsym{/}   \algeffseqoverindex{ \alpha }{ \text{\unboldmath$\mathit{I}$} }   ]   \vdash   \ottnt{E}  [  \mathit{y}  ]   \ottsym{:}  \ottnt{C} \,  |  \, \epsilon$
 \end{itemize}
 for some $ \algeffseqoverindex{ \beta }{ \text{\unboldmath$\mathit{J}$} } $, $ \algeffseqoverindex{ \ottnt{D} }{ \text{\unboldmath$\mathit{I}$} } $, and $\epsilon'$.
 Furthermore, if $ \algeffseqoverindex{ \beta }{ \text{\unboldmath$\mathit{J}$} } $ is not the empty sequence,
 $\mathit{SR} \, \ottsym{(}  \ottsym{\{}  \mathsf{op}  \ottsym{\}}  \ottsym{)}$ holds.
\end{lemma}
\begin{proof}
 By \reflem{eff-ectx-typing},
 \begin{itemize}
  \item $\Gamma  \ottsym{,}   \algeffseqoverindex{ \beta_{{\mathrm{1}}} }{ \text{\unboldmath$\mathit{J_{{\mathrm{1}}}}$} }   \vdash   \textup{\texttt{\#}\relax}  \mathsf{op}   \ottsym{(}   \ottnt{v}   \ottsym{)}   \ottsym{:}  \ottnt{C'} \,  |  \, \epsilon''$ and
  \item $\Gamma  \ottsym{,}  \mathit{y} \,  \mathord{:}  \,  \text{\unboldmath$\forall$}  \,  \algeffseqoverindex{ \beta_{{\mathrm{1}}} }{ \text{\unboldmath$\mathit{J_{{\mathrm{1}}}}$} }   \ottsym{.} \, \ottnt{C'}  \vdash   \ottnt{E}  [  \mathit{y}  ]   \ottsym{:}  \ottnt{C} \,  |  \, \epsilon$ for any $\mathit{y} \,  \not\in  \,  \mathit{dom}  (  \Gamma  ) $, and
  \item if $ \algeffseqoverindex{ \beta_{{\mathrm{1}}} }{ \text{\unboldmath$\mathit{J_{{\mathrm{1}}}}$} } $ is not the empty sequence,
        then $\mathit{SR} \, \ottsym{(}  \ottsym{\{}  \mathsf{op}  \ottsym{\}}  \ottsym{)}$ holds
 \end{itemize}
 for some $ \algeffseqoverindex{ \beta_{{\mathrm{1}}} }{ \text{\unboldmath$\mathit{J_{{\mathrm{1}}}}$} } $ and $\ottnt{C'}$.
 By \reflem{eff-term-inv-op},
 \begin{itemize}
  \item $\Gamma  \ottsym{,}   \algeffseqoverindex{ \beta_{{\mathrm{1}}} }{ \text{\unboldmath$\mathit{J_{{\mathrm{1}}}}$} }   \ottsym{,}   \algeffseqoverindex{ \beta_{{\mathrm{2}}} }{ \text{\unboldmath$\mathit{J_{{\mathrm{2}}}}$} }   \vdash   \algeffseqoverindex{ \ottnt{D} }{ \text{\unboldmath$\mathit{I}$} } $,
  \item $\Gamma  \ottsym{,}   \algeffseqoverindex{ \beta_{{\mathrm{1}}} }{ \text{\unboldmath$\mathit{J_{{\mathrm{1}}}}$} }   \ottsym{,}   \algeffseqoverindex{ \beta_{{\mathrm{2}}} }{ \text{\unboldmath$\mathit{J_{{\mathrm{2}}}}$} }   \vdash  \ottnt{v}  \ottsym{:}   \ottnt{A}    [   \algeffseqoverindex{ \ottnt{D} }{ \text{\unboldmath$\mathit{I}$} }   \ottsym{/}   \algeffseqoverindex{ \alpha }{ \text{\unboldmath$\mathit{I}$} }   ]   \,  |  \, \epsilon'$,
  \item $\Gamma  \ottsym{,}   \algeffseqoverindex{ \beta_{{\mathrm{1}}} }{ \text{\unboldmath$\mathit{J_{{\mathrm{1}}}}$} }   \vdash    \text{\unboldmath$\forall$}  \,  \algeffseqoverindex{ \beta_{{\mathrm{2}}} }{ \text{\unboldmath$\mathit{J_{{\mathrm{2}}}}$} }   \ottsym{.} \, \ottnt{B}    [   \algeffseqoverindex{ \ottnt{D} }{ \text{\unboldmath$\mathit{I}$} }   \ottsym{/}   \algeffseqoverindex{ \alpha }{ \text{\unboldmath$\mathit{I}$} }   ]    \sqsubseteq  \ottnt{C'}$, and
  \item if $ \algeffseqoverindex{ \beta_{{\mathrm{2}}} }{ \text{\unboldmath$\mathit{J_{{\mathrm{2}}}}$} } $ is not the empty sequence,
        $\mathit{SR} \, \ottsym{(}  \ottsym{\{}  \mathsf{op}  \ottsym{\}}  \ottsym{)}$ holds
 \end{itemize}
 for some $ \algeffseqoverindex{ \beta_{{\mathrm{2}}} }{ \text{\unboldmath$\mathit{J_{{\mathrm{2}}}}$} } $, $ \algeffseqoverindex{ \ottnt{D} }{ \text{\unboldmath$\mathit{I}$} } $, and $\epsilon'$.

 We show the conclusion by letting $ \algeffseqoverindex{ \beta }{ \text{\unboldmath$\mathit{J}$} }  \,  =  \,  \algeffseqoverindex{ \beta_{{\mathrm{1}}} }{ \text{\unboldmath$\mathit{J_{{\mathrm{1}}}}$} }   \ottsym{,}   \algeffseqoverindex{ \beta_{{\mathrm{2}}} }{ \text{\unboldmath$\mathit{J_{{\mathrm{2}}}}$} } $.
 It suffices to show that, for any $\mathit{y} \,  \not\in  \,  \mathit{dom}  (  \Gamma  ) $,
 \[
 \Gamma  \ottsym{,}  \mathit{y} \,  \mathord{:}  \,  \text{\unboldmath$\forall$}  \,  \algeffseqoverindex{ \beta_{{\mathrm{1}}} }{ \text{\unboldmath$\mathit{J_{{\mathrm{1}}}}$} }   \ottsym{.} \,  \text{\unboldmath$\forall$}  \,  \algeffseqoverindex{ \beta_{{\mathrm{2}}} }{ \text{\unboldmath$\mathit{J_{{\mathrm{2}}}}$} }   \ottsym{.} \, \ottnt{B} \,  [   \algeffseqoverindex{ \ottnt{D} }{ \text{\unboldmath$\mathit{I}$} }   \ottsym{/}   \algeffseqoverindex{ \alpha }{ \text{\unboldmath$\mathit{I}$} }   ]   \vdash   \ottnt{E}  [  \mathit{y}  ]   \ottsym{:}  \ottnt{C} \,  |  \, \epsilon.
 \]
 Since $\Gamma  \ottsym{,}   \algeffseqoverindex{ \beta_{{\mathrm{1}}} }{ \text{\unboldmath$\mathit{J_{{\mathrm{1}}}}$} }   \vdash    \text{\unboldmath$\forall$}  \,  \algeffseqoverindex{ \beta_{{\mathrm{2}}} }{ \text{\unboldmath$\mathit{J_{{\mathrm{2}}}}$} }   \ottsym{.} \, \ottnt{B}    [   \algeffseqoverindex{ \ottnt{D} }{ \text{\unboldmath$\mathit{I}$} }   \ottsym{/}   \algeffseqoverindex{ \alpha }{ \text{\unboldmath$\mathit{I}$} }   ]    \sqsubseteq  \ottnt{C'}$,
 we have
 \[
 \Gamma  \vdash    \text{\unboldmath$\forall$}  \,  \algeffseqoverindex{ \beta_{{\mathrm{1}}} }{ \text{\unboldmath$\mathit{J_{{\mathrm{1}}}}$} }   \ottsym{.} \,  \text{\unboldmath$\forall$}  \,  \algeffseqoverindex{ \beta_{{\mathrm{2}}} }{ \text{\unboldmath$\mathit{J_{{\mathrm{2}}}}$} }   \ottsym{.} \, \ottnt{B}    [   \algeffseqoverindex{ \ottnt{D} }{ \text{\unboldmath$\mathit{I}$} }   \ottsym{/}   \algeffseqoverindex{ \alpha }{ \text{\unboldmath$\mathit{I}$} }   ]    \sqsubseteq   \text{\unboldmath$\forall$}  \,  \algeffseqoverindex{ \beta_{{\mathrm{1}}} }{ \text{\unboldmath$\mathit{J_{{\mathrm{1}}}}$} }   \ottsym{.} \, \ottnt{C'}
 \]
 by \Srule{Poly}.
 Since $\Gamma  \ottsym{,}  \mathit{y} \,  \mathord{:}  \,  \text{\unboldmath$\forall$}  \,  \algeffseqoverindex{ \beta_{{\mathrm{1}}} }{ \text{\unboldmath$\mathit{J_{{\mathrm{1}}}}$} }   \ottsym{.} \, \ottnt{C'}  \vdash   \ottnt{E}  [  \mathit{y}  ]   \ottsym{:}  \ottnt{C} \,  |  \, \epsilon$,
 we have
 \[
 \Gamma  \ottsym{,}  \mathit{y} \,  \mathord{:}  \,  \text{\unboldmath$\forall$}  \,  \algeffseqoverindex{ \beta_{{\mathrm{1}}} }{ \text{\unboldmath$\mathit{J_{{\mathrm{1}}}}$} }   \ottsym{.} \,  \text{\unboldmath$\forall$}  \,  \algeffseqoverindex{ \beta_{{\mathrm{2}}} }{ \text{\unboldmath$\mathit{J_{{\mathrm{2}}}}$} }   \ottsym{.} \, \ottnt{B} \,  [   \algeffseqoverindex{ \ottnt{D} }{ \text{\unboldmath$\mathit{I}$} }   \ottsym{/}   \algeffseqoverindex{ \alpha }{ \text{\unboldmath$\mathit{I}$} }   ]   \vdash   \ottnt{E}  [  \mathit{y}  ]   \ottsym{:}  \ottnt{C} \,  |  \, \epsilon.
 \]
 by \reflem{eff-var-subtype}.
\end{proof}

\begin{lemma}{eff-val-any-eff}
 If $\Gamma  \vdash  \ottnt{v}  \ottsym{:}  \ottnt{A} \,  |  \, \epsilon$, then
 $\Gamma  \vdash  \ottnt{v}  \ottsym{:}  \ottnt{A} \,  |  \, \epsilon'$ for any $\epsilon'$.
\end{lemma}
\begin{proof}
 Straightforward by induction on the typing derivation.
\end{proof}

\begin{lemma}{eff-subtyping-forall-move}
 Suppose that $\alpha$ does not appear free in $\ottnt{A}$.
 \begin{enumerate}
  \item \label{lem:eff-subtyping-forall-move:neg}
        Suppose that
        (1) the occurrences of $\beta$ in $\ottnt{A}$ are only negative or strictly
        positive and
        (2) for any function type $ \ottnt{C}   \rightarrow ^{ \epsilon }  \ottnt{D} $ occurring at a strictly
        positive position of $\ottnt{A}$, if $\beta \,  \in  \,  \mathit{ftv}  (  \ottnt{D}  ) $, then $\mathit{SR} \, \ottsym{(}  \epsilon  \ottsym{)}$.
        Then $\Gamma  \vdash    \text{\unboldmath$\forall$}  \, \alpha  \ottsym{.} \, \ottnt{A}    [  \ottnt{B}  \ottsym{/}  \beta  ]    \sqsubseteq   \ottnt{A}    [   \text{\unboldmath$\forall$}  \, \alpha  \ottsym{.} \, \ottnt{B}  \ottsym{/}  \beta  ]  $.
  \item \label{lem:eff-subtyping-forall-move:pos}
        If the occurrences of $\beta$ in $\ottnt{A}$ are only positive,
        then $\Gamma  \vdash   \ottnt{A}    [   \text{\unboldmath$\forall$}  \, \alpha  \ottsym{.} \, \ottnt{B}  \ottsym{/}  \beta  ]    \sqsubseteq    \text{\unboldmath$\forall$}  \, \alpha  \ottsym{.} \, \ottnt{A}    [  \ottnt{B}  \ottsym{/}  \beta  ]  $.
 \end{enumerate}
\end{lemma}
\begin{proof}
 By induction on $\ottnt{A}$.
 The second case is proven by \reflem{subtyping-forall-remove}, \Srule{Poly},
 \Srule{Gen}, and \Srule{Trans}.

 Let us consider the second case.  We consider the case that $\ottnt{A} \,  =  \,  \ottnt{C}   \rightarrow ^{ \epsilon }  \ottnt{D} $
 for some $\ottnt{C}$, $\ottnt{D}$, and $\epsilon$; the other cases are shown similarly
 to \reflem{subtyping-forall-move}.
 By the IH on $\ottnt{C}$, $\Gamma  \vdash   \ottnt{C}    [   \text{\unboldmath$\forall$}  \, \alpha  \ottsym{.} \, \ottnt{B}  \ottsym{/}  \beta  ]    \sqsubseteq    \text{\unboldmath$\forall$}  \, \alpha  \ottsym{.} \, \ottnt{C}    [  \ottnt{B}  \ottsym{/}  \beta  ]  $.

 Now, we show that
 \begin{equation}
  \Gamma  \vdash     \text{\unboldmath$\forall$}  \, \alpha  \ottsym{.} \, \ottsym{(}    \text{\unboldmath$\forall$}  \, \alpha  \ottsym{.} \, \ottnt{C}    [  \ottnt{B}  \ottsym{/}  \beta  ]    \ottsym{)}   \rightarrow ^{ \epsilon }  \ottnt{D}     [  \ottnt{B}  \ottsym{/}  \beta  ]    \sqsubseteq     \ottnt{C}    [   \text{\unboldmath$\forall$}  \, \alpha  \ottsym{.} \, \ottnt{B}  \ottsym{/}  \beta  ]     \rightarrow ^{ \epsilon }  \ottnt{D}     [   \text{\unboldmath$\forall$}  \, \alpha  \ottsym{.} \, \ottnt{B}  \ottsym{/}  \beta  ]  .
   \label{eqn:eff-subtyping-forall-move:fun:one}
 \end{equation}
 If $\beta \,  \in  \,  \mathit{ftv}  (  \ottnt{D}  ) $, then $\mathit{SR} \, \ottsym{(}  \epsilon  \ottsym{)}$ by the assumption.
 By the IH on $\ottnt{D}$, $\Gamma  \vdash    \text{\unboldmath$\forall$}  \, \alpha  \ottsym{.} \, \ottnt{D}    [  \ottnt{B}  \ottsym{/}  \beta  ]    \sqsubseteq   \ottnt{D}    [   \text{\unboldmath$\forall$}  \, \alpha  \ottsym{.} \, \ottnt{B}  \ottsym{/}  \beta  ]  $.
 By \Srule{FunEff},
 \[
  \Gamma  \vdash    \ottsym{(}    \text{\unboldmath$\forall$}  \, \alpha  \ottsym{.} \, \ottnt{C}    [  \ottnt{B}  \ottsym{/}  \beta  ]    \ottsym{)}   \rightarrow ^{ \epsilon }   \text{\unboldmath$\forall$}  \, \alpha  \ottsym{.} \, \ottnt{D}     [  \ottnt{B}  \ottsym{/}  \beta  ]    \sqsubseteq     \ottnt{C}    [   \text{\unboldmath$\forall$}  \, \alpha  \ottsym{.} \, \ottnt{B}  \ottsym{/}  \beta  ]     \rightarrow ^{ \epsilon }  \ottnt{D}     [   \text{\unboldmath$\forall$}  \, \alpha  \ottsym{.} \, \ottnt{B}  \ottsym{/}  \beta  ]  .
 \]
 Since $\mathit{SR} \, \ottsym{(}  \epsilon  \ottsym{)}$,
 we have (\ref{eqn:eff-subtyping-forall-move:fun:one})
 by \Srule{DFunEff} and \Srule{Trans}.
 Otherwise, if $\beta \,  \not\in  \,  \mathit{ftv}  (  \ottnt{D}  ) $, then
 $\Gamma  \ottsym{,}  \alpha  \vdash   \ottnt{D}    [  \ottnt{B}  \ottsym{/}  \beta  ]    \sqsubseteq   \ottnt{D}    [   \text{\unboldmath$\forall$}  \, \alpha  \ottsym{.} \, \ottnt{B}  \ottsym{/}  \beta  ]  $ by \Srule{Refl}
 because $ \ottnt{D}    [  \ottnt{B}  \ottsym{/}  \beta  ]   =  \ottnt{D}    [   \text{\unboldmath$\forall$}  \, \alpha  \ottsym{.} \, \ottnt{B}  \ottsym{/}  \beta  ]   = \ottnt{D}$.
 Thus,
 \[
  \Gamma  \vdash     \text{\unboldmath$\forall$}  \, \alpha  \ottsym{.} \, \ottsym{(}    \text{\unboldmath$\forall$}  \, \alpha  \ottsym{.} \, \ottnt{C}    [  \ottnt{B}  \ottsym{/}  \beta  ]    \ottsym{)}   \rightarrow ^{ \epsilon }  \ottnt{D}     [  \ottnt{B}  \ottsym{/}  \beta  ]    \sqsubseteq      \text{\unboldmath$\forall$}  \, \alpha  \ottsym{.} \, \ottnt{C}    [   \text{\unboldmath$\forall$}  \, \alpha  \ottsym{.} \, \ottnt{B}  \ottsym{/}  \beta  ]     \rightarrow ^{ \epsilon }  \ottnt{D}     [   \text{\unboldmath$\forall$}  \, \alpha  \ottsym{.} \, \ottnt{B}  \ottsym{/}  \beta  ]  
 \]
 by \Srule{Poly} and \reflem{eff-weakening} (\ref{lem:eff-weakening:sub}).
 Since $\alpha$ does not occur in $\ottnt{A} =  \ottnt{C}   \rightarrow ^{ \epsilon }  \ottnt{D} $,
 we can have (\ref{eqn:eff-subtyping-forall-move:fun:one})
 by eliminating the outermost $\forall$ on the RHS type with
 \Srule{Inst}.

 By \Srule{Inst},
 \begin{equation}
  \Gamma  \ottsym{,}  \alpha  \vdash    \text{\unboldmath$\forall$}  \, \alpha  \ottsym{.} \, \ottnt{C}    [  \ottnt{B}  \ottsym{/}  \beta  ]    \sqsubseteq   \ottnt{C}    [  \ottnt{B}  \ottsym{/}  \beta  ]  .
   \label{eqn:eff-subtyping-forall-move:fun:two}
 \end{equation}
 By \Srule{FunEff} and \Srule{Poly} with (\ref{eqn:eff-subtyping-forall-move:fun:two}),
 \[
  \Gamma  \vdash      \text{\unboldmath$\forall$}  \, \alpha  \ottsym{.} \, \ottnt{C}    [  \ottnt{B}  \ottsym{/}  \beta  ]     \rightarrow ^{ \epsilon }  \ottnt{D}     [  \ottnt{B}  \ottsym{/}  \beta  ]    \sqsubseteq     \text{\unboldmath$\forall$}  \, \alpha  \ottsym{.} \, \ottsym{(}    \text{\unboldmath$\forall$}  \, \alpha  \ottsym{.} \, \ottnt{C}    [  \ottnt{B}  \ottsym{/}  \beta  ]    \ottsym{)}   \rightarrow ^{ \epsilon }  \ottnt{D}     [  \ottnt{B}  \ottsym{/}  \beta  ]  .
 \]
 Thus, by \Srule{Trans} with (\ref{eqn:eff-subtyping-forall-move:fun:one}),
 \[
  \Gamma  \vdash      \text{\unboldmath$\forall$}  \, \alpha  \ottsym{.} \, \ottnt{C}    [  \ottnt{B}  \ottsym{/}  \beta  ]     \rightarrow ^{ \epsilon }  \ottnt{D}     [  \ottnt{B}  \ottsym{/}  \beta  ]    \sqsubseteq     \ottnt{C}    [   \text{\unboldmath$\forall$}  \, \alpha  \ottsym{.} \, \ottnt{B}  \ottsym{/}  \beta  ]     \rightarrow ^{ \epsilon }  \ottnt{D}     [   \text{\unboldmath$\forall$}  \, \alpha  \ottsym{.} \, \ottnt{B}  \ottsym{/}  \beta  ]  .
 \]
\end{proof}

\begin{lemmap}{Subject reduction}{eff-subject-red}
 \begin{enumerate}
  \item If $\Delta  \vdash  \ottnt{M_{{\mathrm{1}}}}  \ottsym{:}  \ottnt{A} \,  |  \, \epsilon$ and $\ottnt{M_{{\mathrm{1}}}}  \rightsquigarrow  \ottnt{M_{{\mathrm{2}}}}$,
        then $\Delta  \vdash  \ottnt{M_{{\mathrm{2}}}}  \ottsym{:}  \ottnt{A} \,  |  \, \epsilon$.
  \item If $\Delta  \vdash  \ottnt{M_{{\mathrm{1}}}}  \ottsym{:}  \ottnt{A} \,  |  \, \epsilon$ and $\ottnt{M_{{\mathrm{1}}}}  \longrightarrow  \ottnt{M_{{\mathrm{2}}}}$,
        then $\Delta  \vdash  \ottnt{M_{{\mathrm{2}}}}  \ottsym{:}  \ottnt{A} \,  |  \, \epsilon$.
 \end{enumerate}
\end{lemmap}
\begin{proof}
 \begin{enumerate}
  \item By induction on the typing derivation.  Most of the cases are similar to
        \reflem{subject-red}.  We here focus on the cases that need a treatment
        specific to the type-and-effect system.
        \begin{caseanalysis}
         \case \Te{App}/\R{Beta}:
          We are given
          \begin{itemize}
           \item $\ottnt{M_{{\mathrm{1}}}} \,  =  \, \ottsym{(}   \lambda\!  \, \mathit{x}  \ottsym{.}  \ottnt{M}  \ottsym{)} \, \ottnt{v}$,
           \item $\ottnt{M_{{\mathrm{2}}}} \,  =  \,  \ottnt{M}    [  \ottnt{v}  /  \mathit{x}  ]  $,
           \item $\Delta  \vdash  \ottsym{(}   \lambda\!  \, \mathit{x}  \ottsym{.}  \ottnt{M}  \ottsym{)} \, \ottnt{v}  \ottsym{:}  \ottnt{A} \,  |  \, \epsilon$,
           \item $\Delta  \vdash   \lambda\!  \, \mathit{x}  \ottsym{.}  \ottnt{M}  \ottsym{:}   \ottnt{B}   \rightarrow ^{ \epsilon_{{\mathrm{0}}} }  \ottnt{A}  \,  |  \, \epsilon$,
           \item $\Delta  \vdash  \ottnt{v}  \ottsym{:}  \ottnt{B} \,  |  \, \epsilon$, and
           \item $\epsilon_{{\mathrm{0}}} \,  \subseteq  \, \epsilon$
          \end{itemize}
          for some $\mathit{x}$, $\ottnt{M}$, $\ottnt{v}$, $\ottnt{B}$, and $\epsilon_{{\mathrm{0}}}$.
          By \reflem{eff-val-inv-abs}
          $\Delta  \ottsym{,}   \algeffseqoverindex{ \alpha }{ \text{\unboldmath$\mathit{I}$} }   \ottsym{,}  \mathit{x} \,  \mathord{:}  \, \ottnt{B'}  \vdash  \ottnt{M}  \ottsym{:}  \ottnt{A'} \,  |  \, \epsilon'$ and
          $\Delta  \vdash    \text{\unboldmath$\forall$}  \,  \algeffseqoverindex{ \alpha }{ \text{\unboldmath$\mathit{I}$} }   \ottsym{.} \, \ottnt{B'}   \rightarrow ^{ \epsilon' }  \ottnt{A'}   \sqsubseteq   \ottnt{B}   \rightarrow ^{ \epsilon_{{\mathrm{0}}} }  \ottnt{A} $
          for some $ \algeffseqoverindex{ \alpha }{ \text{\unboldmath$\mathit{I}$} } $, $\ottnt{A'}$, $\ottnt{B'}$, and $\epsilon'$.
          By \reflem{eff-subtyping-inv-fun},
          we find $\epsilon' \,  =  \, \epsilon_{{\mathrm{0}}}$, and
          there exist $ \algeffseqoverindex{ \alpha_{{\mathrm{1}}} }{ \text{\unboldmath$\mathit{I_{{\mathrm{1}}}}$} } $, $ \algeffseqoverindex{ \alpha_{{\mathrm{2}}} }{ \text{\unboldmath$\mathit{I_{{\mathrm{2}}}}$} } $, $ \algeffseqoverindex{ \beta }{ \text{\unboldmath$\mathit{J}$} } $, and $ \algeffseqoverindex{ \ottnt{C} }{ \text{\unboldmath$\mathit{I_{{\mathrm{1}}}}$} } $
          such that
          \begin{itemize}
           \item $\ottsym{\{}   \algeffseqoverindex{ \alpha }{ \text{\unboldmath$\mathit{I}$} }   \ottsym{\}} \,  =  \, \ottsym{\{}   \algeffseqoverindex{ \alpha_{{\mathrm{1}}} }{ \text{\unboldmath$\mathit{I_{{\mathrm{1}}}}$} }   \ottsym{\}} \,  \mathbin{\uplus}  \, \ottsym{\{}   \algeffseqoverindex{ \alpha_{{\mathrm{2}}} }{ \text{\unboldmath$\mathit{I_{{\mathrm{2}}}}$} }   \ottsym{\}}$,
           \item $\Delta  \ottsym{,}   \algeffseqoverindex{ \beta }{ \text{\unboldmath$\mathit{J}$} }   \vdash   \algeffseqoverindex{ \ottnt{C} }{ \text{\unboldmath$\mathit{I_{{\mathrm{1}}}}$} } $,
           \item $\Delta  \vdash  \ottnt{B}  \sqsubseteq    \text{\unboldmath$\forall$}  \,  \algeffseqoverindex{ \beta }{ \text{\unboldmath$\mathit{J}$} }   \ottsym{.} \, \ottnt{B'}    [   \algeffseqoverindex{ \ottnt{C} }{ \text{\unboldmath$\mathit{I_{{\mathrm{1}}}}$} }   \ottsym{/}   \algeffseqoverindex{ \alpha_{{\mathrm{1}}} }{ \text{\unboldmath$\mathit{I_{{\mathrm{1}}}}$} }   ]  $,
           \item $\Delta  \vdash    \text{\unboldmath$\forall$}  \,  \algeffseqoverindex{ \alpha_{{\mathrm{2}}} }{ \text{\unboldmath$\mathit{I_{{\mathrm{2}}}}$} }   \ottsym{.} \,  \text{\unboldmath$\forall$}  \,  \algeffseqoverindex{ \beta }{ \text{\unboldmath$\mathit{J}$} }   \ottsym{.} \, \ottnt{A'}    [   \algeffseqoverindex{ \ottnt{C} }{ \text{\unboldmath$\mathit{I_{{\mathrm{1}}}}$} }   \ottsym{/}   \algeffseqoverindex{ \alpha_{{\mathrm{1}}} }{ \text{\unboldmath$\mathit{I_{{\mathrm{1}}}}$} }   ]    \sqsubseteq  \ottnt{A}$, and
           \item type variables in $ \algeffseqoverindex{ \beta }{ \text{\unboldmath$\mathit{J}$} } $ do not appear free in
                 $\ottnt{A'}$ and $\ottnt{B'}$, and
           \item If $ \algeffseqoverindex{ \alpha_{{\mathrm{2}}} }{ \text{\unboldmath$\mathit{I_{{\mathrm{2}}}}$} } $ or $ \algeffseqoverindex{ \beta }{ \text{\unboldmath$\mathit{J}$} } $ is not the empty sequence,
                 $\mathit{SR} \, \ottsym{(}  \epsilon_{{\mathrm{0}}}  \ottsym{)}$.
          \end{itemize}
         By \reflem{eff-weakening},
         $\Delta  \ottsym{,}   \algeffseqoverindex{ \beta }{ \text{\unboldmath$\mathit{J}$} }   \ottsym{,}   \algeffseqoverindex{ \alpha }{ \text{\unboldmath$\mathit{I}$} }   \ottsym{,}  \mathit{x} \,  \mathord{:}  \, \ottnt{B'}  \vdash  \ottnt{M}  \ottsym{:}  \ottnt{A'} \,  |  \, \epsilon'$ and
         $\Delta  \ottsym{,}   \algeffseqoverindex{ \beta }{ \text{\unboldmath$\mathit{J}$} }   \ottsym{,}   \algeffseqoverindex{ \alpha_{{\mathrm{2}}} }{ \text{\unboldmath$\mathit{I_{{\mathrm{2}}}}$} }   \vdash   \algeffseqoverindex{ \ottnt{C} }{ \text{\unboldmath$\mathit{I_{{\mathrm{1}}}}$} } $.
         Thus, by \reflem{eff-ty-subst} (\ref{lem:eff-ty-subst:term}),
         \begin{equation}
          \Delta  \ottsym{,}   \algeffseqoverindex{ \beta }{ \text{\unboldmath$\mathit{J}$} }   \ottsym{,}   \algeffseqoverindex{ \alpha_{{\mathrm{2}}} }{ \text{\unboldmath$\mathit{I_{{\mathrm{2}}}}$} }   \ottsym{,}  \mathit{x} \,  \mathord{:}  \, \ottnt{B'} \,  [   \algeffseqoverindex{ \ottnt{C} }{ \text{\unboldmath$\mathit{I_{{\mathrm{1}}}}$} }   \ottsym{/}   \algeffseqoverindex{ \alpha }{ \text{\unboldmath$\mathit{I_{{\mathrm{1}}}}$} }   ]   \vdash  \ottnt{M}  \ottsym{:}   \ottnt{A'}    [   \algeffseqoverindex{ \ottnt{C} }{ \text{\unboldmath$\mathit{I_{{\mathrm{1}}}}$} }   \ottsym{/}   \algeffseqoverindex{ \alpha_{{\mathrm{1}}} }{ \text{\unboldmath$\mathit{I_{{\mathrm{1}}}}$} }   ]   \,  |  \, \epsilon'
           \label{eqn:eff-subject-red:app:beta:body}
         \end{equation}

         Since $\Delta  \vdash  \ottnt{v}  \ottsym{:}  \ottnt{B} \,  |  \, \epsilon$ and $\Delta  \vdash  \ottnt{B}  \sqsubseteq    \text{\unboldmath$\forall$}  \,  \algeffseqoverindex{ \beta }{ \text{\unboldmath$\mathit{J}$} }   \ottsym{.} \, \ottnt{B'}    [   \algeffseqoverindex{ \ottnt{C} }{ \text{\unboldmath$\mathit{I_{{\mathrm{1}}}}$} }   \ottsym{/}   \algeffseqoverindex{ \alpha_{{\mathrm{1}}} }{ \text{\unboldmath$\mathit{I_{{\mathrm{1}}}}$} }   ]  $,
         we have
         \[
         \Delta  \vdash  \ottnt{v}  \ottsym{:}    \text{\unboldmath$\forall$}  \,  \algeffseqoverindex{ \beta }{ \text{\unboldmath$\mathit{J}$} }   \ottsym{.} \, \ottnt{B'}    [   \algeffseqoverindex{ \ottnt{C} }{ \text{\unboldmath$\mathit{I_{{\mathrm{1}}}}$} }   \ottsym{/}   \algeffseqoverindex{ \alpha_{{\mathrm{1}}} }{ \text{\unboldmath$\mathit{I_{{\mathrm{1}}}}$} }   ]   \,  |  \, \epsilon
         \]
         by \Te{Inst} (note that $\Delta  \vdash    \text{\unboldmath$\forall$}  \,  \algeffseqoverindex{ \beta }{ \text{\unboldmath$\mathit{J}$} }   \ottsym{.} \, \ottnt{B'}    [   \algeffseqoverindex{ \ottnt{C} }{ \text{\unboldmath$\mathit{I_{{\mathrm{1}}}}$} }   \ottsym{/}   \algeffseqoverindex{ \alpha_{{\mathrm{1}}} }{ \text{\unboldmath$\mathit{I_{{\mathrm{1}}}}$} }   ]  $ is
         shown easily with \reflem{type-wf}).
         By \reflem{eff-weakening} (\ref{lem:weakening:term}),
         \Srule{Inst}, and \Te{Inst},
         we have
         \[
         \Delta  \ottsym{,}   \algeffseqoverindex{ \beta }{ \text{\unboldmath$\mathit{J}$} }   \ottsym{,}   \algeffseqoverindex{ \alpha_{{\mathrm{2}}} }{ \text{\unboldmath$\mathit{I_{{\mathrm{2}}}}$} }   \vdash  \ottnt{v}  \ottsym{:}   \ottnt{B'}    [   \algeffseqoverindex{ \ottnt{C} }{ \text{\unboldmath$\mathit{I_{{\mathrm{1}}}}$} }   \ottsym{/}   \algeffseqoverindex{ \alpha }{ \text{\unboldmath$\mathit{I_{{\mathrm{1}}}}$} }   ]   \,  |  \, \epsilon.
         \]
         By Lemmas~\ref{lem:eff-val-any-eff} and \ref{lem:eff-term-subst} (\ref{lem:term-subst:term})
         with (\ref{eqn:eff-subject-red:app:beta:body}),
         \[
         \Delta  \ottsym{,}   \algeffseqoverindex{ \beta }{ \text{\unboldmath$\mathit{J}$} }   \ottsym{,}   \algeffseqoverindex{ \alpha_{{\mathrm{2}}} }{ \text{\unboldmath$\mathit{I_{{\mathrm{2}}}}$} }   \vdash   \ottnt{M}    [  \ottnt{v}  /  \mathit{x}  ]    \ottsym{:}   \ottnt{A'}    [   \algeffseqoverindex{ \ottnt{C} }{ \text{\unboldmath$\mathit{I_{{\mathrm{1}}}}$} }   \ottsym{/}   \algeffseqoverindex{ \alpha_{{\mathrm{1}}} }{ \text{\unboldmath$\mathit{I_{{\mathrm{1}}}}$} }   ]   \,  |  \, \epsilon'.
         \]
         By \Te{Gen} (with permutation of the bindings in the typing context),
         \[
         \Delta  \vdash   \ottnt{M}    [  \ottnt{v}  /  \mathit{x}  ]    \ottsym{:}    \text{\unboldmath$\forall$}  \,  \algeffseqoverindex{ \alpha_{{\mathrm{2}}} }{ \text{\unboldmath$\mathit{I_{{\mathrm{2}}}}$} }   \ottsym{.} \,  \text{\unboldmath$\forall$}  \,  \algeffseqoverindex{ \beta }{ \text{\unboldmath$\mathit{J}$} }   \ottsym{.} \, \ottnt{A'}    [   \algeffseqoverindex{ \ottnt{C} }{ \text{\unboldmath$\mathit{I_{{\mathrm{1}}}}$} }   \ottsym{/}   \algeffseqoverindex{ \alpha_{{\mathrm{1}}} }{ \text{\unboldmath$\mathit{I_{{\mathrm{1}}}}$} }   ]   \,  |  \, \epsilon'
         \]
         (note that If $ \algeffseqoverindex{ \alpha_{{\mathrm{2}}} }{ \text{\unboldmath$\mathit{I_{{\mathrm{2}}}}$} } $ or $ \algeffseqoverindex{ \beta }{ \text{\unboldmath$\mathit{J}$} } $ is not the empty sequence,
         $\mathit{SR} \, \ottsym{(}  \epsilon'  \ottsym{)}$).
         Since $\Delta  \vdash    \text{\unboldmath$\forall$}  \,  \algeffseqoverindex{ \alpha_{{\mathrm{2}}} }{ \text{\unboldmath$\mathit{I_{{\mathrm{2}}}}$} }   \ottsym{.} \,  \text{\unboldmath$\forall$}  \,  \algeffseqoverindex{ \beta }{ \text{\unboldmath$\mathit{J}$} }   \ottsym{.} \, \ottnt{A'}    [   \algeffseqoverindex{ \ottnt{C} }{ \text{\unboldmath$\mathit{I_{{\mathrm{1}}}}$} }   \ottsym{/}   \algeffseqoverindex{ \alpha_{{\mathrm{1}}} }{ \text{\unboldmath$\mathit{I_{{\mathrm{1}}}}$} }   ]    \sqsubseteq  \ottnt{A}$,
         we have $\Delta  \vdash   \ottnt{M}    [  \ottnt{v}  /  \mathit{x}  ]    \ottsym{:}  \ottnt{A} \,  |  \, \epsilon'$ by \Te{Inst}.
         Since $\epsilon' \,  \subseteq  \, \epsilon$,
         we have
         \[
          \Delta  \vdash   \ottnt{M}    [  \ottnt{v}  /  \mathit{x}  ]    \ottsym{:}  \ottnt{A} \,  |  \, \epsilon
         \]
         by \Te{Weak}.

         \case \Te{Gen}: By the IH and \Te{Gen}.
         \case \Te{Handle}/\R{Handle}:
          We are given
          \begin{itemize}
           \item $\ottnt{M_{{\mathrm{1}}}} \,  =  \, \mathsf{handle} \,  \ottnt{E}  [   \textup{\texttt{\#}\relax}  \mathsf{op}   \ottsym{(}   \ottnt{v}   \ottsym{)}   ]  \, \mathsf{with} \, \ottnt{H}$,
           \item $\mathsf{op} \,  \not\in  \, \ottnt{E}$,
           \item $\ottnt{H}  \ottsym{(}  \mathsf{op}  \ottsym{)} \,  =  \, \mathsf{op}  \ottsym{(}  \mathit{x}  \ottsym{,}  \mathit{k}  \ottsym{)}  \rightarrow  \ottnt{M}$,
           \item $\ottnt{M_{{\mathrm{2}}}} \,  =  \,   \ottnt{M}    [  \ottnt{v}  /  \mathit{x}  ]      [   \lambda\!  \, \mathit{y}  \ottsym{.}  \mathsf{handle} \,  \ottnt{E}  [  \mathit{y}  ]  \, \mathsf{with} \, \ottnt{H}  /  \mathit{k}  ]  $,
           \item $\Delta  \vdash  \mathsf{handle} \,  \ottnt{E}  [   \textup{\texttt{\#}\relax}  \mathsf{op}   \ottsym{(}   \ottnt{v}   \ottsym{)}   ]  \, \mathsf{with} \, \ottnt{H}  \ottsym{:}  \ottnt{A} \,  |  \, \epsilon$,
           \item $\Delta  \vdash   \ottnt{E}  [   \textup{\texttt{\#}\relax}  \mathsf{op}   \ottsym{(}   \ottnt{v}   \ottsym{)}   ]   \ottsym{:}  \ottnt{B} \,  |  \, \epsilon'$,
           \item $\Delta  \vdash  \ottnt{H}  \ottsym{:}  \ottnt{B} \,  |  \, \epsilon'  \Rightarrow  \ottnt{A} \,  |  \, \epsilon$
          \end{itemize}
          for some $\ottnt{E}$, $\mathsf{op}$, $\ottnt{v}$, $\ottnt{H}$, $\mathit{x}$, $\mathit{y}$,
          $\mathit{k}$, $\ottnt{M}$, $\ottnt{B}$, and $\epsilon'$.
          Suppose that $\mathit{ty} \, \ottsym{(}  \mathsf{op}  \ottsym{)} \,  =  \,   \text{\unboldmath$\forall$}  \,  \algeffseqover{ \alpha }   \ottsym{.} \,  \ottnt{C}  \hookrightarrow  \ottnt{D} $.
          By inversion of the derivation of
          $\Delta  \vdash  \ottnt{H}  \ottsym{:}  \ottnt{B} \,  |  \, \epsilon'  \Rightarrow  \ottnt{A} \,  |  \, \epsilon$,
          we have $\Delta  \ottsym{,}   \algeffseqover{ \alpha }   \ottsym{,}  \mathit{x} \,  \mathord{:}  \, \ottnt{C}  \ottsym{,}  \mathit{k} \,  \mathord{:}  \,  \ottnt{D}   \rightarrow ^{ \epsilon }  \ottnt{A}   \vdash  \ottnt{M}  \ottsym{:}  \ottnt{A} \,  |  \, \epsilon$.

          By \reflem{eff-ectx-op-typing},
          \begin{itemize}
           \item $\Delta  \ottsym{,}   \algeffseqoverindex{ \beta }{ \text{\unboldmath$\mathit{J}$} }   \vdash   \algeffseqover{ \ottnt{C_{{\mathrm{0}}}} } $,
           \item $\Delta  \ottsym{,}   \algeffseqoverindex{ \beta }{ \text{\unboldmath$\mathit{J}$} }   \vdash  \ottnt{v}  \ottsym{:}   \ottnt{C}    [   \algeffseqover{ \ottnt{C_{{\mathrm{0}}}} }   \ottsym{/}   \algeffseqover{ \alpha }   ]   \,  |  \, \epsilon_{{\mathrm{0}}}$,
           \item $\Gamma  \ottsym{,}  \mathit{y} \,  \mathord{:}  \,  \text{\unboldmath$\forall$}  \,  \algeffseqoverindex{ \beta }{ \text{\unboldmath$\mathit{J}$} }   \ottsym{.} \, \ottnt{D} \,  [   \algeffseqover{ \ottnt{C_{{\mathrm{0}}}} }   \ottsym{/}   \algeffseqover{ \alpha }   ]   \vdash   \ottnt{E}  [  \mathit{y}  ]   \ottsym{:}  \ottnt{B} \,  |  \, \epsilon'$, and
           \item if $ \algeffseqoverindex{ \beta }{ \text{\unboldmath$\mathit{J}$} } $ is not the empty sequence, $\mathit{SR} \, \ottsym{(}  \ottsym{\{}  \mathsf{op}  \ottsym{\}}  \ottsym{)}$
          \end{itemize}
          for some $ \algeffseqoverindex{ \beta }{ \text{\unboldmath$\mathit{J}$} } $, $ \algeffseqover{ \ottnt{C_{{\mathrm{0}}}} } $, and $\epsilon_{{\mathrm{0}}}$.
          Since $\Delta  \vdash   \text{\unboldmath$\forall$}  \,  \algeffseqoverindex{ \beta }{ \text{\unboldmath$\mathit{J}$} }   \ottsym{.} \,  \algeffseqover{ \ottnt{C_{{\mathrm{0}}}} } $,
          \begin{equation}
           \Delta  \ottsym{,}  \mathit{x} \,  \mathord{:}  \, \ottnt{C} \,  [   \text{\unboldmath$\forall$}  \,  \algeffseqoverindex{ \beta }{ \text{\unboldmath$\mathit{J}$} }   \ottsym{.} \,  \algeffseqover{ \ottnt{C_{{\mathrm{0}}}} }   \ottsym{/}   \algeffseqover{ \alpha }   ]   \ottsym{,}  \mathit{k} \,  \mathord{:}  \,   \ottnt{D}    [   \text{\unboldmath$\forall$}  \,  \algeffseqoverindex{ \beta }{ \text{\unboldmath$\mathit{J}$} }   \ottsym{.} \,  \algeffseqover{ \ottnt{C_{{\mathrm{0}}}} }   \ottsym{/}   \algeffseqover{ \alpha }   ]     \rightarrow ^{ \epsilon }  \ottnt{A}   \vdash  \ottnt{M}  \ottsym{:}  \ottnt{A} \,  |  \, \epsilon
            \label{eqn:eff-subject-red:handle:handle:one}
          \end{equation}
          by \reflem{eff-ty-subst} (\ref{lem:eff-ty-subst:term}) (note that
          type variables in $ \algeffseqover{ \alpha } $ do not appear free in $\ottnt{A}$).
          Since $\Delta  \ottsym{,}   \algeffseqoverindex{ \beta }{ \text{\unboldmath$\mathit{J}$} }   \vdash  \ottnt{v}  \ottsym{:}   \ottnt{C}    [   \algeffseqover{ \ottnt{C_{{\mathrm{0}}}} }   \ottsym{/}   \algeffseqover{ \alpha }   ]   \,  |  \, \epsilon_{{\mathrm{0}}}$, we have
          $\Delta  \vdash  \ottnt{v}  \ottsym{:}    \text{\unboldmath$\forall$}  \,  \algeffseqoverindex{ \beta }{ \text{\unboldmath$\mathit{J}$} }   \ottsym{.} \, \ottnt{C}    [   \algeffseqover{ \ottnt{C_{{\mathrm{0}}}} }   \ottsym{/}   \algeffseqover{ \alpha }   ]   \,  |  \, \epsilon_{{\mathrm{0}}}$
          by \reflem{eff-val-any-eff} and \Te{Gen}.

          We show that $\Delta  \vdash  \ottnt{v}  \ottsym{:}   \ottnt{C}    [   \algeffseqover{  \text{\unboldmath$\forall$}  \,  \algeffseqoverindex{ \beta }{ \text{\unboldmath$\mathit{J}$} }   \ottsym{.} \, \ottnt{C_{{\mathrm{0}}}} }   \ottsym{/}   \algeffseqover{ \alpha }   ]   \,  |  \, \epsilon_{{\mathrm{0}}}$.
          If $ \algeffseqoverindex{ \beta }{ \text{\unboldmath$\mathit{J}$} } $ is not empty, then $\mathit{SR} \, \ottsym{(}  \ottsym{\{}  \mathsf{op}  \ottsym{\}}  \ottsym{)}$.
          Thus, we have the derivation
          by \reflem{eff-subtyping-forall-move} (\ref{lem:eff-subtyping-forall-move:neg})
          and \Te{Inst}
          (note that we can suppose that $ \algeffseqoverindex{ \beta }{ \text{\unboldmath$\mathit{J}$} } $ do not appear
          free in $\ottnt{C}$).
          Otherwise, if $ \algeffseqoverindex{ \beta }{ \text{\unboldmath$\mathit{J}$} } $ is empty, we also have it.

          By applying Lemmas~\ref{lem:eff-val-any-eff} and \ref{lem:eff-term-subst} (\ref{lem:eff-term-subst:term})
          to (\ref{eqn:eff-subject-red:handle:handle:one}), we have
          \begin{equation}
           \Delta  \ottsym{,}  \mathit{k} \,  \mathord{:}  \,   \ottnt{D}    [   \text{\unboldmath$\forall$}  \,  \algeffseqoverindex{ \beta }{ \text{\unboldmath$\mathit{J}$} }   \ottsym{.} \,  \algeffseqover{ \ottnt{C_{{\mathrm{0}}}} }   \ottsym{/}   \algeffseqover{ \alpha }   ]     \rightarrow ^{ \epsilon }  \ottnt{A}   \vdash   \ottnt{M}    [  \ottnt{v}  /  \mathit{x}  ]    \ottsym{:}  \ottnt{A} \,  |  \, \epsilon.
            \label{eqn:eff-subject-red:handle:handle:two}
          \end{equation}

         We show that
         \[
         \Delta  \vdash   \lambda\!  \, \mathit{y}  \ottsym{.}  \mathsf{handle} \,  \ottnt{E}  [  \mathit{y}  ]  \, \mathsf{with} \, \ottnt{H}  \ottsym{:}    \ottnt{D}    [   \text{\unboldmath$\forall$}  \,  \algeffseqoverindex{ \beta }{ \text{\unboldmath$\mathit{J}$} }   \ottsym{.} \,  \algeffseqover{ \ottnt{C_{{\mathrm{0}}}} }   \ottsym{/}   \algeffseqover{ \alpha }   ]     \rightarrow ^{ \epsilon }  \ottnt{A}  \,  |  \, \epsilon''
         \]
         for any $\epsilon''$.

         For that, we first show that
         $\Delta  \vdash   \ottnt{D}    [   \text{\unboldmath$\forall$}  \,  \algeffseqoverindex{ \beta }{ \text{\unboldmath$\mathit{J}$} }   \ottsym{.} \,  \algeffseqover{ \ottnt{C_{{\mathrm{0}}}} }   \ottsym{/}   \algeffseqover{ \alpha }   ]    \sqsubseteq    \text{\unboldmath$\forall$}  \,  \algeffseqoverindex{ \beta }{ \text{\unboldmath$\mathit{J}$} }   \ottsym{.} \, \ottnt{D}    [   \algeffseqover{ \ottnt{C_{{\mathrm{0}}}} }   \ottsym{/}   \algeffseqover{ \alpha }   ]  $.
         If $ \algeffseqoverindex{ \beta }{ \text{\unboldmath$\mathit{J}$} } $ is not empty, then $\mathit{SR} \, \ottsym{(}  \ottsym{\{}  \mathsf{op}  \ottsym{\}}  \ottsym{)}$.
         Thus, we have the derivation
         by \reflem{eff-subtyping-forall-move} (\ref{lem:eff-subtyping-forall-move:pos})
         (note that we can suppose that
         $ \algeffseqoverindex{ \beta }{ \text{\unboldmath$\mathit{J}$} } $ do not appear free in $\ottnt{D}$).
         Otherwise, if $ \algeffseqoverindex{ \beta }{ \text{\unboldmath$\mathit{J}$} } $ is empty, we also have it
         by \Srule{Refl}.

         Thus, since $\Gamma  \ottsym{,}  \mathit{y} \,  \mathord{:}  \,  \text{\unboldmath$\forall$}  \,  \algeffseqoverindex{ \beta }{ \text{\unboldmath$\mathit{J}$} }   \ottsym{.} \, \ottnt{D} \,  [   \algeffseqover{ \ottnt{C_{{\mathrm{0}}}} }   \ottsym{/}   \algeffseqover{ \alpha }   ]   \vdash   \ottnt{E}  [  \mathit{y}  ]   \ottsym{:}  \ottnt{B} \,  |  \, \epsilon'$,
         we have
         \[
         \Delta  \ottsym{,}  \mathit{y} \,  \mathord{:}  \, \ottnt{D} \,  [   \text{\unboldmath$\forall$}  \,  \algeffseqoverindex{ \beta }{ \text{\unboldmath$\mathit{J}$} }   \ottsym{.} \,  \algeffseqover{ \ottnt{C_{{\mathrm{0}}}} }   \ottsym{/}   \algeffseqover{ \alpha }   ]   \vdash   \ottnt{E}  [  \mathit{y}  ]   \ottsym{:}  \ottnt{B} \,  |  \, \epsilon'
         \]
         by \reflem{eff-var-subtype}.
         Thus, we have
         \[
         \Delta  \ottsym{,}  \mathit{y} \,  \mathord{:}  \, \ottnt{D} \,  [   \text{\unboldmath$\forall$}  \,  \algeffseqoverindex{ \beta }{ \text{\unboldmath$\mathit{J}$} }   \ottsym{.} \,  \algeffseqover{ \ottnt{C_{{\mathrm{0}}}} }   \ottsym{/}   \algeffseqover{ \alpha }   ]   \vdash  \mathsf{handle} \,  \ottnt{E}  [  \mathit{y}  ]  \, \mathsf{with} \, \ottnt{H}  \ottsym{:}  \ottnt{A} \,  |  \, \epsilon
         \]
         by \reflem{eff-weakening} (\ref{lem:eff-weakening:handler}) and \Te{Handle}.
         By \Te{Abs},
         \[
         \Delta  \vdash   \lambda\!  \, \mathit{y}  \ottsym{.}  \mathsf{handle} \,  \ottnt{E}  [  \mathit{y}  ]  \, \mathsf{with} \, \ottnt{H}  \ottsym{:}    \ottnt{D}    [   \text{\unboldmath$\forall$}  \,  \algeffseqoverindex{ \beta }{ \text{\unboldmath$\mathit{J}$} }   \ottsym{.} \,  \algeffseqover{ \ottnt{C_{{\mathrm{0}}}} }   \ottsym{/}   \algeffseqover{ \alpha }   ]     \rightarrow ^{ \epsilon }  \ottnt{A}  \,  |  \, \epsilon''
         \]
         for any $\epsilon''$.

         By applying \reflem{eff-term-subst} (\ref{lem:eff-term-subst:term}) to
         (\ref{eqn:eff-subject-red:handle:handle:two}), we have
         \[
         \Delta  \vdash    \ottnt{M}    [  \ottnt{v}  /  \mathit{x}  ]      [   \lambda\!  \, \mathit{y}  \ottsym{.}  \mathsf{handle} \,  \ottnt{E}  [  \mathit{y}  ]  \, \mathsf{with} \, \ottnt{H}  /  \mathit{k}  ]    \ottsym{:}  \ottnt{A} \,  |  \, \epsilon,
         \]
         which is what we have to show.

         \case \Te{Fix}/\R{Fix}:  By \reflem{eff-term-subst}.
          Note that the fixed-point combinator can be given any effect.
        \end{caseanalysis}

  \item Straightforward by induction on the typing derivation.
 \end{enumerate}
\end{proof}

\ifrestate
\thmEffTypeSoundness*
\else
\begin{thm}[Type Soundness]
 \label{thm:eff-type-sound}
 If $\Delta  \vdash  \ottnt{M}  \ottsym{:}  \ottnt{A} \,  |  \,  \emptyset $ and $\ottnt{M}  \longrightarrow^{*}  \ottnt{M'}$ and $\ottnt{M'}  \centernot\longrightarrow$, then
 $\ottnt{M'}$ is a value.
\end{thm}
\fi
\begin{proof}
 By Lemmas~\ref{lem:eff-subject-red} and \ref{lem:eff-progress}.
\end{proof}